\documentclass[11pt,reqno]{amsart}

\usepackage[letterpaper, portrait, margin=1in]{geometry}

\usepackage{amssymb}
\usepackage{amsmath}
\usepackage{amsthm}
\usepackage{amsfonts}
\usepackage{comment}
\usepackage{graphicx}
\usepackage{enumerate}
\usepackage[bottom,flushmargin,hang]{footmisc}
\usepackage{comment}
\usepackage{parskip}
\usepackage{bm}
\usepackage{mathrsfs}
\usepackage{bbm}
\usepackage{dsfont}
\usepackage[round,authoryear,longnamesfirst]{natbib}
\usepackage[bookmarks=true, hidelinks]{hyperref}
\PassOptionsToPackage{unicode}{hyperref}
\PassOptionsToPackage{naturalnames}{hyperref}
\hypersetup{
pdfauthor={Harold D. Chiang and Ahnaf Rafi},
pdftitle={Maximum Score and Non-Smooth M-Estimators with Multiway Dependence},
colorlinks=true,
linkcolor=blue,
urlcolor=blue,
citecolor=blue
}

\usepackage[nameinlink]{cleveref}
\usepackage{multicol}
\usepackage{tikz}
\usepackage{pgfplots}
\pgfplotsset{compat=1.17}
\definecolor{simblue}{HTML}{2A78D6}
\definecolor{simorange}{HTML}{EB6834}
\definecolor{simaqua}{HTML}{1BAF7A}
\usepackage{needspace}
\makeatletter
\renewcommand\paragraph{\@startsection{paragraph}{4}\z@{.5\linespacing\@plus.7\linespacing}{-.5em}{\normalfont\itshape}}
\makeatother

\numberwithin{equation}{section}
\newtheorem{theorem}{Theorem}

\newtheorem{corollary}{Corollary}

\newtheorem{proposition}{Proposition}

\newtheorem{lemma}{Lemma}
\theoremstyle{definition}
\newtheorem{assumption}{Assumption}
\newtheorem{example}{Example}
\newtheorem{remark}{Remark}

\DeclareMathOperator{\var}{\mathrm{Var}}
\DeclareMathOperator{\cov}{\mathrm{Cov}}

\renewcommand{\Pr}{\mathrm{Pr}}
\def\ind{\perp\!\!\!\perp}
\newcommand{\supp}{\mathrm{supp}}
\renewcommand{\tilde}{\widetilde}

\newcommand{\calB}[0]{\mathcal{B}}
\newcommand{\calC}[0]{\mathcal{C}}
\newcommand{\calD}[0]{\mathcal{D}}
\newcommand{\calF}[0]{\mathcal{F}}
\newcommand{\calH}[0]{\mathcal{H}}
\newcommand{\calI}[0]{\mathcal{I}}
\newcommand{\calK}[0]{\mathcal{K}}
\newcommand{\calM}[0]{\mathcal{M}}
\newcommand{\calY}[0]{\mathcal{Y}}
\newcommand{\calR}[0]{\mathcal{R}}
\newcommand{\calZ}[0]{\mathcal{Z}}
\newcommand{\R}[0]{\mathbb{R}}
\newcommand{\N}{\mathbb{N}}

\newcommand{\sphere}{\mathbbm{S}}
\newcommand{\sgn}{\mathrm{sgn}}
\newcommand{\pto}{\overset{\mathrm{p}}{\to}}
\newcommand{\ptoout}{\overset{\mathrm{p}\dagger}{\to}}
\newcommand{\dtop}{\overset{\mathrm{p}}{\leadsto}}
\def\d{\mathrm{d}}

\newcommand{\I}[0]{\mathbb{I}}

\crefname{assumption}{assumption}{assumptions}
\crefname{proposition}{proposition}{propositions}
\crefname{theorem}{Theorem}{Theorems}
\crefname{lemma}{Lemma}{Lemmas}
\crefname{corollary}{Corollary}{Corollaries}
\crefname{remark}{Remark}{Remarks}
\crefname{proposition}{Proposition}{Propositions}
\crefname{assumption}{Assumption}{Assumptions}

\DeclareMathOperator*{\argmax}{argmax}

\renewcommand{\ind}{\mathbbm 1}

\newcommand{\Op}{O_P}
\newcommand{\op}{o_P}
\renewcommand{\bar}{\overline}
\renewcommand{\i}{{\bm{i}}}
\newcommand{\En}{\hat{\mathbb{E}}_{n}}
\newcommand{\E}{\mathbb{E}}
\newcommand{\Eout}{\mathbb{E}^{\dagger}}
\newcommand{\Prout}{\Pr^{\dagger}}
\def\hatE{\widehat{\mathbb{E}}}
\newcommand{\e}{{\bm{e}}}
\newcommand{\calE}{\mathcal{E}}
\renewcommand{\hat}{\widehat}
\newcommand{\calN}{\mathcal{N}}

\usepackage{comment}
\allowdisplaybreaks

\newcommand{\dBL}{d_{\mathrm{BL}}}

\title{Maximum Score and Non-Smooth M-Estimators\\
with Multiway Dependence}

\author{Harold D. Chiang}
\address{Department of Economics, University of Wisconsin-Madison, 1180
Observatory Drive, Madison, WI 53706, USA.}
\email{hdchiang@wisc.edu}

\author{Ahnaf Rafi}
\address{Department of Economics, University of Virginia, 248 McCormick Road,
Monroe Hall Room 237, Charlottesville, VA 22903, USA.}
\email{ahnafrafi@virginia.edu}

\thanks{%
We thank Bryan Graham, Bruce Hansen, and Joel Horowitz for their
invaluable comments.
All remaining errors are ours.%
}

\begin{document}

\begin{abstract}
Under i.i.d.\ sampling, the maximum score estimator of
\citet{manski1975maximum} converges at the cube-root rate to a non-Gaussian
limit, and under median independence alone no regular estimator attains the
parametric rate \citep{chamberlain1986asymptotic}.
We show that, under multiway dependence and low-level conditions, the
estimator is asymptotically normal at a parametric rate in the smallest number of clusters
and, when the cluster-level variance is non-degenerate, as fast as regular
estimators such as probit.
The result rests on a general theory for non-smooth M-estimators, under which
a modified bootstrap is first-order valid
given a consistent Hessian estimate.
With two-way clustering a cube-root term can remain.
We then develop two-scale non-standard asymptotics, in which the limit is a
Gaussian vector plus an independent Chernoff-type vector and either term may
dominate, along sequences of laws from a non-degenerate cluster-level variance
to the degenerate case.
Under high-level conditions, the symmetric bootstrap intervals for suitable
scalar parameters control size uniformly over such classes of laws,
asymptotically exactly where the Gaussian term dominates and conservatively
otherwise.
\end{abstract}

\maketitle

\section{Introduction}

The maximum score estimator of \citet{manski1975maximum} estimates the slope
coefficients of a binary response model, up to scale, under a conditional
median restriction and without specifying the error distribution.
Under i.i.d.\ sampling it converges at the cube-root rate to a Chernoff-type
limit \citep{kim1990cube}, the non-parametric bootstrap is inconsistent
\citep{abrevaya2005bootstrap}, and valid inference can use subsampling
\citep{delgado2001subsampling} or the modified bootstrap of
\citet{cattaneo2020bootstrap}.

This paper gives the asymptotic distribution of maximum score, and more
generally of M-estimators with non-smooth criteria, when the data are
clustered along \(K\ge2\) dimensions, as is common in applied work
\citep{cameron2011robust,mackinnon2023cluster}.
With \(K=2\) such data form a bipartite network: two distinct sets of agents,
such as consumers and products or workers and firms
\citep{bonhomme2019distributional}, and an outcome for every pair, such as
whether the two are linked \citep{graham2020network}.
Each agent's cluster shock then plays the role of the unobserved agent-level
heterogeneity in models of network formation
\citep{graham2017econometric,gao2020nonparametric}.
Maximum score suits such binary link decisions because, like the link
formation model of \citet{gao2020nonparametric}, it leaves the distribution of
the pair-specific shock unspecified.
The mechanism is a smoothing induced by the data-generating process: the
covariates keep a density given the shock of a single cluster (in a network,
of a single agent), which smooths
the first-order projections of the criterion.
The estimator is then consistent (\Cref{lem--max-score-consistency}) and
asymptotically normal at rate \(n^{1/2}\), with \(n\) the smallest number of
clusters in any dimension and a non-degenerate limit when the cluster-level
variance \(\Sigma\) of the scores is positive definite
(\Cref{prp--max-score-asymp-normal}(a)).
Low-level conditions on the joint law of the outcome, the covariates, the
error and the cluster shocks verify the high-level conditions of a general
theorem for non-smooth M-estimators (\Cref{thm--m-est-non-smooth}).

Semiparametric flexibility therefore does not have a rate cost in this case.
Regular estimators --- sample means
\citep{2021daveziesEmpiricalProcessResults,menzel2021bootstrap} and smooth
parametric M-estimators such as probit, which \Cref{thm--m-est-non-smooth}
also covers --- converge at the same rate when their cluster-level variances
are non-degenerate, so maximum score, which restricts only the conditional
median of the error, loses nothing in rate.
Under i.i.d.\ sampling it is otherwise.
In \(N\) observations maximum score attains only \(N^{1/3}\), and the slower
rate reflects the restriction rather than the estimator.
By the profound insight of \citet{chamberlain1986asymptotic}, under median
independence alone no regular estimator attains \(N^{1/2}\), so a parametric
rate calls for more structure, such as the error distribution that probit
specifies.
Under multiway dependence the cluster shocks supply the missing smoothing, and
the median restriction alone suffices.

The general theory of \Cref{sec--m-est} may be of independent interest.
It rests on stochastic differentiability \citep{pollard1985new} and on local
bounds for higher-order projections, which for maximum score come from the
maximal inequality of \citet{chen2026cross}.
As the scores, and hence \(\Sigma\), involve derivatives of unknown
population objects, we use the modified bootstrap of
\citet{cattaneo2020bootstrap}, which replaces the sample criterion by a
quadratic centred at the estimate with an estimated Hessian.
Its weights are multiplier weights on the clusters of each dimension, as in
\citet{chiang2023inference}.
It reproduces the Gaussian limit (\Cref{thm--bootstrap}), so its symmetric
intervals are asymptotically exact when \(\Sigma\succ0\)
(\Cref{cor--boot-ci}).

With two-way clustering a second term competes with the Gaussian one, since no
cluster shock smooths the cell-level component of the criterion.
For maximum score it contributes a cube-root term of order
\((N_1N_2)^{-1/3}\), smaller
than the Gaussian term only by a factor of order \(n^{-1/6}\) for a balanced
array and fixed \(\Sigma\succ0\).
\(\Sigma\) can also be small or zero, for instance when the cluster shocks
--- in a network, the agents' unobserved heterogeneity --- move the covariates
of a pair but leave the conditional median of its error unchanged.
In that case the cube-root term is of the same or larger order, and when the row
and column scores vanish the estimator converges at the rate \((N_1N_2)^{1/3}\)
to the limit that \citet{kim1990cube} obtain for \(N_1N_2\) independent
observations (\Cref{prp--max-score-asymp-normal}(b)), as if the pairs of the
network were independent.

For \(K=2\), \Cref{sec--two-scale} keeps both terms, along sequences of
distributions on which \(\Sigma\) is fixed, local to zero or zero.
\Cref{thm--two-scale} approximates the law of the estimator by that of a
Gaussian vector plus an independent Chernoff-type vector, either of which may
dominate.
By \Cref{thm--two-scale-boot} the modified bootstrap reproduces the Gaussian
term and inflates the cube-root term by \(2^{1/3}\).
That term is inflated rather than mis-centred, so an inequality in the spirit
of \citet{anderson1955integral} makes symmetric intervals conservative rather
than invalid.
For a scalar parameter \(a'\theta_0\) whose estimation error is of the same
order as \(\hat\theta_n-\theta_0\), \Cref{thm--size} then gives coverage of at
least \(1-\alpha\) along every such sequence, and hence uniformly over
suitable classes of distributions \citep[cf.][]{andrews2020generic}.
The coverage is exact in the non-degenerate case and conservative in the
nearly degenerate one, in the terminology of \citet{chiang2025multiway}.

This rests on \Cref{lem--Z-law}: the law of the Chernoff-type vector is
equivalent to Lebesgue measure, so every linear functional of it has a
continuous and strictly increasing distribution function.
Continuity alone would not do: the critical value must solve its defining
equation uniquely, and \Cref{thm--size}(c) needs the limit law to charge an
interval.
\citet{cattaneo2026continuity} prove continuity for the argmax of a Gaussian
process by perturbing the mean through the Cameron--Martin theorem.
In contrast, we perturb by clipped linear functions, which move the maximiser by
a prescribed vector.

In the Monte Carlo study of \Cref{sec--simulations}, the coverage of the
modified bootstrap tracks the limit of \Cref{thm--size}(c) closely, and the
main error at small sample sizes is a first-order remainder larger than
\Cref{thm--two-scale} predicts.
\Cref{sec--plain-bootstrap} treats the plain bootstrap, first-order valid for
every \(K\) but not recommended, and three-way designs in which the cube-root
term is of smaller order.

Asymptotic theory under multiway dependence is developed by
\citet{2021daveziesEmpiricalProcessResults} and \citet{chiang2023inference},
among others, and \citet{graham2020network} surveys the econometrics of
dyadic and network data.
Separately exchangeable arrays are the bipartite counterpart of the
exchangeable random graphs that underlie nonparametric models of one-mode
networks \citep{BickelChenLevina2011,choi2014co}.
Non-standard asymptotics, in which the leading H\'ajek projection can be
small, arise already for sample means.
\citet{menzel2021bootstrap} proposes wild bootstrap schemes and
\citet{davezies2025analytic} analytic inference for means and GMM parameters,
both conservative under degeneracy.
The limits are normal in dyadic kernel density estimation
\citep{graham2024kernel} and bipartite logistic regression
\citep{graham2024sparse}.
Our asymptotics are dense, with every pair observed, whereas
\citet{graham2024sparse} lets the probability of a link vanish.
The same competition between a linear and a quadratic term, governed there by
a bandwidth rather than by the cluster sizes, drives the small bandwidth
asymptotics of \citet{CattaneoCrumpJansson2014small} for density-weighted
average derivatives.
Ours are more complex because the residual after the projections is
non-smooth, contributing a cube-root term and making the limit a convolution
of Gaussian and Chernoff-type laws.
\Cref{thm--size} nevertheless obtains the same pattern of asymptotically exact
and conservative coverage as \citet{davezies2025analytic}.

\subsection*{Notation}

Let \(\N\) denote the set of positive integers,
\(\R\) the real line and \(\mathbb C\) the complex numbers.
For \(a, b \in \R\), let \(a\vee b=\max\{a,b\}\) and \(a \wedge b = \min\{a,
b\}\).
For \(K \in \N\) and vectors \(\bm{a}, \bm{b} \in \mathbb{R}^{K}\), write
\(\bm{a} \le \bm{b}\) for coordinate-wise inequality.
The Hadamard product is denoted by \(\odot\), i.e.\ \(\bm{a} \odot \bm{b} =
(a_{1} \cdot b_{1}, \dots, a_{K} \cdot b_{K})'\).
A prime denotes transposition, and \(\I_d\) is the \(d\times d\) identity
matrix.
\(\|\cdot\|\) is the Euclidean norm, and the spectral norm for matrices,
unless otherwise specified; \(\operatorname{tr}\) denotes the trace.
\(\nabla\) and \(\nabla^2\) denote the gradient and the Hessian, and \(Dg\) the
derivative (Jacobian matrix) of a map \(g\).
For symmetric matrices, \(\lambda_{\min}\) and \(\lambda_{\max}\) denote the
smallest and largest eigenvalues, and \(A\succ B\) (\(A\succeq B\)) means
that \(A-B\) is positive definite (semi-definite).
\(\ind\{\cdot\}\) is the indicator function, and, for \(u\in\R\),
\(\sgn(u):=\ind\{u\ge0\}-\ind\{u<0\}\), so that \(\sgn(0)=1\).
For positive \(a\) and \(b\), \(a\lesssim b\) means \(a\le Cb\)
for a constant \(C\in(0,\infty)\) that depends neither on \(n\)
nor on the arguments, and \(a\asymp b\) means \(a\lesssim b\lesssim a\);
for positive sequences, \(a_n\sim b_n\) means \(a_n/b_n\to1\).

The normal law with mean \(\mu\) and covariance matrix \(V\) is \(N(\mu, V)\),
and \(\overset{d}{=}\) denotes equality in law.
Weak convergence of probability distributions is denoted
by \(\leadsto\), and \(\pto\) denotes convergence in probability.
For a metric space \((\mathbb D,d_{\mathbb D})\), \(\mathrm{BL}_1(\mathbb D)\)
is the set of \(h:\mathbb D\to[-1,1]\) with \(|h(x)-h(y)|\le d_{\mathbb D} (x,
y)\), and the bounded Lipschitz distance between laws \(\mu\) and \(\nu\) on
\(\mathbb D\) is \(\dBL(\mu,\nu):=\sup_{h\in\mathrm{BL}_1(\mathbb D)} |\int h \,
\d \mu-\int h\,\d\nu|\); product spaces carry the sum of the metrics.
For random vectors \(T_n\) with laws \(\mathcal L(T_n)\) and a law \(Q\),
\(T_n\leadsto Q\) if and only if \(\dBL(\mathcal L(T_n),Q)\to0\)
\citep[Theorem~1.12.4]{2023vandervaartWeakConvergenceEmpirical}; for a random
vector \(T\), \(T_n\leadsto T\) means \(T_n\leadsto\mathcal L(T)\).
For \(\sigma\)-algebras \(\mathscr B_n\), \(T_n\dtop Q\) conditionally given
\(\mathscr B_n\) means that \(\dBL(\mathcal L(T_n\mid\mathscr B_n),Q)\pto0\),
where \(\mathcal L(T_n\mid\mathscr B_n)\) is the conditional law of \(T_n\)
given \(\mathscr B_n\)
\citep[cf.][Section~1.13]{2023vandervaartWeakConvergenceEmpirical}.
We write \(\leadsto\) and \(\dtop\) for convergence to a fixed law, and
\(\dBL\) where the approximating law depends on \(n\), as in
\Cref{thm--two-scale,thm--two-scale-boot}, or where the distance itself is
bounded.
Conditioning on a random element, such as the data, means conditioning on
the \(\sigma\)-algebra it generates.
For a set \(T\), \(\ell^\infty(T)\) is the space of bounded functions
\(z:T\to\R\), normed by \(\|z\|_T:=\sup_{t\in T}|z(t)|\), and, for a law \(Q\),
\(Qf:=\int f\,\d Q\).
For possibly non-measurable maps, \(\Eout\) and \(\Prout\) denote outer
expectation and outer probability, and \(\op^{\dagger}\), \(\Op^{\dagger}\) and
\(\ptoout\) the corresponding order symbols and convergence in outer
probability \citep[Section~1.2 and
Definition~1.9.1]{2023vandervaartWeakConvergenceEmpirical}.

\section{Non-Smooth M-Estimation under Multiway Dependence}
\label{sec--m-est}

Under multiway clustering, data are indexed by \(K\)-vectors of natural numbers,
\(\i=(i_{1}, \dots, i_{K}) \in \N^{K}\).
For each \(n \in \N\), let \(N_{n,1}, \dots, N_{n,K} \in \N\)
denote sample sizes along cluster dimensions.
We set \(n = \min \left\{ N_{n,1}, \dots, N_{n,K} \right\}\) and let
\(n \to \infty\), with \(K \geq 2\) held fixed throughout.
Let \(\calI_{n} = \prod_{k = 1}^{K} \left\{ 1, \dots, N_{n,k}
\right\}\) denote the set of \(K\)-dimensional indices.
Observed data are \(\bm{W}_{n} = \left\{ W_{\i} : \i \in \calI_{n} \right\}\).
We impose the following standard assumption for multiway dependence.

\begin{assumption}
\label{asm--exchangeable}
\(\{W_{\i}\}_{\i \in \N^{K}}\) are separately exchangeable and dissociated.
Equivalently, the array admits the Aldous--Hoover--Kallenberg
representation
\citep{1981aldousRepresentationsPartiallyExchangeable,%
kallenberg2005probabilistic}: there is a Borel measurable \(\tau (\cdot)\) and
mutually independent latent \(\mathrm{Uniform} [0, 1]\) variables
\(\bm{U}_{\i} = \{U_{\i \odot \e} : \e \in \{0, 1\}^{K} \setminus
\{\bm{0}\}\}\) for \(\i \in \N^{K}\) such that \(\{W_{\i}\}_{\i \in \N^{K}}
\overset{d}{=} \{\tau \left( \bm{U}_{\i} \right)\}_{\i \in \N^{K}}\).
Without loss of generality \(W_{\i}=\tau(\bm U_{\i})\) for every
\(\i\in\N^K\), after enlarging the probability space if necessary: this
changes neither the law of \(\bm W_n\) nor the conditional law of any
bootstrap statistic given \(\bm W_n\). In particular,
\(\sigma(\bm W_n)\subseteq\sigma(\bm U_{\i}:\i\in\calI_n)\).
\end{assumption}

When \(K = 1\), the multiway sampling framework reduces to the standard i.i.d.\
setting.
When \(K = 2\), the Aldous--Hoover--Kallenberg representation reads
\(W_{(i_1,i_2)}\overset{d}{=}\tau(U_{(i_1,0)},U_{(0,i_2)},U_{(i_1,i_2)})\),
which corresponds to two-way clustering or bipartite networks.
In the bipartite case the two indices run over two distinct sets of vertices
--- consumers and products in \Cref{ex--ms-two-way} below, workers and firms,
authors and journals.
The observation \(W_{(i_1,i_2)}\) is an attribute of the edge between them: whether the
pair is linked, or an outcome recorded when it is.
The representation is then a latent-variable model of link formation, with
\(U_{(i_1,0)}\) and \(U_{(0,i_2)}\) the unobserved heterogeneity of the two
vertices and \(U_{(i_1,i_2)}\) the shock specific to the pair, and separate
exchangeability is the bipartite counterpart of the joint exchangeability
imposed on one-mode networks.
Our asymptotics are dense: all \(N_{n,1}N_{n,2}\) pairs are observed, and the
probability of a link does not vanish.
Degeneracy in the sense of \Cref{sec--two-scale} is the statement that
the vertex-level shocks do not move the criterion to first order, so that the
sampling variation of \(\hat\theta_n\) comes from the pairs alone.
This is exactly where the cube-root term dominates.
Under \Cref{asm--exchangeable}, the \(W_{\i}\) are typically not mutually
independent, but they are identically distributed.
Let \(W\) denote a random vector distributed according to their common
marginal distribution \(P_{0}\).

\begin{remark}[Examples of multiway dependence]
Multiway dependence arises whenever each observation is indexed by several
sampling dimensions and observations sharing any one of them may be
dependent, through unobserved shocks specific to a single dimension.
Examples include matched employer--employee data, correlated within workers and
within firms; trade or migration flows indexed by origin and destination;
bilateral financial exposures indexed by lender and borrower; education data
indexed by students and teachers or schools; and demand data indexed by
products, markets and firms, where product characteristics, market-level demand
shocks and firm-level heterogeneity each generate dependence.
Multiway clustering and related exchangeable-array frameworks have therefore
become standard tools for modelling such dependence; see, for example,
\citet{cameron2011robust}, \citet{2021daveziesEmpiricalProcessResults},
\citet{menzel2021bootstrap} and \citet{mackinnon2021wild}.
\(\lozenge\)
\end{remark}

\subsection{Sampling and the Hoeffding-type decomposition}
\label{sec--hoeffding}

The assumptions and the proofs throughout rest on the Hoeffding-type
decomposition
\citep{hoeffding1948class,chiang2023inference} for multiway dependence, which
we introduce here together with the averaging notation.
For a measurable function \(f\) we write
\(\En[f]:=\En[f(\bm U_{\i})]:=|\calI_n|^{-1}\sum_{\i\in\calI_n}f(\bm U_{\i})\)
for sample averaging.
As \(W_{\i}=\tau(\bm U_{\i})\) (\Cref{asm--exchangeable}), a function \(g\) of
the data is identified with the function \(g\circ\tau\) of \(\bm U_{\i}\), so
that, for instance, \(\En[g]=\En[g(W_{\i})]\).

Let \(\bm{\iota}\) denote the \(K\)-vector whose components are all equal to
one.
For \(k \in \{0, \dots, K\}\), define \(\calE_{k} = \{\e \in \{0,1\}^{K} :
\bm{\iota}' \e = k\}\).
Then \(\calE_{1} = \{\e_{1}, \dots, \e_{K}\}\) where \(\e_{k}\) is the
\(k\)th standard unit vector.
Denote \(\calE = \cup_{k=1}^{K} \calE_{k}\), so that
\(\calE = \{0,1\}^{K} \setminus \{\bm{0}\}\).
For \(\i \in \N^{K}\) and \(\e \in \calE\), put
\(\bm U_{\i\odot\e}=(U_{\i\odot\e'}:\e'\in\calE,\ \e'\le\e)\) and, for an
integrable \(f\), define the centred conditional expectation
\(P_{\i\odot\e}[f]:=[P_{\e}f](\bm U_{\i\odot\e}):=\E[f(\bm U_{\i})\mid
\bm U_{\i\odot\e}]-\E[f(\bm U_{\i})]\).
Define orthogonal projections recursively:
\begin{equation*}
  \begin{gathered}
    \pi_{\i\odot\e_{k}} [f] :=
    \left[ \pi_{\e_{k}} f \right] (\bm{U}_{\i\odot\e_{k}}) := \left[ P_{\e_{k}}
    f \right] (\bm{U}_{\i\odot\e_{k}}) := P_{\i\odot\e_{k}} [f] \quad \text{for
    } k \in \{1, \dots, K\}, \\
    \text{and for } \e \in \cup_{k=2}^{K} \calE_{k}, \quad
    \left[ \pi_{\e} f \right] (\bm{U}_{\i\odot\e}) = \left[ P_{\e} f \right]
    (\bm{U}_{\i\odot\e}) - \sum_{\substack{\e' \in \calE, \e' \leq \e, \\ \e'
    \neq \e}} \left[ \pi_{\e'} f \right] (\bm{U}_{\i\odot\e'}).
  \end{gathered}
\end{equation*}
By the same convention,
\(\pi_{\i\odot\e}[f]=P_{\i\odot\e}[f]-\sum_{\e'\in\calE,\,\e'\le\e,\,\e'\ne\e}
\pi_{\i\odot\e'}[f]\).
Define \(\calI_{n,\e} = \{\i \odot \e : \i \in \calI_{n}\}\), so that
\(|\calI_{n,\e}| = \prod_{k \in \supp (\e)} N_{n,k}\), where \(\supp(\e)\) is the
set of \(k\) for which the \(k\)th coordinate of \(\e\) is one.
The Hoeffding-type projection is
\(H_{n,\e}[f]=|\calI_{n,\e}|^{-1}\sum_{\i\in\calI_{n,\e}}\pi_{\i\odot\e}[f]\),
and the resulting decomposition is
\begin{equation}
  \hatE_{n} [f] = \E [f] + \sum_{k = 1}^{K} \sum_{\e \in \calE_{k}} H_{n,\e} [f].
  \label{eqn--gen-hoeffding-decomp}
\end{equation}

\begin{example}[Maximum score with two-way clustering]
\label{ex--ms-two-way}
Let \(K=2\), with \(i_1\) indexing \(N_{n,1}\) consumers and \(i_2\) indexing
\(N_{n,2}\) products, so that \(Y_{(i_1,i_2)}\in\{-1,1\}\) records whether
consumer \(i_1\) buys product \(i_2\) and \(X_{(i_1,i_2)}\in\R^{d}\) collects
characteristics of the pair, such as the price the consumer faces, generated by
the latent-utility model of \Cref{asm--ms-consistency}.
The criterion is \(m_\theta:=m_{\beta(\theta)}\) with
\(m_\beta(y,x)=\ind\{y\cdot x'\beta\ge0\}\) as in
\eqref{eqn--max-score-estimator-def} and the chart \eqref{eqn--beta-local}, so
that \(\theta\in\R^{d-1}\) and \(\theta_0=0\), \(d\) being the dimension of
\(X\).
By \Cref{asm--exchangeable},
\(W_{(i_1,i_2)}=(Y_{(i_1,i_2)},X_{(i_1,i_2)}')'=\tau(U_{(i_1,0)},U_{(0,i_2)},U_{(i_1,i_2)})\)
for independent shocks of the consumer \(U_{(i_1,0)}\) (e.g.\ tastes), of the
product \(U_{(0,i_2)}\) (e.g.\ unobserved quality) and of the pair
\(U_{(i_1,i_2)}\), so observations that share a consumer or a product may be
dependent; \Cref{fig--array} depicts this structure.
Here \(\calE_1=\{(1,0),(0,1)\}\), \(\calE_2=\{(1,1)\}\), and
\eqref{eqn--gen-hoeffding-decomp} with \(f=m_\theta\) reads
\[
  \hat M_n(\theta)=M_0(\theta)+H_{n,(1,0)}[m_\theta]+H_{n,(0,1)}[m_\theta]
  +H_{n,(1,1)}[m_\theta],
\]
where \(H_{n,(1,0)}[m_\theta]=N_{n,1}^{-1}\sum_{i_1=1}^{N_{n,1}}
\{\E[m_\theta(W_{(i_1,i_2)})\mid U_{(i_1,0)}]-M_0(\theta)\}\) does not depend
on \(i_2\) and averages the consumer-level components,
\(H_{n,(0,1)}[m_\theta]\) averages the product-level ones, and
\(H_{n,(1,1)}[m_\theta]\) the remaining consumer--product components over all
\(N_{n,1}N_{n,2}\) pairs.
\(\lozenge\)
\end{example}

\begin{figure}[htbp]
\centering
\color{black}
\begin{tikzpicture}[scale=0.42]
\fill[simblue!15] (0,3) rectangle (7,4);
\fill[simblue!15] (4,0) rectangle (5,7);
\fill[simorange] (4,3) rectangle (5,4);
\draw[gray] (0,0) grid (7,7);
\node[rotate=90] at (-0.8,3.5) {\footnotesize consumers \(i_1\)};
\node at (3.5,7.7) {\footnotesize products \(i_2\)};
\end{tikzpicture}
\caption{Dependence in a two-way array. The observation for the pair
\((i_1,i_2)\) (dark) shares the consumer shock \(U_{(i_1,0)}\) with every pair in
its row and the product shock \(U_{(0,i_2)}\) with every pair in its column
(shaded), so it may depend on all of them; it is independent of every pair
that shares neither its consumer nor its product.}
\label{fig--array}
\end{figure}

\subsection{The estimator and the bootstrap}
\label{sec--bootstrap}

We develop a general asymptotic theory for non-smooth
M-estimators under multiway dependence.
Let \(\calM = \{m_{\theta} : \theta \in \Theta\} \subseteq L_{1} (P_{0})\) be a
collection of integrable real-valued functions of the data \(W_{\i}\), indexed
by \(\Theta \subseteq \R^{d}\).
An M-estimator \(\hat{\theta}_{n}\) satisfies
\begin{equation}
  \begin{gathered}
    \hat{\theta}_{n} \in \Theta \text{ is measurable} \quad \text{and} \quad
    \hat{M}_{n}
    (\hat{\theta}_{n}) \geq \sup_{\theta \in \Theta} \hat{M}_{n} (\theta) -
    \hat{\zeta}_{n}, \\
    \text{for} \quad \hat{M}_{n} (\theta) := \En [m_{\theta}
    (W_{\i})], \quad \hat{\zeta}_{n} \text{ measurable,} \quad \text{and} \quad
    0 \leq \hat{\zeta}_{n} = \op (1/n).
  \end{gathered}
  \label{eqn--m-est-def}
\end{equation}
\Cref{eqn--m-est-def} defines an approximate M-estimator, which optimises
the sample criterion \(\hat{M}_{n}\) up to a non-negative tolerance
\(\hat{\zeta}_{n}\) set by the econometrician.
An exact M-estimator requires \(\hat{\zeta}_{n} \equiv 0\).
For \Cref{thm--m-est-non-smooth,thm--bootstrap} below, it is sufficient to have
\(\hat{\zeta}_{n} = \op (1/n)\), whereas the non-standard asymptotics of
\Cref{sec--two-scale} need the smaller tolerance of
\Cref{asm--two-scale}\eqref{asm--two-scale-rates}.
Here \(\hat\theta_n\) is a measurable function of the data and, possibly,
of auxiliary randomisation, such as random tie-breaking, independent of the
latent variables and of the bootstrap weights introduced below.
Denote the population counterpart to \(\hat{M}_{n}\) by \(M_{0} (\theta) := \E
[m_{\theta} (W)]\), and the parameter of interest by \(\theta_0\in\Theta\).

Inference requires the asymptotic variance of \(\hat\theta_{n}\), which
\Cref{thm--m-est-non-smooth} below expresses through scores \(\Delta_{k}\)
involving derivatives of unknown population objects.
Because \(\theta \mapsto m_{\theta}\) is not smooth, these generally have no
obvious sample counterparts, and the curvature \(A_{0}\) has to be estimated
as well.
We therefore conduct inference by a bootstrap, which the rest of this
subsection defines.

Let \(\{\xi_{n,\i} : \i \in \calI_{n}\}\) denote bootstrap weights.
We extend the averaging notation of \Cref{sec--hoeffding} to the weights,
writing \(\En[f]:=|\calI_n|^{-1}\sum_{\i\in\calI_n}f(\bm U_{\i},\xi_{n,\i})\)
and, for \(g\) with \(f(\bm U,\xi)=\xi g(\bm U)\),
\(\En^{*}[g]:=\En[\xi_{n,\i}\,g(\bm U_{\i})]
=|\calI_n|^{-1}\sum_{\i\in\calI_n}\xi_{n,\i}\,g(\bm U_{\i})\); as before,
\(\En[g(W_{\i},\xi_{n,\i})]:=\En[f]\) for
\(f(\cdot,\cdot)=g(\tau(\cdot),\cdot)\).
We consider two constructions of weights, both targeting the first-order
projections \(\pi_{\e}\) with \(\e \in \calE_{1}\).

\begin{assumption}
\label{asm--boot-wgt}
Let \(\{\xi_{k,i} : k \in \{1, \dots, K\}, i \in \N\}\) be scalar random
variables satisfying:
\begin{equation}
  \begin{gathered}
    \{\xi_{1,i} : i \in \N\}, \dots, \{\xi_{K,i} : i \in \N\}, \{\bm{U}_{\i}
    : \i \in \N^{K}\} \quad \text{are mutually independent}, \\
    \text{for each} \quad k \in \{1, \dots, K\}, \quad
    \xi_{k,i} \text{ are i.i.d.\ across } i \in \N.
    \label{eqn--boot-wgt-indep}
  \end{gathered}
\end{equation}
Denote \(\bar{\xi}_{n,k} = N_{n,k}^{-1} \sum_{i_{k}=1}^{N_{n,k}}
\xi_{k,i_{k}}\).
The bootstrap weights \(\xi_{n,\i}\) are formed from one of the following for
every \(n \in \N\).
\begin{enumerate}[(i)]
  \item Pigeonhole: \(\xi_{n,\i} = \prod_{k = 1}^{K} (\xi_{k,i_{k}} /
    \bar{\xi}_{n,k})\), for non-negative \(\xi_{k,i}\) with bounded support and
    \(\E [\xi_{k,1}] = \var [\xi_{k,1}] = 1\).
    If \(\bar{\xi}_{n,k} = 0\), set every normalised weight in dimension \(k\)
    to one.
  \item Multiplier: \(\xi_{n,\i} = 1 + \sum_{k = 1}^{K} (\xi_{k,i_{k}}
    - \bar{\xi}_{n,k})\), for \(\xi_{k,i}\) with bounded support, \(\E
    [\xi_{k,1}] = 0\) and \(\var [\xi_{k,1}] = 1\).
\end{enumerate}
\end{assumption}

The pigeonhole restrictions hold, for instance, for the following transformed
uniform variate:
\begin{equation*}
  \text{for} \quad \xi_{k,i}^{*} \overset{\mathrm{i.i.d.}}{\sim}
  \mathrm{Uniform} [0, 1], \quad
  \xi_{k,i} = \begin{cases}
    \frac{4}{3} \xi^{*}_{k,i}, & \text{if } 0 \leq \xi^{*}_{k,i} < \frac{3}{4},
    \\
    12 \xi^{*}_{k,i} - 8, & \text{if } \frac{3}{4} \leq \xi^{*}_{k,i} \leq 1.
  \end{cases}
\end{equation*}
The resulting \(\xi_{k,i}\) has a \((3/4) \cdot \mathrm{Uniform} [0, 1] +
(1/4) \cdot \mathrm{Uniform} [1, 4]\) mixture distribution.
For the multiplier bootstrap, one can instead use \(\xi_{k,i} \sim \mathrm{Uniform}
[-\sqrt{3}, \sqrt{3}]\).
The pigeonhole bootstrap is closely related to the bootstraps of
\citet{owen2007pigeonhole}, \citet{2012owenBootstrappingDataArrays} and
\citet{2021daveziesEmpiricalProcessResults}, which use multinomial weights,
whereas we use i.i.d.\ bounded weights for each clustering dimension \(k\), as
in the multiplier bootstraps of \citet{chiang2023inference}.
Bounded support lets us use the symmetrisation and bounded multiplier
inequalities for empirical processes in
\citet{2023vandervaartWeakConvergenceEmpirical}.
We recommend multiplier weights, for which
\Cref{thm--two-scale-boot,thm--size} below are proved.

Write \(\hat M^*_n(\theta):=\En^{*}[m_\theta(W_{\i})]\) for the reweighted
criterion, whose maximiser is the plain bootstrap estimator of
\Cref{sec--plain-bootstrap}.
Following \citet{cattaneo2020bootstrap}, the modified bootstrap instead
replaces \(\hat M_n\), the bootstrap-world analogue of \(M_0\), by an exact
quadratic centred at \(\hat\theta_n\):
\begin{equation}
  \begin{gathered}
    \tilde{\theta}_{n}^{*} \in \Theta \text{ is measurable and }
    \tilde{M}_{n}^{*} (\tilde{\theta}_{n}^{*}) \geq \sup_{\theta \in \Theta}
    \tilde{M}_{n}^{*} (\theta) - \tilde{\zeta}_{n}^{*}, \\
    \text{where} \quad
    \tilde{M}_{n}^{*} (\theta) = \hat{M}_{n}^{*} (\theta) - \hat{M}_{n} (\theta)
    - \frac{1}{2} (\theta - \hat{\theta}_{n})' \tilde{A}_{n} (\theta -
    \hat{\theta}_{n}), \\
    \tilde{A}_{n} \text{ measurable and symmetric}, \quad
    \tilde{\zeta}_{n}^{*} \text{ measurable, and } 0 \leq \tilde{\zeta}_{n}^{*}
    = \op (1/n).
  \end{gathered}
  \label{eqn--m-est-boot-def}
\end{equation}
If \(\sup_{\theta\in\Theta}\tilde M^*_n(\theta)=\infty\), set
\(\tilde\theta^*_n:=\hat\theta_n\).
Since \(|\hat M^*_n(\theta)-\hat M_n(\theta)|\le\max_{\i\in\calI_n}|\xi_{n,\i}-1|\,
\En[m_*(W_{\i})]\) for all \(\theta\in\Theta\), with \(m_*\) the envelope in
\Cref{asm--m-est-general} below, this can happen only if \(\tilde A_n\) is not
positive semi-definite, an event whose probability tends to zero when
\(\tilde A_n\pto A_0\), as \Cref{thm--bootstrap} below assumes.
The convention therefore does not affect the asymptotic results below.

\subsection{Asymptotic normality and bootstrap validity}

We now give conditions under which \(\hat\theta_{n}\) is asymptotically normal
and the bootstrap \eqref{eqn--m-est-boot-def} reproduces its limit, beginning
with basic consistency and quadratic decay assumptions.

\begin{assumption}
\label{asm--m-est-consistent-pop-quad-decay}
\hfill
\begin{enumerate}[(i)]
  \item \label{asm--m-est-consistent}
    \(\hat{\theta}_{n} \pto \theta_{0}\);
  \item \label{asm--m-est-smooth-pop}
    \(\theta_{0} \in \mathrm{int} (\Theta)\) and for a non-stochastic real
    symmetric positive definite \(d \times d\) matrix \(A_{0}\),
    \begin{equation*}
      M_{0} (\theta) = M_{0} (\theta_{0}) - \frac{1}{2} (\theta - \theta_{0})'
      A_{0} (\theta - \theta_{0}) + o (\|\theta - \theta_{0}\|^{2}) \quad
      \text{as} \quad \theta \to \theta_{0}.
    \end{equation*}
\end{enumerate}
\end{assumption}

Condition \eqref{asm--m-est-smooth-pop} of
\Cref{asm--m-est-consistent-pop-quad-decay} is ensured if \(M_{0}\) is twice
differentiable with a zero gradient and a negative definite Hessian at
\(\theta_{0}\); then \(A_0=-\nabla^2M_0(\theta_0)\).
These typically follow from standard identification assumptions.
Condition \eqref{asm--m-est-consistent} requires consistency of
\(\hat{\theta}_{n}\), which can follow from
uniform convergence of \(\hat{M}_{n} - M_{0}\) over \(\Theta\) --- see, for
example, \citet[Section~2]{newey1994large} or
\citet[Corollary~3.2.3, p.~395]{2023vandervaartWeakConvergenceEmpirical}.

The asymptotic theory of general non-smooth M-estimators rests on the
following conditions.

\needspace{8\baselineskip}
\begin{assumption}[Stochastic differentiability]
\label{asm--m-est-general}
\hfill
\begin{enumerate}[(i)]
  \item \label{asm--m-est-stoch-diff}
    (Stochastic differentiability)
    For each unit vector \(\e_{k} \in \calE_{1}\), there exists a measurable
    \(\R^{d}\)-valued function \(\Delta_{k} = \Delta_{k} (U_{\e_{k}})\)
    satisfying \(\E [\Delta_k] = 0\), \(\E [\|\Delta_{k}\|^{2}] < \infty\)
    and such that for any non-stochastic positive real sequence
    \(\{\eta_{n}\}\),
    \begin{equation}
      \begin{gathered}
        \text{if} \quad \eta_{n} \to 0, \quad \text{then} \quad
        \Eout \left[ \sup_{\substack{\theta \in \Theta : \\ 0 < \|\theta -
        \theta_{0}\| \leq \eta_{n}}} \frac{|(\hatE_{n,k} - \E) [r_{k,
        \theta}]|}{\|\theta - \theta_{0}\| (\|\theta - \theta_{0}\| + n^{-1/2})}
        \right] = o (1), \\
        \text{where} \quad r_{k, \theta} = \pi_{\e_{k}} m_{\theta} -
        \pi_{\e_{k}} m_{\theta_{0}} - (\theta - \theta_{0})' \Delta_{k},
      \end{gathered}
      \label{eqn--m-est-stoch-diff}
    \end{equation}
    and \(\hatE_{n,k} [f] = N_{n,k}^{-1} \sum_{i_{k} = 1}^{N_{n,k}} f (U_{i_{k}
    \cdot \e_{k}})\)
    is the sample average over indices of the \(k\)th cluster
    dimension \(i_{k} = 1, \dots, N_{n,k}\).

  \item \label{asm--m-est-stoch-eq}
    The class of functions \(\calM\) has envelope \(m_{*}\) with \(\E
    [m_{*} (W)^{2}] < \infty\) and is of VC-type (defined below) with characteristics \(A \ge
    (\mathrm e^{2 (K - 1)} / 16) \vee \mathrm e\) and \(v \ge 1\).
    In addition, for \(\e \in \calE_{k}\) with \(k \geq 2\) and every
non-stochastic positive real sequence \(\{\eta_{n}\}\), it holds that
    \begin{equation}
      \begin{gathered}
        \text{if } \eta_{n} \to 0, \text{ then} \quad \Eout \left[ \sup_{\theta
        \in \Theta: \|\theta - \theta_{0}\| \leq \eta_{n}} \left| H_{n,\e}
        [m_{\theta} - m_{\theta_{0}}] \right| \right] = O (1/n), \quad
        \text{and} \\
        \text{if } \eta_{n} = O (n^{-1/2}), \text{ then} \quad \Eout \left[
        \sup_{\theta \in \Theta: \|\theta - \theta_{0}\| \leq \eta_{n}}
        \left|H_{n,\e} [m_{\theta} - m_{\theta_{0}}] \right| \right] = o (1/n).
      \end{gathered}
      \label{eqn--m-est-stoch-eq}
    \end{equation}
  \item \label{asm--m-est-meas}
    (Measurability) \(\Theta\) is a Borel subset of \(\R^d\) and
    \((\theta,w)\mapsto m_\theta(w)\) is jointly Borel measurable.
\end{enumerate}
\end{assumption}

Condition \eqref{asm--m-est-meas} is the measurability condition of
\citet[Section~4]{kim1990cube}; \Cref{thm--m-est-non-smooth}, whose proof uses
outer expectations, does not need it.
Under it, every process below is jointly measurable, and its supremum over a
Borel set of \(\theta\) universally measurable, which is what the
symmetrisation and Fubini arguments of the proofs require.
\Cref{sec--measurability-conventions} gives the details and verifies the
condition for maximum score.
A class \(\calF\) of measurable functions with measurable envelope \(F\) is
VC-type with characteristics \((A,v)\) if
\(\sup_QN(\varepsilon\|F\|_{Q,2},\calF,L_2(Q))\le(A/\varepsilon)^v\) for all
\(0<\varepsilon\le1\), where the covering number \(N(\varepsilon,\calF,L_2(Q))\) is
the smallest number of \(L_2(Q)\)-balls of radius \(\varepsilon\) that cover
\(\calF\), and the supremum is over finitely discrete probability
measures \(Q\) with \(\|F\|_{Q,2}>0\).
Here \(\|f\|_{Q,2}:=(\int f^2\,\d Q)^{1/2}\); similarly,
\(\|f\|_{P_0,2}:=(\E[f(W)^2])^{1/2}\), and \(\|f\|_{P_n,2}\) is defined in the
same way when the array has law \(P_n\) as in \Cref{sec--two-scale}.

\begin{theorem}[Non-smooth M-estimation]
\label{thm--m-est-non-smooth}
Suppose \Cref{asm--exchangeable,asm--m-est-consistent-pop-quad-decay,%
asm--m-est-general} hold, and \(n / N_{n,k} \to \lambda_{k} \in [0, \infty)\)
for each \(k=1, \dots, K\).
Then
\begin{equation*}
  n^{1/2} (\hat{\theta}_{n} - \theta_{0}) = A_{0}^{-1} n^{1/2}
  \bar{\psi}_{n} + \op (1) \leadsto N(0, V),
\end{equation*}
where \(\bar{\psi}_{n} = \sum_{k = 1}^{K} \hatE_{n,k} [\Delta_{k}]\),
\(V = A_{0}^{-1} \Sigma A_{0}^{-1}\), and \(\Sigma = \sum_{k=1}^{K} \lambda_{k}
\E [\Delta_{k} \Delta_{k}']\).
\end{theorem}

\Cref{thm--m-est-non-smooth} is proven in
\Cref{sec--proof--thm--m-est-non-smooth}.
\Cref{asm--m-est-general}\eqref{asm--m-est-stoch-diff} imposes a mild
stochastic differentiability condition on the first-order projections.
Importantly, it accommodates a broad class of non-smooth criteria.
See \citet{pollard1985new} for details and sufficient conditions.
Furthermore, condition \eqref{asm--m-est-stoch-eq} imposes mild restrictions
on the complexity of the class \(\calM\).
These complexity restrictions are used to control local maximal deviations of
higher-order terms in the Hoeffding-type decomposition.

Condition \eqref{eqn--m-est-stoch-eq} follows from primitive conditions: the
VC-type condition and the maximal inequality of
\citet[Theorem~4]{chen2026cross} give it for \(k\ge3\), and for \(k=2\) it
holds whenever the localised classes
\(\{m_\theta-m_{\theta_0}:\|\theta-\theta_0\|\le\eta\}\) are VC-type with
characteristics free of \(\eta\) and have envelopes whose
\(\|\cdot\|_{P_0,2}\) norms tend to zero, as they do for maximum score; see
\Cref{sec--measurability-conventions}.

\begin{theorem}[Bootstrapping M-estimators]
\label{thm--bootstrap}
Suppose the conditions of \Cref{thm--m-est-non-smooth} and
\Cref{asm--boot-wgt} hold, and \(\tilde A_n\pto A_0\).
Then \(n^{1/2}(\tilde\theta^*_n-\hat\theta_n)\dtop N(0,V)\) conditionally given
\(\bm W_n\).
\end{theorem}

\begin{corollary}[Bootstrap confidence intervals]
\label{cor--boot-ci}
Suppose the conditions of \Cref{thm--bootstrap} hold with \(\Sigma\succ0\), and
let \(g:\Theta\to\R\) be Borel measurable and continuously differentiable near
\(\theta_0\) with \(\nabla g(\theta_0)\neq0\).
For \(\alpha\in(0,1)\), let \(\hat q_n\) be a \((1-\alpha)\)-quantile of the
conditional law of \(|g(\tilde\theta^*_n)-g(\hat\theta_n)|\) given \(\bm W_n\).
Then \(\Pr\{g(\theta_0)\in[g(\hat\theta_n)-\hat q_n,\,g(\hat\theta_n)+\hat
q_n]\}\to1-\alpha\), for the outer and the inner probability if \(\hat q_n\) is
not measurable.
\end{corollary}

\Cref{thm--bootstrap} is proven in \Cref{sec--prf--thm--bootstrap} and
\Cref{cor--boot-ci} in \Cref{sec--prf--size}.
Our main procedure is the modified bootstrap with multiplier weights and the
symmetric intervals of \Cref{cor--boot-ci}.
The corollary needs \(\Sigma\succ0\) only for
\(\nabla g(\theta_0)'V\nabla g(\theta_0)>0\); if \(\Sigma=0\), then \(V=0\) and
\Cref{thm--m-est-non-smooth} gives only
\(n^{1/2}(\hat\theta_n-\theta_0)=\op(1)\).
\Cref{sec--two-scale} treats this case for \(K=2\), together with \(\Sigma\)
local to zero.
The plain bootstrap, which reweights \(\hat M_n\) without the quadratic
correction, is also first-order valid, but it distorts the cube-root part of
the law of \(\hat\theta_n\), so its basic intervals undercover when that part
matters; \Cref{sec--plain-bootstrap} gives the details.

\subsection{\texorpdfstring{Non-standard asymptotics and size control when
\(K=2\)}{Non-standard asymptotics and size control when K=2}}
\label{sec--two-scale}

When \(K=2\), no cluster-level shock smooths the top-level component
\(H_{n,(1,1)}\) of \eqref{eqn--gen-hoeffding-decomp}.
Throughout this subsection, and in \Cref{sec--max-score} wherever \(K=2\), we
take the array to be asymptotically balanced, \(N_{n,1}\sim N_{n,2}\sim n\), as
\Cref{asm--two-scale}\eqref{asm--two-scale-rates} below requires, so that
\(|\calI_n|\sim n^2\) and the cube-root rate is \(n^{2/3}\).
For maximum score, the increments of that component over a distance
\(\delta\) have variance of order \(\delta/|\calI_n|\), so by the scaling
argument of \citet{kim1990cube} it moves \(\hat\theta_n\) by a term of order
\(n^{-2/3}\).
The Gaussian term of \Cref{thm--m-est-non-smooth} is of order \(n^{-1/2}\).
The ratio of the two is \(n^{-1/6}\), so the cube-root term can matter at
practical sample sizes.
For \(K\ge3\) the analogous term of the top-level component is of relative
order \(n^{1/2-K/3}\) and negligible.
However, a component \(H_{n,\e}\) with
\(\e\in\calE_2\) whose projection \(\pi_{\e}m_\theta\) is discontinuous in
\(\theta\) would again contribute at relative order \(n^{-1/6}\).
\Cref{sec--plain-bootstrap} illustrates the case \(K=3\) numerically for maximum
score: with non-degenerate cluster-level scores the weight of the cube-root
term falls at rate \(n^{-1/2}\) rather than \(n^{-1/6}\), and by \(n=128\)
the modified bootstrap of \Cref{sec--bootstrap} attains nominal coverage.
We develop non-standard asymptotics that keep both terms.
For a scalar parameter \(a'\theta_0\) and multiplier weights, we then show
that the symmetric intervals of \Cref{cor--boot-ci} control size in every
regime, from fixed \(\Sigma\succ0\) to the degenerate case without a Gaussian
term, provided
\(a'(\hat\theta_n-\theta_0)\) is of the same
order as \(\hat\theta_n-\theta_0\) (\Cref{thm--size};
\Cref{cor--functional} treats smooth functionals and confidence sets).

\paragraph{Drifting sequences.}
To cover all regimes, we impose
\Cref{asm--exchangeable,asm--m-est-consistent-pop-quad-decay,%
asm--m-est-general,asm--two-scale} along a sequence \(\{P_n\}\) of laws of the
array, uniformly in \(n\), as made precise in
\Cref{sec--prf--two-scale}.
The quantities \(\theta_0\), \(A_0\) and \(\Lambda\) (below) are fixed, but
the law of the array, and with it the scores \(\Delta_k=\Delta_k^{(n)}\), may
vary with \(n\), with \(\sup_n\E\|\Delta_k^{(n)}\|^2<\infty\) and
\(\{\|\Delta_k^{(n)}\|^2/\E\|\Delta_k^{(n)}\|^2\}\) uniformly integrable whenever
\(\E\|\Delta_k^{(n)}\|^2>0\).
Let \(\Sigma_{n,k}=\E[\Delta_k^{(n)}\Delta_k^{(n)\prime}]\) and
\[
  \Omega_n:=\var(\bar\psi_n)=\frac{\Sigma_{n,1}}{N_{n,1}}+\frac{\Sigma_{n,2}}{N_{n,2}},\quad
  \sigma_n:=\bigl(\|A_0^{-1}\Omega_nA_0^{-1}\|+n^{-4/3}\bigr)^{1/2},\quad
  c_n:=(n^{2/3}\sigma_n)^{-1}\in(0,1],
\]
with \(\|\cdot\|\) the spectral norm, and let
\(\Gamma_n:=\sigma_n^{-2}A_0^{-1}\Omega_nA_0^{-1}\), so that
\(\|\Gamma_n\|+c_n^2=1\).
The quantity \(\sigma_n\) is the scale of \(\hat\theta_n-\theta_0\) and \(c_n\) the weight
of the cube-root term.
Following \citet[Section~2.1]{chiang2025multiway}, with the Gaussian and
cube-root terms in the roles of their linear and quadratic terms, we call
\(\{P_n\}\) non-degenerate if \(\lambda_{\min}(n^{4/3}\Omega_n)\to\infty\),
nearly degenerate if \(\lambda_{\min}(n^{4/3}\Omega_n)=O(1)\), and degenerate if
\(\Omega_n=0\).
In particular, for \(d=1\) the first two cases are \(c_n\to0\) and
\(\liminf_n c_n>0\).
For balanced arrays \(n^{4/3}\Omega_n\sim n^{1/3}\Sigma_n\), with
\(\Sigma_n=\Sigma_{n,1}+\Sigma_{n,2}\): the Gaussian regime \(c_n\to0\) obtains
if \(n^{1/3}\|\Sigma_n\|\to\infty\) (e.g.\ \(\Sigma_n=\Sigma\succ0\)); the two
terms are of the same order if \(\|\Sigma_n\|\asymp n^{-1/3}\); and \(c_n\to1\) if
\(n^{1/3}\Sigma_n\to0\) (e.g.\ \(\Sigma_n=0\)).
As \(\lambda_1=\lambda_2=1\), a fixed law is degenerate if and only if
\(\Sigma=0\) in \Cref{thm--m-est-non-smooth}.
Write \(\bm U^\circ_{\i}:=(U_{\i\odot\e_1},U_{\i\odot\e_2})\) for the row
and column shocks of cell \(\i\), and
\(\bm U^\circ:=\bm U^\circ_{(1,1)}=(U_{(1,0)},U_{(0,1)})\); write \(\Pr_n\)
for probabilities under \(P_n\).
In this subsection \(\E\) is also taken under \(P_n\).

\begin{assumption}[Conditions for non-standard asymptotics]
\label{asm--two-scale}
\(K=2\) and \(\calN:=\{\theta:\|\theta-\theta_0\|\le\eta_0\}\subseteq\Theta\)
for some \(\eta_0>0\).
\begin{enumerate}[(i)]
  \item \label{asm--two-scale-first}
    \(\sup_\theta|m_\theta|\le1\), and
    \Cref{asm--m-est-general}\eqref{asm--m-est-stoch-diff} holds with
    \(n^{-1/2}\) replaced by \(n^{-2/3}\) in \eqref{eqn--m-est-stoch-diff}.
  \item \label{asm--two-scale-modulus}
    There is a measurable \(\dot m\) with \(\E[\dot m(\bm U^\circ)^2]<\infty\) such that,
    for all \(\theta\in\calN\) and \(\delta\in(0,\eta_0]\),
    \[
      \E\Bigl[\sup_{\theta'\in\calN:\|\theta'-\theta\|\le\delta}
      |m_{\theta'}(W_{(1,1)})-m_\theta(W_{(1,1)})|\Bigm|\bm U^\circ\Bigr]
      \le\dot m(\bm U^\circ)\,\delta.
    \]
    Moreover, for \(0<\eta\le\eta_0\), the class
    \(\{m_\theta-m_{\theta_0}:\|\theta-\theta_0\|\le\eta\}\) is VC-type with
    envelope \(F_\eta:=\sup_{\|\theta-\theta_0\|\le\eta}|m_\theta-m_{\theta_0}|\)
    and characteristics that do not depend on \(\eta\).
  \item \label{asm--two-scale-cov}
    There is a continuous \(\Lambda:\R^d\to[0,\infty)\) with \(\Lambda(v)>0\)
    for \(v\ne0\) such that \(\delta^{-1}\E[(m_{\theta+\delta v}-m_\theta)^2]
    \to\Lambda(v)\) as \(\delta\downarrow0\) and \(\theta\to\theta_0\),
    uniformly in \(v\) on compact sets.
  \item \label{asm--two-scale-rates}
    \(N_{n,1}\sim N_{n,2}\sim n\), and \(\hat\theta_n\) satisfies
    \eqref{eqn--m-est-def} with \(\hat\zeta_n=\op(n^{-4/3})\).
\end{enumerate}
\end{assumption}

Condition \eqref{asm--two-scale-first} strengthens the stochastic
differentiability condition to
the finer \(n^{-2/3}\) scale.
Condition \eqref{asm--two-scale-modulus} bounds the expected largest change
of \(m_\theta\), given \(\bm U^\circ\), when \(\theta\) moves by up to
\(\delta\), and keeps the row--column interaction in \(H_{n,(1,1)}\) of
smaller order.
Its local VC-type requirement holds whenever \(\calM\) is a VC-subgraph class
and \(F_\eta\) is measurable, as such a class satisfies the uniform entropy bound
for every measurable envelope
\citep[Section~2.6]{2023vandervaartWeakConvergenceEmpirical}.
Condition \eqref{asm--two-scale-cov}, a uniform version of condition~(v) of
\citet[Main Theorem~1.1]{kim1990cube}, implies \(\Lambda(-v)=\Lambda(v)\),
\(\Lambda(\gamma v)=\gamma\Lambda(v)\) for \(\gamma>0\) and, by polarisation,
\(\delta^{-1}\E[(m_{\theta+\delta s}-m_\theta)(m_{\theta+\delta
t}-m_\theta)]\to C_{\mathbb W}(s,t):=\frac12\{\Lambda(s)+\Lambda(t)-\Lambda(s-t)\}\).

Of the four conditions, (i) is the least demanding: for criteria bounded by
one, it holds whenever
\(|(\hatE_{n,k}-\E)[r_{k,\theta}]|\le\|\theta-\theta_0\|^2R_{n,k}\) near
\(\theta_0\) with \(\Eout R_{n,k}\to0\), because the scale in
\eqref{eqn--m-est-stoch-diff} then does not matter.
The substantive additions are conditions (ii) and (iii), both of which
concern the top-level component \(H_{n,(1,1)}\).
\Cref{asm--m-est-general}\eqref{asm--m-est-stoch-eq} only bounds this
component, which is enough when it is negligible.
Here, however, it produces the cube-root term, so its local behaviour must be specified.
For indicators, the modulus bound in (ii) bounds the conditional probability
that \(m_\theta\) switches by a multiple of \(\delta\), which requires the observation to vary,
through the cell's own shock, given its row and column shocks.
Condition (iii) requires these changes to have a local variance \(\Lambda\),
which, together with \(A_0\), determines the law of \(\calZ\).

For maximum score and a fixed DGP, \Cref{prp--max-score-asymp-normal}(b)
shows that (i)--(iii) hold when the covariates have a suitably bounded density
given the row and column shocks of their cell (\Cref{asm--ms-two-scale},
which also imposes (iv)).
In its proof, (i) follows from \eqref{eqn--P4-stoch-diff}; the VC-type part of
(ii) holds because \(\calM\) is a VC-subgraph class and \(F_\eta\) is Borel; the
modulus bound in (ii)
follows from that density bound; and (iii) holds with \(\Lambda\) determined by
the density of the covariates on the hyperplane \(\{x:x'\beta_0=0\}\).
Condition \eqref{asm--two-scale-rates} holds for balanced arrays and exact
maximisers (\(\hat\zeta_n=0\)).
Its balance requirement is imposed for simplicity of presentation: the proofs
of \Cref{sec--prf--two-scale} are written in the cube-root rate
\(|\calI_n|^{1/3}\), and balance lets \(\sigma_n\), \(c_n\) and the tolerance
be stated as powers of \(n\).

\begin{remark}[Fixed law]
\label{rem--fixed-law}
Let \(\{P_n\}\) be constant, so that the uniformity
requirements of \Cref{sec--prf--two-scale} hold automatically.
Beyond the conditions of \Cref{thm--m-est-non-smooth,thm--bootstrap},
\Cref{thm--two-scale,thm--two-scale-boot} then need only
\Cref{asm--two-scale}, multiplier weights and
\(\tilde\zeta^*_n=\op(n^{-4/3})\), the last two for
\Cref{thm--two-scale-boot}.
Under such a law, \(\Omega_n=\Sigma_{n,1}/N_{n,1}+\Sigma_{n,2}/N_{n,2}\) is
either
zero, so that \(c_n=1\), or, by \Cref{asm--two-scale}\eqref{asm--two-scale-rates}, satisfies
\(n^{4/3}\|\Omega_n\|\to\infty\), so that \(c_n\to0\); intermediate values of
\(c_n\) require drifting laws.
\(\lozenge\)
\end{remark}

Let \(\mathbb W\) be a centred Gaussian process on \(\R^d\) with continuous
sample paths and covariance kernel \(C_{\mathbb W}\).
It has stationary increments, with \(\var\{\mathbb W(s)-\mathbb
W(t)\}=\Lambda(s-t)\), and is self-similar: \(\mathbb W(\gamma\,\cdot)\overset
d=\gamma^{1/2}\mathbb W(\cdot)\) for \(\gamma>0\).
Let
\[
  \calZ:=\argmax_{v\in\R^d}\Bigl\{\mathbb W(v)-\tfrac12v'A_0v\Bigr\}.
\]
For \(d=1\), \(\Lambda(v)=\Lambda(1)|v|\), so
\(\mathbb W\overset d=\Lambda(1)^{1/2}\mathbb B\) for a two-sided standard
Brownian motion \(\mathbb B\), and
\[
  \calZ\overset d=\{2\Lambda(1)^{1/2}/A_0\}^{2/3}\calZ_{\rm Ch},\qquad
  \calZ_{\rm Ch}:=\argmax_{t\in\R}\{\mathbb B(t)-t^2\},
\]
where \(\calZ_{\rm Ch}\) has Chernoff's distribution, whose density is symmetric, positive and
log-concave \citep{balabdaoui2014chernoff}.
\Cref{fig--chernoff} shows a path of
\(\mathbb B(t)-t^2\) with its maximiser and the density of \(\calZ_{\rm Ch}\).

\begin{figure}[htbp]
\centering
\color{black}
\begin{minipage}[b]{0.48\textwidth}
\centering
\begin{tikzpicture}
\begin{axis}[
  width=\linewidth, height=0.53\linewidth,
  xmin=-1.6, xmax=1.6, ymin=-2.459, ymax=1.562, xtick={-1,0,1}, ytick=\empty,
  xlabel={\(t\)}, label style={font=\scriptsize}, tick label style={font=\scriptsize},
  axis line style={draw=black!55}, tick style={draw=black!55}, clip=false,
]
\addplot[simblue, thin] table[row sep=\\] {x y \\ -1.600 -0.9433 \\ -1.592 -0.8433 \\ -1.584 -0.6770 \\ -1.576 -0.6824 \\ -1.568 -0.7621 \\ -1.560 -0.6404 \\ -1.552 -0.5664 \\ -1.544 -0.6499 \\ -1.536 -0.7834 \\ -1.528 -0.9609 \\ -1.520 -1.0254 \\ -1.512 -0.9009 \\ -1.504 -0.8670 \\ -1.496 -0.8085 \\ -1.488 -0.8260 \\ -1.480 -0.9266 \\ -1.472 -0.7818 \\ -1.464 -0.9450 \\ -1.456 -0.7595 \\ -1.448 -0.6714 \\ -1.440 -0.6644 \\ -1.432 -0.6921 \\ -1.424 -0.6184 \\ -1.416 -0.6454 \\ -1.408 -0.6103 \\ -1.400 -0.7104 \\ -1.392 -0.8280 \\ -1.384 -0.7695 \\ -1.376 -0.7534 \\ -1.368 -0.7479 \\ -1.360 -0.8132 \\ -1.352 -0.8055 \\ -1.344 -0.8162 \\ -1.336 -0.6624 \\ -1.328 -0.6940 \\ -1.320 -0.7628 \\ -1.312 -0.7979 \\ -1.304 -0.7777 \\ -1.296 -0.6740 \\ -1.288 -0.5754 \\ -1.280 -0.5837 \\ -1.272 -0.4633 \\ -1.264 -0.4083 \\ -1.256 -0.4186 \\ -1.248 -0.4643 \\ -1.240 -0.3988 \\ -1.232 -0.3302 \\ -1.224 -0.1350 \\ -1.216 -0.0700 \\ -1.208 0.0252 \\ -1.200 0.0760 \\ -1.192 0.0365 \\ -1.184 0.1457 \\ -1.176 -0.0006 \\ -1.168 0.0205 \\ -1.160 -0.0102 \\ -1.152 -0.0346 \\ -1.144 0.1369 \\ -1.136 0.2435 \\ -1.128 0.2072 \\ -1.120 0.0796 \\ -1.112 0.1712 \\ -1.104 0.2405 \\ -1.096 0.2249 \\ -1.088 0.1589 \\ -1.080 0.1463 \\ -1.072 0.2835 \\ -1.064 0.3216 \\ -1.056 0.3006 \\ -1.048 0.3483 \\ -1.040 0.4441 \\ -1.032 0.5276 \\ -1.024 0.6453 \\ -1.016 0.6894 \\ -1.008 0.7730 \\ -1.000 0.7716 \\ -0.992 0.6693 \\ -0.984 0.6191 \\ -0.976 0.7283 \\ -0.968 0.6048 \\ -0.960 0.7234 \\ -0.952 0.7141 \\ -0.944 0.7080 \\ -0.936 0.5981 \\ -0.928 0.6569 \\ -0.920 0.5829 \\ -0.912 0.5608 \\ -0.904 0.6185 \\ -0.896 0.6484 \\ -0.888 0.7033 \\ -0.880 0.6193 \\ -0.872 0.4054 \\ -0.864 0.4295 \\ -0.856 0.4175 \\ -0.848 0.5464 \\ -0.840 0.6759 \\ -0.832 0.6967 \\ -0.824 0.7184 \\ -0.816 0.7526 \\ -0.808 0.6703 \\ -0.800 0.6661 \\ -0.792 0.7313 \\ -0.784 0.6326 \\ -0.776 0.5161 \\ -0.768 0.5349 \\ -0.760 0.4831 \\ -0.752 0.5218 \\ -0.744 0.6258 \\ -0.736 0.5114 \\ -0.728 0.2807 \\ -0.720 0.3275 \\ -0.712 0.3951 \\ -0.704 0.3189 \\ -0.696 0.4666 \\ -0.688 0.5208 \\ -0.680 0.4573 \\ -0.672 0.4461 \\ -0.664 0.4744 \\ -0.656 0.3532 \\ -0.648 0.4393 \\ -0.640 0.3625 \\ -0.632 0.3898 \\ -0.624 0.2761 \\ -0.616 0.1756 \\ -0.608 0.2647 \\ -0.600 0.4405 \\ -0.592 0.5253 \\ -0.584 0.3632 \\ -0.576 0.3907 \\ -0.568 0.4100 \\ -0.560 0.3998 \\ -0.552 0.4870 \\ -0.544 0.4661 \\ -0.536 0.5017 \\ -0.528 0.4428 \\ -0.520 0.5055 \\ -0.512 0.4252 \\ -0.504 0.5022 \\ -0.496 0.4615 \\ -0.488 0.5652 \\ -0.480 0.6181 \\ -0.472 0.6594 \\ -0.464 0.6313 \\ -0.456 0.7955 \\ -0.448 0.6865 \\ -0.440 0.7081 \\ -0.432 0.6526 \\ -0.424 0.6856 \\ -0.416 0.4071 \\ -0.408 0.4048 \\ -0.400 0.2982 \\ -0.392 0.2493 \\ -0.384 0.2502 \\ -0.376 0.2998 \\ -0.368 0.3235 \\ -0.360 0.4813 \\ -0.352 0.5391 \\ -0.344 0.4344 \\ -0.336 0.3331 \\ -0.328 0.3707 \\ -0.320 0.3678 \\ -0.312 0.4631 \\ -0.304 0.4017 \\ -0.296 0.3836 \\ -0.288 0.3513 \\ -0.280 0.1909 \\ -0.272 0.0123 \\ -0.264 0.1295 \\ -0.256 0.2014 \\ -0.248 0.2208 \\ -0.240 0.3298 \\ -0.232 0.3660 \\ -0.224 0.3910 \\ -0.216 0.4507 \\ -0.208 0.2284 \\ -0.200 0.3258 \\ -0.192 0.3127 \\ -0.184 0.3395 \\ -0.176 0.4383 \\ -0.168 0.3265 \\ -0.160 0.2630 \\ -0.152 0.2192 \\ -0.144 0.1637 \\ -0.136 0.2635 \\ -0.128 0.1792 \\ -0.120 0.1254 \\ -0.112 -0.0775 \\ -0.104 -0.0270 \\ -0.096 -0.1382 \\ -0.088 -0.0865 \\ -0.080 0.0358 \\ -0.072 0.0007 \\ -0.064 0.1051 \\ -0.056 0.0004 \\ -0.048 0.0871 \\ -0.040 0.1018 \\ -0.032 0.2177 \\ -0.024 0.2688 \\ -0.016 0.1599 \\ -0.008 0.0793 \\ 0.000 0.0000 \\ 0.008 -0.0792 \\ 0.016 -0.0144 \\ 0.024 0.0930 \\ 0.032 0.2102 \\ 0.040 0.2136 \\ 0.048 0.2177 \\ 0.056 0.2467 \\ 0.064 0.2591 \\ 0.072 0.1205 \\ 0.080 0.2113 \\ 0.088 0.3209 \\ 0.096 0.5396 \\ 0.104 0.6518 \\ 0.112 0.6695 \\ 0.120 0.5958 \\ 0.128 0.6753 \\ 0.136 0.5990 \\ 0.144 0.6416 \\ 0.152 0.6790 \\ 0.160 0.6403 \\ 0.168 0.5124 \\ 0.176 0.5102 \\ 0.184 0.5514 \\ 0.192 0.5228 \\ 0.200 0.6428 \\ 0.208 0.6773 \\ 0.216 0.7310 \\ 0.224 0.8250 \\ 0.232 0.9471 \\ 0.240 0.8994 \\ 0.248 0.8967 \\ 0.256 0.8125 \\ 0.264 0.9279 \\ 0.272 0.7902 \\ 0.280 0.6576 \\ 0.288 0.5509 \\ 0.296 0.6417 \\ 0.304 0.6346 \\ 0.312 0.6417 \\ 0.320 0.5931 \\ 0.328 0.4945 \\ 0.336 0.4560 \\ 0.344 0.4718 \\ 0.352 0.3900 \\ 0.360 0.2873 \\ 0.368 0.3942 \\ 0.376 0.4480 \\ 0.384 0.5529 \\ 0.392 0.4402 \\ 0.400 0.4418 \\ 0.408 0.5770 \\ 0.416 0.4585 \\ 0.424 0.3371 \\ 0.432 0.3911 \\ 0.440 0.3539 \\ 0.448 0.4236 \\ 0.456 0.4196 \\ 0.464 0.4253 \\ 0.472 0.3012 \\ 0.480 0.3594 \\ 0.488 0.2595 \\ 0.496 0.4030 \\ 0.504 0.4782 \\ 0.512 0.4254 \\ 0.520 0.3736 \\ 0.528 0.3893 \\ 0.536 0.4467 \\ 0.544 0.3752 \\ 0.552 0.4147 \\ 0.560 0.3483 \\ 0.568 0.4898 \\ 0.576 0.5129 \\ 0.584 0.5757 \\ 0.592 0.6377 \\ 0.600 0.6110 \\ 0.608 0.5649 \\ 0.616 0.6113 \\ 0.624 0.5381 \\ 0.632 0.7644 \\ 0.640 0.9193 \\ 0.648 0.7235 \\ 0.656 0.4964 \\ 0.664 0.5640 \\ 0.672 0.4867 \\ 0.680 0.4145 \\ 0.688 0.4602 \\ 0.696 0.4604 \\ 0.704 0.3947 \\ 0.712 0.5040 \\ 0.720 0.7263 \\ 0.728 0.8658 \\ 0.736 0.7878 \\ 0.744 0.7150 \\ 0.752 0.7926 \\ 0.760 0.7812 \\ 0.768 0.9092 \\ 0.776 0.7990 \\ 0.784 0.8052 \\ 0.792 0.4843 \\ 0.800 0.4374 \\ 0.808 0.3823 \\ 0.816 0.2808 \\ 0.824 0.2219 \\ 0.832 0.1449 \\ 0.840 0.0958 \\ 0.848 0.0259 \\ 0.856 0.0421 \\ 0.864 0.0851 \\ 0.872 0.1469 \\ 0.880 0.1050 \\ 0.888 0.1011 \\ 0.896 0.2564 \\ 0.904 0.2251 \\ 0.912 0.2111 \\ 0.920 0.1134 \\ 0.928 0.2620 \\ 0.936 0.2149 \\ 0.944 0.2449 \\ 0.952 0.2324 \\ 0.960 0.1302 \\ 0.968 0.0785 \\ 0.976 0.0336 \\ 0.984 0.0792 \\ 0.992 0.1066 \\ 1.000 0.0529 \\ 1.008 0.0380 \\ 1.016 -0.0561 \\ 1.024 -0.2194 \\ 1.032 -0.2545 \\ 1.040 -0.2131 \\ 1.048 -0.3873 \\ 1.056 -0.3816 \\ 1.064 -0.2994 \\ 1.072 -0.2838 \\ 1.080 -0.4213 \\ 1.088 -0.4345 \\ 1.096 -0.5603 \\ 1.104 -0.5503 \\ 1.112 -0.6666 \\ 1.120 -0.6980 \\ 1.128 -0.7427 \\ 1.136 -0.8543 \\ 1.144 -0.8716 \\ 1.152 -0.9996 \\ 1.160 -1.0488 \\ 1.168 -0.9966 \\ 1.176 -0.9684 \\ 1.184 -0.8210 \\ 1.192 -0.8528 \\ 1.200 -0.9706 \\ 1.208 -1.0188 \\ 1.216 -1.0120 \\ 1.224 -1.1228 \\ 1.232 -1.2181 \\ 1.240 -1.2181 \\ 1.248 -1.2540 \\ 1.256 -1.2098 \\ 1.264 -1.2786 \\ 1.272 -1.2729 \\ 1.280 -1.4378 \\ 1.288 -1.4430 \\ 1.296 -1.4314 \\ 1.304 -1.6103 \\ 1.312 -1.5560 \\ 1.320 -1.6448 \\ 1.328 -1.6670 \\ 1.336 -1.8981 \\ 1.344 -1.8745 \\ 1.352 -1.9171 \\ 1.360 -1.9183 \\ 1.368 -2.0465 \\ 1.376 -2.0499 \\ 1.384 -1.9848 \\ 1.392 -1.8840 \\ 1.400 -1.8896 \\ 1.408 -1.9874 \\ 1.416 -1.9426 \\ 1.424 -1.8978 \\ 1.432 -1.8432 \\ 1.440 -1.8670 \\ 1.448 -1.7738 \\ 1.456 -1.8003 \\ 1.464 -1.7699 \\ 1.472 -1.6969 \\ 1.480 -1.8953 \\ 1.488 -1.8566 \\ 1.496 -1.9005 \\ 1.504 -2.0422 \\ 1.512 -2.1410 \\ 1.520 -2.1377 \\ 1.528 -2.2277 \\ 1.536 -2.0199 \\ 1.544 -2.1594 \\ 1.552 -2.1374 \\ 1.560 -2.0373 \\ 1.568 -2.1255 \\ 1.576 -2.2488 \\ 1.584 -2.3590 \\ 1.592 -2.2206 \\ 1.600 -2.3488 \\};
\addplot[black!60, densely dashed, domain=-1.6:1.6, samples=60] {-x^2};
\draw[black!60] (axis cs:0.2680,\pgfkeysvalueof{/pgfplots/ymin}) -- (axis cs:0.2680,0.9621);
\addplot[only marks, mark=*, mark size=1.8pt, color=simorange] coordinates {(0.2680,0.9621)};
\node[font=\scriptsize, anchor=south] at (axis cs:0.2680,1.0421) {\(\argmax\)};
\end{axis}
\end{tikzpicture}
\end{minipage}\hfill
\begin{minipage}[b]{0.48\textwidth}
\centering
\begin{tikzpicture}
\begin{axis}[
  width=\linewidth, height=0.53\linewidth,
  xmin=-2, xmax=2, ymin=0, xtick={-2,-1,0,1,2}, ytick=\empty,
  label style={font=\scriptsize}, tick label style={font=\scriptsize},
  axis line style={draw=black!55}, tick style={draw=black!55},
  legend style={font=\scriptsize, draw=none, fill=none, at={(0.98,0.98)}, anchor=north east, row sep=-2pt},
]
\addplot[simblue, thick] table[row sep=\\] {x y \\ -1.980 0.0001 \\ -1.940 0.0002 \\ -1.900 0.0003 \\ -1.860 0.0005 \\ -1.820 0.0007 \\ -1.780 0.0010 \\ -1.740 0.0014 \\ -1.700 0.0019 \\ -1.660 0.0027 \\ -1.620 0.0036 \\ -1.580 0.0049 \\ -1.540 0.0065 \\ -1.500 0.0086 \\ -1.460 0.0112 \\ -1.420 0.0145 \\ -1.380 0.0185 \\ -1.340 0.0235 \\ -1.300 0.0294 \\ -1.260 0.0365 \\ -1.220 0.0450 \\ -1.180 0.0549 \\ -1.140 0.0664 \\ -1.100 0.0796 \\ -1.060 0.0946 \\ -1.020 0.1116 \\ -0.980 0.1306 \\ -0.940 0.1517 \\ -0.900 0.1748 \\ -0.860 0.1999 \\ -0.820 0.2270 \\ -0.780 0.2560 \\ -0.740 0.2866 \\ -0.700 0.3186 \\ -0.660 0.3520 \\ -0.620 0.3862 \\ -0.580 0.4211 \\ -0.540 0.4562 \\ -0.500 0.4912 \\ -0.460 0.5257 \\ -0.420 0.5592 \\ -0.380 0.5913 \\ -0.340 0.6217 \\ -0.300 0.6499 \\ -0.260 0.6755 \\ -0.220 0.6981 \\ -0.180 0.7175 \\ -0.140 0.7334 \\ -0.100 0.7455 \\ -0.060 0.7537 \\ -0.020 0.7578 \\ 0.020 0.7578 \\ 0.060 0.7537 \\ 0.100 0.7455 \\ 0.140 0.7334 \\ 0.180 0.7175 \\ 0.220 0.6981 \\ 0.260 0.6755 \\ 0.300 0.6499 \\ 0.340 0.6217 \\ 0.380 0.5913 \\ 0.420 0.5592 \\ 0.460 0.5257 \\ 0.500 0.4912 \\ 0.540 0.4562 \\ 0.580 0.4211 \\ 0.620 0.3862 \\ 0.660 0.3520 \\ 0.700 0.3186 \\ 0.740 0.2866 \\ 0.780 0.2560 \\ 0.820 0.2270 \\ 0.860 0.1999 \\ 0.900 0.1748 \\ 0.940 0.1517 \\ 0.980 0.1306 \\ 1.020 0.1116 \\ 1.060 0.0946 \\ 1.100 0.0796 \\ 1.140 0.0664 \\ 1.180 0.0549 \\ 1.220 0.0450 \\ 1.260 0.0365 \\ 1.300 0.0294 \\ 1.340 0.0235 \\ 1.380 0.0185 \\ 1.420 0.0145 \\ 1.460 0.0112 \\ 1.500 0.0086 \\ 1.540 0.0065 \\ 1.580 0.0049 \\ 1.620 0.0036 \\ 1.660 0.0027 \\ 1.700 0.0019 \\ 1.740 0.0014 \\ 1.780 0.0010 \\ 1.820 0.0007 \\ 1.860 0.0005 \\ 1.900 0.0003 \\ 1.940 0.0002 \\ 1.980 0.0001 \\};
\addlegendentry{Chernoff}
\addplot[black!60, densely dashed] table[row sep=\\] {x y \\ -1.980 0.0005 \\ -1.940 0.0006 \\ -1.900 0.0008 \\ -1.860 0.0011 \\ -1.820 0.0015 \\ -1.780 0.0019 \\ -1.740 0.0025 \\ -1.700 0.0032 \\ -1.660 0.0042 \\ -1.620 0.0053 \\ -1.580 0.0068 \\ -1.540 0.0086 \\ -1.500 0.0109 \\ -1.460 0.0136 \\ -1.420 0.0169 \\ -1.380 0.0210 \\ -1.340 0.0258 \\ -1.300 0.0315 \\ -1.260 0.0382 \\ -1.220 0.0462 \\ -1.180 0.0554 \\ -1.140 0.0660 \\ -1.100 0.0783 \\ -1.060 0.0922 \\ -1.020 0.1080 \\ -0.980 0.1257 \\ -0.940 0.1454 \\ -0.900 0.1672 \\ -0.860 0.1910 \\ -0.820 0.2170 \\ -0.780 0.2450 \\ -0.740 0.2750 \\ -0.700 0.3067 \\ -0.660 0.3401 \\ -0.620 0.3748 \\ -0.580 0.4105 \\ -0.540 0.4469 \\ -0.500 0.4836 \\ -0.460 0.5202 \\ -0.420 0.5561 \\ -0.380 0.5909 \\ -0.340 0.6241 \\ -0.300 0.6551 \\ -0.260 0.6836 \\ -0.220 0.7089 \\ -0.180 0.7308 \\ -0.140 0.7487 \\ -0.100 0.7625 \\ -0.060 0.7718 \\ -0.020 0.7765 \\ 0.020 0.7765 \\ 0.060 0.7718 \\ 0.100 0.7625 \\ 0.140 0.7487 \\ 0.180 0.7308 \\ 0.220 0.7089 \\ 0.260 0.6836 \\ 0.300 0.6551 \\ 0.340 0.6241 \\ 0.380 0.5909 \\ 0.420 0.5561 \\ 0.460 0.5202 \\ 0.500 0.4836 \\ 0.540 0.4469 \\ 0.580 0.4105 \\ 0.620 0.3748 \\ 0.660 0.3401 \\ 0.700 0.3067 \\ 0.740 0.2750 \\ 0.780 0.2450 \\ 0.820 0.2170 \\ 0.860 0.1910 \\ 0.900 0.1672 \\ 0.940 0.1454 \\ 0.980 0.1257 \\ 1.020 0.1080 \\ 1.060 0.0922 \\ 1.100 0.0783 \\ 1.140 0.0660 \\ 1.180 0.0554 \\ 1.220 0.0462 \\ 1.260 0.0382 \\ 1.300 0.0315 \\ 1.340 0.0258 \\ 1.380 0.0210 \\ 1.420 0.0169 \\ 1.460 0.0136 \\ 1.500 0.0109 \\ 1.540 0.0086 \\ 1.580 0.0068 \\ 1.620 0.0053 \\ 1.660 0.0042 \\ 1.700 0.0032 \\ 1.740 0.0025 \\ 1.780 0.0019 \\ 1.820 0.0015 \\ 1.860 0.0011 \\ 1.900 0.0008 \\ 1.940 0.0006 \\ 1.980 0.0005 \\};
\addlegendentry{normal}
\end{axis}
\end{tikzpicture}
\end{minipage}
\caption{Chernoff's distribution.
Left: one simulated path of \(\mathbb B(t)-t^2\), with \(\mathbb B\) a
two-sided standard Brownian motion started at zero, the drift \(-t^2\)
(dashed) and the maximiser, a draw of \(\calZ_{\rm Ch}\).
Right: the density of \(\calZ_{\rm Ch}\), computed as in
\citet{groeneboom2001computing}, and the normal density with the same
variance; Chernoff's distribution has lighter tails (kurtosis 2.84).}
\label{fig--chernoff}
\end{figure}

\begin{theorem}[Non-standard asymptotics]
\label{thm--two-scale}
Suppose
\Cref{asm--exchangeable,asm--m-est-consistent-pop-quad-decay,asm--m-est-general,asm--two-scale}
hold along \(\{P_n\}\).
Let \(G_n\sim N(0,\Omega_n)\) be independent of \(\calZ\).
Then the maximiser defining \(\calZ\) is almost surely unique, \(\calZ\) is
symmetric about zero, and
\[
  \dBL\Bigl(\mathcal L\bigl(\sigma_n^{-1}(\hat\theta_n-\theta_0)\bigr),\
  \mathcal L\bigl(\sigma_n^{-1}A_0^{-1}G_n+c_n\calZ\bigr)\Bigr)\to0.
\]
Moreover:
(a) if \(c_n\to0\), then \(\hat\theta_n-\theta_0=A_0^{-1}\bar\psi_n+\op(\sigma_n)\);
(b) if \(n^{4/3}\Omega_n\to\Omega\), then \(n^{2/3}(\hat\theta_n-\theta_0)\leadsto
A_0^{-1}G+\calZ\), with \(G\sim N(0,\Omega)\) independent of \(\calZ\).
\end{theorem}

Uniformly over the regimes, on the scale \(\sigma_n\), the law of
\(\hat\theta_n-\theta_0\) is thus approximated by the Gaussian law
\(N(0,A_0^{-1}\Omega_nA_0^{-1})\) of the linear term of
\Cref{thm--m-est-non-smooth} convolved with a Chernoff-type law scaled by
\(n^{-2/3}=c_n\sigma_n\).
In the degenerate case, \Cref{thm--two-scale} gives
\(n^{2/3}(\hat\theta_n-\theta_0)\leadsto\calZ\): the limit that
\citet{kim1990cube} obtain for \(|\calI_n|\) independent observations.
For maximum score, \Cref{prp--max-score-asymp-normal}(b) derives this limit
under low-level conditions.

\begin{theorem}[Non-standard asymptotics of the modified bootstrap]
\label{thm--two-scale-boot}
Suppose the conditions of \Cref{thm--two-scale} hold, the weights are of the
multiplier type in \Cref{asm--boot-wgt}, \(\tilde A_n\pto A_0\), and
\(\tilde\theta^*_n\) satisfies \eqref{eqn--m-est-boot-def} with
\(\tilde\zeta^*_n=\op(n^{-4/3})\).
Then
\[
  \dBL\Bigl(\mathcal L\bigl(\sigma_n^{-1}(\tilde\theta^*_n-\hat\theta_n)\bigm|\bm W_n\bigr),\
  \mathcal L\bigl(\sigma_n^{-1}A_0^{-1}G_n+2^{1/3}c_n\calZ\bigr)\Bigr)\pto0.
\]
Moreover:
(a) if \(c_n\to0\), then \(\tilde\theta^*_n-\hat\theta_n=\tilde
A_n^{-1}\tilde\Delta^*_n+\op(\sigma_n)\) (with \(\tilde A_n^{-1}:=0\) where
\(\tilde A_n\) is singular), where
\(\tilde\Delta^*_n:=\sum_{k=1}^2N_{n,k}^{-1}\sum_{i=1}^{N_{n,k}}(\xi_{k,i}-\bar\xi_{n,k})\Delta_k(U_{i\cdot\e_k})\)
is the centred bootstrap score of \eqref{eqn--bootstrap-derivatives};
(b) if \(n^{4/3}\Omega_n\to\Omega\), then \(n^{2/3}(\tilde\theta^*_n-\hat\theta_n)
\dtop A_0^{-1}G+2^{1/3}\calZ\) conditionally given \(\bm W_n\), with
\(G\sim N(0,\Omega)\) independent of \(\calZ\).
\end{theorem}

The bootstrap reproduces the Gaussian term exactly but inflates the cube-root
term by \(2^{1/3}\approx1.26\).
Subtracting \(\hat M_n\) in \eqref{eqn--m-est-boot-def} removes the sample's
own top-level process, leaving that process reweighted by
\(\xi_{n,\bm i}-1=(\xi_{1,i_1}-\bar\xi_{n,1})+(\xi_{2,i_2}-\bar\xi_{n,2})\); this
multiplies its variance by \(\lim_n\var(\xi_{n,\bm i})=2\) (the limit is \(3\)
for pigeonhole weights and \(K\) for multiplier weights in general;
\Cref{rem--plain-bootstrap}).
Weights calibrated to the first-order terms thus count the top-level term
twice, just as the jackknife variance estimate is unbiased for linear
statistics but doubles the quadratic term in expectation
\citep[Theorem~1 and Comment~2, p.~589]{EfronStein1981}.
By self-similarity, \(\argmax_v\{\gamma^{1/2}\mathbb
W(v)-\frac12v'A_0v\}\overset d=\gamma^{1/3}\calZ\) for every \(\gamma>0\),
and here \(\gamma=2\).
As the cube-root term is inflated, not mis-centred, a comparison inequality
in the spirit of \citet{anderson1955integral} turns the inflation into
conservativeness.

\begin{lemma}[Anderson-type comparison]
\label{lem--anderson}
Let \(G\) be a centred Gaussian vector in \(\R^d\), possibly degenerate, and
\(\calY\) a random vector independent of \(G\).
Then \(\Pr(G+\calY\in C)\ge\Pr(G+\gamma\calY\in C)\) for every convex Borel
set \(C\subseteq\R^d\) symmetric about the origin and every \(\gamma\ge1\).
\end{lemma}

The limit laws must also be continuous, so that bootstrap quantiles
converge.
For \(d=1\), \(\calZ\) is a multiple of \(\calZ_{\rm Ch}\), which has a
positive density \citep{groeneboom1989brownian}.
The next lemma covers every \(d\).

\begin{lemma}[Law of \(\calZ\)]
\label{lem--Z-law}
Under the conditions of \Cref{thm--two-scale}, the law of \(\calZ\) and
Lebesgue measure on \(\R^d\) are mutually absolutely continuous.
Hence, for every \(a\in\R^d\setminus\{0\}\), the law of \(a'\calZ\) and
Lebesgue measure on \(\R\) are mutually absolutely continuous; in particular,
the distribution function of \(a'\calZ\) is continuous and strictly
increasing.
\end{lemma}

\Cref{lem--Z-law} is proven in \Cref{sec--prf--size}; its idea is short
enough to give here.
Adding a linear function \(\ell_u(t)=u't\) to \(\mathbb W\) translates the
law of the maximiser by \(A_0^{-1}u\), as completing the square and the
stationarity of the increments show.
If adding \(\ell_u\) also preserved null sets of the law of \(\calZ\), that
law would then be quasi-invariant under every translation --- and a measure on
\(\R^d\) quasi-invariant under all translations is equivalent to Lebesgue
measure.
What preserves null sets is the Cameron--Martin theorem
\citep[Lemma~2.4.1 and Corollary~2.4.3]{bogachev1998gaussian}, but it does not
apply to \(\ell_u\), whose Cameron--Martin norm is infinite: an element \(h\)
of that space satisfies \(|h(t)|\le C_h\Lambda(t)^{1/2}\), and so grows at
most like \(\|t\|^{1/2}\) along a ray, whereas \(\ell_u\) grows linearly
when \(u\ne0\).
The proof therefore replaces \(\ell_u\) by a function \(h\) that does lie in
the Cameron--Martin space and agrees with \(\ell_u\) on the ball of radius
\(R\).
Writing \(\Lambda(v)=\int|v's|\,\nu(\d s)\) for a finite measure
\(\nu\) on the unit sphere
\citep[Theorem~2.4.3]{samorodnitsky1994stable} reproduces the covariance
structure of \(\mathbb W\) concretely enough both to build \(h\) by clipping
and to bound its Cameron--Martin norm
\citep[Theorem~2.4.5]{bogachev1998gaussian}.
The two maximisers then agree unless one of them escapes the ball, which has
probability tending to zero as \(R\) grows, so the null set transfers from
\(h\) back to \(\ell_u\).
This follows the strategy of \citet[proof of Theorem~1]{cattaneo2026continuity},
who combine the Cameron--Martin theorem with shift equivariance to show that
the maximiser of a Gaussian process has a continuous distribution function.
Applying the theorem to clipped linear functions, rather than to the mean,
gives the stronger conclusion that \Cref{thm--size} uses when the Gaussian
term vanishes.

Two features of the argument are worth isolating.
First, the perturbation must range over all of \(\R^d\).
The shift that \(\ell_u\) induces is \(A_0^{-1}u\), so, \(A_0\) being positive
definite, the choice \(u=A_0y\) makes every \(y\in\R^d\) available, and it is
quasi-invariance under the whole translation group that forces equivalence
with Lebesgue measure.
A thinner set of shifts would not do: quasi-invariance holds on a subgroup of
\(\R^d\), so a generating set would suffice, but for \(d=1\) the unit vectors
generate only a lattice, and a probability measure that charges every point of
one coset of a lattice is quasi-invariant under it and singular with respect
to Lebesgue measure.
Second, although \(\ell_u\) is not itself a Cameron--Martin shift, Step~2 of
the proof produces one that agrees with it on any given ball, much as
\citet{cattaneo2026continuity} place the mean function in the Cameron--Martin
space of the process restricted to a cube.
The clipping is what turns such a local statement into one about the maximiser over \(\R^d\).

\paragraph{A scalar parameter.}
Let \(a\in\R^d\) with \(\|a\|=1\), and consider \(a'\theta_0\), for instance
a coordinate of \(\theta_0\).
The scale of \(a'(\hat\theta_n-\theta_0)\), under the condition below, and
the weight of its cube-root term are
\[
  \sigma_{n,a}:=\bigl(a'A_0^{-1}\Omega_nA_0^{-1}a+n^{-4/3}\bigr)^{1/2},
  \qquad
  c_{n,a}:=(n^{2/3}\sigma_{n,a})^{-1}\in(0,1],
\]
so that \((\sigma_{n,a}/\sigma_n)^2=a'\Gamma_na+c_n^2\le1\) and \(c_{n,a}\ge c_n\); for
\(d=1\), \(\sigma_{n,a}=\sigma_n\) and \(c_{n,a}=c_n\).
We call \(\{P_n\}\) non-degenerate in direction \(a\) if \(c_{n,a}\to0\), that
is, \(n^{4/3}a'A_0^{-1}\Omega_nA_0^{-1}a\to\infty\).
If \(\liminf_n\sigma_{n,a}/\sigma_n>0\), that is, if \(a'(\hat\theta_n-\theta_0)\) is
of the same order as \(\hat\theta_n-\theta_0\) (recall that
\(\sigma_{n,a}\le \sigma_n\)),
\Cref{thm--two-scale,thm--two-scale-boot} determine the law of
\(a'(\hat\theta_n-\theta_0)\) on the scale \(\sigma_{n,a}\).
This holds for every \(a\) when \(d=1\), and when \(\Omega_n=\chi_n\bar\Omega\)
for scalars \(\chi_n\ge0\) and a fixed \(\bar\Omega\succ0\), as when the
\(\Sigma_{n,k}\) are a common scalar multiple of fixed \(\Sigma_k\) with
\(\Sigma_1+\Sigma_2\succ0\), from fixed \(\Sigma\succ0\) to
\(\Sigma_n=0\).
Indeed, with \(\bar V:=A_0^{-1}\bar\Omega A_0^{-1}\),
\((\sigma_{n,a}/\sigma_n)^2=(\chi_na'\bar Va+n^{-4/3})/(\chi_n\|\bar V\|+n^{-4/3})
\ge\lambda_{\min}(\bar V)/\lambda_{\max}(\bar V)\), and
\(c_{n,a}=(1+n^{4/3}\chi_na'\bar Va)^{-1/2}\), so that \(c_{n,a}\to0\) if and
only if \(\{P_n\}\) is non-degenerate.
The condition fails when the Gaussian term dominates in some direction but
is negligible in direction \(a\), for example when \(\Sigma_n=\Sigma\ne0\)
with \(a'A_0^{-1}\Sigma A_0^{-1}a=0\).
Then \(\hat\theta_n-\theta_0\) is of order \(n^{-1/2}\), and
\Cref{thm--two-scale,thm--two-scale-boot} determine
\(a'(\hat\theta_n-\theta_0)\) only up to \(\op(n^{-1/2})\), which is not
enough.
We leave this mixed case open.

\begin{theorem}[Size control for a scalar parameter regardless of degeneracy]
\label{thm--size}
Let \(a\in\R^d\) with \(\|a\|=1\).
Suppose the conditions of \Cref{thm--two-scale-boot} hold and
\(\liminf_n\sigma_{n,a}/\sigma_n>0\).
For \(\alpha\in(0,1)\), let \(\hat q_n\) be a \((1-\alpha)\)-quantile of
the conditional law of \(|a'(\tilde\theta^*_n-\hat\theta_n)|\) given
\(\bm W_n\), and let \(\mathrm{CI}_n:=[a'\hat\theta_n-\hat q_n,a'\hat\theta_n+\hat
q_n]\).
Then:
\begin{enumerate}[(a)]
  \item \(\liminf_n\Pr_n(a'\theta_0\in \mathrm{CI}_n)\ge1-\alpha\);
  \item if \(c_{n,a}\to0\), then \(\Pr_n(a'\theta_0\in \mathrm{CI}_n)\to1-\alpha\);
  \item if \(c_{n,a}\to c\in(0,1]\), then \(\Pr_n(a'\theta_0\in
    \mathrm{CI}_n)\to\Pr(|T_{c,a}|\le q_{c,a})>1-\alpha\), where
    \(T_{c,a}:=(1-c^2)^{1/2}Z_0+c\,a'\calZ\), \(q_{c,a}\) is the
    \((1-\alpha)\)-quantile of
    \(|(1-c^2)^{1/2}Z_0+2^{1/3}c\,a'\calZ|\), and \(Z_0\sim N(0,1)\) is
    independent of \(\calZ\);
  \item \(\Pr_n(a'\theta_0\in \mathrm{CI}_n)-\Pr(|T_{c_{n,a},a}|\le q_{c_{n,a},a})\to0\),
    where \(T_{c,a}\) and \(q_{c,a}\) are defined as in (c) for every
    \(c\in[0,1]\); in particular, if \(\liminf_nc_{n,a}>0\), then
    \(\liminf_n\Pr_n(a'\theta_0\in \mathrm{CI}_n)>1-\alpha\).
\end{enumerate}
\end{theorem}

Thus, by (b) and (d), the intervals are asymptotically exact when
\(\{P_n\}\) is non-degenerate in direction \(a\) and conservative when
\(\liminf_nc_{n,a}>0\); by (d), the limit coverage in (c), evaluated at
\(c=c_{n,a}\), approximates the coverage along every sequence.
When \(\Omega_n=\chi_n\bar\Omega\) with \(\bar\Omega\succ0\), these are the
non-degenerate and nearly degenerate cases.
For \(d=1\), \Cref{thm--size} with \(a=1\) concerns \(\theta_0\) itself.
If \(\hat q_n\) is not measurable, (a)--(d) hold for the outer probability
\(\Prout_n(a'\theta_0\in\mathrm{CI}_n)\) and for the inner probability
\(1-\Prout_n(a'\theta_0\notin\mathrm{CI}_n)\).
Since (a) holds along every sequence that satisfies the conditions, it
implies uniform asymptotic validity,
\[
  \liminf_{n\to\infty}\inf_{P\in\mathcal P}P(a'\theta_0\in \mathrm{CI}_n)\ge1-\alpha,
\]
for any class \(\mathcal P\) of distributions with common \(\theta_0\),
\(A_0\) and \(\Lambda\) such that every sequence \(\{P_n\}\) in \(\mathcal P\)
satisfies the conditions of the theorem \citep[cf.][]{andrews2020generic}.
Because \(T_{c,a}\) and \((1-c^2)^{1/2}Z_0+2^{1/3}c\,a'\calZ\) are
symmetric, and the proof of \Cref{thm--size} shows that, along subsequences
with \(c_{n,a}\to c\), \(\sigma_{n,a}^{-1}a'(\tilde\theta^*_n-\hat\theta_n)\dtop
(1-c^2)^{1/2}Z_0+2^{1/3}c\,a'\calZ\) conditionally given \(\bm W_n\),
equal-tailed intervals from the
modified bootstrap have the same limit coverage.
We recommend symmetric intervals, which are less sensitive to finite-sample asymmetry of the
bootstrap law.

The same bootstrap covers smooth functionals of \(\theta_0\) and the whole
vector.

\begin{corollary}[Smooth functionals and confidence sets]
\label{cor--functional}
Suppose the conditions of \Cref{thm--two-scale-boot} hold, and let
\(\alpha\in(0,1)\).
\begin{enumerate}[(a)]
  \item Let \(g:\Theta\to\R\) be Borel measurable and continuously
    differentiable near \(\theta_0\) with \(\nabla
    g(\theta_0)\ne0\), let \(a:=\nabla g(\theta_0)/\|\nabla
    g(\theta_0)\|\) satisfy \(\liminf_n\sigma_{n,a}/\sigma_n>0\), and let
    \(\hat q_n\) be a \((1-\alpha)\)-quantile of the conditional law of
    \(|g(\tilde\theta^*_n)-g(\hat\theta_n)|\) given \(\bm W_n\).
    Then \Cref{thm--size}(a)--(d) hold with \(g(\theta_0)\) in place of
    \(a'\theta_0\), \(\mathrm{CI}_n:=[g(\hat\theta_n)-\hat
    q_n,g(\hat\theta_n)+\hat q_n]\), and \(\sigma_{n,a}\) and
    \(c_{n,a}\) formed from this \(a\).
  \item Let \(S\) be a fixed non-singular \(d\times d\) matrix, let \(\hat
    q_n\) be a \((1-\alpha)\)-quantile of the conditional law of
    \(\|S(\tilde\theta^*_n-\hat\theta_n)\|\) given \(\bm W_n\), and let
    \(\mathrm{CS}_n:=\{\theta:\|S(\theta-\hat\theta_n)\|\le\hat q_n\}\).
    Then \(\liminf_n\Pr_n(\theta_0\in \mathrm{CS}_n)\ge1-\alpha\), and
    \(\Pr_n(\theta_0\in \mathrm{CS}_n)\to1-\alpha\) if \(c_n\to0\).
\end{enumerate}
If \(\hat q_n\) is not measurable, (a) and (b) hold for the outer and for the
inner probability.
\end{corollary}

Part (b) asks nothing of a direction: on the scale \(\sigma_n\) of the whole
vector no direction needs a finer expansion, and \Cref{lem--Z-law} makes the
limit laws continuous.
Proofs for this subsection are in \Cref{sec--prf--two-scale,sec--prf--size}.

\begin{remark}[A rescaled bootstrap]
The factor \(2^{1/3}\) in \Cref{thm--two-scale-boot} comes from counting
the top-level term twice, as \citet{chiang2025multiway} find for the multiway
jackknife and remove in their modified empirical likelihood.
A bootstrap analogue reshapes \eqref{eqn--m-est-boot-def} with \(\tilde
A_n/2\) in place of \(\tilde A_n\) and uses the conditional quantiles of
\(t^*_n:=(\tilde\theta^{*}_{n,\mathrm r}-\hat\theta_n)/2\), with
\(\tilde\theta^{*}_{n,\mathrm r}\) the resulting maximiser.
The proof of \Cref{thm--size} then applies without the inflation, so symmetric
intervals built from \(a't^*_n\) are asymptotically exact in every regime,
the degenerate case included.
We do not recommend them in practice.
The inflation is also the margin that absorbs the first-order remainder and an
overestimated \(\tilde A_n\), and it is where that remainder is large --- DM
and DL of \Cref{sec--simulations} --- that the margin is needed most.
\(\lozenge\)
\end{remark}

\section{Asymptotic Distribution of the Maximum Score Estimator}
\label{sec--max-score}

This section establishes the asymptotic normality of the maximum score
estimator under multiway dependence from low-level conditions
(\Cref{asm--ms-consistency,asm--X-ac-smooth-m-diff}) on the joint
distribution of the outcome, the covariates, the error and the latent cluster
shocks.
We first show that the estimator is consistent
(\Cref{lem--max-score-consistency}) and then verify the high-level conditions
of \Cref{thm--m-est-non-smooth} under these low-level conditions
(\Cref{prp--max-score-asymp-normal}(a)).
For \(K=2\), further conditions (\Cref{asm--ms-two-scale}) also yield the
high-level conditions of \Cref{sec--two-scale}, and hence the cube-root limit
when the row and column scores vanish (\Cref{prp--max-score-asymp-normal}(b)).
Observations are \(W_{\i} = \left( Y_{\i}, X_{\i}' \right)'\), where \(Y_{\i}\)
is a binary outcome with values in \(\{- 1, 1\}\) and \(X_{\i}\) is a
\(d\)-dimensional vector of covariates.
Here \(d\) is the dimension of \(X_{\i}\) rather than of \(\theta\), which lies
in \(\R^{d-1}\), so the results of \Cref{sec--m-est} are applied with \(d-1\)
in place of \(d\).
Denote the surface of the unit sphere by \(\sphere^{d - 1} = \left\{ v : \|v\| =
1 \right\}\) and let
\(\mathcal{B} \subseteq \sphere^{d - 1}\).
The maximum score estimator of \citet{manski1975maximum} solves
\begin{equation}
  \begin{split}
    \hat{M}_{n} (\hat{\beta}_{n})  &= \sup_{\beta \in \calB} \hat{M}_{n} (\beta)  - \op
    (1) \quad \text{where} \quad \hat{M}_{n} (\beta) = \hatE_{n} [m_{\beta}], \\
    m_{\beta} (y, x)  &= \ind \{y \cdot x'\beta \geq 0\} = \begin{cases}
      1 & \text{if } \sgn (y) = \sgn (x' \beta) \text{ or if } x' \beta = 0, \\
      0 & \text{otherwise}.
    \end{cases}
  \end{split}
  \label{eqn--max-score-estimator-def}
\end{equation}
Here \(\hat\beta_n\in\calB\) is measurable, as \eqref{eqn--m-est-def}
requires.
Manski's own criterion \(\ind\{y=1,x'\beta\ge0\}+\ind\{y=-1,x'\beta<0\}\)
equals \(m_\beta(y,x)-\ind\{y=-1,x'\beta=0\}\), and, writing
\(\ind\{x'\beta=0\}=\ind\{x'\beta\le0\}+\ind\{x'(-\beta)\le0\}-1\), the proof
of \Cref{lem--exchangeable-ulln-for-max-score} applied to half-spaces in
\(X_{\i}\) gives
\(\sup_{\beta\in\calB}\hatE_n[\ind\{X_{\i}'\beta=0\}]\pto0\) under
\Cref{asm--exchangeable} and
\Cref{asm--ms-consistency}\eqref{asm--continuity-under-constraint} below.
A measurable maximiser of Manski's criterion over \(\calB\) therefore
satisfies \eqref{eqn--max-score-estimator-def}, with that supremum as its
tolerance.
The sample criterion \(\hat M_n\) is a step function of \(\beta\), so
expansions based on its derivatives are not available.

As in \Cref{sec--m-est}, \(W = (Y, X')'\) has the common marginal
distribution \(P_0\) of the \(W_{\i}\); denote
\begin{equation}
  g_{0} (X) = \E [Y \mid X].
  \label{eqn--g0-def}
\end{equation}
We impose the following modelling and regularity conditions on the joint
distribution of \((Y,X',\epsilon)'\) for a generic cell.
Only the implied law of \(W=(Y,X')'\) enters
\Cref{lem--max-score-consistency}.

\begin{assumption}
\label{asm--ms-consistency}
\hfill
\begin{enumerate}[(i)]
  \item \label{asm--latent-linear-median-regression}
    \(Y = 2 \ind \left\{ X' \beta_{0} - \epsilon \geq 0 \right\} -
    1\) for some \(\beta_{0} \in \calB\), and the conditional median
    of the error satisfies \(\mathrm{Med} [\epsilon \mid X] = 0\) almost surely.
  \item \label{asm--index-sgn-disagree-pos-pr}
    For any \(b \in \sphere^{d - 1} \setminus \{\beta_{0}\}\), \(\Pr \left\{
    \left| g_{0} (X) \right| > 0, \,\sgn \left( X' \beta_{0} \right)
    \neq \sgn \left( X' b \right) \right\} > 0\).
  \item \label{asm--continuity-under-constraint}
    The estimation set \(\mathcal{B}\) in \eqref{eqn--max-score-estimator-def}
    is a closed subset of \(\sphere^{d - 1}\).
    Furthermore, for every \(b \in \mathcal{B}\), \(\Pr \left\{ X' b =
    0 \right\} = 0\).
\end{enumerate}
\end{assumption}

Lower-level sufficient conditions for
\Cref{asm--ms-consistency}\eqref{asm--index-sgn-disagree-pos-pr} and
\eqref{asm--continuity-under-constraint} are well known.
\citet{1985manskiSemiparametricAnalysisDiscrete} couples an index restriction
with a special regressor \(X[l]\), \(\beta_{0,l}\ne0\), of full support
conditional on the others, and restricts the parameter space to
\(\calB=\{b\in\sphere^{d-1}:|b_l|\ge\eta\}\) for some
\(0<\eta<|\beta_{0,l}|\).
What then gives \eqref{asm--continuity-under-constraint} is that \(X[l]\) have
no atoms conditional on the other regressors \(X[-l]\), as when it has a
conditional density: then
\(\Pr\{X'b=0\}=\E[\Pr\{X[l]=-X[-l]'b_{-l}/b_l\mid X[-l]\}]=0\) for every
\(b\in\calB\), where \(b_{-l}\) omits the \(l\)th entry of \(b\).
Full support alone would not suffice: with an atom at zero,
\(\Pr\{X'b=0\}>0\) at the \(l\)th standard basis vector, which lies in
\(\calB\).
What gives \eqref{asm--index-sgn-disagree-pos-pr} is that \(g_0(X)\ne0\) where
the signs of \(X'\beta_0\) and \(X'b\) differ, as under a unique conditional
median of \(\epsilon\).
These conditions allow some discrete regressors, whereas \citet{kim1990cube}
assume directly that \(X\) has a differentiable Lebesgue density and that
\(X/\|X\|\) has a density with respect to surface measure on the sphere,
which give \eqref{asm--continuity-under-constraint} with
\(\calB=\sphere^{d-1}\).
The following result provides consistency of the maximum score estimator under
multiway dependence.

\begin{lemma}[Consistency]
\label{lem--max-score-consistency}
Suppose \Cref{asm--exchangeable,asm--ms-consistency} hold.
Then \(\widehat{\beta}_{n}\) in \eqref{eqn--max-score-estimator-def} satisfies
\(\widehat{\beta}_{n} \pto \beta_{0}\).
\end{lemma}

\Cref{lem--max-score-consistency} is proven in
\Cref{sec--prf--lem--max-score-consistency}.
Consistency lets us work with \(\beta\) close to \(\beta_0\), in the local
parameter space
\begin{equation*}
  \calB_{0} = \{\beta \in \sphere^{d - 1} : \|\beta - \beta_{0}\| <
  \sqrt{2}\} = \{\beta \in \sphere^{d - 1} : \beta_{0}' \beta > 0\}.
\end{equation*}
\(\calB_{0}\) admits a smooth parametrisation by \(\Theta = \{\theta \in \R^{d
- 1} : \|\theta\| < 1\}\).
Let \(B_{0}\) be an orthonormal \(d \times (d - 1)\) basis matrix for the
subspace \(\{\beta : \beta' \beta_{0} = 0\}\),
i.e.\ \(B_{0}' B_{0} = \I_{d - 1}\) and \(\beta_{0}' B_{0} = 0\).
Define
\begin{equation}
  \beta (\theta) = B_{0} \theta + \sqrt{1 - \|\theta\|^{2}}\, \beta_{0}.
  \label{eqn--beta-local}
\end{equation}
The parametrisation in \eqref{eqn--beta-local} is a diffeomorphism from
\(\Theta\) to \(\calB_{0}\) and satisfies \(\beta_{0} = \beta (0)\), \(\|\beta
(\theta)\| = 1\) for every \(\theta\), and \(\beta_{0}' \beta (\theta) =
\sqrt{1 - \|\theta\|^{2}}\), so \(\beta_0\) is interior, being identified with
\(\theta=0\in\mathrm{int}(\Theta)\).
Finally, let \(\hat\theta_n=B_0'\hat\beta_n=B_0'(\hat\beta_n-\beta_0)\), so
that \(\hat\theta_n\pto0\).
As \((\beta_0,B_0)\) is an orthonormal basis of \(\R^d\) and
\(\|\hat\beta_n\|=1\),
\[
  \hat\beta_n=B_0\hat\theta_n+\sgn(\hat\beta_n'\beta_0)
  \sqrt{1-\|\hat\theta_n\|^2}\,\beta_0 ,
\]
which is \(\beta(\hat\theta_n)\) on the event
\(\{\hat\beta_n\in\calB_0\}\).

We further impose the following regularity conditions on the joint distributions
of \(Y\), \(X\) and the latent shocks \((U_{\e'})_{\e' \in \calE, \e' \le \e}\)
arising from \Cref{asm--exchangeable}, for \(\e\in\calE_1\).
They impose no topological structure on the latent shocks.

\begin{assumption}
\label{asm--X-ac-smooth-m-diff}
The following approximate maximiser condition holds at rate \(n^{-1}\):
\begin{equation*}
  \hat{M}_{n} (\hat{\beta}_{n}) \geq
  \sup_{\theta \in \Theta} \hat{M}_{n} (\beta (\theta)) - \op (n^{-1}).
\end{equation*}
Furthermore, for a generic cell let \(\bm u_{\e}=(u_{\e'})_{\e'\in\calE,\e'\le\e}\) denote a
value of its latent shocks \(\bm U_{\e}=(U_{\e'})_{\e'\in\calE,\e'\le\e}\), the vector
\(\bm U_{\i\odot\e}\) of \Cref{sec--m-est} with \(\i=\bm\iota\).
The following conditions hold.
\begin{enumerate}[(i)]
  \item \label{asm--X-ac-smooth-m-diff-p-smooth}
    For each \(\bm{e} \in \calE_{1}\), the conditional distribution \(X \mid \bm U_{\e}\) is absolutely continuous with respect to Lebesgue measure, and the
    associated conditional density, \(p_{\bm{e}} (x \mid \bm{u}_\e)\), is
    continuously differentiable with respect to \(x\) on \(\R^{d}\).
  \item \label{asm--X-ac-smooth-m-diff-m-smooth}
    For each \(\bm{e} \in \calE_{1}\), let \(g_{0, \bm{e}} (x, \bm{u}_{\e}) = \E
    [Y \mid X = x, \bm U_{\e} = \bm{u}_{\e}]\).
    Then \(g_{0, \bm{e}} (x, \bm{u}_\e)\) is continuously differentiable with
    respect to \(x\) for each \(\bm{u}_{\e}\) with derivative \(\dot{g}_{0,
    \bm{e}} (x, \bm{u}_{\e}) = \nabla_{x} g_{0, \bm{e}} (x,
    \bm{u}_{\e})\).
  \item \label{asm--X-ac-smooth-m-diff-m-dominance}
    For each \(\bm{e} \in \calE_{1}\), there are \(\eta>0\) and a
    measurable \(D_{\e}\) with \(\E[D_{\e}(\bm U_{\e})^{2}]<\infty\) such that,
    for all \(x\in\R^{d}\) and all \(\bm u_{\e}\),
    \begin{equation*}
      p_{\e}(x\mid\bm u_{\e})+\|\nabla_{x}p_{\e}(x\mid\bm u_{\e})\|
      +p_{\e}(x\mid\bm u_{\e})\,\|\dot g_{0,\e}(x,\bm u_{\e})\|
      \le D_{\e}(\bm u_{\e})\,(1+\|x\|)^{-d-2-\eta}.
    \end{equation*}
  \item \label{asm--X-ac-smooth-m-diff-primitive-density}
    \(X\) has a continuously differentiable density \(f_{X}\) on
    \(\R^{d}\), and the conditional cumulative distribution function
    \(F (t \mid x) = \Pr (\epsilon \leq t \mid X = x)\) admits a version that is
    continuously differentiable jointly in \((t, x)\).
    Write \(f_{\epsilon|X} (t \mid x) = \partial_{t} F (t \mid x)\).
    There exist \(a_F \geq 0\), \(\eta > 0\) (not
    necessarily the \(\eta\) of
    \eqref{asm--X-ac-smooth-m-diff-m-dominance}), and
    \(C_{X}, C_{F} < \infty\) such that, for every
    \(x \in \R^{d}\),
    \begin{equation*}
      |f_{X} (x)| + \|\nabla f_{X} (x)\| \leq C_{X} (1 + \|x\|)^{-d - 2 - \eta -
      a_F}
    \end{equation*}
    and
    \begin{equation*}
      f_{\epsilon|X} (x' \beta_{0} \mid x)
      + \bigl\|\nabla_{x} F (t \mid x)\big|_{t = x' \beta_{0}}\bigr\|
      \leq C_{F} (1 + \|x\|)^{a_F}.
    \end{equation*}
  \item \label{asm--X-ac-smooth-m-diff-m-increasing}
    For \(g_{0} (x) = 2 F (x' \beta_{0} \mid x) - 1\), a version of
    \eqref{eqn--g0-def}, \(f_X(x)\,\beta_0'\nabla g_0(x)>0\) for all \(x\) in
    a non-empty relatively open subset of \(\{x:x'\beta_0=0\}\).
\end{enumerate}
\end{assumption}

The approximate maximiser condition at rate \(n^{-1}\) is mild and can be
ensured by the econometrician during numerical optimisation.
It holds if \(\hat\beta_n\) maximises \(\hat M_n\) over \(\calB\) to within
\(\op(n^{-1})\) and either \(\calB=\sphere^{d-1}\) or \(\beta_0\) lies in the
relative interior of \(\calB\).
In the latter case, \(\hat M_n\to M_0\) uniformly on \(\sphere^{d-1}\) in
probability (as in
\Cref{lem--exchangeable-ulln-for-max-score}), and \(M_0\), continuous on
\(\sphere^{d-1}\) when \(X\) has a density, has its unique maximum at
\(\beta_0\) (as in \Cref{lem--ms-b0-is-maximizer}, by
\Cref{asm--ms-consistency}\eqref{asm--index-sgn-disagree-pos-pr}).
With probability tending to one, the supremum of \(\hat M_n\) over
\(\calB_0\) therefore equals that over a neighbourhood of \(\beta_0\)
contained in \(\calB\).
Conditions \eqref{asm--X-ac-smooth-m-diff-p-smooth} and
\eqref{asm--X-ac-smooth-m-diff-m-smooth} rule out point masses in
\(X\mid\bm U_{\e}\) and make the conditional regression function
\(g_{0,\e}(x,\bm u_\e)\) smooth in \(x\), so that the population criterion
admits local smooth approximations and its local curvature is given by
\(\dot g_{0,\e}(x,\bm u_\e)\).
Condition \eqref{asm--X-ac-smooth-m-diff-m-dominance}, the conditional
analogue of \eqref{asm--X-ac-smooth-m-diff-primitive-density}, lets us differentiate
\(\theta\mapsto\pi_{\e_k}m_\theta\) twice under the integral, with derivatives
dominated by a multiple of \(D_{\e_k}+1\).
Here \(m_\theta:=m_{\beta(\theta)}\).
The conditions concern only \(\e\in\calE_1\), because higher-order projections
need no smoothness (\Cref{asm--m-est-general}\eqref{asm--m-est-stoch-eq}).
For \(K=2\), they could not hold for \(\e=(1,1)\) in any case: \(\bm U_{(1,1)}\)
comprises all the latent variables of a cell, so \(X\) is degenerate given it.
Condition \eqref{asm--X-ac-smooth-m-diff-m-increasing} makes \(A_0\) in
\eqref{eqn--A0-ms} positive definite
(\Cref{sec--prf--thm-max-score-asymp-normal}).
As \(F(0\mid x)=1/2\) where \(f_X(x)>0\),
\(\beta_0'\nabla g_0(x)=2f_{\epsilon|X}(0\mid x)\) on
\(\{x:x'\beta_0=0,f_X(x)>0\}\), so the condition holds if
\(f_X(x)f_{\epsilon|X}(0\mid x)>0\) for some \(x\) with \(x'\beta_0=0\); it
corresponds to the non-singularity condition of
\citet[Example~6.4, p.~215]{kim1990cube}.

For \(K=2\), part (b) of the next result also uses the following conditions,
where, as in \Cref{sec--two-scale}, \(\bm U^\circ=(U_{(1,0)},U_{(0,1)})\)
denotes the row and column shocks of cell \((1,1)\).

\begin{assumption}
\label{asm--ms-two-scale}
\(K=2\), and the following conditions hold.
\begin{enumerate}[(i)]
  \item \label{asm--ms-two-scale-density}
    Given \(\bm U^\circ\), \(X_{(1,1)}\) has a Lebesgue density bounded by
    \(\bar p(\bm U^\circ)(1+\|x\|)^{-d-\eta^\circ}\) at every \(x\in\R^d\),
    for some \(\eta^\circ>0\) and a measurable \(\bar p\) with
    \(\E[\bar p(\bm U^\circ)^2]<\infty\).
  \item \label{asm--ms-two-scale-rates}
    \(N_{n,1}\sim N_{n,2}\sim n\), and the approximate maximiser condition of
    \Cref{asm--X-ac-smooth-m-diff} holds at rate \(n^{-4/3}\):
    \(\hat M_n(\hat\beta_n)\ge\sup_{\theta\in\Theta}\hat M_n(\beta(\theta))
    -\op(n^{-4/3})\).
\end{enumerate}
\end{assumption}

Condition \eqref{asm--ms-two-scale-density} requires the cell's own shock to
move the covariates given the row and column shocks, so that these shocks do
not pin down \(X\).
It therefore excludes covariates, or linear combinations of them, that are
functions of the row and column shocks alone, such as a distance between the
row and column units.
It holds, for instance, if \(X_{(i_1,i_2)}=\mu(U_{(i_1,0)},U_{(0,i_2)})
+h(U_{(i_1,i_2)})\), where \(h(U_{(1,1)})\) has a density bounded by
\(C(1+\|x\|)^{-d-\eta^\circ}\) and
\(\E\|\mu(\bm U^\circ)\|^{2(d+\eta^\circ)}<\infty\): as
\(1+\|x\|\le(1+\|x-\mu\|)(1+\|\mu\|)\), one can take
\(\bar p(\bm U^\circ)=C\{1+\|\mu(\bm U^\circ)\|\}^{d+\eta^\circ}\).
Condition \eqref{asm--ms-two-scale-density} makes the density bound of
\Cref{lem--max-score-local-envelope} hold given the row and column shocks,
and so yields the modulus bound in
\Cref{asm--two-scale}\eqref{asm--two-scale-modulus}.
Condition \eqref{asm--ms-two-scale-rates} is
\Cref{asm--two-scale}\eqref{asm--two-scale-rates} for maximum score.
As \(n^{-4/3}\le n^{-1}\), its tolerance is smaller than that of
\Cref{asm--X-ac-smooth-m-diff}; by the argument given after that
assumption, its approximate maximiser condition holds if \(\hat\beta_n\)
maximises \(\hat M_n\) over \(\calB\) to within \(\op(n^{-4/3})\) and either
\(\calB=\sphere^{d-1}\) or \(\beta_0\) lies in the relative interior of
\(\calB\).
A maximiser of Manski's criterion qualifies: under
\eqref{asm--ms-two-scale-density}, given all row and column shocks, the
\(X_{\i}\) are independent with Lebesgue densities, so almost surely no \(d\)
of them lie in a hyperplane through the origin, and the tolerance
\(\sup_{\beta\in\calB}\hatE_n[\ind\{X_{\i}'\beta=0\}]\) noted after
\eqref{eqn--max-score-estimator-def} is at most
\((d-1)/|\calI_n|=o(n^{-4/3})\).
Its remaining requirement, \(N_{n,1}\sim N_{n,2}\sim n\), restricts the array
rather than the model, as in \Cref{sec--two-scale}.

\begin{proposition}
\label{prp--max-score-asymp-normal}
Let \Cref{asm--exchangeable,asm--ms-consistency,asm--X-ac-smooth-m-diff}
hold.
\begin{enumerate}[(a)]
  \item If \(n/N_{n,k}\to\lambda_k\in[0,\infty)\) for each
    \(k\in\{1,\dots,K\}\), then, in the chart \eqref{eqn--beta-local},
    \(n^{1/2}\hat\theta_n\leadsto N(0,V)\), where
    \(V=A_0^{-1}\Sigma A_0^{-1}\), \(\Sigma=\sum_{k=1}^K\lambda_k\E[\Delta_k
    \Delta_k']\), \(A_0\succ0\) is given in \eqref{eqn--A0-ms} and
    \(\Delta_k\) in \eqref{eqn--P4-Delta-ms}.
    If the weights satisfy \Cref{asm--boot-wgt} and \(\tilde A_n\pto A_0\),
    the modified bootstrap \eqref{eqn--m-est-boot-def} in this chart satisfies
    \(n^{1/2}(\tilde\theta^*_n-\hat\theta_n)\dtop N(0,V)\) conditionally
    given \(\bm W_n\).
  \item If \Cref{asm--ms-two-scale} holds, then, for a fixed DGP (a constant
    sequence \(\{P_n\}\)) and in the chart \eqref{eqn--beta-local}, once
    \(\hat\theta_n\) is redefined as \(0\) on the event
    \(\{\hat\beta_n'\beta_0\le0\}\), whose probability tends to zero,
    \Cref{asm--two-scale} holds with \(\eta_0=1/2\) and
    \[
      \Lambda(v)=\int_{\R^{d-1}}|v't|\,f_X(B_0t)\,\d t,\qquad v\in\R^{d-1}.
    \]
    If, in addition, the weights are of the multiplier type in
    \Cref{asm--boot-wgt}, \(\tilde A_n\pto A_0\) and
    \(\tilde\zeta^*_n=\op(n^{-4/3})\), then \Cref{thm--two-scale-boot} holds
    in this chart.
    If, moreover, \(\Delta_1=\Delta_2=0\) almost surely, then
    \[
      n^{2/3}\hat\theta_n\leadsto\calZ=\argmax_{v\in\R^{d-1}}
      \Bigl\{\mathbb W(v)-\tfrac12v'A_0v\Bigr\},
    \]
    where \(\mathbb W\) is the Gaussian process of \Cref{sec--two-scale} with
    this \(\Lambda\); for \(d=2\),
    \[
      \calZ\overset d=\{2\Lambda(1)^{1/2}/A_0\}^{2/3}\calZ_{\rm Ch}.
    \]
    Under the bootstrap conditions above, also
    \(n^{2/3}(\tilde\theta^*_n-\hat\theta_n)\dtop2^{1/3}\calZ\) conditionally
    given \(\bm W_n\).
\end{enumerate}
\end{proposition}

Parts (a) and (b) are proven in
\Cref{sec--prf--thm-max-score-asymp-normal,sec--prf--prp--ms-two-scale}; the
bootstrap claims follow from \Cref{thm--bootstrap,thm--two-scale-boot}, whose
remaining conditions those proofs verify.
The proof of part (a) verifies the conditions of
\Cref{thm--m-est-non-smooth} for maximum score, after a redefinition of
\(\hat\theta_n\) on an event whose probability tends to zero;
\Cref{asm--m-est-general}\eqref{asm--m-est-meas} holds by the paragraph after
that assumption.
Its two demanding steps are stochastic differentiability, which draws on the
analytic techniques for maximum score of \citet{kim1990cube}, and the local
bounds on higher-order projections in
\Cref{asm--m-est-general}\eqref{asm--m-est-stoch-eq}, which rest on a
localisation argument for empirical processes based on the maximal inequality
of \citet[Theorem~4]{chen2026cross}, applied to classes with shrinking
envelopes.

Part (b) concerns two-way clustering: its first claim makes the theory of
\Cref{sec--two-scale} available for maximum score, and its last concerns the
degenerate case of that section, since
\(\Delta_1=\Delta_2=0\) almost surely if and only if
\(\Sigma_{n,1}=\Sigma_{n,2}=0\), that is, \(\Omega_n=0\), so that
\(\sigma_n=n^{-2/3}\) and \(c_n=1\).
By \eqref{eqn--P4-Delta-ms}, \(\Delta_k=0\) almost surely if, for almost every
value \(\bm u\) of \(\bm U_{\e_k}\), \(g_{0,\e_k}(x,\bm u)=0\) at every \(x\)
with \(x'\beta_0=0\) and \(p_{\e_k}(x\mid\bm u)>0\), so that the cluster shock
leaves the conditional median of the error at zero on the hyperplane
\(\{x:x'\beta_0=0\}\).
This holds, for instance, if \(\E[Y\mid X,\bm U_{\e_k}]=\E[Y\mid X]\) almost
surely (\Cref{sec--prf--prp--ms-two-scale}), as when the error is independent
of the cluster shocks given the covariates: these shocks then affect the
covariates but not the conditional law of the response given the
covariates.
In the degenerate case, the estimator converges at the rate \(n^{2/3}\) to
the limit that \citet[Main Theorem~1.1]{kim1990cube} give for
\(|\calI_n|\) independent observations from the law of \(W\).
In the notation of their Example~6.4 (pp.~215--216), \(\Lambda(v)=L(B_0v)\) and
\(v'A_0v=Q(B_0v)\).\footnote{Their display of the limit on p.~216 has
\(-Q(\theta)\) where their Main Theorem~1.1 gives \(-\frac12Q(\theta)\), as
their second derivative of \(\Gamma\) on p.~215 shows.}
In this case, the dependence across cells changes neither the rate, as a
function of the number of cells, nor the limit.
The function \(\Lambda\) depends only on the density of \(X\) on the
hyperplane: \(\Lambda(v)=f_{X'\beta_0}(0)\,\E[|v'B_0'X|\mid X'\beta_0=0]\),
with \(f_{X'\beta_0}\) the continuous density of \(X'\beta_0\) and the
conditional expectation under the conditional density
\(f_X(B_0t)/f_{X'\beta_0}(0)\) of \(B_0'X\) given \(X'\beta_0=0\)
(\Cref{sec--prf--prp--ms-two-scale}).
For centred isotropic Gaussian \(X\), \(\Lambda(v)=\|v\|/\pi\).

\Cref{prp--max-score-asymp-normal}(b) is stated for a fixed law, so the
uniform size control of \Cref{thm--size} calls for a drifting sequence, such
as the following family.

\begin{example}[A nearly degenerate sequence]
Let \(\{P_n\}\) be the Gaussian design of \Cref{sec--simulations}, with
\(K=2\), \(d=2\) and \(N_{n,1}=N_{n,2}=n\), in which the covariate scales
\(s_{x1},s_{x2},s_{ex}\) and the error scale \(\sigma_\epsilon\) are held
fixed while
\[
  s_u=s_v=\rho_n\downarrow0,\qquad s_\nu^2=\sigma_\epsilon^2-2\rho_n^2,
\]
so that \(\epsilon_{\i}\sim N(0,\sigma_\epsilon^2)\) for every \(n\) and the
drift falls on the dependence of the error on the row and column shocks
alone.
Then \(\theta_0=\pi/4\), \(A_0=(\sigma_X/\sigma_\epsilon)/\pi\) and
\(\Lambda(1)=1/\pi\) are the same for every \(n\), while
\(\Sigma_{n,k}\asymp\rho_n^2\) by the closed form in
\Cref{sec--simulations}; as \(\theta\) is scalar, \(\sigma_{n,a}=\sigma_n\),
\(c_{n,a}=c_n\) and \(c_n^2=(1+n^{1/3}\Sigma_n/A_0^2)^{-1}\) with
\(\Sigma_n=\Sigma_{n,1}+\Sigma_{n,2}\).
Given \(c\in(0,1)\), that closed form determines, for all sufficiently large
\(n\), a unique \(\rho_n\asymp n^{-1/6}\) solving
\(n^{1/3}\Sigma_n=A_0^2(1-c^2)/c^2\), and then \(c_n=c\) at those \(n\); the
endpoint \(c=1\) is the degenerate case \(\rho_n\equiv0\) of
\Cref{prp--max-score-asymp-normal}(b), which a fixed law with \(\Sigma=0\)
also delivers.
Along \(\{P_n\}\), \Cref{asm--two-scale} holds with the uniformity of
\Cref{sec--prf--two-scale}, so, with multiplier weights,
\(\tilde A_n\pto A_0\) and \(\tilde\zeta^*_n=\op(n^{-4/3})\),
\Cref{thm--two-scale,thm--two-scale-boot,thm--size} apply; in particular the
symmetric intervals have limit coverage \(\Pr(|T_{c,a}|\le q_{c,a})>1-\alpha\)
at this \(c\), a case that no fixed law delivers, since for a fixed law
\(c_n\to0\) when \(\Sigma\ne0\) and \(c_n=1\) when \(\Sigma=0\).
Two features of the design give the uniformity.
First, \eqref{eqn--P4-Delta-ms} integrates over \(\{x:x'\beta_0=0\}\), on
which \(g_{0,\e_1}(x,\bm u)=2\Phi\{-u/(\sigma_\epsilon^2-\rho_n^2)^{1/2}\}-1\)
does not depend on \(x\), with \(\Phi\) the standard normal distribution
function and \(u\) the row shock of the error; so \(\Delta_1^{(n)}\) is that
factor times \(\int t\,p_{\e_1}(B_0t\mid\bm u)\,\d t\), which does not depend
on \(n\), and \(\rho_n^{-1}\Delta_k^{(n)}\) converges in \(L_2\) to a fixed
square-integrable variable.
This gives the uniform square integrability of the normalised scores.
Second, the conditional densities and regression functions of the design are
smooth in \((\theta,\rho)\) on a compact neighbourhood of \((\theta_0,0)\),
so the dominating functions of
\Cref{asm--X-ac-smooth-m-diff}\eqref{asm--X-ac-smooth-m-diff-m-dominance} and
\eqref{asm--X-ac-smooth-m-diff-primitive-density}, the remainder in
\Cref{asm--m-est-consistent-pop-quad-decay}\eqref{asm--m-est-smooth-pop} and
the limit in \Cref{asm--two-scale}\eqref{asm--two-scale-cov} can be taken
free of \(n\); the envelope in
\Cref{asm--m-est-general}\eqref{asm--m-est-stoch-eq} is one and the classes
are VC-subgraph of dimension at most three for every \(n\).
Identification does not drift either: \(g_0(x)\ne0\) if and only if
\(x'\beta_0\ne0\), so the probability in
\Cref{asm--ms-consistency}\eqref{asm--index-sgn-disagree-pos-pr} is the angle
between \(b\) and \(\beta_0\) divided by \(\pi\).
\(\lozenge\)
\end{example}

When \(\Sigma\succ0\), \Cref{cor--boot-ci} turns the conditional limit of
\Cref{prp--max-score-asymp-normal}(a) into asymptotically exact symmetric
intervals.
Inference thus needs no estimate of \(\Sigma\), which depends on various
unknown quantities.

Under the conditions of \Cref{prp--max-score-asymp-normal}(b),
\Cref{thm--size} also applies in this chart, for every unit vector \(a\) with
\(\liminf_n\sigma_{n,a}/\sigma_n>0\).
In the degenerate case \(\sigma_{n,a}=\sigma_n=n^{-2/3}\) for every unit vector \(a\),
so that condition holds for every \(a\), and \Cref{thm--size}(c) with
\(c=1\) gives the symmetric interval for \(a'\theta_0\) asymptotic coverage
\(\Pr(|a'\calZ|\le2^{1/3}q_a)>1-\alpha\), where \(q_a\) is the
\((1-\alpha)\)-quantile of \(|a'\calZ|\): the inflation of the cube-root
term by \(2^{1/3}\) makes the interval conservative.

\paragraph{Charts.}
The chart \eqref{eqn--beta-local} involves the unknown \(\beta_0\), so
computation uses a fixed chart.
For \(d=2\) we use the angle chart \(\beta(\theta)=(\cos\theta,\sin\theta)'\),
with \(\theta\) taken modulo \(2\pi\) and angle differences wrapped to
\((-\pi,\pi]\); in the general theory, \(\Theta=(\theta_0-\pi,\theta_0+\pi]\).
For \(d>2\) we use the normalisation chart \(\beta(\theta)=(1,\theta')'\)
on \(\Theta=\R^{d-1}\), which requires \(\beta_{0,1}>0\) and uses the scale invariance
\(m_{t\beta}=m_\beta\) for \(t>0\), as in \citet{horowitz1992smoothed},
\citet{abrevaya2005bootstrap} and \citet{cattaneo2020bootstrap}.
In either chart, \(\hat\theta_n\) is the coordinate of \(\hat\beta_n\): its
angle in \(\Theta\) if \(d=2\); and, if \(d>2\),
\((\hat\beta_{n,2},\dots,\hat\beta_{n,d})'/\hat\beta_{n,1}\) when
\(\hat\beta_{n,1}>0\), and \(0\) on the event \(\{\hat\beta_{n,1}\le0\}\),
whose probability tends to zero by \Cref{lem--max-score-consistency}.
Then \(\hat\theta_n\pto\theta_0\).
Moreover, as each \(\beta(\theta)\) is a positive multiple of a point of
\(\sphere^{d-1}\), \(\hat\theta_n\) satisfies
\eqref{eqn--m-est-def} in the chart whenever
\(\sup_{\beta\in\sphere^{d-1}}\hat M_n(\beta)-\hat M_n(\hat\beta_n)=\op(n^{-1})\),
as when \(\calB=\sphere^{d-1}\) and \(\hat\beta_n\) maximises \(\hat M_n\) to
within \(\op(n^{-1})\).
This condition also implies the approximate maximiser condition of
\Cref{asm--X-ac-smooth-m-diff}.
Near \(\theta_0\), where \(\beta(\theta_0)\) is a positive multiple of
\(\beta_0\), each chart reparametrises \eqref{eqn--beta-local} smoothly
through \(\vartheta(\theta):=B_0'\beta(\theta)/\|\beta(\theta)\|\), with
\(D\vartheta:=D\vartheta(\theta_0)\) invertible.
\Cref{asm--m-est-consistent-pop-quad-decay,asm--m-est-general} are invariant
under such reparametrisations, with \(A_0\) and \(\Delta_k\) replaced by
\(D\vartheta'A_0D\vartheta\) and \(D\vartheta'\Delta_k\).
The classes \(\{m_\theta\}\) remain subclasses of the VC-subgraph class
\(\{m_\beta:\beta\in\R^d\}\); \(\vartheta\) is bi-Lipschitz near
\(\theta_0\), so the conditions on shrinking balls in
\Cref{asm--m-est-general} carry over; and its curvature, which does not enter
\(A_0\) because \(\nabla M_0(\theta_0)=0\), adds to
\((\hatE_{n,k}-\E)[r_{k,\theta}]\) only a negligible term, bounded by
\(C\|\theta-\theta_0\|^2\|\hatE_{n,k}[\Delta_k]\|\) with
\(\E\|\hatE_{n,k}[\Delta_k]\|\le(\E\|\Delta_k\|^2/N_{n,k})^{1/2}\to0\).
Hence, under the condition above, \Cref{thm--bootstrap,cor--boot-ci} also
apply in these charts, in which \eqref{eqn--m-est-boot-def} can be computed,
with \(\theta_0\), \(\hat\theta_n\), \(A_0\), \(\Delta_k\), \(\Sigma\) and
\(V\) referring to the chart in use; in particular,
\(n^{1/2}(\hat\theta_n-\theta_0)\leadsto N(0,V)\) there, by
\Cref{thm--m-est-non-smooth}.
In the angle chart, the quadratic in \eqref{eqn--m-est-boot-def} is computed
with the wrapped difference.
On \(\Theta\) it differs from the unwrapped one
only within \(|\hat\theta_n-\theta_0|\) of \(\theta_0\pm\pi\), where both
criteria lie below \(\tilde M^*_n(\hat\theta_n)\) with probability tending to
one, by \Cref{lem--boot-consistent}(a) and \(\tilde A_n\pto A_0>0\).
Hence the computed draws satisfy \eqref{eqn--m-est-boot-def} on \(\Theta\), after a
redefinition of the tolerance on an event whose probability tends to
zero.
The estimation of \(A_0\), which \eqref{eqn--m-est-boot-def} needs, is
discussed at the end of this section.

For \(K=2\) and a fixed DGP, under \Cref{asm--ms-two-scale} with
\(\op(n^{-4/3})\) in place of \(\op(n^{-1})\) in the condition above,
\Cref{asm--two-scale} also holds in these charts, with \(\eta_0\) reduced if
necessary and \(\Lambda\) replaced by \(v\mapsto\Lambda(D\vartheta v)\).
Condition (i) holds as \Cref{asm--m-est-general}\eqref{asm--m-est-stoch-diff}
does.
In (ii), the classes remain VC-subgraph with Borel envelopes, and the modulus
bound holds by the argument of \Cref{sec--prf--prp--ms-two-scale}, with
\(\dot m\) multiplied by a Lipschitz constant of
\(\theta\mapsto\beta(\theta)/\|\beta(\theta)\|\) near \(\theta_0\).
Condition (iii) holds because
\(\vartheta(\theta+\delta v)-\vartheta(\theta)=\delta D\vartheta v+o(\delta)\),
uniformly in \(v\) on compact sets, as \(\theta\to\theta_0\) and
\(\delta\downarrow0\), whilst the limit in (iii) for the chart
\eqref{eqn--beta-local} is uniform on compact sets and \(\Lambda\) is
continuous.
In the degenerate case, as \(\vartheta(\hat\theta_n)=B_0'\hat\beta_n\) with
probability tending to one, the delta method gives
\(n^{2/3}(\hat\theta_n-\theta_0)\leadsto(D\vartheta)^{-1}\calZ\) there, which is
the \(\calZ\) of the chart in use.
The conclusions above on the modified bootstrap hold there too, with \(\tilde A_n\pto D\vartheta'A_0D\vartheta\).

\paragraph{Hessian estimation.}
The modified bootstrap \eqref{eqn--m-est-boot-def} needs an estimator
\(\tilde A_n\) of \(A_0\), which for maximum score cannot be obtained by
differentiating the criterion.
As in \citet[Section~3.1]{cattaneo2020bootstrap}, we estimate \(A_0\) by
finite differences of \(\hat M_n\) (on numerical derivatives in extremum
estimation, see also \citealp{hong2015extremum}).
We use the 4-point difference on the nodes \(\pm h,\pm2h\), with the weights
given by the algorithm of \citet[Section~2]{fornberg1988generation},
\begin{equation}
  \label{eqn--Dhat-4pt}
  \hat D_h(u):=-\frac{\hat M_n(\hat\theta_n+2hu)-\hat
  M_n(\hat\theta_n+hu)-\hat M_n(\hat\theta_n-hu)+\hat
  M_n(\hat\theta_n-2hu)}{3h^2},
\end{equation}
which estimates \(u'A_0u\) for a direction \(u\), is exact for quadratics and
cancels odd-order terms around \(\hat\theta_n\).
Since \(4u'Av=(u+v)'A(u+v)-(u-v)'A(u-v)\), we set, with \(e_j\) the \(j\)th
standard basis vector of \(\R^{d-1}\),
\begin{equation}
  \label{eqn--A-4pt}
  \tilde A^{\rm 4pt}_{n,jj}:=\hat D_h(e_j),\qquad
  \tilde A^{\rm 4pt}_{n,jl}:=\tfrac14\{\hat D_h(e_j+e_l)-\hat
  D_h(e_j-e_l)\}\quad(j\ne l),
\end{equation}
and take \(\tilde A_n:=\tilde A^{\rm 4pt}_n\).
If \(\tilde A_n\) is not positive definite, we double \(h\) and recompute, at
most a fixed number of times, and report a failure if it is still not positive
definite.

The step must shrink fast enough to control the bias of
\eqref{eqn--Dhat-4pt}, which under further smoothness of \(M_0\) is of order
\(h^2\), and slowly enough to control its noise, which heuristically is of
order \((|\calI_n|h^3)^{-1/2}\) from the top-level component of
\(\hat M_n\) --- the step condition of
\citet[Lemma~1]{cattaneo2020bootstrap} with \(|\calI_n|\) as the sample
size --- plus an error of order \(n^{-1/2}\) from the lower-level components
that does not depend on the step to first order.
Minimising the step-dependent part of the resulting approximate mean squared
error, \(h^4\mathfrak B(u)^2+\mathfrak V(u)/(|\calI_n|h^3)\) with
\[
  \mathfrak B(u):=-\tfrac5{12}\,\partial_t^4M_0(\theta_0+tu)\big|_{t=0},
  \qquad
  \mathfrak V(u):=\tfrac29\,\Lambda(u),
\]
gives \(h=\{3\mathfrak V(u)/(4\mathfrak B(u)^2)\}^{1/7}|\calI_n|^{-1/7}\),
the rule stated after \citet[Lemma~2]{cattaneo2020bootstrap} with
\(|\calI_n|\) for their \(r_n^3\) and with the constants of the 4-point
stencil.
Here \(\Lambda\) is the function of \Cref{prp--max-score-asymp-normal}(b) for
the chart \eqref{eqn--beta-local}, replaced by \(v\mapsto\Lambda(D\vartheta
v)\) in the charts above.
In the isotropic Gaussian designs of \Cref{sec--simulations}, let
\(\sigma_X\) and \(\sigma_\epsilon\) be the standard deviations of each
coordinate of \(X\) and of \(\epsilon\), and
\(\varsigma:=\sigma_X/\sigma_\epsilon\).
Then, in the angle chart,
\(\partial_t^4M_0(\theta_0+t)|_{t=0}=A_0(1+3\varsigma^2)\) and
\(\Lambda(1)=1/\pi\), so
\(h=\{24\pi/(25\varsigma^2(1+3\varsigma^2)^2)\}^{1/7}|\calI_n|^{-1/7}\), or
\(0.79\,|\calI_n|^{-1/7}\) when \(\varsigma=1\).
For \(d>2\), one step can be used for all entries.

\begin{proposition}[Consistency of the curvature estimator]
Let \Cref{asm--exchangeable} hold with \(K=2\); let
\(\{m_\theta:\theta\in\Theta\}\) be uniformly bounded and VC-type; let
\(\theta_0\) be the unique maximiser of \(M_0\) over \(\Theta\), with
\(A_0=-\nabla^2M_0(\theta_0)\succ0\) and \(M_0\) five times continuously
differentiable on a neighbourhood \(\Theta_0\) of \(\theta_0\); and let
\[
  \operatorname{Cov}\bigl\{H_{n,(1,1)}[m_\theta],H_{n,(1,1)}[m_{\theta'}]\bigr\}
  =|\calI_n|^{-1}\bigl\{A(\theta,\theta')
  +B(\theta,\theta')\|\theta-\theta'\|\bigr\},
  \qquad\theta,\theta'\in\Theta_0,
\]
with \(A\) and \(B\) having bounded derivatives of order five on
\(\Theta_0\times\Theta_0\). If the step satisfies
\[
  h\to0,\qquad \|\hat\theta_n-\theta_0\|/h\pto0,\qquad
  \frac{|\calI_n|h^3}{\log|\calI_n|}\to\infty,
\]
then \(\tilde A^{\rm 4pt}_n\pto A_0\), and likewise for a random step
\(\hat h_n\) with \(\hat h_n/h\pto c_h\) for some constant \(c_h>0\).
\end{proposition}

This is \citet[Theorem~1]{chiang2026curvature}, whose local polynomial
curvature estimators contain \eqref{eqn--Dhat-4pt} as the interpolating
stencil with derivative order \(p=2\), degree \(q=3\) and nodes
\(\{\pm\frac12,\pm1\}\) on a window of radius \(2h\); the exponent of
\(h\) in the third condition is \(2p-\min(\gamma,2)\) at the exponent
\(\gamma=1\) of \(\|\theta-\theta'\|\) above, and the clause on \(\hat h_n\) covers the doubling
convention, which applies with probability tending to zero.
The rule \(h\asymp|\calI_n|^{-1/7}\) meets the three conditions in both
regimes: \(|\calI_n|h^3\asymp|\calI_n|^{4/7}\), while \(h\asymp n^{-2/7}\)
exceeds the order of \(\hat\theta_n-\theta_0\), which is \(n^{-1/2}\) under
\Cref{prp--max-score-asymp-normal}(a) and \(n^{-2/3}\) in the degenerate case.
That paper also treats the error decomposition, the admissibility of the
window rule and the joint law of \(\tilde A_n\) with \(\hat\theta_n\).
None of this is needed here, as
\Cref{thm--bootstrap,cor--boot-ci,thm--two-scale-boot} use \(\tilde A_n\) only
through \(\tilde A_n\pto A_0\).
\Cref{sec--simulations} reports its behaviour in our designs.

\section{Monte Carlo Study}
\label{sec--simulations}

To check \Cref{thm--size}(c)--(d) and \Cref{rem--plain-bootstrap}, we
simulate the main procedure of \Cref{sec--m-est} and the plain bootstrap.

\paragraph{Designs.}
Let \(K=2\), \(d=2\) (so \(\theta\) is scalar) and \(N_{n,1}=N_{n,2}=n\).
For \(\bm i\in\calI_n\),
\[
  X_{\bm i}=\alpha_{i_1}+\gamma_{i_2}+e_{\bm i}\in\R^2,\qquad
  \epsilon_{\bm i}=u_{i_1}+v_{i_2}+\nu_{\bm i},\qquad
  Y_{\bm i}=2\ind\{X_{\bm i}'\beta_0\ge\epsilon_{\bm i}\}-1,
\]
with \(\beta_0=(\cos\theta_0,\sin\theta_0)'\), \(\theta_0=\pi/4\), and
\(\theta\) estimated in the angle chart \(\beta(\theta)=(\cos\theta,\sin\theta)'\).
All components are independent centred Gaussians: \(\alpha,\gamma,e\) are
isotropic bivariate with standard deviations \(s_{x1},s_{x2},s_{ex}\), and
\(u,v,\nu\) are scalar with standard deviations \(s_u,s_v,s_\nu\); let
\(\sigma_X^2=s_{x1}^2+s_{x2}^2+s_{ex}^2\) and
\(\sigma_\epsilon^2=s_u^2+s_v^2+s_\nu^2\).
Then \(A_0=(\sigma_X/\sigma_\epsilon)/\pi\) and \(\Lambda(1)=1/\pi\), so
\(\calZ\overset d=(2\pi^{-1/2}/A_0)^{2/3}\calZ_{\rm Ch}\) and
\(\var(\calZ)^{1/2}=(2\pi^{-1/2}/A_0)^{2/3}\times0.5134\)
\citep{groeneboom2001computing}.
Moreover \(V=\Sigma/A_0^2\) with \(\Sigma=\Sigma_{n,1}+\Sigma_{n,2}\), where
\(\Sigma_{n,k}\) does not depend on \(n\) here and
\[
  \Sigma_{n,1}=\frac{\arcsin(s_u^2/\sigma_\epsilon^2)}{\pi^2}\,
  \frac{s_{x1}^2}{\bar s_{x1}(\bar s_{x1}^2+2s_{x1}^2)^{1/2}},
\]
\(\Sigma_{n,2}\) being the same with \((s_v,s_{x2},\bar s_{x2})\) in place of
\((s_u,s_{x1},\bar s_{x1})\), for \(\bar s_{x1}^2=s_{x2}^2+s_{ex}^2\) and
\(\bar s_{x2}^2=s_{x1}^2+s_{ex}^2\), the variances of a coordinate of \(X\)
within a row and within a column.
Thus \(\Sigma_{n,k}>0\) if and only if both the \(X\)-shock and the
\(\epsilon\)-shock of way \(k\) are non-degenerate, so that the shock moves
the conditional median of \(\epsilon\) on \(\{x:x'\beta_0=0\}\).
This is the maximum-score form of the condition \(\Sigma\succ0\) of
\Cref{prp--max-score-asymp-normal}(a).

What matters at a given \(n\) is not the limiting regime but the cube-root
share of the approximating variance in \Cref{thm--two-scale},
\[
  \omega_n:=\frac{\var(\calZ)n^{-4/3}}{V/n+\var(\calZ)n^{-4/3}}
  =\frac{\var(\calZ)n^{-1/3}}{V+\var(\calZ)n^{-1/3}}
  =\frac{c_n^2\var(\calZ)}{1-c_n^2+c_n^2\var(\calZ)},
\]
since \Cref{thm--size}(d) approximates the coverage by the limit coverage of
\Cref{thm--size}(c) at \(c=c_n\).
\Cref{tab--designs} therefore groups the designs by \(\omega_n\) over our
\(n\).
D3 is degenerate: \(s_u=s_v=0\), so \(\Sigma=0\) and \(\omega_n=1\).
D1, with equal scales, is non-degenerate but behaves as nearly degenerate at
these \(n\), with \(\omega_n\) between 0.45 and 0.67, so that the two terms
are of the same order.
DS (with strong cluster shocks), DM and DL are non-degenerate, with
\(\omega_n\le0.32\); DM and DL are the designs with the smallest \(\omega_n\)
among \(400{,}000\) random draws of the scales at \(n=50\) and \(n=200\),
subject to bounds on the approximating standard deviation
\(\mathrm{sd}_n:=\{(V+\var(\calZ)n^{-1/3})/n\}^{1/2}\) and on
\(\bar s_{xk}/s_{xk}\) relative to it.
Their small \(\omega_n\) rests on a small \(s_{ex}\), which inflates the
factor multiplying \(\delta\) in
\Cref{asm--two-scale}\eqref{asm--two-scale-modulus}, whose second moment is of
order \(1/s_{ex}\): the row--column interaction in \(H_{n,(1,1)}\) need not be
negligible unless \(n^{2/3}s_{ex}\) is large, so we expect the first-order
remainder \(\hat\theta_n-\theta_0-A_0^{-1}\bar\psi_n\) to exceed what
\Cref{thm--two-scale} predicts in DM and DL, and to a lesser extent in DS at
small \(n\).

\begin{table}[htbp]
\centering\small
\caption{Designs and cube-root shares \(\omega_n\). The headings group the
designs by \(\omega_n\) at these \(n\); as \(n\to\infty\), every design except
D3 is non-degenerate in the sense of \Cref{sec--two-scale}.}
\label{tab--designs}
\begin{tabular}{lllcccccc}
\hline
Design & \((s_{x1},s_{x2},s_{ex},s_u,s_v,s_\nu)\) & \(\sqrt V\) & \(\var(\calZ)^{1/2}\)
& \(\omega_{25}\) & \(\omega_{50}\) & \(\omega_{100}\) & \(\omega_{200}\) & \(\omega_{400}\)\\
\hline
\multicolumn{9}{l}{\emph{Degenerate}}\\
D3 & (1, 1, 1, 0, 0, 1) & 0 & 0.828 & 1 & 1 & 1 & 1 & 1\\
\multicolumn{9}{l}{\emph{Nearly degenerate}}\\
D1 & (1, 1, 1, 1, 1, 1) & 0.490 & 1.194 & 0.67 & 0.62 & 0.56 & 0.50 & 0.45\\
\multicolumn{9}{l}{\emph{Non-degenerate}}\\
DS & (1, .5, .25, 1, .5, .25) & 1.028 & 1.194 & 0.32 & 0.27 & 0.23 & 0.19 & 0.15\\
DM & (1, .37, .11, .53, .02, .01) & 0.809 & 0.747 & 0.23 & 0.19 & 0.16 & 0.13 & 0.10\\
DL & (1, .354, .012, .99, .026, .028) & 1.600 & 1.140 & 0.15 & 0.12 & 0.10 & 0.08 & 0.06\\
\hline
\end{tabular}
\end{table}

\paragraph{Procedures.}
We use bounded weights satisfying \Cref{asm--boot-wgt}: multiplier weights
with \(\xi_{k,i}\sim\mathrm{Uniform}[-\sqrt3,\sqrt3]\) and pigeonhole (PH)
weights with the mixture given after that assumption.
The main procedure is the modified bootstrap \eqref{eqn--m-est-boot-def} with
multiplier weights, the symmetric intervals of \Cref{cor--boot-ci} and the
4-point Hessian \(\tilde A^{\rm 4pt}_n\) of \eqref{eqn--A-4pt} at the step
given for these designs in \Cref{sec--max-score}, evaluated at the true
\(\varsigma=\sigma_X/\sigma_\epsilon\) (an infeasible choice) and doubled
while \(\tilde A^{\rm 4pt}_n\le0\).
We compare it with the same procedure using the true \(A_0\), and with PH
weights; with the plain bootstrap \eqref{eqn--m-est-boot-def-basic}, in basic
and in symmetric form; with the normal interval
\(\hat\theta_n\pm z_{1-\alpha/2}(V/n)^{1/2}\) at the true \(V\); and with the
limit coverage \(\Pr(|T_{c,1}|\le q_{c,1})\) of \Cref{thm--size}(c) at
\(c=c_n\), by quadrature against Chernoff's density
\citep{groeneboom2001computing}.

\paragraph{Computation.}
In the angle chart, \(m_{\beta(\theta)}(W_{\i})=1\) if and only if \(\theta\)
lies in the closed arc of half-width \(\pi/2\) centred at the angle of
\(Z_{\i}=Y_{\i}X_{\i}\), so the \(2|\calI_n|\) arc endpoints cut the circle
into open pieces on which every weighted criterion
\(\sum_{\i}v_{\i}m_{\beta(\theta)}(W_{\i})\) is constant; almost surely the
endpoints are distinct and the criterion at an endpoint equals its value on an
adjacent piece, so suprema over \(\theta\) are maxima over pieces.
We therefore compute every estimator and bootstrap draw exactly, at a cost of
\(O(|\calI_n|)\) in cumulative sums after one \(O(|\calI_n|\log|\calI_n|)\)
sort.
Ties are common, so we take \(\hat\theta_n\) to be the midpoint of a
maximising piece chosen uniformly at random, independently of the data and the
weights; a fixed rule would shift \(\hat\theta_n\) systematically.
Then \(\hat\zeta_n=0\), and all conditional statements hold given \(\bm W_n\) and the tie-breaking draw.
For the modified bootstrap, \(\tilde\theta^*_n\) is the point closest to
\(\hat\theta_n\) of the piece maximising \(\hat M^*_n-\hat M_n-\frac12\tilde
A_n\delta^2\), with \(\delta\) the distance from \(\hat\theta_n\), so that
\(\tilde\zeta^*_n=O(|\calI_n|^{-1})\); the plain bootstrap maximises \(\hat
M^*_n\) over the same pieces, and we again take the midpoint.

\paragraph{Sample sizes and output.}
We take \(n\in\{25,50,100,200,400\}\) (625 to 160{,}000 cells), \(R=4{,}000\)
replications, \(B=999\) bootstrap draws and \(\alpha\in\{0.05,0.10\}\).
The Monte Carlo standard error of a coverage near 95\% is about 0.34 points.
\Cref{tab--sim-coverage} reports \(\alpha=0.05\); the replication files also
give \(\alpha=0.10\), which shows the same pattern, the achieved optimisation
tolerance, which never exceeded \(4.6|\calI_n|^{-1}\), and diagnostics for the
Hessian estimate, the first-order remainder and the centring of each
bootstrap.

\paragraph{Results.}
Even with the true \(V\), the normal interval covers only 73.7--86.1\% in the
nearly degenerate D1 and 89.6--92.7\% in the non-degenerate DS, so at these
sample sizes the Gaussian approximation alone does not give valid intervals.
In D3, D1 and DS, the Monte Carlo standard deviation of \(\hat\theta_n\) is
within 2\% of \(\mathrm{sd}_n\), and the main procedure covers within 0.7
points of the limit of \Cref{thm--size}(c).
As \Cref{thm--size}(d) predicts, it is conservative in the degenerate and nearly degenerate designs
(97.6--99.5\% in D3 and D1) and close to nominal in DS (95.4--96.6\%), with
median lengths 15--38\% and 6--9\% above \(2\times1.96\) Monte Carlo standard
deviations.
With the true \(A_0\) in place of \(\tilde A^{\rm 4pt}_n\), coverage is lower
in every cell of these designs, by up to 2.6 points.
The 4-point estimate is low on average, which widens the intervals, so part of
the agreement with the limit at small \(n\) comes from that error.

Pigeonhole weights add up to a further 0.9 points there, DS at \(n=25\) apart,
but lower coverage in DM at \(n\le50\) and in DL at \(n\le100\).
As \Cref{rem--plain-bootstrap} predicts, the plain bootstrap's basic intervals
undercover in every design (77.7--88.8\%, and 77.7--80.5\% in the degenerate
D3), while its symmetric intervals overcover in D3, D1 and DS
(96.7--99.7\%).
In DM and DL the first-order remainder is larger than \Cref{thm--two-scale}
predicts, as anticipated above, and the main procedure covers 92.7--95.0\%, up
to 3.2 points below the limit in DL at \(n=25\), the gap falling to 0.7--1.1
points at \(n=400\).

\begin{table}[htbp]
\centering\footnotesize
\setlength{\tabcolsep}{3pt}
\caption{Coverage (\%) of nominal 95\% intervals and median length relative
to \(2\times1.96\times\) the Monte Carlo standard deviation of \(\hat\theta_n\).}
\label{tab--sim-coverage}
\begin{tabular}{llcccccccccc}
\hline
& & & \multicolumn{7}{c}{Coverage} & \multicolumn{2}{c}{Length}\\
& & & & & \multicolumn{2}{c}{Main, mult.} & PH & \multicolumn{2}{c}{Plain, mult.}
& \multicolumn{2}{c}{Main, mult.}\\
Design & \(n\) & \(\omega_n\) & Limit & Normal & 4pt & \(A_0\) & 4pt
& basic & symm. & 4pt & \(A_0\)\\
\hline
\multicolumn{12}{l}{\emph{Degenerate}}\\
D3 & 25 & 1 & 98.8 & -- & 99.2 & 97.9 & 99.5 & 80.5 & 99.6 & 1.38 & 1.22\\
   & 50 & 1 & 98.8 & -- & 99.5 & 98.5 & 99.7 & 79.7 & 99.5 & 1.36 & 1.25\\
   & 100 & 1 & 98.8 & -- & 99.3 & 98.7 & 99.8 & 79.8 & 99.6 & 1.33 & 1.26\\
   & 200 & 1 & 98.8 & -- & 99.1 & 98.7 & 99.8 & 77.7 & 99.7 & 1.28 & 1.23\\
   & 400 & 1 & 98.8 & -- & 99.1 & 98.8 & 99.7 & 78.8 & 99.7 & 1.30 & 1.27\\
\multicolumn{12}{l}{\emph{Nearly degenerate}}\\
D1 & 25 & 0.67 & 98.0 & 73.7 & 98.2 & 96.7 & 98.9 & 80.8 & 99.2 & 1.23 & 1.11\\
   & 50 & 0.62 & 97.8 & 77.6 & 98.0 & 97.0 & 98.8 & 81.0 & 98.9 & 1.22 & 1.13\\
   & 100 & 0.56 & 97.6 & 80.5 & 98.3 & 97.4 & 99.0 & 81.5 & 98.7 & 1.20 & 1.14\\
   & 200 & 0.50 & 97.4 & 83.5 & 98.0 & 97.5 & 98.9 & 82.4 & 98.5 & 1.17 & 1.13\\
   & 400 & 0.45 & 97.2 & 86.1 & 97.6 & 97.0 & 98.4 & 83.0 & 98.4 & 1.15 & 1.12\\
\multicolumn{12}{l}{\emph{Non-degenerate}}\\
DS & 25 & 0.32 & 96.7 & 89.6 & 96.4 & 93.8 & 96.0 & 82.7 & 97.7 & 1.07 & 0.99\\
   & 50 & 0.27 & 96.5 & 91.4 & 96.5 & 95.2 & 96.7 & 83.7 & 97.2 & 1.09 & 1.04\\
   & 100 & 0.23 & 96.3 & 91.5 & 96.6 & 95.5 & 96.6 & 85.4 & 97.2 & 1.08 & 1.04\\
   & 200 & 0.19 & 96.1 & 92.5 & 96.2 & 95.2 & 96.5 & 87.1 & 96.7 & 1.06 & 1.04\\
   & 400 & 0.15 & 95.9 & 92.7 & 95.4 & 95.3 & 96.2 & 88.8 & 96.8 & 1.06 & 1.04\\
DM & 25 & 0.23 & 96.3 & 92.5 & 94.8 & 91.2 & 93.5 & 79.4 & 97.1 & 0.98 & 0.92\\
   & 50 & 0.19 & 96.1 & 93.1 & 94.8 & 92.9 & 94.5 & 80.0 & 97.0 & 0.99 & 0.95\\
   & 100 & 0.16 & 95.9 & 93.5 & 95.0 & 94.1 & 95.0 & 81.7 & 96.2 & 1.02 & 1.00\\
   & 200 & 0.13 & 95.8 & 93.3 & 94.6 & 94.3 & 94.9 & 83.8 & 95.4 & 1.01 & 0.99\\
   & 400 & 0.10 & 95.6 & 93.9 & 95.0 & 94.3 & 95.5 & 85.1 & 95.7 & 1.01 & 1.00\\
DL & 25 & 0.15 & 95.9 & 94.8 & 92.7 & 91.5 & 91.6 & 78.6 & 97.0 & 0.91 & 0.90\\
   & 50 & 0.12 & 95.7 & 94.5 & 93.0 & 91.9 & 92.3 & 77.9 & 96.3 & 0.93 & 0.93\\
   & 100 & 0.10 & 95.6 & 94.7 & 93.5 & 93.8 & 93.2 & 79.0 & 96.2 & 0.93 & 0.95\\
   & 200 & 0.08 & 95.5 & 94.5 & 93.2 & 94.3 & 93.5 & 81.5 & 94.8 & 0.95 & 0.97\\
   & 400 & 0.06 & 95.4 & 94.1 & 94.3 & 94.4 & 94.5 & 85.7 & 95.3 & 0.98 & 0.99\\
\hline
\end{tabular}

\smallskip
\begin{minipage}{0.95\linewidth}\footnotesize\raggedright
\textit{Notes.} Limit: \(\Pr(|T_{c,1}|\le q_{c,1})\) of \Cref{thm--size}(c) at
\(c=c_n\), multiplier weights. Normal: \(\hat\theta_n\pm1.96(V/n)^{1/2}\). 4pt,
\(A_0\): Hessian \(\tilde A^{\rm 4pt}_n\) or the true \(A_0\).
Symmetric intervals unless stated; basic:
\([\hat\theta_n-q^*_{0.975},\hat\theta_n-q^*_{0.025}]\), where \(q^*_p\) is the
\(p\)-quantile of the bootstrap law of \(\hat\theta^*_n-\hat\theta_n\).
\end{minipage}
\end{table}

\clearpage
\appendix

\section{\texorpdfstring{Proofs for \Cref{sec--m-est}}{Proofs for Section 2}}

\subsection{A general rate and linearisation theorem}

As before, let \(\Theta \subseteq \R^{d}\) for \(d \in \N\).
An extremum estimator is an approximate maximiser of a random criterion
\(M_{n}\):
\begin{equation}
  \theta_{n} \in \Theta \text{ is measurable} \quad \text{and} \quad
  M_{n} (\theta_{n}) \geq \sup_{\theta \in \Theta} M_{n}
  (\theta) - \zeta_{n} \quad \text{for} \quad 0 \leq \zeta_{n} = \op (1).
  \label{eqn--extremum-def}
\end{equation}
In this section, we let \(M_{n,0} (\theta)\) denote a (possibly non-stochastic)
`centring' function for \(M_{n} (\theta)\).
Corresponding to the estimator \(\theta_{n}\) is a \(\Theta\)-valued
target \(\theta_{n,0}\), which may be random: in the proof of
\Cref{thm--bootstrap}, \(\theta_{n,0}\) is the sample M-estimator and
\(\theta_{n}\) the bootstrap estimator.
In \Cref{asm--extremum-consistent-pop-quad-decay,asm--extremum-stoch-diff}
below and in \Cref{lem--extremum-stoch-diff-Rn-neg-rate}, \(\{\eta_n\}\)
denotes a non-stochastic positive sequence.

\begin{assumption}
\label{asm--extremum-consistent-pop-quad-decay}
\hfill
\begin{enumerate}[(i)]
  \item \label{asm--extremum-consistency}
    \(\|\theta_{n} - \theta_{n,0}\| \pto 0\) as \(n \to \infty\).
  \item \label{asm--extremum-pop-quad-decay-strong}
    There exist \(d \times d\) (real, possibly random)
    symmetric matrices \(A_{n,0}\) such that
    \begin{equation}
      \begin{gathered}
        \text{if } \eta_{n} \to 0, \text{ then } \phi_{n,0} (\eta_{n}) =
        \op^{\dagger} (1) \quad \text{where} \quad
        \phi_{n,0} (\eta_{n}) := \sup_{\theta \in \Theta: \|\theta -
        \theta_{n,0}\| \leq \eta_{n}} |R_{n,0} (\theta)|, \\
        R_{n,0} (\theta) = \frac{M_{n,0} (\theta) - M_{n,0} (\theta_{n,0}) +
        \frac{1}{2} (\theta - \theta_{n,0})' A_{n,0} (\theta -
        \theta_{n,0})}{\|\theta - \theta_{n,0}\|^{2}} \quad \text{if} \quad
        \theta \neq \theta_{n,0}, \\
        \text{and} \quad
        R_{n,0} (\theta_{n,0}) = 0.
      \end{gathered}
      \label{eqn--extremum-pop-quad-decay-strong}
    \end{equation}
    Furthermore, defining \(\underline{\lambda}_{n} = \lambda_{\min} (A_{n,0})\),
    \(\lim_{\varepsilon \to 0} \limsup_{n \to \infty} \Pr \{\underline{\lambda}_{n} <
    \varepsilon\} = 0\), that is, \((\underline\lambda_n\vee0)^{-1}=\Op(1)\), with
    \(1/0:=\infty\).
\end{enumerate}
\end{assumption}

\begin{assumption}
\label{asm--extremum-stoch-diff}
The centred oscillation \(M_{n} - M_{n,0}\) can be decomposed
into two parts:
\begin{equation}
  M_{n} (\theta) - M_{n,0} (\theta) = Z_{n,1} (\theta) + Z_{n,2} (\theta),
\end{equation}
and there exist a non-stochastic positive real sequence \(\{\delta_{n}\}\) and a
sequence of random \(\R^{d}\)-vectors \(\overline{\Delta}_{n}\) such that the
following hold.
\begin{enumerate}[(i)]
  \item \label{asm--extremum-stoch-diff-Delta-Op-delta}
    As \(n \to \infty\), \(\delta_{n} \to 0\) and \(\overline{\Delta}_{n} =
    \Op (\delta_{n})\).
  \item \label{asm--extremum-stoch-diff-phin1-small}
    If \(\eta_{n} \to 0\), then \(\phi_{n,1} (\eta_{n}) = \op^{\dagger} (1)\)
    for
    \begin{equation}
      \begin{gathered}
        \phi_{n,1} (\eta) = \sup_{\substack{\theta \in \Theta : \\ 0 < \|\theta
        - \theta_{n,0}\| \leq \eta}} \frac{|R_{n,1} (\theta)|}{\|\theta -
        \theta_{n,0}\| (\|\theta - \theta_{n,0}\| + \delta_{n})}, \\
        \text{and} \quad R_{n,1} (\theta) = Z_{n,1} (\theta) - Z_{n,1}
        (\theta_{n,0}) - (\theta - \theta_{n,0})' \overline{\Delta}_{n}.
      \end{gathered}
      \label{eqn--extremum-stoch-diff}
    \end{equation}
  \item \label{asm--extremum-stoch-diff-rem-local-asymp-equicont}
    \([\eta_{n} \to 0 \implies \phi_{n,2} (\eta_{n}) = \Op^{\dagger}
    (\delta_{n}^{2})]\) and \([\eta_{n} = O (\delta_{n}) \implies \phi_{n,2}
    (\eta_{n}) = \op^{\dagger} (\delta_{n}^{2})]\) for
    \begin{equation}
      \phi_{n,2} (\eta) = \sup_{\substack{\theta \in \Theta : \\
      0 < \|\theta - \theta_{n,0}\| \leq \eta}} |Z_{n,2} (\theta) - Z_{n,2}
      (\theta_{n,0})|.
      \label{eqn--extremum-rem-local-ball-envelope}
    \end{equation}
\end{enumerate}
\end{assumption}

\begin{theorem}
\label{thm--extremum-diff-rate}
Let \(\theta_{n}\), \(\zeta_{n}\) satisfy \eqref{eqn--extremum-def}.
Let \Cref{asm--extremum-consistent-pop-quad-decay,asm--extremum-stoch-diff}
hold, and let \(\delta_{n}\) and \(\overline{\Delta}_{n}\) be as defined
therein.
Suppose furthermore that \(\zeta_{n} = \op (\delta_{n}^{2})\) and
\(\lim_{n \to \infty} \Prout \{\theta_{n,0} +
A_{n,0}^{-1} \overline{\Delta}_{n} \notin \Theta\} = 0\), with
\(A_{n,0}^{-1}:=0\) where \(A_{n,0}\) is singular.
Then \(\theta_{n} - \theta_{n,0} = A_{n,0}^{- 1} \overline{\Delta}_{n} + \op
(\delta_{n})\).
\end{theorem}

With non-stochastic \(\theta_{n,0}\), \(M_{n,0}\) and \(A_{n,0}\),
\Cref{thm--extremum-diff-rate} is close to
\citet[Theorem~3.2.16, p.~407]{2023vandervaartWeakConvergenceEmpirical}, and
its proof follows the same two steps (a rate, then the linearisation); we allow
random centring, which the proof of \Cref{thm--bootstrap} needs.

The proof is in \Cref{sec--prf--thm--extremum-diff-rate}.

\subsection{Proof of \Cref{thm--m-est-non-smooth}}
\label{sec--proof--thm--m-est-non-smooth}

By the central limit theorem
\(A_{0}^{-1} n^{1/2} \bar{\psi}_{n} \leadsto N(0, V)\),
and so by Slutsky's lemma it suffices to show that the linearisation
\(\hat{\theta}_{n} - \theta_{0} = A_{0}^{- 1} \bar{\psi}_{n} + \op (1 /
\sqrt{n})\) holds.
To that end, we verify the conditions of \Cref{thm--extremum-diff-rate} under
the conditions of \Cref{thm--m-est-non-smooth} and with \(\delta_{n} =
1/\sqrt{n}\), \(\theta_{n} = \hat{\theta}_{n}\), \(\zeta_{n} =
\hat{\zeta}_{n}\), \(M_{n} = \hat{M}_{n}\), \(\theta_{n,0} = \theta_{0}\),
\(M_{n,0} = M_{0}\), \(A_{n,0} = A_{0}\) and \(\overline{\Delta}_{n} =
\bar{\psi}_{n}\).
The central limit theorem applies because the \(K\) summands of
\(\bar\psi_n\) are independent, being functions of the disjoint families
\(\{U_{i\cdot\e_k}\}_{i\ge1}\).
The \(k\)th summand is an average of \(N_{n,k}\)
i.i.d.\ mean-zero vectors with covariance \(\E[\Delta_k\Delta_k']\); since
\(n/N_{n,k}\to\lambda_k\), the Lindeberg--L\'evy theorem and Slutsky's lemma
give \(n^{1/2}\bar\psi_n\leadsto N(0,\Sigma)\).

By assumption, \(\hat{\zeta}_{n} = \op (1/n)\) and by the central limit theorem,
\(\bar{\psi}_{n} = \Op (1 / \sqrt{n})\).
Since \(\theta_{0} \in \mathrm{int} (\Theta)\) and \(A_{0}\) is positive
definite (both assumed in
\Cref{asm--m-est-consistent-pop-quad-decay}\eqref{asm--m-est-smooth-pop}), it
follows that \(\Prout \{\theta_{0} + A_{0}^{- 1}
\bar{\psi}_{n} \notin \Theta\} \to 0\).
\Cref{asm--extremum-consistent-pop-quad-decay}\eqref{asm--extremum-consistency}
holds by
\Cref{asm--m-est-consistent-pop-quad-decay}\eqref{asm--m-est-consistent}.
\Cref{asm--extremum-consistent-pop-quad-decay}\eqref{asm--extremum-pop-quad-decay-strong}
holds by
\Cref{asm--m-est-consistent-pop-quad-decay}\eqref{asm--m-est-smooth-pop}
since \(A_{0}\) is positive definite.
\Cref{asm--extremum-stoch-diff}\eqref{asm--extremum-stoch-diff-Delta-Op-delta}
has already been verified above.
In what follows, we take \(Z_{n,1}\) and \(Z_{n,2}\) in
\Cref{asm--extremum-stoch-diff} to be
\begin{equation*}
  Z_{n,1} (\theta) = \sum_{k = 1}^{K} \hatE_{n,k} [\pi_{\e_{k}}
  m_{\theta}] \qquad \text{and} \qquad
  Z_{n,2} (\theta) = \sum_{k = 2}^{K} \sum_{\e \in \calE_{k}} H_{n,\e}
  [m_{\theta}].
\end{equation*}
By \eqref{eqn--gen-hoeffding-decomp} and
\(H_{n,\e_k}[f]=\hatE_{n,k}[\pi_{\e_k}f]\),
\(\hat M_n-M_0=Z_{n,1}+Z_{n,2}\).
Since \(\E[\pi_{\e_k}m_\theta]=0\) and \(\E[\Delta_k]=0\), \(\E[r_{k,\theta}]=0\),
so the remainder of \Cref{asm--extremum-stoch-diff} is
\(R_{n,1}(\theta)=Z_{n,1}(\theta)-Z_{n,1}(\theta_0)-(\theta-\theta_0)'\bar\psi_n
=\sum_{k=1}^K(\hatE_{n,k}-\E)[r_{k,\theta}]\).
Hence \(\phi_{n,1}(\eta)\) is at most the sum over \(k\) of the suprema in
\eqref{eqn--m-est-stoch-diff}, and
\Cref{asm--extremum-stoch-diff}\eqref{asm--extremum-stoch-diff-phin1-small}
follows from \Cref{asm--m-est-general}\eqref{asm--m-est-stoch-diff},
subadditivity of \(\Eout\) and Markov's inequality.
Finally \eqref{eqn--m-est-stoch-eq} in
\Cref{asm--m-est-general}\eqref{asm--m-est-stoch-eq}
together with Markov's inequality and the triangle inequality implies
\Cref{asm--extremum-stoch-diff}\eqref{asm--extremum-stoch-diff-rem-local-asymp-equicont}.
\qed

\subsection{Common definitions and useful auxiliary results for proofs of
\Cref{thm--bootstrap} and \texorpdfstring{\Cref{rem--plain-bootstrap}}{Remark \ref{rem--plain-bootstrap}}}
Let \(\mathscr{U}_{n}\) be the \(\sigma\)-algebra generated by \(\{\bm{U}_{\i} :
\i \in \calI_{n}\}\).
\label{sec--P3-common}%
By \Cref{asm--exchangeable} we take \(W_{\i}=\tau(\bm U_{\i})\) for every
\(\i\), so \(\sigma(\bm W_n)\subseteq\mathscr U_n\).
As explained in \Cref{sec--measurability-conventions}, all suprema over \(\theta\)
below are then measurable, outer expectations are expectations, and
conditional bounds can be integrated by Fubini's theorem.
Write \(\Delta_{k,i} = \Delta_{k} (U_{i \cdot \e_{k}})\) and set
\begin{equation}
  \hat{\Delta}_{n}^{*} = \sum_{k = 1}^{K} \hatE_{n,k}^{*} [\Delta_{k}],
  \qquad \tilde{\Delta}_{n}^{*} = \hat{\Delta}_{n}^{*} - \bar{\psi}_{n}
  = \sum_{k = 1}^{K} (\hatE_{n,k}^{*} - \hatE_{n,k}) [\Delta_{k}].
  \label{eqn--bootstrap-derivatives}
\end{equation}
For the two weighting schemes,
\begin{equation}
  \begin{split}
    \hatE_{n,k}^{*} [f]  &= \frac{1}{N_{n,k}} \sum_{i =
    1}^{N_{n,k}} w_{n,k,i} f (U_{i \cdot \e_{k}}), \quad \text{where} \\
    w_{n,k,i}  &=
    \begin{cases}
      \xi_{k,i} / \bar{\xi}_{n,k} & \text{for pigeonhole, if } \bar{\xi}_{n,k} \neq 0, \\
      1 & \text{for pigeonhole, if } \bar{\xi}_{n,k} = 0, \\
      1 + \xi_{k,i} - \bar{\xi}_{n,k} & \text{for multiplier}.
    \end{cases}
  \end{split}
  \label{eqn--hajek-weighted-avg}
\end{equation}
For both, \(\hatE_{n,k}^{*} [1] = 1\).
We also adjust definitions of Hoeffding-type projections accordingly.
In particular,
\begin{equation}
  \begin{split}
    H_{n,\e}^{*} [f]  &= \frac{1}{|\calI_{n,\e}|}
    \sum_{\i \in \calI_{n,\e}} w_{n,\i\odot\e}\, \pi_{\i\odot\e} [f], \quad \text{where} \\
    w_{n,\i\odot\e}  &= \begin{cases}
      \prod_{\ell \in \supp (\e)} w_{n, \ell, i_{\ell}} &
      \text{for pigeonhole,} \\
      1 + \sum_{\ell \in \supp (\e)} (w_{n,\ell,i_{\ell}} - 1) & \text{for
      multiplier}.
    \end{cases}
  \end{split}
  \label{eqn--hoeffding-proj-weighted-avg}
\end{equation}

Both weighting schemes satisfy \(\hatE_n[\xi_{n,\i}]=1\), and averaging
\(\xi_{n,\i}\) over the coordinates \(i_k\), \(k\notin\supp(\e)\), gives
\(w_{n,\i\odot\e}\): for pigeonhole weights because
\(N_{n,k}^{-1}\sum_{i=1}^{N_{n,k}}w_{n,k,i}=1\) for every \(k\), and for
multiplier weights because
\(N_{n,k}^{-1}\sum_{i=1}^{N_{n,k}}(w_{n,k,i}-1)=0\).
Since \(\pi_{\i\odot\e}[f]\) depends on \(\i\) only through \(\i\odot\e\), it
follows that \(|\calI_n|^{-1}\sum_{\i\in\calI_n}\xi_{n,\i}\,\pi_{\i\odot\e}[f]=
H^*_{n,\e}[f]\).
Moreover, \(f(W_{\i})=\E[f]+\sum_{\e\in\calE}\pi_{\i\odot\e}[f]\)
identically, by the recursive definition of \(\pi_{\e}\) and
\(W_{\i}=\tau(\bm U_{\i})\).
Hence, for all \(\theta\in\Theta\),
\begin{equation}
  \hat M^*_n(\theta)=M_0(\theta)+\sum_{\e\in\calE}H^*_{n,\e}[m_\theta],
  \qquad
  \hat M^*_n(\theta)-\hat M_n(\theta)=\sum_{\e\in\calE}
  (H^*_{n,\e}-H_{n,\e})[m_\theta],
  \label{eqn--P3-boot-hoeffding}
\end{equation}
where \((H^*_{n,\e_k}-H_{n,\e_k})[f]=(\hatE^*_{n,k}-\hatE_{n,k})
[\pi_{\e_k}f]\).

Furthermore, for \(r_{k,\theta}\) in \eqref{eqn--m-est-stoch-diff},
we set
\begin{equation}
  \begin{gathered}
    \hat{Z}_{n,1}^{*} (\theta) = \sum_{k = 1}^{K} \hatE_{n,k}^{*}
    [\pi_{\e_{k}} m_{\theta}], \quad
    \hat{Z}_{n,2}^{*} (\theta) = \sum_{k = 2}^{K} \sum_{\e \in \calE_{k}}
    H_{n,\e}^{*} [m_{\theta}], \\
    \hat{R}_{n,1}^{*} (\theta) = \hat{Z}_{n,1}^{*} (\theta) - \hat{Z}_{n,1}^{*}
    (\theta_{0}) - (\theta - \theta_{0})' \hat{\Delta}_{n}^{*} =
    \sum_{k = 1}^{K} \hatE_{n,k}^{*} [r_{k,\theta}], \\
    \tilde{Z}_{n,1}^{*} (\theta) = \sum_{k = 1}^{K} (\hatE_{n,k}^{*} -
    \hatE_{n,k}) [\pi_{\e_{k}} m_{\theta}], \quad
    \tilde{Z}_{n,2}^{*} (\theta) = \sum_{k = 2}^{K} \sum_{\e \in \calE_{k}}
    (H_{n,\e}^{*} - H_{n,\e}) [m_{\theta}], \\
    \text{and} \qquad
    \tilde{R}_{n,1}^{*} (\theta) = \tilde{Z}_{n,1}^{*} (\theta) -
    \tilde{Z}_{n,1}^{*} (\theta_{0}) - (\theta - \theta_{0})'
    \tilde{\Delta}_{n}^{*} = \sum_{k = 1}^{K} (\hatE_{n,k}^{*} -
    \hatE_{n,k}) [r_{k,\theta}].
  \end{gathered}
  \label{eqn--bootstrap-hoeffding-decomps-and-rems}
\end{equation}
It will be useful to define extremal remainder oscillations centred at
\(\theta_{0}\):
\begin{equation}
  \begin{split}
    \tilde{\phi}_{n,1}^{*} (\eta)  &= \sup_{\substack{\theta \in \Theta : 0 < \|\theta - \theta_{0}\| \leq
    \eta}} \frac{|\tilde{R}_{n,1}^{*} (\theta)|}{\|\theta - \theta_{0}\|
    (\|\theta - \theta_{0}\| + n^{-1/2})}, \\
    \tilde{\phi}_{n,2}^{*} (\eta)  &= \sup_{\substack{\theta \in \Theta : 0 < \|\theta - \theta_{0}\| \leq
    \eta}} |\tilde{Z}^{*}_{n,2} (\theta) - \tilde{Z}^{*}_{n,2} (\theta_{0})|.
  \end{split}
  \label{eqn--bootstrap-tildephinjstar-def}
\end{equation}

\begin{lemma}
\label{lem--cond-equiv-uncond-op1-Op1}
Let \(\{\tau_{n}\}\) be a sequence of random \(\R^{d}\)-vectors on a probability
space \((\Omega, \mathscr{F}, \Pr)\) and let \(\{\mathscr{U}_{n}\}\) be a
sequence of sub-\(\sigma\)-algebras of \(\mathscr{F}\).
Let \(\eta_{n}\) be an arbitrary non-stochastic positive real
sequence.
Then \(\tau_{n} = \op (\eta_{n})\) if and only if
\(\Pr \{\|\tau_{n}\| > \varepsilon \eta_{n} \mid \mathscr{U}_{n}\} = \op (1)\)
for every \(\varepsilon \in (0, \infty)\), and \(\tau_{n} = \Op (\eta_{n})\)
if and only if \(\Pr \{\|\tau_{n}\| > L_{n} \eta_{n} \mid \mathscr{U}_{n}\} =
\op (1)\) for every \(L_{n} \to \infty\).
\end{lemma}

\begin{proof}[Proof of \Cref{lem--cond-equiv-uncond-op1-Op1}]
In both cases, the `only if' direction follows from the law of iterated
expectations and Markov's inequality.
The `if' direction follows from the bounded convergence theorem for
convergence in probability since conditional probabilities take values in \([0,
1]\).
\end{proof}

\begin{lemma}
\label{lem--strong-cond-weak-suff}
Let \(\{\tau_{n}\}\) be a sequence of random \(\R^{d}\)-vectors on a probability
space \((\Omega, \mathscr{F}, \Pr)\).
Let \(\{\mathscr{W}_{n}\}\) and \(\{\mathscr{U}_{n}\}\) be sequences of
sub-\(\sigma\)-algebras of \(\mathscr{F}\) with \(\mathscr{W}_{n} \subseteq
\mathscr{U}_{n}\) for every \(n\).
Suppose that there is a probability distribution \(Q\) on \(\R^{d}\) such that
\(\tau_{n} \dtop Q\) conditionally given \(\mathscr{U}_{n}\).
Then \(\tau_{n} \dtop Q\) conditionally given \(\mathscr{W}_{n}\).
\end{lemma}

The proof is in \Cref{sec--prf--lem--strong-cond-weak-suff}.

\begin{lemma}
\label{lem--bootstrap-cond-weak-conv}
Suppose the conditions of \Cref{thm--m-est-non-smooth} hold.
Under \Cref{asm--boot-wgt},
\(\sqrt{n} \tilde{\Delta}_{n}^{*} \dtop N
(0, \Sigma)\) conditionally given \(\mathscr{U}_{n}\).
By \Cref{lem--strong-cond-weak-suff} (applicable since
\(\sigma(\bm W_n)\subseteq\mathscr U_n\)), it follows that
\(\sqrt{n} \tilde{\Delta}_{n}^{*} \dtop N
(0, \Sigma)\) conditionally given \(\bm{W}_{n}\).
\end{lemma}

\begin{proof}[Proof of \Cref{lem--bootstrap-cond-weak-conv}]
A proof of the first claim is given in
\Cref{sec--prf--lem--bootstrap-prev-cond-weak-conv}.
\end{proof}

\begin{lemma}
\label{lem--bootstrap-1st-order-control}
Suppose the conditions of \Cref{thm--m-est-non-smooth} hold.
Under \Cref{asm--boot-wgt},
for every positive non-stochastic sequence \(\eta_{n} \to 0\),
\begin{equation}
  \Eout \left[ \sup_{\substack{\theta \in \Theta : \\ 0 < \|\theta -
  \theta_{0}\| \leq \eta_{n}}} \frac{\bigl|N_{n,k}^{-1} \sum_{i =
  1}^{N_{n,k}} (\xi_{k,i} - \overline{\xi}_{n,k}) r_{k,\theta} (U_{i \cdot
  \e_{k}}) \bigr|}{\|\theta - \theta_{0}\|
  (\|\theta - \theta_{0}\| + n^{-1/2})}
  \right] = o (1) \quad \text{for each } k \in \{1, \dots, K\}.
  \label{eqn--bootstrap-1st-order-control}
\end{equation}
Consequently, \(\tilde{\phi}_{n,1}^{*} (\cdot)\) in
\eqref{eqn--bootstrap-tildephinjstar-def} satisfies \(\tilde{\phi}_{n,1}^{*}
(\eta_{n}) \ptoout 0\).
\end{lemma}

The proof is in \Cref{sec--prf--lem--bootstrap-1st-order-control}.
It uses only \Cref{asm--exchangeable,asm--boot-wgt},
\Cref{asm--m-est-general}\eqref{asm--m-est-meas},
\eqref{eqn--m-est-stoch-diff} and \(\E[\Delta_k]=0\).
In that proof, \(n^{-1/2}\) enters only through the denominator of the class
\(\calR_{n,k}\) defined there, and the constants depend only on \(K\) and the
bounds \(L_k\) on \(|\xi_{k,i}|\).
Hence, for any non-stochastic \(\delta_n>0\), if \eqref{eqn--m-est-stoch-diff}
holds with \(n^{-1/2}\) replaced by \(\delta_n\), then
\eqref{eqn--bootstrap-1st-order-control} and \(\tilde\phi^*_{n,1}(\eta_n)
\ptoout0\) hold with the same replacement, also along a drifting sequence
\(\{P_n\}\).

\begin{lemma}
\label{lem--bootstrap-2nd-order-control}
Let the conditions of \Cref{thm--m-est-non-smooth} hold.
Under \Cref{asm--boot-wgt},
\(\tilde{\phi}_{n,2}^{*} (\cdot)\) in \eqref{eqn--bootstrap-tildephinjstar-def}
satisfies \([\eta_{n} \to 0 \implies \tilde{\phi}_{n,2}^{*} (\eta_{n}) =
\Op^{\dagger} (1/n)]\) and \([\eta_{n} = O (n^{-1/2}) \implies
\tilde{\phi}_{n,2}^{*} (\eta_{n}) = \op^{\dagger} (1/n)]\).
\end{lemma}

The proof is in \Cref{sec--prf--lem--bootstrap-2nd-order-control}.

\Needspace{8\baselineskip}
\subsection{Proofs of \Cref{thm--bootstrap} and \texorpdfstring{\Cref{rem--plain-bootstrap}}{Remark \ref{rem--plain-bootstrap}}}
\subsubsection{Proof of \Cref{thm--bootstrap}}
\label{sec--prf--thm--bootstrap}

We apply \Cref{thm--extremum-diff-rate} with
\begin{gather*}
  \delta_{n} = n^{-1/2}, \quad
  M_{n} = \tilde M_{n}^{*}, \quad
  M_{n,0} (\theta) = - \frac{1}{2} (\theta - \hat{\theta}_{n})'
  \tilde{A}_{n} (\theta - \hat{\theta}_{n}), \\
  \theta_{n} = \tilde{\theta}_{n}^{*}, \quad
  \theta_{n,0} = \hat{\theta}_{n}, \quad
  \overline{\Delta}_{n} = \tilde{\Delta}_{n}^{*}, \quad
  A_{n,0} = \tilde A_{n}, \quad \zeta_{n} = \tilde\zeta^{*}_{n},
\end{gather*}
where \(\tilde\Delta_{n}^{*}\) is the centred bootstrap score of
\eqref{eqn--bootstrap-derivatives}.
Let \(E^\infty_n=\{\sup_{\theta\in\Theta}\tilde M^*_n(\theta)=\infty\}\), on
which \eqref{eqn--m-est-boot-def} has no solution and
\(\tilde\theta^*_n=\hat\theta_n\) by convention; \(\Prout(E^\infty_n)\to0\)
by the paragraph after \eqref{eqn--m-est-boot-def}.
On \(E^\infty_n\), replace \(M_n\) above by
\(\theta\mapsto-\|\theta-\hat\theta_n\|^2\), which \(\theta_n=\hat\theta_n\)
maximises, and replace \(Z_{n,2}\), defined below, by
\(M_n-M_{n,0}-Z_{n,1}\).
Then \(\theta_n\) satisfies the approximate maximiser condition of
\Cref{thm--extremum-diff-rate} everywhere, and each condition of that theorem
verified below is a statement in outer probability that a change confined to
\(E^\infty_n\) does not affect.
Likewise, the proof of \Cref{lem--boot-consistent}(b) uses its
approximate maximiser hypothesis only through
\(\tilde M^*_n(\tilde\theta^*_n)\ge\tilde M^*_n(\hat\theta_n)-\tilde\zeta^*_n\),
which also holds on \(E^\infty_n\), where \(\tilde\theta^*_n=\hat\theta_n\).
Below we verify the conditions of \Cref{thm--extremum-diff-rate}.
Given a successful verification of these conditions,
that theorem gives
\(\tilde{\theta}_{n}^{*} - \hat{\theta}_{n} = \tilde{A}_{n}^{- 1}
\tilde{\Delta}_{n}^{*} + \op (1/\sqrt{n})\).
Furthermore, since \(A_{0}\) is positive definite and \(\tilde{A}_{n} = A_{0} +
\op (1)\), it follows that \(\tilde{A}_{n}^{- 1} = A_{0}^{-1} + \op (1)\) by the
continuous mapping theorem.
By \Cref{lem--bootstrap-cond-weak-conv}, which gives conditional
convergence in law to a tight limit and hence conditional tightness, and
\Cref{lem--cond-equiv-uncond-op1-Op1}, \(\sqrt{n}
\tilde{\Delta}_{n}^{*} = \Op (1)\) and so
\(\sqrt{n} (\tilde{\theta}_{n}^{*} - \hat{\theta}_{n}) = A_{0}^{-
1} \sqrt{n} \tilde{\Delta}_{n}^{*} + \op (1)\).
The conclusion of \Cref{thm--bootstrap} then follows from
\Cref{lem--bootstrap-cond-weak-conv},
\Cref{lem--cond-equiv-uncond-op1-Op1} applied to the \(\op (1)\) remainder,
and the conditional version of Slutsky's lemma.

By assumption, \(\tilde{\zeta}_{n}^{*} = \op (1/n)\) and, as shown above,
\(\tilde{\Delta}^{*}_{n} = \Op (1 / \sqrt{n})\).
Since \(\hat{\theta}_{n} \pto \theta_{0} \in \mathrm{int} (\Theta)\) and
\(\tilde{A}_{n}^{-1} \tilde{\Delta}_{n}^{*} = \{A_{0}^{-1} + \op (1)\}
\Op (1/\sqrt{n}) = \op (1)\) by the above, it
follows that \(\Pr \{\hat{\theta}_{n} + \tilde{A}_{n}^{- 1}
\tilde{\Delta}_{n}^{*} \notin \Theta\} \to 0\).
\Cref{asm--extremum-consistent-pop-quad-decay}\eqref{asm--extremum-consistency}
holds by \Cref{lem--boot-consistent}(b), since \(\tilde A_n\pto A_0\succ0\)
and \(\tilde\zeta^*_n=\op(1/n)\).
For the above definitions, \(M_{n,0}\) satisfies
\Cref{asm--extremum-consistent-pop-quad-decay}%
\eqref{asm--extremum-pop-quad-decay-strong} by hypothesis for \(\tilde{A}_{n}\).
\Cref{asm--extremum-stoch-diff}\eqref{asm--extremum-stoch-diff-Delta-Op-delta}
has already been verified above.

\begin{remark}[On the choice of \(Z_{n,j}\)]
One can write, by \eqref{eqn--P3-boot-hoeffding} and \eqref{eqn--bootstrap-hoeffding-decomps-and-rems},
\begin{equation*}
  \tilde{M}_{n}^{*} (\theta) - M_{n,0} (\theta) =
  \tilde{Z}_{n,1}^{*} (\theta) + \tilde{Z}_{n,2}^{*} (\theta), \quad \text{and}
  \quad \tilde{Z}^{*}_{n,1} (\theta) - \tilde{Z}^{*}_{n,1} (\theta_{0}) =
  (\theta - \theta_{0})' \tilde{\Delta}_{n}^{*} + \tilde{R}_{n,1}^{*} (\theta).
\end{equation*}
However, in this proof, we do not set \(Z_{n,j} = \tilde{Z}^{*}_{n,j}\).
Doing so for \(j = 1\) would require control of
\begin{equation*}
  \frac{|\tilde{R}_{n,1}^{*} (\theta) - \tilde{R}_{n,1}^{*}
  (\hat{\theta}_{n})|}{\|\theta - \hat{\theta}_{n}\| (\|\theta -
  \hat{\theta}_{n}\| + n^{-1/2})},
\end{equation*}
even when \(\|\theta - \hat{\theta}_{n}\| \ll n^{-1/2}\).
Direct control of the above display for such values is difficult to establish
from our stochastic differentiability assumptions alone --- the latter
control behaviour for neighbourhoods \(\|\theta - \theta_{0}\| \leq \eta_{n}\)
centred at \(\theta_{0}\).
An alternative definition of \(Z_{n,j}\) allows us to carry out the proof
under the current set of assumptions.
\(\lozenge\)
\end{remark}

For this proof, use the following definition of \(Z_{n,1}\), \(Z_{n,2}\) and
\(R_{n,1}\) in \Cref{asm--extremum-stoch-diff}:
\begin{equation}
  \begin{split}
    Z_{n,1} (\theta)  &= \tilde{Z}_{n,1}^{*} (\hat{\theta}_{n}) + (\theta - \hat{\theta}_{n})'
    \tilde{\Delta}_{n}^{*} + \ind \{\|\theta - \hat{\theta}_{n}\| >
    n^{-1/2}\} \{\tilde{R}_{n,1}^{*} (\theta) - \tilde{R}_{n,1}^{*}
    (\hat{\theta}_{n}) \}, \\
    Z_{n,2} (\theta)  &= \tilde{Z}_{n,2}^{*} (\theta) + \ind \{\|\theta - \hat{\theta}_{n}\|
    \leq n^{-1/2}\} \{\tilde{R}_{n,1}^{*} (\theta) - \tilde{R}_{n,1}^{*}
    (\hat{\theta}_{n})\}, \\
    R_{n,1} (\theta)  &= \ind \{\|\theta - \hat{\theta}_{n}\| > n^{-1/2}\}
    \{\tilde{R}_{n,1}^{*} (\theta) - \tilde{R}_{n,1}^{*} (\hat{\theta}_{n})\}.
  \end{split}
  \label{eqn--bootstrap-Znj-def}
\end{equation}
We can replace \(\phi_{n,j} (\eta)\) in
\eqref{eqn--extremum-stoch-diff} and
\eqref{eqn--extremum-rem-local-ball-envelope}
of \Cref{asm--extremum-stoch-diff} with appropriately convergent
upper bounds \(\bar\phi_{n,j}(\eta)\).
Thus, we define
\begin{equation}
  \begin{gathered}
    \bar\phi_{n,1} (\eta) = \sup_{\substack{\theta \in \Theta : \\ n^{-1/2} <
    \|\theta - \hat{\theta}_{n}\| \leq \eta}} \frac{|\tilde{R}_{n,1}^{*}
    (\theta)| + |\tilde{R}_{n,1}^{*} (\hat{\theta}_{n})|}{\|\theta -
    \hat{\theta}_{n}\| (\|\theta - \hat{\theta}_{n}\| + n^{-1/2})} \quad
    \text{if} \quad n^{-1/2} < \eta \quad \text{and} \\
    \bar\phi_{n,1} (\eta) = 0 \quad \text{otherwise}; \\
    \quad
    \bar\phi_{n,2} (\eta) = \sup_{\substack{\theta \in \Theta: \\ \|\theta -
    \hat{\theta}_{n}\| \leq \eta}} |\tilde{Z}_{n,2}^{*} (\theta) -
    \tilde{Z}_{n,2}^{*} (\hat{\theta}_{n})| + \sup_{\substack{\theta \in \Theta:
    \\ \|\theta - \hat{\theta}_{n}\| \leq n^{-1/2}}} (|\tilde{R}_{n,1}^{*}
    (\theta)| + |\tilde{R}_{n,1}^{*} (\hat{\theta}_{n})|).
  \end{gathered}
  \label{eqn--bootstrap-phinj-def}
\end{equation}
By the triangle inequality and the complementary indicators in
\eqref{eqn--bootstrap-Znj-def}, which make \(R_{n,1} (\theta) = 0\) for
\(\|\theta - \hat{\theta}_{n}\| \leq n^{-1/2}\),
\(\phi_{n,j}\le\bar\phi_{n,j}\) for \(j=1,2\).
Take an arbitrary positive non-stochastic sequence \(\eta_{n} \to 0\).
Fix \(L \geq 1\) and define the event \(B_{n,L} = \{\|\hat{\theta}_{n} -
\theta_{0}\| \leq L n^{-1/2}\}\).
If \(\|\theta - \hat{\theta}_{n}\| \leq \eta_{n}\),
then \(\max \{ \|\theta - \theta_{0}\|, \|\hat{\theta}_{n} - \theta_{0}\| \}
\leq \eta_{n} + L n^{-1/2}\) on \(B_{n,L}\) by the triangle inequality.
Thus, both arguments of \(\tilde{R}_{n,1}^{*}\) in
\eqref{eqn--bootstrap-phinj-def} lie within \(\eta_{n} + L n^{-1/2}\) of
\(\theta_{0}\) and so, on \(B_{n,L}\),
\begin{equation}
  \begin{gathered}
    |\tilde{R}_{n,1}^{*} (\theta)| \leq \tilde{\phi}_{n,1}^{*} (\eta_{n} + L
    n^{-1/2})\, \|\theta - \theta_{0}\| (\|\theta - \theta_{0}\| +
    n^{-1/2}),
    \\
    |\tilde{R}_{n,1}^{*} (\hat{\theta}_{n})| \leq \tilde{\phi}_{n,1}^{*}
    (\eta_{n} + L n^{-1/2})\, \|\hat{\theta}_{n} - \theta_{0}\|
    (\|\hat{\theta}_{n} - \theta_{0}\| + n^{-1/2}), \\
    \bar\phi_{n,2} (\eta_{n}) \leq 2 \tilde{\phi}_{n,2}^{*} (\eta_{n} + L n^{-1/2})
    + \frac{2 (L + 1) (L + 2)}{n} \tilde{\phi}_{n,1}^{*} ((L + 1)
    n^{-1/2}),
  \end{gathered}
  \label{eqn--bootstrap-prelim-stoch-diff-bounds}
\end{equation}
for \(\tilde{\phi}_{n,j}^{*}\) in \eqref{eqn--bootstrap-tildephinjstar-def}.
In \eqref{eqn--bootstrap-prelim-stoch-diff-bounds}, the second term in the third
bound follows from evaluating the first two bounds at \(\eta_{n} = n^{-1/2}\)
and using \(\max \{ \|\theta - \theta_{0}\|, \|\hat{\theta}_{n} - \theta_{0}\|
\} \leq (L + 1) n^{-1/2}\) if \(\eta_{n} = n^{-1/2}\).

We now show \Cref{asm--extremum-stoch-diff}%
\eqref{asm--extremum-stoch-diff-phin1-small}.
On \(B_{n,L}\) and for \(n^{-1/2}
< \|\theta - \hat{\theta}_{n}\| \leq \eta_{n}\),
\(\|\hat{\theta}_{n} - \theta_{0}\| < L \|\theta - \hat{\theta}_{n}\|\) and by
the triangle inequality,
\(\|\theta - \theta_{0}\| \leq
\|\theta - \hat{\theta}_{n}\| + \|\hat{\theta}_{n} - \theta_{0}\| < (1 + L)
\|\theta - \hat{\theta}_{n}\|\).
Since \(n^{-1/2} \leq (L + 1) n^{-1/2}\) trivially, combine these with
\eqref{eqn--bootstrap-phinj-def} and
\eqref{eqn--bootstrap-prelim-stoch-diff-bounds} to see that
\begin{equation*}
  |\tilde{R}_{n,1}^{*} (\theta)| + |\tilde{R}_{n,1}^{*} (\hat{\theta}_{n})|
  \leq 2 (L + 1)^{2} \tilde{\phi}_{n,1}^{*} (\eta_{n} + L n^{-1/2})\,
  \|\theta - \hat{\theta}_{n}\| (\|\theta - \hat{\theta}_{n}\| + n^{-1/2}),
\end{equation*}
and taking suprema gives
\begin{equation*}
  \bar\phi_{n,1} (\eta_{n}) \leq 2 (L + 1)^{2} \tilde{\phi}_{n,1}^{*} (\eta_{n} +
  L n^{-1/2}).
\end{equation*}
Therefore for any \(\varepsilon \in (0, \infty)\), by the union bound
\begin{equation*}
  \Prout \{\bar\phi_{n,1} (\eta_{n}) > \varepsilon\} \leq
  \Prout \left\{ \tilde{\phi}_{n,1}^{*} (\eta_{n} + L n^{-1/2}) >
  \frac{\varepsilon}{2 (L + 1)^{2}}
  \right\} +
  \Pr \{B_{n,L}^{c}\},
\end{equation*}
whence by \Cref{lem--bootstrap-1st-order-control}
\begin{equation*}
  \limsup_{n \to \infty} \Prout \{\bar\phi_{n,1} (\eta_{n}) > \varepsilon\} \leq
  \limsup_{n \to \infty} \Pr \{B_{n,L}^{c}\}.
\end{equation*}
Since the left-hand side of the above does not depend on the choice of \(L\), we
can take limits on the right as \(L \to \infty\) to get
\(\bar\phi_{n,1} (\eta_{n}) \ptoout 0\), since \(\lim_{L \to \infty} \limsup_{n \to
\infty} \Pr \{B_{n,L}^{c}\} = 0\) by \Cref{thm--m-est-non-smooth}; hence
\(\phi_{n,1} (\eta_{n}) \ptoout 0\).
This verifies \Cref{asm--extremum-stoch-diff}%
\eqref{asm--extremum-stoch-diff-phin1-small}.

We now show \Cref{asm--extremum-stoch-diff}%
\eqref{asm--extremum-stoch-diff-rem-local-asymp-equicont}.
By the third bound in \eqref{eqn--bootstrap-prelim-stoch-diff-bounds},
given any \(t \in (0, \infty)\),
\begin{equation*}
  \begin{split}
    \Prout \{n \bar\phi_{n,2} (\eta_{n}) > t\}  &\leq \Pr \{B_{n,L}^{c}\} + \Prout \{ n \tilde{\phi}_{n,2}^{*}
    (\eta_{n} + L n^{-1/2}) > t / 4 \} \\
    &\quad+ \Prout \left\{ \tilde{\phi}_{n,1}^{*} ((L + 1) n^{-1/2}) >
    \frac{t}{4 (L + 1) (L + 2)} \right\}.
  \end{split}
\end{equation*}
Fix \(L\ge1\).
By \Cref{lem--bootstrap-1st-order-control}, the last term tends to zero as
\(n\to\infty\) for every \(t\).
If \(\eta_n\to0\), the middle term vanishes as \(n\to\infty\) and then
\(t\to\infty\), by the first claim of \Cref{lem--bootstrap-2nd-order-control};
if \(\eta_n=O(n^{-1/2})\), then \(\eta_n+Ln^{-1/2}=O(n^{-1/2})\) and the middle
term tends to zero as \(n\to\infty\) for every fixed \(t\), by its second claim.
Letting \(L\to\infty\), as in the previous paragraph, gives \(\bar\phi_{n,2}(\eta_n)=\Op^\dagger(1/n)\) for every \(\eta_n\to0\),
and \(\bar\phi_{n,2}(\eta_n)=\op^\dagger(1/n)\) if \(\eta_n=O(n^{-1/2})\); as
\(\phi_{n,2}\le\bar\phi_{n,2}\), the same holds for \(\phi_{n,2}(\eta_n)\).
This verifies \Cref{asm--extremum-stoch-diff}%
\eqref{asm--extremum-stoch-diff-rem-local-asymp-equicont}.
\qed

\subsubsection{\texorpdfstring{Proof of the conditional limit in \Cref{rem--plain-bootstrap}}{Proof of the conditional limit in Remark \ref{rem--plain-bootstrap}}}
\label{sec--prf--prp--bootstrap-prev}

\begin{lemma}
\label{lem--bootstrap-prev-linearize}
Suppose the conditions of \Cref{thm--m-est-non-smooth} hold, and that
\(M_0(\theta_0)>\sup_{\|\theta-\theta_0\|\ge\varepsilon}M_0(\theta)\) for every
\(\varepsilon>0\).
Under \Cref{asm--boot-wgt},
\(\hat{\theta}_{n}^{*} - \theta_{0} = A_{0}^{-1} \hat{\Delta}_{n}^{*} + \op
(n^{-1/2})\).
\end{lemma}

The proof is in \Cref{sec--prf--lem--bootstrap-prev-linearize}.

By \Cref{thm--m-est-non-smooth}, \(\hat{\theta}_{n} - \theta_{0} = A_{0}^{-1}
\bar{\psi}_{n} + \op (n^{-1/2})\).
Since \(\tilde{\Delta}_{n}^{*} = \hat{\Delta}_{n}^{*} - \bar{\psi}_{n}\),
subtracting this from \Cref{lem--bootstrap-prev-linearize} yields
\(\sqrt{n} (\hat{\theta}_{n}^{*} - \hat{\theta}_{n}) = A_{0}^{-1} \sqrt{n}
\tilde{\Delta}_{n}^{*} + \op (1)\).
By \Cref{lem--cond-equiv-uncond-op1-Op1}, the \(\op (1)\) term is also \(\op
(1)\) conditionally on the data \(\bm{W}_{n}\).
By the conditional version of Slutsky's lemma and \Cref{lem--bootstrap-cond-weak-conv},
the conditional limit in \Cref{rem--plain-bootstrap} follows.
\qed

\subsubsection{Proof of \Cref{lem--bootstrap-prev-linearize}}
\label{sec--prf--lem--bootstrap-prev-linearize}

We apply \Cref{thm--extremum-diff-rate} with \(M_{n} = \hat{M}_{n}^{*}\),
\(M_{n,0} = M_{0}\), \(\theta_{n} = \hat{\theta}_{n}^{*}\), \(\theta_{n,0} =
\theta_{0}\), \(A_{n,0} = A_{0}\), \(\zeta_{n} = \hat{\zeta}_{n}^{*}\), \(\delta_{n} = n^{-1/2}\) and
\(\overline{\Delta}_{n} = \hat{\Delta}_{n}^{*}\) in
\eqref{eqn--bootstrap-derivatives}.
By assumption, \(\hat{\zeta}_{n}^{*} = \op (1/n)\), and
\(\hat{\Delta}^{*}_{n}=\bar\psi_n+\tilde\Delta^*_n=\Op(1/\sqrt n)\) by the
central limit theorem for \(\bar\psi_n\) (\Cref{sec--proof--thm--m-est-non-smooth})
and \Cref{lem--bootstrap-cond-weak-conv,lem--cond-equiv-uncond-op1-Op1}; this
is \Cref{asm--extremum-stoch-diff}\eqref{asm--extremum-stoch-diff-Delta-Op-delta}.
As in \Cref{sec--proof--thm--m-est-non-smooth}, with \(\bar\psi_n\)
replaced by \(\hat\Delta^*_n\), \(\Pr\{\theta_0+A_0^{-1}\hat\Delta^*_n\notin
\Theta\}\to0\) and
\Cref{asm--extremum-consistent-pop-quad-decay}\eqref{asm--extremum-pop-quad-decay-strong}
holds; \Cref{asm--extremum-consistent-pop-quad-decay}\eqref{asm--extremum-consistency}
holds by \Cref{lem--boot-consistent}(c).

It remains to verify the two local oscillation conditions.
By the first identity in \eqref{eqn--P3-boot-hoeffding} and
\(H^*_{n,\e_k}[f]=\hatE^*_{n,k}[\pi_{\e_k}f]\) (take \(\e=\e_k\) in
\eqref{eqn--hoeffding-proj-weighted-avg}),
\(\hat M^*_n(\theta)=M_0(\theta)+\hat Z^*_{n,1}(\theta)+\hat Z^*_{n,2}(\theta)\)
under either weighting scheme, with \(\hat Z^*_{n,j}\) as in
\eqref{eqn--bootstrap-hoeffding-decomps-and-rems}.
Thus we take \(Z_{n,j} = \hat{Z}_{n,j}^{*}\) for \(j \in \{1, 2\}\) and
\(R_{n,1} = \hat{R}_{n,1}^{*}\) in \Cref{asm--extremum-stoch-diff}.

Note that \(\E [r_{k,\theta}] = 0\) by the centring of \(\pi_{\e_{k}}\) and
\(\E [\Delta_{k}] = 0\).
Hence \eqref{eqn--bootstrap-hoeffding-decomps-and-rems} gives
\begin{equation*}
  \hat{R}_{n,1}^{*} (\theta) = \sum_{k = 1}^{K} (\hatE_{n,k} - \E)
  [r_{k,\theta}] + \tilde{R}_{n,1}^{*} (\theta).
\end{equation*}
Take any positive non-stochastic sequence \(\eta_{n} \to 0\).
By the triangle inequality, the local oscillation in
\eqref{eqn--extremum-stoch-diff} satisfies
\begin{equation*}
  \begin{split}
    \phi_{n,1} (\eta_{n})  &= \sup_{\substack{\theta \in \Theta : \\ 0 < \|\theta - \theta_{0}\|
    \leq \eta_{n}}} \frac{|\hat{R}_{n,1}^{*} (\theta)|}{\|\theta -
    \theta_{0}\| (\|\theta - \theta_{0}\| + n^{-1/2})} \\
     &\leq \sum_{k = 1}^{K} \sup_{\substack{\theta \in \Theta : \\ 0 < \|\theta
    - \theta_{0}\| \leq \eta_{n}}} \frac{|(\hatE_{n,k} - \E)
    [r_{k,\theta}]|}{\|\theta - \theta_{0}\| (\|\theta - \theta_{0}\| +
    n^{-1/2})} + \tilde{\phi}_{n,1}^{*} (\eta_{n}).
  \end{split}
\end{equation*}
The finite sum is \(\op^{\dagger} (1)\) by \eqref{eqn--m-est-stoch-diff} and
Markov's inequality.
The last term is \(\op^{\dagger} (1)\) by
\Cref{lem--bootstrap-1st-order-control}.
This verifies \Cref{asm--extremum-stoch-diff}%
\eqref{asm--extremum-stoch-diff-phin1-small}.

For the higher-order condition,
\eqref{eqn--bootstrap-hoeffding-decomps-and-rems} implies
\begin{equation*}
  \begin{split}
    \hat{Z}_{n,2}^{*} (\theta) - \hat{Z}_{n,2}^{*} (\theta_{0})  &= \sum_{k = 2}^{K} \sum_{\e \in \calE_{k}} H_{n,\e} [m_{\theta} -
    m_{\theta_{0}}] + \tilde{Z}_{n,2}^{*} (\theta) - \tilde{Z}_{n,2}^{*}
    (\theta_{0}).
  \end{split}
\end{equation*}
Therefore, the local oscillation in
\eqref{eqn--extremum-rem-local-ball-envelope} satisfies
\begin{equation*}
  \begin{split}
    \phi_{n,2} (\eta_{n})  &= \sup_{\substack{\theta \in \Theta : \\ 0 < \|\theta - \theta_{0}\|
    \leq \eta_{n}}} |\hat{Z}_{n,2}^{*} (\theta) - \hat{Z}_{n,2}^{*}
    (\theta_{0})| \\
     &\leq \sum_{k = 2}^{K} \sum_{\e \in \calE_{k}} \sup_{\substack{\theta \in
    \Theta : \\ \|\theta - \theta_{0}\| \leq \eta_{n}}} |H_{n,\e} [m_{\theta}
    - m_{\theta_{0}}]| + \tilde{\phi}_{n,2}^{*} (\eta_{n}).
  \end{split}
\end{equation*}
By \eqref{eqn--m-est-stoch-eq} and Markov's inequality, the sum on the right is
\(\Op^{\dagger} (1/n)\) for every \(\eta_{n} \to 0\), and it is
\(\op^{\dagger} (1/n)\) if \(\eta_{n} = O (n^{-1/2})\).
By \Cref{lem--bootstrap-2nd-order-control}, the same two bounds hold for
\(\tilde{\phi}_{n,2}^{*} (\eta_{n})\).
This verifies \Cref{asm--extremum-stoch-diff}%
\eqref{asm--extremum-stoch-diff-rem-local-asymp-equicont} under either weighting
scheme.

By \Cref{thm--extremum-diff-rate}, \(\hat{\theta}_{n}^{*} - \theta_{0} =
A_{0}^{-1} \hat{\Delta}_{n}^{*} + \op (n^{-1/2})\) where the stochastic order is
unconditional, as asserted.
\qed

\Needspace{8\baselineskip}
\subsection{Proofs of auxiliary results used in proofs of
\Cref{thm--bootstrap} and \texorpdfstring{\Cref{rem--plain-bootstrap}}{Remark \ref{rem--plain-bootstrap}}}
\subsubsection{Proof of \Cref{lem--strong-cond-weak-suff}}
\label{sec--prf--lem--strong-cond-weak-suff}

Let \(\mathrm{BL}_{1} := \mathrm{BL}_{1} (\R^{d})\).
By the Arzel\`a--Ascoli theorem, \(\mathrm{BL}_{1}\) is compact and
metrisable, hence separable, in the topology of uniform convergence on
compacta; as probability measures on \(\R^{d}\) are
tight, there is a countable \(\calD_{0} \subseteq \mathrm{BL}_{1}\) such that
\(\dBL (Q_{1}, Q_{2}) = \sup_{g \in \calD_{0}} |(Q_{1} - Q_{2}) g|\) for all
probability measures \(Q_{1}\) and \(Q_{2}\) on \(\R^{d}\).
Let \(D_{n} = \dBL (\mathcal{L} (\tau_{n} \mid \mathscr{U}_{n}), Q)\), which is
therefore measurable.
By hypothesis \(D_{n} \pto 0\), and \(0 \leq D_{n} \leq 2\), hence \(\E [D_{n}]
\to 0\), and by Markov's inequality \(\E [D_{n} \mid \mathscr{W}_{n}] \pto 0\).
For every \(g \in \calD_{0}\), as \(\mathscr{W}_{n} \subseteq
\mathscr{U}_{n}\), the tower property gives, almost surely,
\(|\E [g (\tau_{n}) \mid \mathscr{W}_{n}] - Q g| = |\E [\E [g (\tau_{n}) \mid
\mathscr{U}_{n}] - Q g \mid \mathscr{W}_{n}]| \leq \E [D_{n} \mid
\mathscr{W}_{n}]\).
Taking the supremum over the countable \(\calD_{0}\) gives
\(\dBL (\mathcal{L} (\tau_{n} \mid \mathscr{W}_{n}), Q) \leq \E [D_{n} \mid
\mathscr{W}_{n}] \pto 0\).
\qed

\subsubsection{Proof of \Cref{lem--bootstrap-cond-weak-conv}}
\label{sec--prf--lem--bootstrap-prev-cond-weak-conv}

By \eqref{eqn--boot-wgt-indep} in \Cref{asm--boot-wgt}, the vectors
\((\hatE_{n,k}^{*} - \hatE_{n,k}) [\Delta_{k}]\), \(k = 1, \dots, K\), are
conditionally independent given \(\mathscr{U}_{n}\).
Fix \(k\) and let \(\bar{\Delta}_{n,k} = N_{n,k}^{-1} \sum_{i = 1}^{N_{n,k}}
\Delta_{k,i}\).
By \Cref{asm--exchangeable} and
\Cref{asm--m-est-general}\eqref{asm--m-est-stoch-diff}, the \(\Delta_{k,i}\),
\(i \in \N\), are i.i.d., have mean zero and have finite second moments, so
\(\bar{\Delta}_{n,k} \to 0\) almost surely by the strong law of large numbers,
as \(N_{n,k} \geq n \to \infty\).

\emph{Multiplier weights.}
By \eqref{eqn--hajek-weighted-avg},
\begin{equation*}
  (\hatE_{n,k}^{*} - \hatE_{n,k}) [\Delta_{k}] = \frac{1}{N_{n,k}} \sum_{i =
  1}^{N_{n,k}} (\xi_{k,i} - \bar{\xi}_{n,k}) \Delta_{k,i} =
  \frac{1}{N_{n,k}} \sum_{i = 1}^{N_{n,k}} \xi_{k,i}
  \left( \Delta_{k,i} - \bar{\Delta}_{n,k} \right).
\end{equation*}
By \Cref{asm--boot-wgt}, the multipliers \(\xi_{k,i}\) are i.i.d., have
mean zero and variance one, and are independent of \(\mathscr{U}_{n}\).
By Lemma~2.9.5 of \citet[pp.~270--271]{2023vandervaartWeakConvergenceEmpirical},
which holds for almost every sequence \(\Delta_{k,1}, \Delta_{k,2}, \dots\),
\begin{equation*}
  N_{n,k}^{-1/2} \sum_{i = 1}^{N_{n,k}} \xi_{k,i} \Delta_{k,i} \dtop N
  \left(0, \E [\Delta_{k} \Delta_{k}']\right) \quad
  \text{conditionally given } \mathscr{U}_{n}.
\end{equation*}
To account for the empirical centring, write
\begin{equation*}
  N_{n,k}^{-1/2} \sum_{i = 1}^{N_{n,k}} \xi_{k,i} (\Delta_{k,i} -
  \bar{\Delta}_{n,k}) =
  N_{n,k}^{-1/2} \sum_{i = 1}^{N_{n,k}} \xi_{k,i} \Delta_{k,i} -
  \bar{\Delta}_{n,k} N_{n,k}^{-1/2} \sum_{i = 1}^{N_{n,k}}
  \xi_{k,i}.
\end{equation*}
The second term converges to zero in conditional probability given
\(\mathscr{U}_{n}\), since \(\bar{\Delta}_{n,k} \to 0\) almost surely and the
scalar sum \(N_{n,k}^{-1/2} \sum_{i} \xi_{k,i}\) has conditional mean zero and
variance one.
As \(n / N_{n,k} \to \lambda_{k}\), the conditional version of Slutsky's lemma gives
\begin{equation}
  n^{1/2} (\hatE_{n,k}^{*} - \hatE_{n,k}) [\Delta_{k}] \dtop
  N \left(0, \lambda_{k} \E [\Delta_{k} \Delta_{k}']\right) \quad
  \text{conditionally given } \mathscr{U}_{n}.
  \label{eqn--boot-clt-k}
\end{equation}

\emph{Pigeonhole weights.}
Let \(B_{n,k} = \{\bar{\xi}_{n,k} > 1/2\}\).
By Chebyshev's inequality and the independence of the weights from
\(\mathscr{U}_{n}\), \(\Pr \{B_{n,k}^{c} \mid \mathscr{U}_{n}\} \leq 4 / N_{n,k}
\leq 4 / n \to 0\), so, by the conditional version of Slutsky's lemma, it suffices to work on
\(B_{n,k}\).
On \(B_{n,k}\), by \eqref{eqn--hajek-weighted-avg},
\begin{equation*}
  (\hatE_{n,k}^{*} - \hatE_{n,k}) [\Delta_{k}] = \frac{1}{\bar{\xi}_{n,k}
  N_{n,k}} \sum_{i = 1}^{N_{n,k}}
  (\xi_{k,i} - \bar{\xi}_{n,k}) \Delta_{k,i} =
  \frac{1}{\bar{\xi}_{n,k} N_{n,k}} \sum_{i = 1}^{N_{n,k}} (\xi_{k,i} - 1)
  \left(\Delta_{k,i} - \bar{\Delta}_{n,k}\right).
\end{equation*}
By \Cref{asm--boot-wgt}, the \(\xi_{k,i} - 1\) are i.i.d., have mean zero and
variance one, and are independent of \(\mathscr{U}_{n}\).
Lemma~2.9.5 of \citet{2023vandervaartWeakConvergenceEmpirical} and the
centring argument above, with \(\xi_{k,i} - 1\) in place of \(\xi_{k,i}\),
give
\(N_{n,k}^{-1/2} \sum_{i} (\xi_{k,i} - 1) (\Delta_{k,i} - \bar{\Delta}_{n,k})
\dtop N (0, \E [\Delta_{k} \Delta_{k}'])\) conditionally given
\(\mathscr{U}_{n}\).
As \(\bar{\xi}_{n,k} \to 1\) in conditional probability given
\(\mathscr{U}_{n}\), by Chebyshev's inequality, and \(n / N_{n,k} \to
\lambda_{k}\), the conditional version of Slutsky's lemma again gives \eqref{eqn--boot-clt-k}.

\emph{Sum over \(k\).}
By the conditional independence noted at the start, \eqref{eqn--boot-clt-k}
for \(k = 1, \dots, K\) and \eqref{eqn--bootstrap-derivatives} give
\begin{equation*}
  n^{1/2} \tilde{\Delta}^{*}_{n} = n^{1/2} \sum_{k = 1}^{K}
  (\hatE_{n,k}^{*} - \hatE_{n,k}) [\Delta_{k}] \dtop N (0, \Sigma),
\end{equation*}
conditionally given \(\mathscr{U}_{n}\).
\qed

\subsubsection{Proof of \Cref{lem--bootstrap-1st-order-control}}
\label{sec--prf--lem--bootstrap-1st-order-control}

For a non-stochastic sequence \(\eta_{n} \to 0\), define
\begin{equation*}
  \calR_{n,k} = \left\{\frac{r_{k,\theta}}{\|\theta - \theta_{0}\|
  \left(\|\theta - \theta_{0}\| + n^{-1/2}\right)} : \theta \in \Theta,\
  0 < \|\theta - \theta_{0}\| \leq \eta_{n}\right\}.
\end{equation*}
For a linear functional \(T\), write \(\|T\|_{\calR} = \sup_{r \in \calR}
|T [r]|\), and for the purpose of this proof, define \(Z_{k,i} [r] = r (U_{i
\cdot \e_{k}})\).
Showing \eqref{eqn--bootstrap-1st-order-control} is equivalent to showing that
\begin{equation*}
  \Eout \left[ \left\|\frac{1}{N_{n,k}} \sum_{i = 1}^{N_{n,k}}
  (\xi_{k,i} - \overline{\xi}_{n,k})
  Z_{k,i} \right\|_{\calR_{n,k}} \right] = o (1).
\end{equation*}
By the bounded-support condition in \Cref{asm--boot-wgt}, there are
non-stochastic, positive and finite \(L_{k}\) such that \(|\xi_{k,i}| \leq
L_{k}\) almost surely.
It follows that \(|\overline{\xi}_{n,k}| \leq L_{k}\) as well.
By the triangle inequality,
\begin{equation*}
  \begin{split}
    \Eout \left[ \left\|\frac{1}{N_{n,k}} \sum_{i = 1}^{N_{n,k}}
    (\xi_{k,i} - \overline{\xi}_{n,k}) Z_{k,i} \right\|_{\calR_{n,k}} \right]
     &\leq \Eout \left[ \left\|\frac{1}{N_{n,k}} \sum_{i = 1}^{N_{n,k}}
    \xi_{k,i} Z_{k,i} \right\|_{\calR_{n,k}} \right] \\
    &\quad+ L_{k} \Eout \left[ \left\|\frac{1}{N_{n,k}} \sum_{i
    = 1}^{N_{n,k}} Z_{k,i} \right\|_{\calR_{n,k}} \right] \\
     &= \Eout \left[ \left\|\frac{1}{N_{n,k}} \sum_{i = 1}^{N_{n,k}}
    \xi_{k,i} Z_{k,i} \right\|_{\calR_{n,k}} \right] + o (1),
  \end{split}
\end{equation*}
where the last equality follows by \eqref{eqn--m-est-stoch-diff}.
Thus, it suffices to show
\begin{equation}
  \Eout \left[ \left\|\frac{1}{N_{n,k}} \sum_{i = 1}^{N_{n,k}} \xi_{k,i}
  Z_{k,i} \right\|_{\calR_{n,k}} \right] = o (1).
  \label{eqn--bootstrap-1st-order-control-suffice}
\end{equation}
As processes on \(\calR_{n,k}\), \(Z_{k,i}\) are i.i.d.\ across \(i\),
independent of the weights and mean zero since \(\E [r_{k,\theta}] = 0\).
Hence, regardless of the mean of the weights, the processes \(r \mapsto
\xi_{k,i} Z_{k,i}\) are also i.i.d.\ across \(i\) and mean zero.
Let \(\{\varepsilon_{i}\}_{i=1}^{\infty}\) be i.i.d.\ Rademacher variables,
independent of the data and weights.
By \Cref{sec--P3-common}, the suprema over \(\calR_{n,k}\) below are
measurable, so outer expectations are expectations.
The contraction step is applied with \(\{(\xi_{k,i},Z_{k,i})\}_i\) fixed,
and the resulting bound is integrated by Fubini's theorem.
The distributional identity below holds jointly with
\(\{Z_{k,i}\}_i\), because the \(\varepsilon_i\,\mathrm{sgn}(\xi_{k,i})\)
are i.i.d.\ Rademacher variables independent of
\(\{(|\xi_{k,i}|,Z_{k,i})\}_i\).
By \citet[Lemma~2.3.6, p.~175]{2023vandervaartWeakConvergenceEmpirical},
\begin{equation*}
  \begin{split}
    \Eout \left[ \left\| \frac{1}{N_{n,k}} \sum_{i = 1}^{N_{n,k}}
    \xi_{k,i} Z_{k,i} \right\|_{\calR_{n,k}} \right]  &\leq 2 \Eout \left[ \left\| \frac{1}{N_{n,k}} \sum_{i = 1}^{N_{n,k}}
    \varepsilon_{i} \xi_{k,i} Z_{k,i} \right\|_{\calR_{n,k}} \right] \\
     &= 2 L_{k} \Eout \left[ \left\| \frac{1}{N_{n,k}} \sum_{i = 1}^{N_{n,k}}
    \varepsilon_{i} \frac{|\xi_{k,i}|}{L_{k}} Z_{k,i} \right\|_{\calR_{n,k}}
    \right].
  \end{split}
\end{equation*}
The equality above follows by \(\varepsilon_{i} |\xi_{k,i}|
\overset{d}{=} \varepsilon_{i} \xi_{k,i}\).
By Proposition~A.1.11 of
\citet[pp.~600--601]{2023vandervaartWeakConvergenceEmpirical},
\begin{equation*}
  \Eout \left[ \left\| \frac{1}{N_{n,k}} \sum_{i = 1}^{N_{n,k}}
  \xi_{k,i} Z_{k,i} \right\|_{\calR_{n,k}} \right] \leq
  2 L_{k} \Eout \left[ \left\| \frac{1}{N_{n,k}} \sum_{i = 1}^{N_{n,k}}
  \varepsilon_{i} Z_{k,i} \right\|_{\calR_{n,k}} \right].
\end{equation*}
Then \eqref{eqn--bootstrap-1st-order-control-suffice} follows by
Lemma~2.3.6 again:
\begin{equation*}
  \Eout \left[ \left\| \frac{1}{N_{n,k}} \sum_{i = 1}^{N_{n,k}} \xi_{k,i}
  Z_{k,i} \right\|_{\calR_{n,k}} \right] \leq 4 L_{k} \Eout \left[ \left\|
  \frac{1}{N_{n,k}} \sum_{i = 1}^{N_{n,k}} Z_{k,i} \right\|_{\calR_{n,k}}
  \right] = o (1).
\end{equation*}

Note that
\begin{equation*}
  \tilde{R}_{n,1}^{*} (\theta) =
  \sum_{k=1}^{K} (\hatE_{n,k}^{*} - \hatE_{n,k}) [r_{k,\theta}] =
  \sum_{k=1}^{K} \frac{1}{N_{n,k}} \sum_{i_{k}=1}^{N_{n,k}} (w_{n,k,i_{k}} -
  1) r_{k,\theta} (U_{i_{k} \cdot \e_{k}}),
\end{equation*}
for the weights \(w_{n,k,i_{k}}\) defined in \eqref{eqn--hajek-weighted-avg}.
For multiplier weights, \(w_{n,k,i}-1=\xi_{k,i}-\overline\xi_{n,k}\); let
\(B_n\) be the whole sample space.
For pigeonhole weights, let \(B_{n} = \{\min_{1 \leq k \leq K}
\overline{\xi}_{n,k} > 1/2\}\); on \(B_n\),
\(w_{n,k,i}-1=\overline\xi_{n,k}^{-1}(\xi_{k,i}-\overline\xi_{n,k})\) with
\(0<\overline\xi_{n,k}^{-1}\le2\), and, since \(\E [\xi_{k,i}] =
\var[\xi_{k,i}] = 1\), Chebyshev's inequality and the union bound give
\(\Pr \{B_{n}^{c}\} \leq 4\sum_{k = 1}^{K} N_{n,k}^{-1} \leq (4K) / n \to 0\).
In both cases the triangle inequality gives
\begin{equation*}
  \Eout \left[ \tilde{\phi}_{n,1}^{*} (\eta_{n}) \ind_{B_{n}} \right]
  \leq 2 \sum_{k=1}^{K} \Eout \left[ \sup_{\substack{\theta \in \Theta : \\ 0
  < \|\theta - \theta_{0}\| \leq \eta_{n}}} \frac{\bigl|N_{n,k}^{-1}
  \sum_{i = 1}^{N_{n,k}} (\xi_{k,i} - \overline{\xi}_{n,k}) r_{k,\theta} (U_{i
  \cdot \e_{k}}) \bigr|}{\|\theta - \theta_{0}\|
  (\|\theta - \theta_{0}\| + n^{-1/2})} \right].
\end{equation*}
Thus we have \(\tilde{\phi}_{n,1}^{*} (\eta_{n}) \ind_{B_{n}} = \op^{\dagger}
(1)\) by \eqref{eqn--bootstrap-1st-order-control} and Markov's inequality,
and hence \(\tilde\phi^*_{n,1}(\eta_n)=\op^\dagger(1)\), since \(\ind_{B_n}\)
is zero-one with \(\ind_{B_n}\pto1\).
\qed

\subsubsection{Proof of \Cref{lem--bootstrap-2nd-order-control}}
\label{sec--prf--lem--bootstrap-2nd-order-control}

Since the expression for \(\tilde{Z}_{n,2}^{*} (\theta)\)
in \eqref{eqn--bootstrap-hoeffding-decomps-and-rems} is linear in
\(m_{\theta}\), write
\begin{equation*}
  \tilde{Z}_{n,2}^{*} (\theta) - \tilde{Z}_{n,2}^{*} (\theta_{0}) = \sum_{k =
  2}^{K} \sum_{\e \in \calE_{k}} (H_{n,\e}^{*} - H_{n,\e}) [m_{\theta} -
  m_{\theta_{0}}].
\end{equation*}
For \(\eta > 0\), let \(\calM (\eta) = \{m_{\theta} - m_{\theta_{0}} : \theta
\in \Theta, \|\theta - \theta_{0}\| \leq \eta\}\).
From the above display and the triangle inequality, we can bound
\(\tilde{\phi}_{n,2}^{*} (\eta)\) in \eqref{eqn--bootstrap-tildephinjstar-def}
by
\begin{equation*}
  \tilde{\phi}_{n,2}^{*} (\eta) \leq \sum_{k=2}^{K} \sum_{\e \in \calE_{k}}
  \|H_{n,\e}^{*}\|_{\calM (\eta)} + \sum_{k=2}^{K} \sum_{\e \in \calE_{k}}
  \|H_{n,\e}\|_{\calM (\eta)}.
\end{equation*}
By \eqref{eqn--m-est-stoch-eq} in \Cref{asm--m-est-general}\eqref{asm--m-est-stoch-eq} together with Markov's inequality,
it follows that for each \(k \in \{2, \dots, K\}\) and \(\e \in \calE_{k}\),
\begin{equation*}
  [\eta_{n} \to 0 \implies \|H_{n,\e}\|_{\calM (\eta_{n})} = \Op^{\dagger}
  (1/n)] \quad \text{and} \quad
  [\eta_{n} = O (n^{-1/2}) \implies \|H_{n,\e}\|_{\calM (\eta_{n})} =
  \op^{\dagger} (1/n)].
\end{equation*}
Thus, to prove \Cref{lem--bootstrap-2nd-order-control}, it suffices to show that
for each \(k \in \{2, \dots, K\}\) and \(\e \in \calE_{k}\),
\begin{equation}
  \begin{gathered}
    \text{if} \quad \eta_{n} \to 0, \quad \text{then} \quad
    \|H_{n,\e}^{*}\|_{\calM (\eta_{n})} = \Op^{\dagger}
    (1/n), \\
    \text{if} \quad \eta_{n} = O (n^{-1/2}), \quad \text{then} \quad
    \|H_{n,\e}^{*}\|_{\calM (\eta_{n})} = \op^{\dagger} (1/n).
  \end{gathered}
  \label{eqn--bootstrap-2nd-order-control-suff}
\end{equation}
Henceforth, fix \(k \in \{2, \dots, K\}\) and \(\e \in \calE_{k}\).
We first establish a lemma that compares \(H_{n,\e}\) to partially
weighted averages, weighted by the raw bootstrap weights
\(\xi_{\ell,i}\).

\begin{lemma}
\label{lem--bootstrap-2nd-order-partial-control}
Let the conditions of \Cref{thm--m-est-non-smooth} hold, and
fix \(k \in \{2, \dots, K\}\) and \(\e \in \calE_{k}\).
Under \Cref{asm--boot-wgt},
fix \(A \subseteq \supp (\e)\), possibly empty, and define%
\begin{equation}
  H_{n,\e}^{A} [f] = \frac{1}{|\calI_{n,\e}|} \sum_{\i \in \calI_{n,\e}} \left(
  \prod_{\ell \in A} \xi_{\ell,i_{\ell}} \right) \pi_{\i\odot\e} [f], \quad
  \text{so
  that} \quad H_{n,\e}^{\varnothing} = H_{n,\e}.
  \label{eqn--bootstrap-2nd-order-partial}
\end{equation}
Normalise the empty product to equal one.
Let \(L_{\ell} \in (0, \infty)\) be such that
\(|\xi_{\ell,i}| \leq L_{\ell}\) almost surely.
For \(\eta \in (0, \infty)\) and \(\calM (\eta) = \{m_{\theta} - m_{\theta_{0}} :
\theta \in \Theta, \|\theta - \theta_{0}\| \leq \eta\}\),
\begin{equation}
  \Eout \left[ \|H_{n,\e}^{A}\|_{\calM (\eta)} \right] \leq \left( \prod_{\ell
  \in A} 4 L_{\ell} \right) \Eout \left[ \|H_{n,\e}\|_{\calM (\eta)} \right].
  \label{eqn--bootstrap-2nd-order-partial-control}
\end{equation}
\end{lemma}

The proof is in \Cref{sec--prf--lem--bootstrap-2nd-order-partial-control}.

For pigeonhole weights, let \(B_{n} = \{\min_{1 \leq \ell \leq K}
\bar{\xi}_{n,\ell} > 1/2\}\) as in the proof of \Cref{lem--bootstrap-1st-order-control}, so that
\(\Pr \{B_{n}^{c}\} \leq (4K) / n \to 0\).
On \(B_{n}\), the pigeonhole weight definition gives
\begin{equation*}
  H_{n,\e}^{*} = \left( \prod_{\ell \in \supp (\e)} \bar{\xi}_{n,\ell}^{-1}
  \right) H_{n,\e}^{\supp (\e)}, \quad \text{with} \quad 0 < \prod_{\ell \in
  \supp (\e)} \bar{\xi}_{n,\ell}^{-1} \leq 2^{k}.
\end{equation*}
Hence \eqref{eqn--bootstrap-2nd-order-partial-control} with \(A = \supp (\e)\)
implies
\begin{equation}
  \Eout \left[ \ind_{B_{n}} \|H_{n,\e}^{*}\|_{\calM (\eta_{n})} \right] \leq
  2^{k} \left( \prod_{\ell \in \supp (\e)} 4 L_{\ell} \right) \Eout \left[
  \|H_{n,\e}\|_{\calM (\eta_{n})} \right].
  \label{eqn--bootstrap-2nd-order-pigeonhole-comparison}
\end{equation}

For multiplier weights, recall from \eqref{eqn--hajek-weighted-avg}
and \eqref{eqn--hoeffding-proj-weighted-avg} the additive definition
\[w_{n,\i\odot\e} = 1 + \sum_{\ell \in \supp (\e)} (\xi_{\ell,i_{\ell}} -
\overline{\xi}_{n,\ell}).\]
Thus, for every \(f \in \calM (\eta)\),
\begin{equation*}
  H_{n,\e}^{*} [f] = \left( 1 - \sum_{\ell \in \supp (\e)} \bar{\xi}_{n,\ell}
  \right) H_{n,\e} [f] + \sum_{\ell \in \supp (\e)} H_{n,\e}^{\{\ell\}} [f].
\end{equation*}
Apply the triangle inequality, \(|\bar{\xi}_{n,\ell}| \leq L_{\ell}\) and
\eqref{eqn--bootstrap-2nd-order-partial-control} with \(A = \{\ell\}\) to get
\begin{equation}
  \begin{split}
    \Eout \left[ \|H_{n,\e}^{*}\|_{\calM (\eta_{n})} \right]  &\leq \left( 1 + \sum_{\ell \in \supp (\e)} L_{\ell} \right) \Eout \left[
    \|H_{n,\e}\|_{\calM (\eta_{n})} \right] + \sum_{\ell \in \supp (\e)} \Eout
    \left[ \|H_{n,\e}^{\{\ell\}}\|_{\calM (\eta_{n})} \right] \\
     &\leq \left( 1 + 5 \sum_{\ell \in \supp (\e)} L_{\ell} \right) \Eout \left[
    \|H_{n,\e}\|_{\calM (\eta_{n})} \right].
  \end{split}
  \label{eqn--bootstrap-2nd-order-jackknife-comparison}
\end{equation}

By \eqref{eqn--m-est-stoch-eq}, \(n\Eout[\|H_{n,\e}\|_{\calM(\eta_n)}]\)
is \(O(1)\) if \(\eta_n\to0\) and \(o(1)\) if \(\eta_n=O(n^{-1/2})\).
Let \(B_n\) be the whole sample space for multiplier weights, and let
\(C_\e=2^k\prod_{\ell\in\supp(\e)}4L_\ell\) for pigeonhole weights and
\(C_\e=1+5\sum_{\ell\in\supp(\e)}L_\ell\) for multiplier weights.
For every \(t>0\), Markov's inequality and
\eqref{eqn--bootstrap-2nd-order-pigeonhole-comparison} or
\eqref{eqn--bootstrap-2nd-order-jackknife-comparison} give
\begin{equation*}
  \Prout \{n \|H_{n,\e}^{*}\|_{\calM (\eta_{n})} > t\} \leq \Pr \{B_{n}^{c}\} +
  \frac{C_{\e}}{t} n \Eout \left[ \|H_{n,\e}\|_{\calM (\eta_{n})} \right].
\end{equation*}
When \(\eta_{n} \to 0\), take the limit superior as \(n \to \infty\), then let
\(t \to \infty\), to obtain \(\|H_{n,\e}^{*}\|_{\calM (\eta_{n})} =
\Op^{\dagger} (1/n)\) under either weighting scheme.
When \(\eta_{n} = O (n^{-1/2})\), keep any \(t > 0\) fixed and let \(n \to
\infty\) to obtain \(\|H_{n,\e}^{*}\|_{\calM (\eta_{n})} = \op^{\dagger}
(1/n)\).
The summands are finitely many for fixed \(K\) and so by
\eqref{eqn--bootstrap-2nd-order-control-suff},
we have proved both claims of \Cref{lem--bootstrap-2nd-order-control}.
\qed

\subsubsection{Proof of \Cref{lem--bootstrap-2nd-order-partial-control}}
\label{sec--prf--lem--bootstrap-2nd-order-partial-control}
The proof uses only \Cref{asm--exchangeable,asm--boot-wgt} and
\Cref{asm--m-est-general}\eqref{asm--m-est-meas}, and no property of
\(\calM(\eta)\) beyond an integrable envelope and the measurability of suprema
over Borel index sets (\Cref{sec--P3-common}).
Hence \eqref{eqn--bootstrap-2nd-order-partial-control} holds with
\(\calM(\eta)\) replaced by any class of the form
\(\{m_\theta:\theta\in B'\}\) with \(B'\subseteq\Theta\) Borel (for example
\(\calM\)) or \(\{m_{\theta'}-m_\theta:(\theta,\theta')\in B\}\) with
\(B\subseteq\Theta\times\Theta\) Borel, provided it has an integrable
envelope, and also for \(\e\in\calE_1\): then
each slice below has a single term and \(\mathscr D_{n,\e,\ell}\) is generated
by the weights alone, so the argument reduces to the proof of
\Cref{lem--bootstrap-1st-order-control}.

Take any \(\ell \in A\).
Let \(\mathscr{D}_{n,\e,\ell}\) be the \(\sigma\)-algebra generated by
\begin{equation}
  \{\xi_{j,i_{j}} : j \in \{1, \dots, K\} \setminus \{\ell\},\ 1 \leq
  i_{j} \leq N_{n,j}\}
  \cup
  \{U_{\i \odot \e'} : \bm{0} \neq \e' \leq \e,\ e'_{\ell} = 0,\ \i \in
  \calI_{n}\}.
\end{equation}
Thus, \(\mathscr{D}_{n,\e,\ell}\) is generated by all bootstrap weights
outside dimension \(\ell\) and all latent variables used by the projection
whose index supports exclude \(\ell\).
Define the `slice' processes by%
\footnote{The dependence of \(Y_{n,\ell,t}\) on the fixed \(\e\) and \(A\) is
suppressed in the notation --- these processes are defined for the sole purpose
of proving \eqref{eqn--bootstrap-2nd-order-partial-control}.}
\begin{equation*}
  Y_{n,\ell,t} [f] = \frac{N_{n,\ell}}{|\calI_{n,\e}|} \sum_{\substack{\i \in
  \calI_{n,\e} : \\ i_{\ell} = t}} \left( \prod_{j \in A \setminus \{\ell\}}
  \xi_{j,i_{j}} \right) \pi_{\i\odot\e} [f], \quad f \in \calM (\eta), \quad t
  \in \{1, \dots, N_{n,\ell}\}.
\end{equation*}
Conditionally on \(\mathscr{D}_{n,\e,\ell}\), \(Y_{n,\ell,t}\) are independent
across \(t\).
Indeed, every latent variable entering the slice indexed by \(t\) that is not
already conditioned on has \(\ell\)th index equal to \(t\).
Distinct slices therefore depend on disjoint collections of remaining
independent latent variables.
The raw bootstrap weights \(\{\xi_{\ell,t}\}\) are independent of both the
slices \(\{Y_{n,\ell,t}\}\) and \(\mathscr{D}_{n,\e,\ell}\), so the weighted
slices have conditional mean zero whenever the slices do, under either
weighting scheme, regardless of the mean of \(\xi_{\ell,t}\).
By the slice definition,
\begin{equation}
  H_{n,\e}^{A} = \frac{1}{N_{n,\ell}} \sum_{t = 1}^{N_{n,\ell}} \xi_{\ell,t}
  Y_{n,\ell,t}, \quad \text{and} \quad H_{n,\e}^{A \setminus \{\ell\}} =
  \frac{1}{N_{n,\ell}} \sum_{t = 1}^{N_{n,\ell}} Y_{n,\ell,t}.
  \label{eqn--bootstrap-partial-weighted-avg-hoeffding-slice}
\end{equation}
Conditionally on \(\mathscr{D}_{n,\e,\ell}\), the slices \(Y_{n,\ell,t}\)
and the weighted slices \(\xi_{\ell,t}Y_{n,\ell,t}\), \(t=1,\dots,N_{n,\ell}\),
are independent mean-zero processes indexed by \(\calM(\eta)\), and the average
of the \(Y_{n,\ell,t}\) is \(H_{n,\e}^{A\setminus\{\ell\}}\) (see
\eqref{eqn--bootstrap-partial-weighted-avg-hoeffding-slice}).
Conditioning on \(\mathscr D_{n,\e,\ell}\) amounts to fixing its generating
variables, since these are independent of the remaining latent variables and
of \(\{\xi_{\ell,t}\}_t\) (Fubini's theorem).
The steps below need, at a fixed value, mean zero of the slices for all
\(f\in\calM(\eta)\) at once; we verify it directly, with the jointly measurable
versions of the projections described after \Cref{asm--m-est-general}.
M\"obius inversion of the recursive definition gives
\(\pi_{\i\odot\e}[f]=\sum_{\e'\le\e}(-1)^{\bm{\iota}'(\e-\e')}
P_{\i\odot\e'}[f]\), with \(P_{\i\odot\bm{0}}[f]:=0\).
At a fixed value, the latent variables in \(\pi_{\i\odot\e}[f]\) that remain
random are the \(U_{\i\odot\e''}\) with \(\e''\le\e\) and \(e''_\ell=1\).
Integrating them out maps \(P_{\i\odot\e'}[f]\) to
\(P_{\i\odot(\e'-e'_\ell\e_\ell)}[f]\) by Fubini's theorem, so the terms for
\(\e'\) and \(\e'+\e_\ell\) with \(e'_\ell=0\) cancel.
With \(F\) an integrable envelope of \(\calM(\eta)\), these steps are valid
for every \(f\in\calM(\eta)\) at each fixed value at which
\(\int F(\tau(\bm U_{\i\odot\e''},u'))\,\d u'<\infty\) for all
\(\i\in\calI_{n,\e}\) and \(\bm{0}\neq\e''\le\e\) with \(e''_\ell=0\); by
Tonelli's theorem, the other values form a null set, which does not depend
on \(f\).
For every fixed value outside this null set we apply the symmetrisation,
contraction and
desymmetrisation steps of the proof of \Cref{lem--bootstrap-1st-order-control},
with \(\varepsilon_i\), \(\xi_{k,i}\), \(Z_{k,i}\), \(L_k\), \(N_{n,k}\) and
\(\calR_{n,k}\) replaced by \(\varepsilon_t\), \(\xi_{\ell,t}\),
\(Y_{n,\ell,t}\), \(L_\ell\), \(N_{n,\ell}\) and \(\calM(\eta)\); these steps use only independence and mean zero of the
summands, not identical distribution.
The identity \(\varepsilon_t|\xi_{\ell,t}|\overset{d}{=}
\varepsilon_t\xi_{\ell,t}\) holds jointly with \(\{Y_{n,\ell,t}\}_t\) and
\(\mathscr D_{n,\e,\ell}\), because the \(\varepsilon_t\,
\mathrm{sgn}(\xi_{\ell,t})\) are i.i.d.\ Rademacher variables independent of
\(\{|\xi_{\ell,t}|\}_t\), \(\{Y_{n,\ell,t}\}_t\) and \(\mathscr D_{n,\e,\ell}\).
The suprema are measurable (\Cref{sec--P3-common}), so integrating the
conditional bound gives
\begin{equation*}
  \Eout \left[ \|H_{n,\e}^{A}\|_{\calM (\eta)} \right] \leq 4 L_{\ell} \Eout
  \left[ \|H_{n,\e}^{A \setminus \{\ell\}}\|_{\calM (\eta)} \right].
\end{equation*}
Iterating over the finite set \(A\), removing one index dimension at a time,
proves \eqref{eqn--bootstrap-2nd-order-partial-control}.
\qed

\subsubsection{\texorpdfstring{Consistency of the bootstrap estimators}{Consistency of the bootstrap estimators}}

\begin{lemma}
\label{lem--boot-consistent}
Suppose \Cref{asm--exchangeable,asm--boot-wgt} and
\Cref{asm--m-est-general}\eqref{asm--m-est-stoch-eq}--\eqref{asm--m-est-meas}
hold.
\begin{enumerate}[(a)]
  \item \(\sup_{\theta\in\Theta}|\hat M^*_n(\theta)-\hat M_n(\theta)|\pto0\).
  \item Let \(\tilde M^*_n\) be as in \eqref{eqn--m-est-boot-def}, with
    \(\tilde A_n\) measurable and symmetric and \(\tilde A_n\pto A_0\) for
    some symmetric \(A_0\succ0\).
    If \(\tilde\theta^*_n\in\Theta\) is measurable and
    \(\tilde M^*_n(\tilde\theta^*_n)\ge\sup_{\theta\in\Theta}\tilde
    M^*_n(\theta)-\tilde\zeta^*_n\) with \(\tilde\zeta^*_n=\op(1)\), then
    \(\tilde\theta^*_n-\hat\theta_n\pto0\).
  \item If \(\theta_0\in\Theta\) satisfies
    \(\sup_{\theta\in\Theta:\|\theta-\theta_0\|\ge\varepsilon}M_0(\theta)<
    M_0(\theta_0)\) for every \(\varepsilon>0\), and \(\hat\theta^*_n\) satisfies
    \eqref{eqn--m-est-boot-def-basic} with \(\hat\zeta^*_n=\op(1)\) in place
    of \(\op(1/n)\), then \(\hat\theta^*_n\pto\theta_0\).
\end{enumerate}
\end{lemma}

\begin{proof}[Proof of \Cref{lem--boot-consistent}]
Let \(L_\ell\) bound \(|\xi_{\ell,i}|\).
The suprema below are measurable (\Cref{sec--P3-common}).
For pigeonhole weights let \(B_n:=\{\min_k\bar\xi_{n,k}>1/2\}\), so that
\(\Pr(B_n^c)\le4K/n\) by Chebyshev's inequality, as in the proof of
\Cref{lem--bootstrap-1st-order-control}; for multiplier weights let
\(B_n\) be the whole sample space.

\emph{Step 1: comparison.}
For every \(\e\in\calE\),
\begin{equation}
  \E\big[\ind_{B_n}\|H^*_{n,\e}-H_{n,\e}\|_{\calM}\big]\le\Big(1+\prod_{\ell\in\supp(\e)}8L_\ell
  +5\sum_{\ell\in\supp(\e)}L_\ell\Big)\E\|H_{n,\e}\|_{\calM}.
  \label{eqn--P3-boot-comparison}
\end{equation}
By the remark at the start of the proof of
\Cref{lem--bootstrap-2nd-order-partial-control}, that lemma holds with \(\calM\) in place
of \(\calM(\eta)\) and for every \(\e\in\calE\).
For multiplier weights,
\(w_{n,\i\odot\e}=1+\sum_{\ell\in\supp(\e)}(\xi_{\ell,i_\ell}-\bar\xi_{n,\ell})\)
gives \(H^*_{n,\e}-H_{n,\e}=\sum_{\ell\in\supp(\e)}(H^{\{\ell\}}_{n,\e}-
\bar\xi_{n,\ell}H_{n,\e})\), as in
\eqref{eqn--bootstrap-2nd-order-jackknife-comparison}; since
\(|\bar\xi_{n,\ell}|\le L_\ell\), that lemma
with \(A=\{\ell\}\) bounds the left-hand side of \eqref{eqn--P3-boot-comparison} by
\(5\sum_{\ell}L_\ell\,\E\|H_{n,\e}\|_{\calM}\).
For pigeonhole weights, on \(B_n\),
\(H^*_{n,\e}=(\prod_{\ell\in\supp(\e)}\bar\xi_{n,\ell}^{-1})
H^{\supp(\e)}_{n,\e}\) with \(0<\prod_{\ell}\bar\xi_{n,\ell}^{-1}\le
2^{|\supp(\e)|}\); the same lemma with
\(A=\supp(\e)\) and the triangle inequality bound the left-hand side by
\((1+\prod_{\ell}8L_\ell)\,\E\|H_{n,\e}\|_{\calM}\), as in
\eqref{eqn--bootstrap-2nd-order-pigeonhole-comparison}.

\emph{Step 2: maximal inequality.}
By \Cref{asm--m-est-general}\eqref{asm--m-est-stoch-eq}, \(\calM\) is of
VC-type with envelope \(m_*\) and \(\|m_*\|_{P_0,2}<\infty\).
The maximal inequality of \citet[Theorem~4]{chen2026cross} applies to every
\(\e\in\calE\), including \(\e\in\calE_1\), and, by
their Lemma~5, its entropy integral is bounded in terms of
\(K\) and the VC characteristics.
As in the discussion of \eqref{eqn--m-est-stoch-eq} following
\Cref{thm--m-est-non-smooth}, there is therefore a constant \(C\), depending
only on \(K\) and the VC characteristics, such that, for every \(\e\in\calE\),
\begin{equation*}
  \E\|H_{n,\e}\|_{\calM}\le C\|m_*\|_{P_0,2}\,|\calI_{n,\e}|^{-1/2}
  \le C\|m_*\|_{P_0,2}\,n^{-1/2}.
\end{equation*}

\emph{Part (a).}
By \eqref{eqn--P3-boot-hoeffding}, \(\sup_{\theta}|\hat M^*_n(\theta)-\hat
M_n(\theta)|\le\sum_{\e\in\calE}\|H^*_{n,\e}-H_{n,\e}\|_{\calM}\).
By Steps 1 and 2 and Markov's inequality, \(\ind_{B_n}\) times the right-hand side
is \(\Op(n^{-1/2})\); since \(\Pr(B_n)\to1\), (a) follows, indeed with rate
\(n^{-1/2}\).
Likewise, by \eqref{eqn--gen-hoeffding-decomp} and Step 2,
\(\sup_\theta|\hat M_n(\theta)-M_0(\theta)|\le\sum_{\e\in\calE}
\|H_{n,\e}\|_{\calM}=\Op(n^{-1/2})\).

\emph{Part (b).}
Let \(D_n:=\sup_\theta|\hat M^*_n(\theta)-\hat M_n(\theta)|\).
Since \(\hat\theta_n\in\Theta\),
\(\tilde M^*_n(\tilde\theta^*_n)\ge\tilde M^*_n(\hat\theta_n)-\tilde\zeta^*_n
\ge-D_n-\tilde\zeta^*_n\), whilst
\(\tilde M^*_n(\tilde\theta^*_n)\le D_n-\frac12\lambda_{\min}(\tilde A_n)
\|\tilde\theta^*_n-\hat\theta_n\|^2\).
Hence \(\lambda_{\min}(\tilde A_n)\|\tilde\theta^*_n-\hat\theta_n\|^2\le
4D_n+2\tilde\zeta^*_n=\op(1)\) by (a).
By Weyl's inequality, \(|\lambda_{\min}(\tilde A_n)-\lambda_{\min}(A_0)|\le
\|\tilde A_n-A_0\|\pto0\), and \(\lambda_{\min}(A_0)>0\); so
\(\tilde\theta^*_n-\hat\theta_n\pto0\).

\emph{Part (c).}
By (a) and the bound on \(\sup_\theta|\hat M_n(\theta)-M_0(\theta)|\) in the proof of part (a),
\(D'_n:=\sup_\theta|\hat M^*_n(\theta)-M_0(\theta)|\pto0\).
Since \(\theta_0\in\Theta\),
\(M_0(\hat\theta^*_n)\ge\hat M^*_n(\hat\theta^*_n)-D'_n\ge\hat
M^*_n(\theta_0)-\hat\zeta^*_n-D'_n\ge M_0(\theta_0)-\hat\zeta^*_n-2D'_n\).
For \(\varepsilon>0\), let \(\varrho_\varepsilon:=M_0(\theta_0)-
\sup_{\|\theta-\theta_0\|\ge\varepsilon}M_0(\theta)>0\).
Then \(\{\|\hat\theta^*_n-\theta_0\|\ge\varepsilon\}\subseteq
\{\hat\zeta^*_n+2D'_n\ge\varrho_\varepsilon\}\), whose probability tends to zero.

The constants in Steps 1 and 2 and the bound on \(\Pr(B_n^c)\) depend on the
distribution of the data only through \(\|m_*\|_{P_0,2}\) and the VC
characteristics of \(\calM\).
Hence (a) and (b) also hold along a drifting sequence \(\{P_n\}\) satisfying
the assumptions with \(\sup_n\|m_*\|_{P_n,2}<\infty\) and VC characteristics
not depending on \(n\), as in
\Cref{sec--two-scale}.
\end{proof}

\subsection{Proof of \Cref{thm--extremum-diff-rate}}
\label{sec--prf--thm--extremum-diff-rate}

By \Cref{asm--extremum-consistent-pop-quad-decay}\eqref{asm--extremum-pop-quad-decay-strong},
\begin{equation*}
  \limsup_{n \to \infty} \Pr \{\underline{\lambda}_{n} \leq 0\} \leq
  \lim_{\varepsilon \downarrow 0} \limsup_{n \to \infty}
  \Pr \{\underline{\lambda}_{n} < \varepsilon\} = 0.
\end{equation*}
Thus \(\Pr \{\underline{\lambda}_{n} > 0\} \to 1\).

In what follows, let \(\phi_{n,j}^{\dagger} (\eta)\) denote the minimal
measurable majorant for \(\phi_{n,j} (\eta)\) for \(j \in \{0, 1, 2\}\).
We break the proof into two steps:
\begin{enumerate}[Step (i)]
  \item \label{thm--extremum-diff-rate-Op}
    \emph{Show that \(\|\theta_{n} - \theta_{n,0}\| = \Op (\delta_{n})\)}.
  \item \label{thm--extremum-diff-rate-lin}
    \emph{Show that if \(\|\theta_{n} - \theta_{n,0}\| = \Op (\delta_{n})\),
    then \(\theta_{n} - \theta_{n,0} = A_{n,0}^{- 1} \overline{\Delta}_{n} + \op
    (\delta_{n})\)}.
\end{enumerate}
For both, we will use a common decomposition.
For every \(\theta \in \Theta\), \(M_{n} (\theta_{n}) - M_{n} (\theta) +
\zeta_{n} \geq 0\).
Combine this with \(M_{n} = M_{n,0} + Z_{n,1} + Z_{n,2}\) to see that
\begin{equation*}
  0 \leq M_{n,0} (\theta_{n}) - M_{n,0} (\theta) + Z_{n,1} (\theta_{n}) -
  Z_{n,1} (\theta) + Z_{n,2} (\theta_{n}) - Z_{n,2} (\theta) + \zeta_{n} \quad
  \text{for all} \quad \theta \in \Theta.
\end{equation*}
Apply a quadratic expansion to \(M_{n,0} (\theta_{n})
- M_{n,0} (\theta_{n,0})\) and \(M_{n,0} (\theta) - M_{n,0} (\theta_{n,0})\)
using \eqref{eqn--extremum-pop-quad-decay-strong},
and a linear expansion to \(Z_{n,1} (\theta_{n}) - Z_{n,1} (\theta)\) using
\eqref{eqn--extremum-stoch-diff} to get, for all \(\theta \in \Theta\),
\begin{equation}
  \begin{split}
    0  &\leq - \frac{1}{2} (\theta_{n} - \theta_{n,0})' A_{n,0} (\theta_{n} -
    \theta_{n,0}) + \frac{1}{2} (\theta - \theta_{n,0})' A_{n,0} (\theta -
    \theta_{n,0}) \\
    &\quad+ (\theta_{n} - \theta)' \overline{\Delta}_{n} + R_{n,1} (\theta_{n}) -
    R_{n,1} (\theta) + \|\theta_{n} - \theta_{n,0}\|^{2} R_{n,0} (\theta_{n})
    \\
    &\quad- \|\theta - \theta_{n,0}\|^{2} R_{n,0} (\theta) + Z_{n,2}
    (\theta_{n}) - Z_{n,2} (\theta) + \zeta_{n}.
  \end{split}
  \label{eqn--extremum-diff-rate-2}
\end{equation}
Let \(\{\eta_{n}\} \subseteq (0, \infty)\) be a non-stochastic sequence chosen
to ensure:
\begin{equation}
  \eta_{n} \to 0, \quad
  \Pr \{\theta_{n} \in C_{n,1}^{c}\} \to 0 \quad \text{for} \quad
  C_{n,1} = \{\theta \in \R^{d} : \|\theta - \theta_{n,0}\| \leq \eta_{n}\}, \quad
  \text{and} \quad \delta_{n} / \eta_{n} \to 0.
  \label{eqn--extremum-diff-rate-eta-n}
\end{equation}
By \Cref{asm--extremum-consistent-pop-quad-decay}\eqref{asm--extremum-consistency}
and \Cref{lem--det-rate-exists}, some
non-stochastic sequence in \((0, \infty)\) satisfies the first two conditions,
and \(\eta_{n}\), the larger of its \(n\)th term and \(\sqrt{\delta_{n}}\),
satisfies all three.

\subsubsection{Proof of Step (\ref{thm--extremum-diff-rate-Op}):
\texorpdfstring{\(\|\theta_{n} - \theta_{n,0}\| = \Op (\delta_{n})\)}{}}

At \(\theta = \theta_{n,0}\), \eqref{eqn--extremum-diff-rate-2} becomes
\begin{equation*}
  \begin{split}
    0  &\leq - \frac{1}{2} (\theta_{n} - \theta_{n,0})' A_{n,0} (\theta_{n} -
    \theta_{n,0}) + (\theta_{n} - \theta_{n,0})' \overline{\Delta}_{n} \\
    &\quad+ R_{n,1} (\theta_{n}) + \|\theta_{n} - \theta_{n,0}\|^{2} R_{n,0}
    (\theta_{n}) + Z_{n,2} (\theta_{n}) - Z_{n,2} (\theta_{n,0}) + \zeta_{n}.
  \end{split}
\end{equation*}
Let \(a_{n,0} = \underline{\lambda}_{n} / 6\).
By a standard eigenvalue norm bound, the Cauchy--Schwarz and triangle
inequalities, and the definition of \(\phi_{n,1}\) in
\eqref{eqn--extremum-stoch-diff}, on the event \(\{\theta_{n} \in C_{n,1}\}\) (on which \(|R_{n,1} (\theta_{n})| \leq \|\theta_{n} - \theta_{n,0}\| (\|\theta_{n} - \theta_{n,0}\| + \delta_{n}) \phi_{n,1}^{\dagger} (\eta_{n})\)),
\begin{equation}
  \begin{split}
    0  &\leq - (3 a_{n,0} - |R_{n,0} (\theta_{n})| - \phi_{n,1}^{\dagger}
    (\eta_{n}))
    \|\theta_{n} - \theta_{n,0}\|^{2} + \|\theta_{n} - \theta_{n,0}\|
    (\|\overline{\Delta}_{n}\| + \delta_{n} \phi_{n,1}^{\dagger} (\eta_{n})) \\
    &\quad+ |Z_{n,2} (\theta_{n}) - Z_{n,2} (\theta_{n,0})| + \zeta_{n}.
  \end{split}
  \label{eqn--extremum-diff-rate-3}
\end{equation}
For an event \(B_{n,1}\) such that \(\Pr \{B_{n,1}^{c}\} \to 0\), it suffices to
show that
\begin{equation}
  \lim_{L \to \infty} \limsup_{n \to \infty} \Pr \left[ \left\{ \|\theta_{n} -
  \theta_{n,0}\| > L \delta_{n} \right\} \cap B_{n,1} \right] = 0.
  \label{eqn--extremum-diff-rate-suff-1}
\end{equation}
Define \(B_{n,1} = \{\underline{\lambda}_{n} > 0, \theta_{n} \in C_{n,1},
\max_{j \in \{0, 1\}} \phi_{n,j}^{\dagger} (\eta_{n}) \leq a_{n,0}\}\).
Then \(\Pr \{B_{n,1}^{c}\} \to 0\), by
\Cref{asm--extremum-consistent-pop-quad-decay}\eqref{asm--extremum-pop-quad-decay-strong},
\Cref{asm--extremum-stoch-diff}\eqref{asm--extremum-stoch-diff-phin1-small} and
\eqref{eqn--extremum-diff-rate-eta-n}, together with
\Cref{lem--stoch-compare-op1-inv-Op1} applied to
\(6\max_{j\in\{0,1\}}\phi_{n,j}^{\dagger}(\eta_n)\) and
\(\underline\lambda_n\vee0\).

On \(B_{n,1}\), \(|Z_{n,2} (\theta_{n}) - Z_{n,2} (\theta_{n,0})|
\leq \phi_{n,2}^{\dagger} (\eta_{n})\).
Thus \eqref{eqn--extremum-diff-rate-3} can be specialised to the following
inequality on \(B_{n,1}\):
\begin{equation}
  a_{n,0} \|\theta_{n} - \theta_{n,0}\|^{2} -
  \|\theta_{n} - \theta_{n,0}\| (\|\overline{\Delta}_{n}\| + \delta_{n}
  \phi_{n,1}^{\dagger} (\eta_{n})) - \phi_{n,2}^{\dagger} (\eta_{n}) - \zeta_{n}
  \leq 0.
  \label{eqn--extremum-diff-rate-4}
\end{equation}
This is a quadratic inequality in \(\|\theta_{n} - \theta_{n,0}\|\).
For \(a > 0\) and \(b, c \geq 0\), the set equality holds:
\begin{equation*}
  \{x \in [0, \infty) : a x^{2} - b x - c \leq 0\} = \left[ 0, \left( b +
  \sqrt{b^{2} + 4 a c} \right)/(2 a) \right].
\end{equation*}
Combine the above with \(\sqrt{x_{1} + x_{2}} \leq \sqrt{x_{1}} + \sqrt{x_{2}}\)
for all real \(x_{1}, x_{2} \geq 0\) and apply to
\eqref{eqn--extremum-diff-rate-4} to get, on \(B_{n,1}\),
\begin{equation*}
  \|\theta_{n} - \theta_{n,0}\| \leq \frac{\|\overline{\Delta}_{n}\| + \delta_{n}
  \phi_{n,1}^{\dagger} (\eta_{n})}{a_{n,0}}
  + \sqrt{\frac{\phi_{n,2}^{\dagger} (\eta_{n})}{a_{n,0}}}
  + \sqrt{\frac{\zeta_{n}}{a_{n,0}}}.
\end{equation*}
Extending these terms by zero outside \(B_{n,1}\), the first two are
\(\Op (\delta_{n})\), because \(\phi_{n,1}^{\dagger} (\eta_{n}) \leq a_{n,0}\)
on \(B_{n,1}\), \((\underline{\lambda}_{n}\vee0)^{-1} = \Op (1)\), and, by
\Cref{asm--extremum-stoch-diff}\eqref{asm--extremum-stoch-diff-Delta-Op-delta}
and \eqref{asm--extremum-stoch-diff-rem-local-asymp-equicont} with
\(\eta_{n} \to 0\), \(\overline{\Delta}_{n} = \Op (\delta_{n})\) and
\(\phi_{n,2}^{\dagger} (\eta_{n}) = \Op (\delta_{n}^{2})\), whereas the third is
\(\op (\delta_{n})\) as \(\zeta_{n} = \op (\delta_{n}^{2})\), which implies
\eqref{eqn--extremum-diff-rate-suff-1}.
\qed

\subsubsection{Proof of Step (\ref{thm--extremum-diff-rate-lin}):
If \texorpdfstring{\(\|\theta_{n} - \theta_{n,0}\| = \Op (\delta_{n})\)}{}, then
\texorpdfstring{\(\theta_{n} - \theta_{n,0} = A_{n,0}^{- 1}
\overline{\Delta}_{n} + \op (\delta_{n})\)}{}}

Define \(h_{n} = \theta_{n} - \theta_{n,0}\), \(\overline{h}_{n} = A_{n,0}^{- 1}
\overline{\Delta}_{n}\).
Let \(D_{n} = \{\theta_{n,0} + \overline{h}_{n} \in \Theta\} \cap
\{\underline{\lambda}_{n} > 0\} \cap \{\theta_{n,0} + \overline{h}_{n} \in
C_{n,1}\}\).
By hypothesis, \(\Prout \{\theta_{n,0} + \overline{h}_{n} \notin \Theta\}
\to 0\).
Furthermore, \(\Pr \{\theta_{n,0} + \overline{h}_{n} \in C_{n,1}\} \to 1\) since
\((\underline{\lambda}_{n}\vee0)^{-1} = \Op (1)\), \(\overline{\Delta}_{n} = \Op (\delta_{n})\)
and \((\eta_{n}, \delta_{n} / \eta_{n}) \to (0, 0)\).
Together with \(\Pr \{\underline{\lambda}_{n} > 0\} \to 1\), this gives
\(\Prout \{D_{n}^{c}\} \to 0\), by subadditivity of outer probability.
Let \(\overline{\theta}_{n} = \theta_{n,0} + \ind_{D_{n}}
\overline{h}_{n}\).
Then by construction, \(\overline{\theta}_{n} \in \Theta\) always.

Note that
\begin{equation*}
  \theta_{n} - \theta_{n,0} - A_{n,0}^{-1} \overline{\Delta}_{n} = h_{n} -
  \overline{h}_{n} = \ind_{D_{n}}
  (h_{n} - \overline{h}_{n}) + \ind_{D_{n}^{c}} (h_{n} -
  \overline{h}_{n}).
\end{equation*}
We will show that the left-hand side is \(\op (\delta_{n})\).
Since \(\ind_{D_{n}^{c}}\) is a \(\{0, 1\}\)-valued sequence and \(\Prout
\{D_{n}^{c}\} \to 0\), \(\|h_{n} - \overline{h}_{n}\|
\ind_{D_{n}^{c}} = \op^{\dagger} (\delta_{n})\).
Thus, as \(h_{n} - \overline{h}_{n}\) is measurable, it suffices to show that \(\|h_{n} - \overline{h}_{n}\| \ind_{D_{n}} = \op^{\dagger}
(\delta_{n})\).

By \eqref{eqn--extremum-diff-rate-2}
applied to \(\theta = \overline{\theta}_{n}\), on \(D_{n}\),
\begin{equation}
  \begin{split}
    0  &\leq - \frac{1}{2} h_{n}' A_{n,0} h_{n} + \frac{1}{2} \overline{h}_{n}' A_{n,0}
    \overline{h}_{n} + (h_{n} - \overline{h}_{n})' \overline{\Delta}_{n}
    \\
    &\quad+ R_{n,1} (\theta_{n}) - R_{n,1} (\overline{\theta}_{n})
    + \|h_{n}\|^{2} R_{n,0} (\theta_{n}) - \|\overline{h}_{n}\|^{2} R_{n,0}
    (\overline{\theta}_{n}) \\
    &\quad+ Z_{n,2} (\theta_{n}) - Z_{n,2} (\overline{\theta}_{n})
    + \zeta_{n}.
  \end{split}
\end{equation}
Rearrange the above using \(\overline{h}_{n}' \overline{\Delta}_{n} =
\overline{h}_{n}' A_{n,0} \overline{h}_{n}\) and \(h_{n}'
\overline{\Delta}_{n} = h_{n}' A_{n,0} \overline{h}_{n}\), and bound the
remainders by their absolute values:
\begin{equation*}
  \begin{split}
    \frac{1}{2} (h_{n} - \overline{h}_{n})' A_{n,0} (h_{n} - \overline{h}_{n})
    \ind_{D_{n}}
    &= \frac{1}{2} (h_{n}' A_{n,0} h_{n} - 2 h_{n}' A_{n,0}
    \overline{h}_{n} + \overline{h}_{n}' A_{n,0} \overline{h}_{n})
    \ind_{D_{n}}\\
    &\leq \zeta_{n}^{+},
  \end{split}
\end{equation*}
where \(\zeta_{n}^{+} = 0\) on \(D_{n}^{c}\) and, on \(D_{n}\),
\begin{equation}
  \begin{split}
    \zeta_{n}^{+}  &= \zeta_{n} + \|h_{n}\|^{2} |R_{n,0} (\theta_{n})| +
    \|\overline{h}_{n}\|^{2} |R_{n,0} (\overline{\theta}_{n})| + |R_{n,1}
    (\theta_{n})| + |R_{n,1} (\overline{\theta}_{n})| \\
    &\quad+ |Z_{n,2} (\theta_{n}) - Z_{n,2} (\theta_{n,0})| + |Z_{n,2}
    (\overline{\theta}_{n}) - Z_{n,2} (\theta_{n,0})|.
  \end{split}
\end{equation}
Note that \(\zeta_{n}^{+}\) is non-negative, and on \(D_{n}\),
\(\|h_{n} - \overline{h}_{n}\| \leq \sqrt{(2 \zeta_{n}^{+}) / \underline{\lambda}_{n}}\).
Thus by \((\underline{\lambda}_{n}\vee0)^{-1} = \Op (1)\)
(\Cref{asm--extremum-consistent-pop-quad-decay}\eqref{asm--extremum-pop-quad-decay-strong}),
it suffices to show that \(\zeta_{n}^{+} = \op^{\dagger} (\delta_{n}^{2})\).
By Step (\ref{thm--extremum-diff-rate-Op}) (now proved), \(\|h_{n}\| = \Op
(\delta_{n})\).
Furthermore, \(\|\overline{h}_{n}\| = \Op (\delta_{n})\) by
\(\overline{\Delta}_{n} = \Op (\delta_{n})\)
(\Cref{asm--extremum-stoch-diff}) and \((\underline{\lambda}_{n}\vee0)^{-1} = \Op (1)\).
By \eqref{eqn--extremum-pop-quad-decay-strong}, since
\(|R_{n,0}(\theta_n)|\le\phi_{n,0}(\eta_n)\) on \(\{\theta_n\in C_{n,1}\}\),
\(|R_{n,0}(\overline\theta_n)|\le\phi_{n,0}(\eta_n)\) on \(D_n\) and
\(\phi_{n,0}(\eta_n)=\op^\dagger(1)\),
\(\|h_{n}\|^{2 }R_{n,0} (\theta_{n})\) and
\(\|\overline{h}_{n}\|^{2} R_{n,0} (\overline{\theta}_{n})\) are both \(\op^{\dagger}
(\delta_{n}^{2})\).
By \Cref{lem--extremum-stoch-diff-Rn-neg-rate},
\(R_{n,1} (\theta_{n,1}) = \op^{\dagger} (\delta_{n}^{2})\) for any sequence of
\(\Theta\)-valued maps \(\{\theta_{n,1}\}\) with
\(\|\theta_{n,1} - \theta_{n,0}\| = \Op^{\dagger} (\delta_{n})\).
By \Cref{lem--extremum-stoch-diff-rem-local-asymp-equicont-rate},
\(|Z_{n,2} (\theta_{n,1}) - Z_{n,2} (\theta_{n,0})| = \op^{\dagger} (\delta_{n}^{2})\) for
any sequence of \(\Theta\)-valued maps \(\{\theta_{n,1}\}\) with
\(\|\theta_{n,1} - \theta_{n,0}\| = \Op^{\dagger} (\delta_{n})\).
Deduce that \(\zeta_{n}^{+} = \op^{\dagger} (\delta_{n}^{2})\) since \(\zeta_{n} = \op
(\delta_{n}^{2})\) by hypothesis, and we have shown that each of its other six
summands is \(\op^{\dagger} (\delta_{n}^{2})\).
\qed

\subsection{Useful auxiliary results}

\begin{lemma}
\label{lem--det-rate-exists}
If \(\zeta_{n} \ptoout 0\), then there is a non-increasing sequence of
non-stochastic positive real numbers \(\{\eta_{n}\}\) such that \(\eta_{n} \to
0\) and \(\zeta_{n} = \op^{\dagger} (\eta_{n})\).
As a consequence, \(\Prout \{|\zeta_{n}| > \eta_{n}\} \to 0\).
\end{lemma}

\begin{proof}[Proof of \Cref{lem--det-rate-exists}]
The claim \(\Prout \{|\zeta_{n}| > \eta_{n}\} \to 0\) follows from
\(\zeta_{n} = \op^{\dagger} (\eta_{n})\).
Since \(\zeta_{n} \ptoout 0\), for any \(j \in \N\) there is \(N_{j} \in \N\)
such that
\begin{equation}
  \Prout \{|\zeta_{n}| > j^{-2}\} \leq j^{-1} \quad \text{for all } n \geq
  N_{j}.
  \label{eqn--det-rate-lemma-main}
\end{equation}
Set \(n_{1} = N_{1}\) and \(n_{j} = \max \{N_{j}, n_{j - 1} + 1\}\) for \(j
\geq 2\).
Then \(\{n_{j}\}\) is a strictly increasing sequence of natural numbers, and
\eqref{eqn--det-rate-lemma-main} still holds with \(n_{j}\) in place of
\(N_{j}\).
Next, set
\begin{equation*}
  \eta_{n} =
  \begin{cases}
    1 & \text{if } n < n_{2}, \\
    1 / j & \text{if } n_{j} \leq n < n_{j + 1},\ j \geq 2.
  \end{cases}
\end{equation*}
Then \(\{\eta_{n}\}\) is non-increasing, positive and tends to zero.
Take any \(\varepsilon \in (0, \infty)\), let \(J (\varepsilon)\) be an integer
with \(J (\varepsilon) \geq \varepsilon^{-1}\), and let \(n \geq n_{J
(\varepsilon)}\).
Then \(n_{j} \leq n < n_{j + 1}\) for some \(j \geq J (\varepsilon)\), so that
\(\eta_{n} = 1 / j\) (also when \(j = 1\)) and \(\varepsilon \eta_{n} =
\varepsilon / j \geq j^{-2}\), as \(\varepsilon \geq j^{-1}\).
Hence, by \eqref{eqn--det-rate-lemma-main},
\begin{equation*}
  \Prout \{|\zeta_{n}| > \varepsilon \eta_{n}\} \leq \Prout \{|\zeta_{n}| >
  j^{-2}\} \leq j^{-1} \leq J (\varepsilon)^{-1} \leq \varepsilon
  \qquad \text{for all } n \geq n_{J (\varepsilon)}.
\end{equation*}
By \Cref{lem--det-rate-exists-shortcut} below, applied to \(\bar{\zeta}_{n} =
\zeta_{n} / \eta_{n}\), it follows that \(\zeta_{n} = \op^{\dagger}
(\eta_{n})\).
\end{proof}

\begin{lemma}
\label{lem--det-rate-exists-shortcut}
Let \(\{\bar{\zeta}_{n}\}\) be real-valued maps.
Then \(\bar{\zeta}_{n} \ptoout 0\) if and only if for all \(\varepsilon \in
(0, \infty)\) there is \(N (\varepsilon) \in \N\) such that \(\Prout
\{|\bar{\zeta}_{n}| > \varepsilon\} \leq \varepsilon\) for all \(n \geq
N (\varepsilon)\).
\end{lemma}

\begin{proof}[Proof of \Cref{lem--det-rate-exists-shortcut}]
\(\bar{\zeta}_{n} \ptoout 0\) if and only if for every \(\varepsilon, \delta
\in (0, \infty)\) there is \(N_{1} (\varepsilon, \delta)\) such that \(\Prout
\{|\bar{\zeta}_{n}| > \varepsilon\} \leq \delta\) for all \(n \geq N_{1}
(\varepsilon, \delta)\).
The `only if' direction follows by setting \(N (\varepsilon) = N_{1}
(\varepsilon, \varepsilon)\).
For the converse, let \(\varepsilon, \delta \in (0, \infty)\), put \(\varepsilon_0 =
\min \{\varepsilon, \delta\}\) and \(N_{1} (\varepsilon, \delta) = N
(\varepsilon_0)\).
As \(\{|\bar{\zeta}_{n}| > \varepsilon\} \subseteq \{|\bar{\zeta}_{n}| >
\varepsilon_0\}\), \(\Prout \{|\bar{\zeta}_{n}| > \varepsilon\} \leq \Prout
\{|\bar{\zeta}_{n}| > \varepsilon_0\} \leq \varepsilon_0 \leq \delta\) for all \(n \geq N_{1}
(\varepsilon, \delta)\).
\end{proof}

\begin{lemma}
\label{lem--stoch-compare-op1-inv-Op1}
Let \(\zeta_{n} = \op (1)\), and let \(\lambda_{n} \geq 0\) satisfy
\(\lambda_{n}^{-1} = \Op (1)\), with the convention \(1/0 = \infty\).
Then \(\Pr \{\zeta_{n} \leq \lambda_{n}\} \to 1\) as \(n \to \infty\).
\end{lemma}

\begin{proof}[Proof of \Cref{lem--stoch-compare-op1-inv-Op1}]
Fix \(\varepsilon > 0\).
Since \(\lambda_{n}^{-1} = \Op (1)\), there exists \(\delta > 0\) such that
\(\Pr \{\lambda_{n} < \delta\} < \varepsilon\) for all sufficiently large \(n\).
Consequently, \(\Pr \{\zeta_{n} > \lambda_{n}\} \leq \Pr \{\zeta_{n} > \delta\}
+ \Pr \{\lambda_{n} < \delta\}\).
Since \(\zeta_{n} = \op (1)\), the first term tends to zero.
Thus, \(\limsup_{n \to \infty} \Pr \{\zeta_{n} > \lambda_{n}\} \leq
\varepsilon\).
The conclusion follows from the arbitrariness of \(\varepsilon > 0\).
\end{proof}

\begin{lemma}
\label{lem--extremum-stoch-diff-Rn-neg-rate}
Let \(\theta_{n,1}\) be \(\Theta\)-valued maps such that \(\|\theta_{n,1}
- \theta_{n,0}\| = \Op^{\dagger} (\delta_{n})\) for a non-stochastic positive
real sequence \(\{\delta_{n}\}\) with \(\delta_{n} \to 0\).
Let \(R_{n,1} (\cdot)\) satisfy \(R_{n,1} (\theta_{n,0}) = 0\) and
\begin{equation}
  \phi_{n,1} (\eta_{n}) = \op^{\dagger} (1) \quad \text{if} \quad \eta_{n} = o
  (1), \quad \text{where} \quad \phi_{n,1} (\eta) = \sup_{\substack{\theta \in
  \Theta : \\ 0 < \|\theta - \theta_{n,0}\| \leq \eta}}
  \frac{|R_{n,1} (\theta)|}{\|\theta - \theta_{n,0}\| (\|\theta - \theta_{n,0}\|
  + \delta_{n})}.
  \label{eqn--extremum-stoch-diff-Rn-neg-rate-cond}
\end{equation}
Then \(R_{n,1} (\theta_{n,1}) = \op^{\dagger} (\delta_{n}^{2})\).
\end{lemma}

\begin{proof}[Proof of \Cref{lem--extremum-stoch-diff-Rn-neg-rate}]
Fix \(\varepsilon \in (0, \infty)\).
For \(J \in (0, \infty)\), set \(B_{n, J} = \{ \|\theta_{n,1} -
\theta_{n,0}\| > J \delta_{n}\}\).
By the union bound,
\begin{equation*}
  \Prout \{|R_{n,1} (\theta_{n,1})| > \varepsilon \delta_{n}^{2}\} \leq
  \Prout \{|R_{n,1} (\theta_{n,1})| > \varepsilon \delta_{n}^{2}, \ \|\theta_{n,1}
  - \theta_{n,0}\| \leq J \delta_{n}\} + \Prout \left\{ B_{n, J} \right\}.
\end{equation*}
Given \(\theta \in \Theta\) with \(\|\theta - \theta_{n,0}\| \leq
\eta\), \(|R_{n,1} (\theta)| \leq \|\theta - \theta_{n,0}\| (\|\theta -
\theta_{n,0}\| + \delta_{n}) \phi_{n,1} (\eta) \leq \eta (\eta + \delta_{n})
\phi_{n,1} (\eta)\) (trivially if \(\theta = \theta_{n,0}\), as
\(R_{n,1} (\theta_{n,0}) = 0\)).
Therefore, on \(\{\|\theta_{n,1} - \theta_{n,0}\| \leq J \delta_{n}\}\),
\(|R_{n,1} (\theta_{n,1})| \leq
J \delta_{n} (J \delta_{n} + \delta_{n}) \phi_{n,1} (J \delta_{n})
= \delta_{n}^{2} J (J + 1) \phi_{n,1} (J \delta_{n})\).
Where, in addition, \(|R_{n,1} (\theta_{n,1})| > \varepsilon \delta_{n}^{2}\),
we have \(\phi_{n,1} (J \delta_{n}) > \varepsilon / \{J (J + 1)\}\).
Combined with the above display, this gives
\begin{equation*}
  \Prout \{|R_{n,1} (\theta_{n,1})| > \varepsilon \delta_{n}^{2}\} \leq
  \Prout \{\phi_{n,1} (J \delta_{n}) > \varepsilon / \{J (J + 1)\}\} +
  \Prout \left\{ B_{n, J} \right\}.
\end{equation*}
For fixed \(J \in (0, \infty)\), the first term tends to zero by
\eqref{eqn--extremum-stoch-diff-Rn-neg-rate-cond}, as \(J \delta_{n} \to 0\),
and hence
\begin{equation*}
  \limsup_{n \to \infty} \Prout \{|R_{n,1} (\theta_{n,1})| >
  \varepsilon \delta_{n}^{2}\} \leq \limsup_{n \to \infty} \Prout \left\{ B_{n,
  J} \right\} \quad \text{for every } J \in (0, \infty).
\end{equation*}
As \(\|\theta_{n,1} - \theta_{n,0} \| = \Op^{\dagger} (\delta_{n})\), \(\lim_{J
\to \infty} \limsup_{n \to \infty} \Prout \{B_{n, J}\} = 0\).
Since the left-hand side does not depend on \(J\), it is zero.
\end{proof}

\begin{lemma}
\label{lem--extremum-stoch-diff-rem-local-asymp-equicont-rate}
Let \Cref{asm--extremum-stoch-diff}\eqref{asm--extremum-stoch-diff-rem-local-asymp-equicont} hold and
\(\delta_{n}\) satisfy the conditions therein.
Then for any (possibly stochastic) \(\Theta\)-valued sequence \(\theta_{n,1} =
\theta_{n,0} + \Op^{\dagger} (\delta_{n})\), \(|Z_{n,2} (\theta_{n,1}) - Z_{n,2}
(\theta_{n,0})| = \op^{\dagger} (\delta_{n}^{2})\).
\end{lemma}

\begin{proof}[Proof of
\Cref{lem--extremum-stoch-diff-rem-local-asymp-equicont-rate}]
Fix \(\varepsilon \in (0, \infty)\).
For \(C \in (0, \infty)\), the absolute value of \(Z_{n,2} (\theta_{n,1}) -
Z_{n,2} (\theta_{n,0})\) is at most \(\phi_{n,2} (C \delta_{n})\) on
\(\{\|\theta_{n,1} - \theta_{n,0}\| \leq C \delta_{n}\}\), so
\begin{equation*}
  \begin{split}
    \Prout \{|Z_{n,2} (\theta_{n,1}) - Z_{n,2} (\theta_{n,0})| > \varepsilon
    \delta_{n}^{2}\}
    &\leq \Prout \{|Z_{n,2} (\theta_{n,1}) - Z_{n,2} (\theta_{n,0})| > \varepsilon
    \delta_{n}^{2},\ \|\theta_{n,1} - \theta_{n,0}\| \leq C \delta_{n}\} \\
    &\quad+ \Prout \{\|\theta_{n,1} - \theta_{n,0}\| > C \delta_{n}\} \\
    &\leq \Prout \{\phi_{n,2} (C \delta_{n}) > \varepsilon \delta_{n}^{2}\}
    + \Prout \{\|\theta_{n,1} - \theta_{n,0}\| > C \delta_{n}\}.
  \end{split}
\end{equation*}
The first term tends to zero by
\Cref{asm--extremum-stoch-diff}\eqref{asm--extremum-stoch-diff-rem-local-asymp-equicont},
since \(C \delta_{n} = O (\delta_{n})\).
Thus,
\begin{equation*}
  \limsup_{n \to \infty} \Prout \{|Z_{n,2} (\theta_{n,1}) - Z_{n,2}
  (\theta_{n,0})| > \varepsilon \delta_{n}^{2}\} \leq \limsup_{n \to \infty}
  \Prout \{\|\theta_{n,1} - \theta_{n,0}\| > C \delta_{n}\}.
\end{equation*}
As \(\|\theta_{n,1} - \theta_{n,0}\| = \Op^{\dagger} (\delta_{n})\), the
right-hand side tends to zero as \(C \to \infty\); the left-hand side does not
depend on \(C\), and hence is zero.
\end{proof}

\subsection{\texorpdfstring{Proofs of \Cref{thm--two-scale,thm--two-scale-boot}}{Proofs of Theorems \ref{thm--two-scale} and \ref{thm--two-scale-boot}}}
\label{sec--prf--two-scale}

\paragraph{Conventions.}
Throughout this subsection \(K=2\), so that \(\bm\iota=(1,1)\),
\(\pi_{\i}:=\pi_{\i\odot\bm\iota}\), and, unless stated otherwise, the
conditions of \Cref{thm--two-scale} hold along \(\{P_n\}\); \(\E\), \(\Pr\),
\(M_0\) and the projections are taken under \(P_n\).
The proofs are written in terms of
\[
  r_n:=|\calI_n|^{1/3}=(N_{n,1}N_{n,2})^{1/3},
\]
which keeps the identities relating \(r_n\), \(N_{n,1}\), \(N_{n,2}\) and
\(|\calI_n|\) exact; by
\Cref{asm--two-scale}\eqref{asm--two-scale-rates}, \(r_n\sim n^{2/3}\), the
rate in which \Cref{sec--two-scale} states its results, so the two are
interchangeable in every statement of order.
Uniformity along \(\{P_n\}\) means the following: (i)
\(\sup_n\E\|\Delta_k^{(n)}\|^2<\infty\);
(ii) the function \(\dot m\) in \Cref{asm--two-scale}\eqref{asm--two-scale-modulus} may
depend on \(n\), with \(\|\dot m\|_2:=\sup_n\{\E\dot m(\bm U^\circ)^2\}^{1/2}<\infty\),
and the VC characteristics there do not depend on \(n\); (iii) likewise, the
envelope \(m_*\) in \Cref{asm--m-est-general}\eqref{asm--m-est-stoch-eq} may
depend on \(n\), with \(\sup_n\E[m_*(W)^2]<\infty\) and VC characteristics
of \(\calM\) that do not depend on \(n\), as \Cref{lem--boot-consistent}
requires along \(\{P_n\}\); (iv) the \(o(\|\theta-\theta_0\|^2)\)
term in \Cref{asm--m-est-consistent-pop-quad-decay}\eqref{asm--m-est-smooth-pop}
is uniform in \(n\), and the expectation in \eqref{eqn--m-est-stoch-diff}, with
\(r_n^{-1}\) in place of \(n^{-1/2}\) as in
\Cref{asm--two-scale}\eqref{asm--two-scale-first} and taken under \(P_n\),
tends to zero for every non-stochastic \(\eta_n\to0\);
(v) the limit in \Cref{asm--two-scale}\eqref{asm--two-scale-cov}, with the
expectation taken under \(P_n\), holds along all sequences \(\delta_n>0\) with
\(\delta_n\to0\), not necessarily monotone, and \(\theta_n\to\theta_0\); and
(vi) the uniform integrability of
\(\{\|\Delta_k^{(n)}\|^2/\E\|\Delta_k^{(n)}\|^2\}\) required in
\Cref{sec--two-scale} holds.
Each statement concerns, for every \(n\), only the joint law of the \(n\)th
array, its bootstrap weights and the estimators \(\hat\theta_n\),
\(\hat\zeta_n\), \(\tilde A_n\), \(\tilde\theta^*_n\) and \(\tilde\zeta^*_n\),
so we realise all of them on one complete probability space:
\(W_{\i}=\tau(\bm U_{\i})\) with one array of latent uniforms and a
Borel measurable \(\tau\) that may depend on \(n\) (cf.\
\Cref{asm--exchangeable}), the weights of \Cref{asm--boot-wgt}, and the
estimators as measurable functions of these and of an independent uniform
variable that carries any auxiliary randomisation.
By \Cref{asm--m-est-general}\eqref{asm--m-est-meas} the suprema below are
measurable, so outer probabilities and expectations are ordinary ones.
The exception is the bounded Lipschitz distance between a conditional law and
a fixed law on a non-separable space such as \(\ell^\infty(\calK)\), which
need not be measurable.
There we use convergence in outer probability, \(\ptoout\), which holds if
and only if every subsequence has a further subsequence along which the
distance tends to zero outer almost surely
\citep[Lemma~1.9.2(ii)]{2023vandervaartWeakConvergenceEmpirical}.
We write \(\E_{\mathscr U}\), \(\Pr_{\mathscr U}\) and \(\mathcal L_{\mathscr U}\)
for expectation, probability and law conditional on \(\mathscr U_n\), the
\(\sigma\)-algebra generated by the latent variables of the \(n\)th array
(\Cref{sec--P3-common}); as the weights are independent of \(\mathscr U_n\),
these integrate over the weights with the latent variables fixed.
Let \(\underline\lambda>0\) be the smallest eigenvalue of \(A_0\).

\paragraph{Scales.}
In this subsection \(\sigma_n\) and \(c_n\) are read with \(r_n\) in place of
\(n^{2/3}\), that is \(\sigma_n^2=\|A_0^{-1}\Omega_nA_0^{-1}\|+r_n^{-2}\) and
\(c_n=(r_n\sigma_n)^{-1}\), which keeps the identities among \(\sigma_n\),
\(c_n\) and \(\Gamma_n\) exact; by
\Cref{asm--two-scale}\eqref{asm--two-scale-rates} each is the corresponding
quantity of \Cref{sec--two-scale} times a factor tending to one, so the two
readings give the same limits, orders and bounds, and the conclusions proved
here are those stated there.
By definition \(r_n^3=|\calI_n|\) and \(\sigma_n^2\ge r_n^{-2}\), so
\(\sigma_n\ge r_n^{-1}\) and \(c_n\le1\).
Since \(\Omega_n=A_0(A_0^{-1}\Omega_nA_0^{-1})A_0\), \(\|\Omega_n\|\le
\|A_0\|^2\sigma_n^2\).
As \(\hatE_{n,1}[\Delta_1]\) and \(\hatE_{n,2}[\Delta_2]\) are independent
averages of i.i.d.\ centred vectors, \(\var(\bar\psi_n)=\Omega_n\), and
\begin{equation}
  \E\|\bar\psi_n\|^2=\operatorname{tr}(\Omega_n)\le d\|A_0\|^2\sigma_n^2,
  \qquad
  \Big(\frac{\sigma_n}{|\calI_n|}\Big)^{1/2}=\sigma_n^{1/2}r_n^{-3/2}
  =\sigma_n^2(r_n\sigma_n)^{-3/2}=\sigma_n^2c_n^{3/2}.
  \label{eqn--P1a-orders}
\end{equation}
Also \(\sigma_n\to0\), because \(\|\Omega_n\|\le\sum_k\E\|\Delta_k^{(n)}\|^2/N_{n,k}
=O(n^{-1})\) and \(r_n\ge n^{2/3}\).
We use the following central limit theorem: if \(B_n\) are non-stochastic
\(d\times d\) matrices with \(\sup_n\sigma_n\|B_n\|<\infty\) and
\(B_n\Omega_nB_n'\to\Xi\), then
\begin{equation}
  B_n\bar\psi_n\leadsto N(0,\Xi).
  \label{eqn--P1a-clt}
\end{equation}
Indeed, \(B_n\bar\psi_n=\sum_{k=1}^2N_{n,k}^{-1}\sum_{i=1}^{N_{n,k}}
B_n\Delta_{k,i}\) is a sum of independent centred vectors with covariance
\(B_n\Omega_nB_n'\).
Let \(s_{n,k}^2:=\E\|\Delta_k^{(n)}\|^2\).
As \(s_{n,k}^2/N_{n,k}\le d\|\Omega_n\|\le d\|A_0\|^2\sigma_n^2\),
\(\|B_n\|^2s_{n,k}^2/N_{n,k}\le q_1\), where
\(q_1:=d\|A_0\|^2\sup_m\sigma_m^2\|B_m\|^2\), so the Lindeberg
sum of dimension \(k\) is at most
\[
  q_1\E[(\|\Delta_k^{(n)}\|^2/s_{n,k}^2)\ind\{\|\Delta_k^{(n)}\|^2/
  s_{n,k}^2>\epsilon^2N_{n,k}/q_1\}],
\]
which tends to zero by uniform
integrability (it is zero if \(s_{n,k}=0\)).

\paragraph{Local classes.}
For \(0<\eta\le\eta_0\), the class \(\calM(\eta)\) of
\Cref{lem--bootstrap-2nd-order-partial-control} is
\(\{m_\theta-m_{\theta_0}:\|\theta-\theta_0\|\le\eta\}\), as \(\calN\subseteq\Theta\).
By \Cref{asm--two-scale}\eqref{asm--two-scale-modulus} with
\(\theta=\theta_0\), \(\E[F_\eta(W_{(1,1)})\mid\bm U^\circ]\le\dot m(\bm U^\circ)\,\eta\); as
\(0\le F_\eta\le2\),
\begin{equation}
  \E F_\eta^2\le2\E F_\eta\le2\|\dot m\|_2\,\eta.
  \label{eqn--P1a-env}
\end{equation}
We also use the local VC-type condition in that part of the assumption: \(\calM(\eta)\) is VC-type
with envelope \(F_\eta\) and characteristics \((A_{\mathrm{loc}},v_{\mathrm{loc}})\) that do not depend on
\(\eta\) or \(n\).

\paragraph{The cell--interaction split.}
For \(\i\in\calI_n\), \(\bm U^\circ_{\i}=(U_{\i\odot\e_1},U_{\i\odot\e_2})
=(U_{(i_1,0)},U_{(0,i_2)})\) are the row and column shocks of cell \(\i\)
(\Cref{sec--two-scale}), so that \(\bm U^\circ=\bm U^\circ_{(1,1)}\); let
\(\mathscr G_n:=\sigma(\bm U^\circ_{\i}:\i\in\calI_n)\).
Then \(W_{\i}=\tau(\bm U^\circ_{\i},U_{\i\odot\bm\iota})\), where the cells'
own shocks \(U_{\i\odot\bm\iota}\) are i.i.d.\ and independent of \(\mathscr
G_n\).
For integrable \(f\) let \(\bar f(x):=\int_0^1f(\tau(x,w))\,\d w\), so that
\(\E[f(W_{\i})\mid\mathscr G_n]=\bar f(\bm U^\circ_{\i})\), and split
\(\pi_{\i}[f]=\mathsf C_{\i}[f]+\mathsf I_{\i}[f]\) into the cell's own noise
\(\mathsf C_{\i}[f]:=f(W_{\i})-\bar f(\bm U^\circ_{\i})\) and the
row--column interaction \(\mathsf I_{\i}[f]:=\bar
f(\bm U^\circ_{\i})-\E f-\pi_{\i\odot\e_1}[f]-\pi_{\i\odot\e_2}[f]\).
Since \(\E\bar f=\E f\) and \(\pi_{\i\odot\e_k}[\bar f]=\pi_{\i\odot\e_k}[f]\),
\(\mathsf I_{\i}[f]=\pi_{\i}[\bar f]\).
Hence
\begin{equation}
  H_{n,\bm\iota}[f]=\frac{1}{|\calI_n|}\sum_{\i\in\calI_n}\mathsf C_{\i}[f]
  +H_{n,\bm\iota}[\bar f],
  \label{eqn--P1a-split}
\end{equation}
where, conditionally on \(\mathscr G_n\), the \(\mathsf C_{\i}[f]\) are
independent across \(\i\) with mean zero.
\Cref{asm--two-scale}\eqref{asm--two-scale-modulus} holds for every cell by
exchangeability, so for \(\theta,\theta'\in\calN\) with
\(\|\theta'-\theta\|\le\eta_0\),
\begin{equation}
  \big|\overline{m_{\theta'}}-\overline{m_\theta}\big|(\bm U^\circ_{\i})
  \le\E\big[|m_{\theta'}(W_{\i})-m_\theta(W_{\i})|\bigm|\bm U^\circ_{\i}\big]
  \le\dot m(\bm U^\circ_{\i})\,\|\theta'-\theta\|.
  \label{eqn--P1a-lip}
\end{equation}
As the supremum over \(\theta'\) in that condition lies inside the conditional
expectation, applying it, with the version
\(\int_0^1|m_{\theta'}-m_\theta|(\tau(x,w))\,\d w\) of the conditional
expectation, with rational \(\theta\in\calN\) and rational
\(\delta\in(0,\eta_0]\), and using the triangle inequality (through the
midpoint of \(\theta\) and \(\theta'\), which lies in the ball \(\calN\), when
\(\|\theta'-\theta\|\) is close to \(\eta_0\)), shows that
\eqref{eqn--P1a-lip} holds for all \(\theta,\theta'\) whenever
\(\bm U^\circ_{\i}\) lies outside a single Lebesgue-null set.
On this null set the map \(\theta\mapsto\overline{m_\theta}\) need not be
Lipschitz, so we set \(\dot m:=\infty\) there; this changes no integral of \(\dot m\),
and, as there are countably many cells, almost surely no
\(\bm U^\circ_{\i}\) lies in it.
With this convention \eqref{eqn--P1a-lip} holds identically.

\paragraph{Localised processes.}
For \(\|u\|\le\eta_0r_n\) let \(f_u:=m_{\theta_0+u/r_n}-m_{\theta_0}\),
\(\mathbb W_n(u):=r_n^2H_{n,\bm\iota}[f_u]\) and, for multiplier
weights, \(\mathbb W^*_n(u):=r_n^2(H^*_{n,\bm\iota}-H_{n,\bm\iota})[f_u]\).
Since \(r_n^2/|\calI_n|=r_n^{-1}\), \eqref{eqn--P1a-split} gives \(\mathbb
W_n=\mathbb W^{\mathsf C}_n+\mathbb W^{\mathsf I}_n\) with
\[
  \mathbb W^{\mathsf C}_n(u):=r_n^{-1}\sum_{\i\in\calI_n}\mathsf C_{\i}[f_u],
  \qquad
  \mathbb W^{\mathsf I}_n(u):=r_n^2H_{n,\bm\iota}[\bar f_u].
\]

\paragraph{The limit process.}
By \Cref{asm--two-scale}\eqref{asm--two-scale-cov}, \(\Lambda(0)=0\);
replacing \(\delta\) by \(\gamma\delta\) there gives
\(\Lambda(\gamma v)=\gamma\Lambda(v)\) for \(\gamma>0\), and replacing \(\theta\) by \(\theta-\delta v\) gives
\(\Lambda(-v)=\Lambda(v)\).
By \Cref{asm--two-scale}\eqref{asm--two-scale-first}--\eqref{asm--two-scale-modulus},
\(\E(m_{\theta_0+\delta v}-m_{\theta_0})^2\le2\E|m_{\theta_0+\delta
v}-m_{\theta_0}|\le2\|\dot m\|_2\delta\|v\|\), so \(\Lambda(v)\le2\|\dot m\|_2\|v\|\).
Applying \Cref{asm--two-scale}\eqref{asm--two-scale-cov} at \(\theta\) and at
\(\theta+\delta t\), and using \(2xy=x^2+y^2-(x-y)^2\),
\begin{equation}
  \delta^{-1}\E[(m_{\theta+\delta s}-m_\theta)(m_{\theta+\delta t}-m_\theta)]
  \to C_{\mathbb W}(s,t)\qquad(\delta\to0,\ \theta\to\theta_0),
  \label{eqn--P1a-polar}
\end{equation}
along every sequence of \((\delta,\theta)\) as in (v) of the conventions,
with \(\E\) taken under \(P_n\).
Hence \(C_{\mathbb W}\) is positive semi-definite, as a pointwise limit of positive
semi-definite kernels.
Moreover \(C_{\mathbb W}(\gamma s,\gamma t)=\gamma C_{\mathbb W}(s,t)\) for \(\gamma>0\), \(C_{\mathbb W}(-s,-t)=C_{\mathbb W}(s,t)\), and, as
\(\Lambda\) is even, \(C_{\mathbb W}(a+s,a+t)-C_{\mathbb W}(a+s,a)-C_{\mathbb W}(a,a+t)+C_{\mathbb W}(a,a)=C_{\mathbb W}(s,t)\) for
\(a\in\R^d\).
A centred Gaussian process with covariance \(C_{\mathbb W}\) has \(\E|\mathbb
W(s)-\mathbb W(t)|^{2p}=C_p\Lambda(s-t)^p\le C_p(2\|\dot m\|_2\|s-t\|)^p\), so
Kolmogorov's criterion with \(p>d\) gives a version \(\mathbb W\) with
continuous paths.
By the identities for \(C_{\mathbb W}\), \(\mathbb W(0)=0\), \(\mathbb W(\gamma\,\cdot)\overset
d=\gamma^{1/2}\mathbb W\), \(\mathbb W(-\cdot)\overset d=\mathbb W\) and
\(\mathbb W(a+\cdot)-\mathbb W(a)\overset d=\mathbb W\) (stationary
increments).
By \citet[Lemma~2.5]{kim1990cube}, with their \(V=A_0/2\), almost surely
\(\mathbb W(v)-\frac14v'A_0v\to-\infty\) as \(\|v\|\to\infty\), since
\(A_0\succ0\), \(\mathbb W\) is centred Gaussian with continuous paths and
\(C_{\mathbb W}(\gamma s,\gamma t)=\gamma C_{\mathbb W}(s,t)\) for \(\gamma>0\).
Hence the process \(v\mapsto\mathbb W(v)-\frac12v'A_0v\) almost surely tends
to \(-\infty\) as \(\|v\|\to\infty\) and so attains its supremum; it is
Gaussian with continuous paths and \(\var\{\mathbb W(s)-\mathbb
W(t)\}=\Lambda(s-t)>0\) for \(s\ne t\), so its maximiser is almost surely
unique \citep[Lemma~2.6]{kim1990cube}.
In particular, \(\calZ\) is almost surely unique, and
\(-\calZ=\argmax_v\{\mathbb W(-v)-\frac12v'A_0v\}\overset d=\calZ\); the
argmax map is measurable on the set \(D\) of Step~4 of the proof of
\Cref{lem--Z-law}, which contains \(\mathbb W\) and \(\mathbb W(-\cdot)\)
almost surely.

\begin{lemma}[Maximal inequalities for the top-level component]
\label{lem--cell-maximal}
Let \(K=2\), and let \Cref{asm--exchangeable},
\Cref{asm--m-est-general}\eqref{asm--m-est-meas},
\(\sup_\theta|m_\theta|\le1\) and
\Cref{asm--two-scale}\eqref{asm--two-scale-modulus} hold along \(\{P_n\}\),
with (ii) of the conventions above.
For \(\eta,\delta>0\) with \(\eta+\delta\le\eta_0\), let
\(\calF(\eta,\delta):=\{m_{\theta'}-m_\theta:\|\theta-\theta_0\|\le\eta,\
\|\theta'-\theta\|\le\delta\}\).
There is \(C<\infty\), depending only on \(\|\dot m\|_2\), the local VC
characteristics and, for the second bound, \(L_1+L_2\), where
\(|\xi_{k,i}|\le L_k\) almost surely (so \(C\) does not depend on \(n\), \(\eta\)
or \(\delta\)), such that
\[
  \E\|H_{n,\bm\iota}\|_{\calF(\eta,\delta)}
  \le C\Big(\frac{\eta+\delta}{|\calI_n|}\Big)^{1/2}
\]
and, for weights of the multiplier type in \Cref{asm--boot-wgt},
\(\E\|H^*_{n,\bm\iota}-H_{n,\bm\iota}\|_{\calF(\eta,\delta)}
\le C\{(\eta+\delta)/|\calI_n|\}^{1/2}\).
\end{lemma}

\begin{proof}[Proof of \Cref{lem--cell-maximal}]
Let \(t:=\eta+\delta\).
If \(m_{\theta'}-m_\theta\in\calF(\eta,\delta)\), then
\(\|\theta-\theta_0\|\le\eta\) and
\(\|\theta'-\theta_0\|\le\|\theta'-\theta\|+\|\theta-\theta_0\|\le t\), so
\(m_{\theta'}-m_\theta=(m_{\theta'}-m_{\theta_0})-(m_\theta-m_{\theta_0})\in
\calM(t)-\calM(t)\), whose elements are bounded by \(2F_t\).
If \(\{g_j\}_{j\le N}\) is an \(\epsilon\|F_t\|_{Q,2}\)-net of \(\calM(t)\) for
\(\|\cdot\|_{Q,2}\), then \(\{g_j-g_l\}_{j,l\le N}\) is an
\(\epsilon\|2F_t\|_{Q,2}\)-net of \(\calM(t)-\calM(t)\).
By the local VC-type condition at radius \(t\le\eta_0\),
\(\calF(\eta,\delta)\) is therefore VC-type with envelope \(2F_t\) and
characteristics \((A_{\mathrm{loc}},2v_{\mathrm{loc}})\).
By the maximal inequality of \citet[Theorem~4]{chen2026cross}, as in the
proof of \Cref{prp--max-score-asymp-normal}(a) (the measurability and centring
arguments given after \Cref{asm--m-est-general} and \Cref{thm--m-est-non-smooth} apply to this class of
differences indexed by a Borel set of pairs; centring replaces \(2F_t\) by
\(2F_t+2\E F_t\), whose \(\|\cdot\|_{P_n,2}\) norm is at most \(4\|F_t\|_{P_n,2}\),
which at most doubles \(C_0\)), and by \eqref{eqn--P1a-env} at radius \(t\),
\[
  \E\|H_{n,\bm\iota}\|_{\calF(\eta,\delta)}
  \le C_0|\calI_n|^{-1/2}\|2F_t\|_{P_n,2}
  \le 2C_0(2\|\dot m\|_2)^{1/2}\Big(\frac{t}{|\calI_n|}\Big)^{1/2},
\]
where \(C_0\) depends only on \((A_{\mathrm{loc}},2v_{\mathrm{loc}})\), and
\(|\calI_{n,\bm\iota}|=|\calI_n|\).

For multiplier weights, the identity preceding
\eqref{eqn--bootstrap-2nd-order-jackknife-comparison}, with \(\e=\bm\iota\),
gives
\(H^*_{n,\bm\iota}-H_{n,\bm\iota}=-(\bar\xi_{n,1}+\bar\xi_{n,2})H_{n,\bm\iota}
+H^{\{1\}}_{n,\bm\iota}+H^{\{2\}}_{n,\bm\iota}\), with \(H^A_{n,\e}\) as in
\eqref{eqn--bootstrap-2nd-order-partial}.
By the note at the start of its proof,
\Cref{lem--bootstrap-2nd-order-partial-control} applies with
\(\calF(\eta,\delta)\), a class of differences indexed by a Borel set of
pairs, in place of \(\calM(\eta)\); it gives
\(\E\|H^{\{\ell\}}_{n,\bm\iota}\|_{\calF(\eta,\delta)}\le4L_\ell
\E\|H_{n,\bm\iota}\|_{\calF(\eta,\delta)}\).
As \(|\bar\xi_{n,\ell}|\le L_\ell\),
\(\E\|H^*_{n,\bm\iota}-H_{n,\bm\iota}\|_{\calF(\eta,\delta)}\le5(L_1+L_2)
\E\|H_{n,\bm\iota}\|_{\calF(\eta,\delta)}\), and the first bound applies.
\end{proof}

\begin{lemma}[Rate and linearisation on the scale \(\sigma_n\)]
\label{lem--two-scale-rate}
\hfill
\begin{enumerate}[(a)]
  \item Under the conditions of \Cref{thm--two-scale},
    \(\hat\theta_n-\theta_0=\Op(\sigma_n)\) and
    \(\hat\theta_n-\theta_0-A_0^{-1}\bar\psi_n=\op(\sigma_n)+\Op(\sigma_nc_n^{3/4})\).
  \item Under the conditions of \Cref{thm--two-scale-boot},
    \(\tilde\theta^*_n-\hat\theta_n=\Op(\sigma_n)\) and
    \(\tilde\theta^*_n-\hat\theta_n-\tilde A_n^{-1}\tilde\Delta^*_n
    =\op(\sigma_n)+\Op(\sigma_nc_n^{3/4})\), with \(\tilde A_n^{-1}:=0\) where
    \(\tilde A_n\) is singular.
\end{enumerate}
Here \(X_n=\op(\sigma_n)+\Op(\sigma_nc_n^{3/4})\) means \(\|X_n\|\le Y_n+Z_n\) with
\(Y_n=\op(\sigma_n)\) and \(Z_n=\Op(\sigma_nc_n^{3/4})\).
The orders in (b) are unconditional.
\end{lemma}

\begin{proof}[Proof of \Cref{lem--two-scale-rate}]
(a) \emph{Decomposition.}
By \eqref{eqn--gen-hoeffding-decomp} with \(K=2\),
\(H_{n,\e_k}[f]=\hatE_{n,k}[\pi_{\e_k}f]\) and the definition of
\(r_{k,\theta}\) in \eqref{eqn--m-est-stoch-diff}, for
\(\theta_0+v\in\calN\),
\begin{equation}
  \hat M_n(\theta_0+v)-\hat M_n(\theta_0)
  =-\tfrac12v'A_0v+v'\bar\psi_n+\rho_n(v)+\mathbb G_n(v),
  \label{eqn--P1a-decomp}
\end{equation}
where \(\mathbb G_n(v):=H_{n,\bm\iota}[m_{\theta_0+v}-m_{\theta_0}]\) and
\(\rho_n(v):=M_0(\theta_0+v)-M_0(\theta_0)+\frac12v'A_0v
+\sum_{k}(\hatE_{n,k}-\E)[r_{k,\theta_0+v}]\) (note \(\E r_{k,\theta}=0\)).
By \Cref{asm--m-est-consistent-pop-quad-decay}\eqref{asm--m-est-smooth-pop}
and \Cref{asm--two-scale}\eqref{asm--two-scale-first}, along \(\{P_n\}\) with
the uniformity in (iv) of the conventions, and Markov's inequality, for every
non-stochastic \(\eta_n\to0\) with \(\eta_n\le\eta_0\),
\begin{equation}
  \epsilon_n(\eta_n):=\sup_{0<\|v\|\le\eta_n}
  \frac{|\rho_n(v)|}{\|v\|(\|v\|+r_n^{-1})}=\op(1).
  \label{eqn--P1a-rho}
\end{equation}

\emph{Rate (peeling).}
Let \(\hat v:=\hat\theta_n-\theta_0\).
As \(\hat v\pto0\), there are non-stochastic \(\eta_n\downarrow0\) with
\(\sigma_n/\eta_n\to0\) and \(\Pr(\|\hat v\|>\eta_n)\to0\) (take
\(\eta_n:=\sup_{m\ge n}\max\{\eta'_m,\sigma_m^{1/2}\}\), with \(\eta'_m\) from
\Cref{lem--det-rate-exists}), and \(\eta_n\le\eta_0/4\) for large \(n\);
by
\eqref{eqn--P1a-rho}, \(\Pr(E_n)\to1\) for \(E_n:=\{\|\hat v\|\le\eta_n,\
\epsilon_n(\eta_n)\le\underline\lambda/8\}\).
Since \(\hat M_n(\hat\theta_n)\ge\hat M_n(\theta_0)-\hat\zeta_n\),
\eqref{eqn--P1a-decomp} gives
\begin{equation}
  \tfrac12\hat v'A_0\hat v-|\rho_n(\hat v)|
  \le\|\hat v\|\,\|\bar\psi_n\|+|\mathbb G_n(\hat v)|+\hat\zeta_n.
  \label{eqn--P1a-basic}
\end{equation}
We split \(\{v:\|v\|>\sigma_n\}\) into the shells
\(S_j:=\{v:2^{j-1}\sigma_n<\|v\|\le2^j\sigma_n\}\), \(j\ge1\).
For \(v\in S_j\), \(\|v\|>\sigma_n\ge r_n^{-1}\), so
\(\|v\|(\|v\|+r_n^{-1})\le2\|v\|^2\).
Hence, on \(E_n\cap\{\hat v\in S_j\}\), the left-hand side of
\eqref{eqn--P1a-basic} is at least
\(\frac{\underline\lambda}2\|\hat v\|^2-\frac{\underline\lambda}4\|\hat
v\|^2>\frac{\underline\lambda}{16}2^{2j}\sigma_n^2\), so one of the three terms on
the right exceeds \(\beta_j:=\frac{\underline\lambda}{48}2^{2j}\sigma_n^2\).
That is, \(\|\bar\psi_n\|>\frac{\underline\lambda}{48}2^j\sigma_n\), or
\(\hat\zeta_n>\beta_j\), or \(\|H_{n,\bm\iota}\|_{\calF(2^j\sigma_n,2^j\sigma_n)}>\beta_j\);
for the last, note that \(|\mathbb G_n(v)|\le
\|H_{n,\bm\iota}\|_{\calF(\delta,\delta)}\) if \(\|v\|\le\delta\) (take
\(\theta=\theta_0\), \(\theta'=\theta_0+v\)).
On \(E_n\), \(\hat v\in S_j\) forces \(2^{j-1}\sigma_n<\eta_n\), so
\(2^{j+1}\sigma_n<4\eta_n\le\eta_0\) and \Cref{lem--cell-maximal} applies.
With Markov's inequality and \eqref{eqn--P1a-orders} it gives
\[
  \Pr\big(\|H_{n,\bm\iota}\|_{\calF(2^j\sigma_n,2^j\sigma_n)}>\beta_j\big)
  \le\frac{48C}{\underline\lambda}\cdot
  \frac{(2^{j+1}\sigma_n/|\calI_n|)^{1/2}}{2^{2j}\sigma_n^2}
  =\frac{48\sqrt2\,C}{\underline\lambda}\,2^{-3j/2}c_n^{3/2}.
\]
Taking the union over \(j>J\), using Chebyshev's inequality and
\eqref{eqn--P1a-orders} for \(\bar\psi_n\), \(\beta_j\ge
\frac{\underline\lambda}{48}\sigma_n^2\) and \(c_n\le1\), we get, for every
integer \(J\ge1\),
\[
  \Pr(\|\hat v\|>2^J\sigma_n)
  \le\Pr(E_n^c)+\frac{48^2d\|A_0\|^2}{\underline\lambda^2}2^{-2J}
  +\Pr\Big(\hat\zeta_n>\frac{\underline\lambda}{48}\sigma_n^2\Big)
  +\frac{48\sqrt2\,C}{\underline\lambda}\sum_{j>J}2^{-3j/2}.
\]
By \Cref{asm--two-scale}\eqref{asm--two-scale-rates},
\(\hat\zeta_n=\op(r_n^{-2})=\op(\sigma_n^2)\).
Letting \(n\to\infty\) and then \(J\to\infty\) gives \(\hat v=\Op(\sigma_n)\).

\emph{Linearisation.}
Let \(\check h:=A_0^{-1}\bar\psi_n=\Op(\sigma_n)\); as \(\sigma_n\to0\),
\(\theta_0+\check h\in\calN\) with probability tending to one.
Since \(-\frac12v'A_0v+v'\bar\psi_n=-\frac12(v-\check h)'A_0(v-\check
h)+\frac12\check h'A_0\check h\), the inequality \(\hat
M_n(\hat\theta_n)\ge\hat M_n(\theta_0+\check h)-\hat\zeta_n\) and
\eqref{eqn--P1a-decomp} give
\[
  \tfrac{\underline\lambda}2\|\hat v-\check h\|^2
  \le|\rho_n(\hat v)|+|\rho_n(\check h)|+|\mathbb G_n(\hat v)|+|\mathbb
  G_n(\check h)|+\hat\zeta_n.
\]
Let \(h\in\{\hat v,\check h\}\), so that \(h=\Op(\sigma_n)\).
By \eqref{eqn--P1a-rho} and \(r_n^{-1}\le \sigma_n\), \(|\rho_n(h)|\le
\epsilon_n(\eta_n)\|h\|(\|h\|+\sigma_n)=\op(\sigma_n^2)\) on \(\{\|h\|\le\eta_n\}\).
For \(L>0\) and \(n\) with \(2L\sigma_n\le\eta_0\), on \(\{\|h\|\le L\sigma_n\}\),
\(|\mathbb G_n(h)|\le\|H_{n,\bm\iota}\|_{\calF(L\sigma_n,L\sigma_n)}\), whose mean is at most
\(C(2L\sigma_n/|\calI_n|)^{1/2}=C(2L)^{1/2}\sigma_n^2c_n^{3/2}\) by
\Cref{lem--cell-maximal} and \eqref{eqn--P1a-orders}; as \(L\) is arbitrary,
\(|\mathbb G_n(h)|=\Op(\sigma_n^2c_n^{3/2})\).
Hence \(\|\hat v-\check h\|^2=\op(\sigma_n^2)+\Op(\sigma_n^2c_n^{3/2})\), and (a)
follows by taking square roots.

(b) The argument is the same, around the data-dependent centre
\(\hat\theta_n\).
Recall \(\hat v=\hat\theta_n-\theta_0=\Op(\sigma_n)\), by (a), and let
\(\tilde v:=\tilde\theta^*_n-\hat\theta_n\pto0\), by
\Cref{lem--boot-consistent}(b).
Choose \(\eta_n\downarrow0\) with \(\sigma_n/\eta_n\to0\) and
\(\Pr(\|\tilde v\|>\eta_n/2)\to0\), as in (a), so that \(\eta_n\le\eta_0/4\)
for large \(n\); then also \(\Pr(\|\hat
v\|>\eta_n/2)\to0\).
By \eqref{eqn--P3-boot-hoeffding},
\(\hat M^*_n-\hat M_n=\tilde Z^*_{n,1}+\tilde Z^*_{n,2}\) exactly, with the
notation of \eqref{eqn--bootstrap-hoeffding-decomps-and-rems} and
\(\tilde Z^*_{n,2}(\theta)=(H^*_{n,\bm\iota}-H_{n,\bm\iota})[m_\theta]\).
Hence, for \(\hat\theta_n+v\in\calN\),
\[
  \tilde M^*_n(\hat\theta_n+v)-\tilde M^*_n(\hat\theta_n)
  =-\tfrac12v'\tilde A_nv+v'\tilde\Delta^*_n+\rho^*_n(v)+\mathbb G^*_n(v),
\]
with \(\rho^*_n(v):=\tilde R^*_{n,1}(\hat\theta_n+v)-\tilde
R^*_{n,1}(\hat\theta_n)\) and \(\mathbb
G^*_n(v):=(H^*_{n,\bm\iota}-H_{n,\bm\iota})[m_{\hat\theta_n+v}-m_{\hat\theta_n}]\).

\emph{Remainder.}
Let \(\phi^*_n(\eta):=\sup_{0<\|\theta-\theta_0\|\le\eta}|\tilde
R^*_{n,1}(\theta)|/\{\|\theta-\theta_0\|(\|\theta-\theta_0\|+r_n^{-1})\}\).
For multiplier weights,
\(\tilde R^*_{n,1}(\theta)=\sum_kN_{n,k}^{-1}\sum_i(\xi_{k,i}-\bar\xi_{n,k})
r_{k,\theta}(U_{i\cdot\e_k})\), so
\Cref{lem--bootstrap-1st-order-control}, whose proof applies verbatim with
\(n^{-1/2}\) replaced by \(r_n^{-1}\) under
\Cref{asm--two-scale}\eqref{asm--two-scale-first}, gives
\(\phi^*_n(\eta_n)=\op(1)\).
Using \(x(x+y)\le\frac32x^2+\frac12y^2\), \(\|\hat v+v\|^2\le2\|\hat
v\|^2+2\|v\|^2\) and \(r_n^{-2}\le \sigma_n^2\), on \(\{\|\hat
v\|\le\eta_n/2\}\) and for \(\|v\|\le\eta_n/2\),
\[
  |\rho^*_n(v)|\le\phi^*_n(\eta_n)\big\{3\|v\|^2+\tfrac92\|\hat
  v\|^2+r_n^{-2}\big\}\le\epsilon^*_n(\|v\|^2+\sigma_n^2),
  \qquad
  \epsilon^*_n:=\phi^*_n(\eta_n)\big(4+\tfrac92\|\hat v\|^2/\sigma_n^2\big)=\op(1).
\]

\emph{Cell term.}
On \(\{\|\hat v\|\le L\sigma_n\}\),
\(\sup_{\|v\|\le\delta}|\mathbb G^*_n(v)|\le
\|H^*_{n,\bm\iota}-H_{n,\bm\iota}\|_{\calF(L\sigma_n,\delta)}\), a supremum over a
non-stochastic class, so \Cref{lem--cell-maximal} applies despite the random
centre.

\emph{Linear term.}
Since \(\sum_i(\xi_{k,i}-\bar\xi_{n,k})\Delta_{k,i}
=\sum_i\xi_{k,i}(\Delta_{k,i}-\hatE_{n,k}[\Delta_k])\), and the \(\xi_{k,i}\) are
independent of the data, centred and of unit variance,
\(\E\|\tilde\Delta^*_n\|^2\le\sum_kN_{n,k}^{-2}\E\sum_i\|\Delta_{k,i}\|^2
=\operatorname{tr}(\Omega_n)\le d\|A_0\|^2\sigma_n^2\).

\emph{Peeling.}
On \(E^*_n:=\{\|\tilde v\|\le\eta_n/2,\ \|\hat v\|\le\eta_n/2\wedge L\sigma_n,\
\lambda_{\min}(\tilde A_n)\ge\underline\lambda/2,\
\epsilon^*_n\le\underline\lambda/16\}\) and for \(v\in S_j\) with \(\|v\|\le\eta_n/2\),
\(\frac12v'\tilde A_nv-|\rho^*_n(v)|\ge\frac{\underline\lambda}4\|v\|^2
-\frac{\underline\lambda}8\|v\|^2>\frac{\underline\lambda}{32}2^{2j}\sigma_n^2\),
because \(\|v\|^2+\sigma_n^2\le2\|v\|^2\).
The peeling argument of (a), with \(\tilde v\), \(\tilde\Delta^*_n\),
\(\tilde\zeta^*_n=\op(\sigma_n^2)\), \(\rho^*_n\), \(\mathbb G^*_n\) and
\(\beta^*_j:=\frac{\underline\lambda}{96}2^{2j}\sigma_n^2\), and with
\Cref{lem--cell-maximal} applied to \(\calF(L\sigma_n,2^j\sigma_n)\) (valid for large
\(n\), as \(L\sigma_n\le\eta_0/2\) and \(2^j\sigma_n<\eta_n\)), bounds the cell
contribution of shell \(j\) by
\[
  \frac{96C}{\underline\lambda}\cdot
  \frac{\{(L+2^j)\sigma_n/|\calI_n|\}^{1/2}}{2^{2j}\sigma_n^2}
  =\frac{96C}{\underline\lambda}(L+2^j)^{1/2}2^{-2j}c_n^{3/2},
  \qquad
  \sum_{j>J}(L+2^j)^{1/2}2^{-2j}\le L^{1/2}2^{-2J}+2^{-3J/2}.
\]
Since \(\tilde A_n\pto A_0\), \(\limsup_n\Pr(E^{*c}_n)\le
\limsup_n\Pr(\|\hat v\|>L\sigma_n)\), which is small for large \(L\) by (a).
Choosing \(L\) and then \(J\) large gives \(\tilde v=\Op(\sigma_n)\).
For the linearisation, let \(\check h^*:=\tilde A_n^{-1}\tilde
\Delta^*_n=\Op(\sigma_n)\), defined when \(\tilde A_n\) is invertible.
Let \(h\in\{\tilde v,\check h^*\}\), so that \(h=\Op(\sigma_n)\) and
\(\Pr(\|h\|>\eta_n/2)\to0\).
On \(\{\|h\|\le\eta_n/2,\ \|\hat v\|\le\eta_n/2\}\), the remainder bound gives
\(|\rho^*_n(h)|\le\epsilon^*_n(\|h\|^2+\sigma_n^2)=\op(\sigma_n^2)\).
For \(L>0\) and \(n\) with \(2L\sigma_n\le\eta_0\), on \(\{\|\hat v\|\le L\sigma_n,\
\|h\|\le L\sigma_n\}\), the cell-term bound gives \(|\mathbb G^*_n(h)|\le
\|H^*_{n,\bm\iota}-H_{n,\bm\iota}\|_{\calF(L\sigma_n,L\sigma_n)}\), whose mean is at
most \(C(2L)^{1/2}\sigma_n^2c_n^{3/2}\) by \Cref{lem--cell-maximal} and
\eqref{eqn--P1a-orders}; as \(\hat v,h=\Op(\sigma_n)\) and \(L\) is arbitrary,
\(|\mathbb G^*_n(h)|=\Op(\sigma_n^2c_n^{3/2})\).
On \(\{\lambda_{\min}(\tilde A_n)\ge\underline\lambda/2\}\),
\(\sup_\theta\tilde M^*_n(\theta)\le\sup_\theta|\hat M^*_n(\theta)-\hat
M_n(\theta)|<\infty\), so \eqref{eqn--m-est-boot-def} gives
\(\tilde M^*_n(\tilde\theta^*_n)\ge\tilde M^*_n(\hat\theta_n+\check
h^*)-\tilde\zeta^*_n\) whenever \(\hat\theta_n+\check h^*\in\calN\).
As in (a), on \(\{\lambda_{\min}(\tilde A_n)\ge\underline\lambda/2\}\cap\{\hat
\theta_n+\check h^*\in\calN\}\), whose probability tends to one,
\[
  \tfrac{\underline\lambda}4\|\tilde v-\check h^*\|^2
  \le|\rho^*_n(\tilde v)|+|\rho^*_n(\check h^*)|+|\mathbb G^*_n(\tilde v)|
  +|\mathbb G^*_n(\check h^*)|+\tilde\zeta^*_n
  =\op(\sigma_n^2)+\Op(\sigma_n^2c_n^{3/2}),
\]
which proves (b).
\end{proof}

\begin{lemma}[Local limits]
\label{lem--two-scale-local}
Suppose the conditions of \Cref{thm--two-scale} hold and
\(r_n^2\Omega_n\to\Omega\).
Let \(G\sim N(0,\Omega)\) be independent of \(\mathbb W\), and let
\(\calK\subset\R^d\) be compact.
\begin{enumerate}[(a)]
  \item \((r_n\bar\psi_n,\mathbb W_n)\leadsto(G,\mathbb W)\) in
    \(\R^d\times\ell^\infty(\calK)\), and
    \[
      \lim_{\delta\downarrow0}\limsup_{n\to\infty}
      \E\sup_{u,u'\in\calK,\|u-u'\|\le\delta}|\mathbb W_n(u)-\mathbb W_n(u')|=0.
    \]
  \item If, in addition, the weights are of the multiplier type in
    \Cref{asm--boot-wgt}, then \((r_n\tilde\Delta^*_n,\mathbb
    W^*_n)\dtop(G,\kappa^{1/2}\mathbb W)\) conditionally given
    \(\mathscr U_n\), where \(\kappa\), the limit of \(\var(\xi_{n,\i})\),
    equals \(2\) for such weights, in the sense that
    \[
      \sup_{h\in\mathrm{BL}_1(\R^d\times\ell^\infty(\calK))}
      \big|\E_{\mathscr U}h(r_n\tilde\Delta^*_n,\mathbb W^*_n)
      -\E h(G,\kappa^{1/2}\mathbb W)\big|\ptoout0.
    \]
\end{enumerate}
\end{lemma}

\begin{proof}[Proof of \Cref{lem--two-scale-local}(a)]
It suffices to take \(\calK=\{u:\|u\|\le R\}\), since restriction to a
subset is continuous and the modulus in (a) is monotone in \(\calK\).
Let \(n\) be large enough that \(2R/r_n\le\eta_0\).
Then \(f_u\in\calM(R/r_n)\) for \(u\in\calK\), \(f_0=0\), and, by
\eqref{eqn--P1a-lip}, \(|\bar f_s-\bar f_t|\le\dot m(\bm U^\circ)\|s-t\|/r_n\)
and \(|\bar f_u|\le\min\{2,\dot m(\bm U^\circ)\|u\|/r_n\}\).
Let \(\omega_\delta(z):=\sup_{u,u'\in\calK,\|u-u'\|\le\delta}|z(u)-z(u')|\),
a \(2\)-Lipschitz function on \(\ell^\infty(\calK)\).

\emph{Step 1: the interaction is negligible.}
An \(\epsilon R\)-net of \(\calK\) with at most \((3/\epsilon)^d\) points
yields an \(\epsilon R\|\dot m\|_{Q,2}/r_n\)-net of \(\{\bar f_u:u\in\calK\}\) in
every \(L_2(Q)\), so this class is VC-type with envelope
\(\dot m(\bm U^\circ)R/r_n\) and characteristics \((3,d)\).
By \citet[Theorem~4]{chen2026cross}, applied to the separately exchangeable
and dissociated array \(\bm U^\circ_{\i}\), whose top-level projection of
\(\bar f_u\) is \(\pi_{\i}[\bar f_u]=\mathsf I_{\i}[f_u]\), and to the centred
class \(\{\bar f_u-\E\bar f_u:u\in\calK\}\), whose envelope has \(L_2\) norm
at most \(2R\|\dot m\|_2/r_n\) (see the discussion after \Cref{thm--m-est-non-smooth}),
\[
  \E\sup_{u\in\calK}|\mathbb W^{\mathsf I}_n(u)|
  \le C_0r_n^2|\calI_n|^{-1/2}\frac{2R\|\dot m\|_2}{r_n}
  =2C_0R\|\dot m\|_2\,r_n^{-1/2}\to0,
\]
where \(C_0\) depends only on \(d\).

\emph{Step 2: conditional limit of \(\mathbb W^{\mathsf C}_n\).}
Write \(\mathbb W^{\mathsf C}_n(u)=\sum_{\i}Y_{n,\i}(u)\) with
\(Y_{n,\i}(u):=r_n^{-1}\mathsf C_{\i}[f_u]\).
Conditionally on \(\mathscr G_n\), the processes \(Y_{n,\i}\) are
independent and centred, with \(|Y_{n,\i}|\le4/r_n\).
We check the conditions of the central limit theorem for sums of independent
processes under a random-entropy condition
\citep[Theorem~2.11.1]{2023vandervaartWeakConvergenceEmpirical}, with the
Euclidean metric on \(\calK\).
\begin{enumerate}[(i)]
  \item \emph{Lindeberg.} Trivial, as \(\sup_{\i,u}|Y_{n,\i}(u)|\le4/r_n\to0\).
  \item \emph{Variance modulus.} As \((f_s-f_t)^2\le2|f_s-f_t|\),
    \eqref{eqn--P1a-lip} and \(r_n^{-3}|\calI_n|=1\) give
    \[
      \sum_{\i}\E\big[\{Y_{n,\i}(s)-Y_{n,\i}(t)\}^2\bigm|\mathscr G_n\big]
      \le\frac{2}{r_n^2}\sum_{\i}\E\big[|f_s-f_t|(W_{\i})\bigm|
      \bm U^\circ_{\i}\big]\le2\|s-t\|\,\En[\dot m(\bm U^\circ_{\i})].
    \]
  \item \emph{Random entropy.} Let \(d_n(s,t):=\{\sum_{\i}(Y_{n,\i}(s)
    -Y_{n,\i}(t))^2\}^{1/2}\) and let \(Q_n\) be the empirical
    measure of \(\{\bm U_{\i}\}_{\i\in\calI_n}\).
    Since \(\mathsf C_{\i}[f]^2\le2f(W_{\i})^2+2\bar f(\bm U^\circ_{\i})^2\)
    and \(|\calI_n|/r_n^2=r_n\),
    \(d_n(s,t)\le(2r_n)^{1/2}\|f_s-f_t\|_{Q_n,2}
    +(2r_n)^{1/2}\|\bar f_s-\bar f_t\|_{Q_n,2}\).
    Let \(a_n:=(2r_n)^{1/2}\|F_{R/r_n}\|_{Q_n,2}\) and
    \(b_n:=(2r_n)^{1/2}R\|\dot m\|_{Q_n,2}/r_n\).
    By the local VC-type condition and the Lipschitz bound,
    \(N(x(a_n+b_n),\calK,d_n)\le(C_1/x)^{v_{\mathrm{loc}}+d}\) for
    \(0<x\le2\), and
    \(N(x(a_n+b_n),\calK,d_n)=1\) for \(x>2\) because the \(d_n\)-diameter of
    \(\calK\) is at most \(2(a_n+b_n)\); here \(C_1\) depends only on
    \((A_{\mathrm{loc}},v_{\mathrm{loc}},d)\).
    Hence, if \(a_n+b_n\le T\), then, with \(\log_+x:=\max\{\log x,0\}\),
    \begin{equation}
      \int_0^\delta\sqrt{\log N(\epsilon,\calK,d_n)}\,\d\epsilon
      \le\int_0^{\delta}\sqrt{(v_{\mathrm{loc}}+d)\log_+(C_1T/\epsilon)}\,\d\epsilon,
      \label{eqn--P1a-entropy}
    \end{equation}
    which tends to zero as \(\delta\downarrow0\).
    Moreover \(\E[a_n^2\mid\mathscr G_n]\le4R\,\En[\dot m(\bm U^\circ_{\i})]\) by
    \Cref{asm--two-scale}\eqref{asm--two-scale-modulus} and \(F^2_\eta\le2F_\eta\),
    and \(b_n\) is \(\mathscr G_n\)-measurable with \(\E
    b_n^2\le2R^2\|\dot m\|_2^2/r_n\to0\).
  \item \emph{Covariances.} The conditional covariance is
    \[
      \mathcal C_n(s,t):=\sum_{\i}\E[Y_{n,\i}(s)Y_{n,\i}(t)\mid\mathscr G_n]
      =\frac{1}{|\calI_n|}\sum_{\i}g_n(s,t;\bm U^\circ_{\i}),
    \]
    where \(g_n(s,t;\bm U^\circ):=r_n\cov(f_s(W_{(1,1)}),f_t(W_{(1,1)})\mid\bm U^\circ)\).
    By \eqref{eqn--P1a-polar} with \(\delta=r_n^{-1}\) and
    \(\theta=\theta_0\), \(r_n\E[f_sf_t]\to C_{\mathbb W}(s,t)\); and \(r_n\E|\bar
    f_s\bar f_t|\le\|s\|\|t\|\|\dot m\|_2^2/r_n\to0\).
    Hence \(\E g_n(s,t;\bm U^\circ)\to C_{\mathbb W}(s,t)\).
    Also \(|g_n|\le4r_n\E[|f_s(W_{(1,1)})-\bar f_s(\bm U^\circ)|\mid\bm U^\circ]\le8\dot m(\bm U^\circ)\|s\|\).
    For a square-integrable \(g\) of \((U_{(1,0)},U_{(0,1)})\), the
    Hoeffding decomposition of \(g\) gives
    \(\var\{|\calI_n|^{-1}\sum_{\i}g(\bm U^\circ_{\i})\}\le
    (N_{n,1}^{-1}+N_{n,2}^{-1}+|\calI_n|^{-1})\E g^2\).
    Therefore \(\mathcal C_n(s,t)-C_{\mathbb W}(s,t)\pto0\) and
    \(\En[\dot m(\bm U^\circ_{\i})]-\E\dot m(\bm U^\circ)\pto0\).
\end{enumerate}
Every subsequence has a further subsequence along which, almost surely,
\(\limsup_n\En[\dot m(\bm U^\circ_{\i})]<\infty\), \(b_n\to0\) and \(\mathcal C_n(s,t)\to C_{\mathbb W}(s,t)\)
for all \(s,t\) in a fixed countable dense subset of \(\calK\).
Fix a realisation of \(\{\bm U^\circ_{\i}\}\) in this event and argue under
the conditional law given \(\mathscr G_n\), along the further subsequence.
Conditions (i) and (ii) hold, and so does the random-entropy condition,
since \(\Pr(a_n+b_n>T\mid\mathscr G_n)\le\Pr(a_n>T/2\mid\mathscr
G_n)\le16R\,\En[\dot m(\bm U^\circ_{\i})]/T^2\) for large \(n\), for every \(T\): for
\(\delta_n\downarrow0\) and \(\epsilon>0\), by \eqref{eqn--P1a-entropy}, the
conditional probability that
\(\int_0^{\delta_n}(\log N(\epsilon',\calK,d_n))^{1/2}\,\d\epsilon'>\epsilon\)
is at most \(16R\,\En[\dot m(\bm U^\circ_{\i})]/T^2\) once
\(\int_0^{\delta_n}((v_{\mathrm{loc}}+d)\log_+(C_1T/\epsilon'))^{1/2}\,
\d\epsilon'\le\epsilon\), and \(T\) is arbitrary.
By the Cauchy--Schwarz inequality, (ii) and \(Y_{n,\i}(0)=0\),
\(|\mathcal C_n(s,t)-\mathcal C_n(s',t')|\le2\En[\dot m(\bm U^\circ_{\i})]
\{(\|s-s'\|R)^{1/2}+(\|t-t'\|R)^{1/2}\}\), so the \(\mathcal C_n\) are
equicontinuous on \(\calK^2\); as \(C_{\mathbb W}\) is continuous, \(\mathcal C_n\to C_{\mathbb W}\)
on all of \(\calK^2\).
The finite-dimensional laws then converge to those of \(\mathbb W\) by the
Lindeberg--Feller theorem, and the theorem cited above (whose measurability
conditions hold because, by
\Cref{asm--m-est-general}\eqref{asm--m-est-meas}, the suprema involved are
measurable for the completed measures) gives
\(\mathbb W^{\mathsf C}_n\leadsto\mathbb W\) in \(\ell^\infty(\calK)\) under
the conditional law.
For \(u_1,\dots,u_k\in\calK\) and \(\lambda\in\R^k\) let
\(\chi_n:=\E[\exp\{i\sum_j\lambda_j\mathbb W^{\mathsf C}_n(u_j)\}\mid\mathscr
G_n]\) and \(\chi:=\E\exp\{i\sum_j\lambda_j\mathbb W(u_j)\}\).
Then \(\chi_n\to\chi\) and, by the portmanteau theorem,
\(\limsup_n\Pr\{\omega_\delta(\mathbb W^{\mathsf C}_n)\ge x\mid\mathscr
G_n\}\le\Pr\{\omega_\delta(\mathbb W)\ge x\}\) for \(x>0\).
By dominated convergence and Fatou's lemma, along the further subsequence,
\begin{equation}
  \E|\chi_n-\chi|\to0
  \qquad\text{and}\qquad
  \limsup_n\Pr\{\omega_\delta(\mathbb W^{\mathsf C}_n)\ge x\}
  \le\Pr\{\omega_\delta(\mathbb W)\ge x\}.
  \label{eqn--P1a-cond}
\end{equation}
As the subsequence was arbitrary, \eqref{eqn--P1a-cond} holds along the
whole sequence.

\emph{Step 3: joint limit.}
\(r_n\bar\psi_n\leadsto G\) by \eqref{eqn--P1a-clt} with \(B_n=r_n\I_d\),
since \(r_n\sigma_n=c_n^{-1}=(\|A_0^{-1}r_n^2\Omega_nA_0^{-1}\|+1)^{1/2}\) is
bounded.
As \(r_n\bar\psi_n\) is \(\mathscr G_n\)-measurable, for \(a\in\R^d\),
\[
  \big|\E\exp\{ia'r_n\bar\psi_n+i\textstyle\sum_j\lambda_j\mathbb
  W^{\mathsf C}_n(u_j)\}-\E\exp(ia'G)\,\chi\big|
  \le\E|\chi_n-\chi|+|\E\exp(ia'r_n\bar\psi_n)-\E\exp(ia'G)|\to0,
\]
so the finite-dimensional laws of \((r_n\bar\psi_n,\mathbb W^{\mathsf C}_n)\)
converge to those of \((G,\mathbb W)\), with \(G\) independent of
\(\mathbb W\).
By \eqref{eqn--P1a-cond} and the uniform continuity of the paths of
\(\mathbb W\) on \(\calK\),
\(\lim_{\delta\downarrow0}\limsup_n\Pr\{\omega_\delta(\mathbb
W^{\mathsf C}_n)\ge x\}=0\) for every \(x>0\).
Hence \((r_n\bar\psi_n,\mathbb W^{\mathsf C}_n)\) is asymptotically tight and
converges weakly to \((G,\mathbb W)\) in \(\R^d\times\ell^\infty(\calK)\),
identified with \(\ell^\infty\) of the disjoint union of \(\{1,\dots,d\}\)
and \(\calK\) \citep[Section~1.5]{2023vandervaartWeakConvergenceEmpirical}.
By Step 1 and Slutsky's lemma, the same holds for \((r_n\bar\psi_n,\mathbb
W_n)\).

\emph{Step 4: modulus in mean.}
Step 3 controls \(\omega_\delta(\mathbb W^{\mathsf C}_n)\) in probability; to
pass to expectations we bound its second moment.
Symmetrise given \(\mathscr G_n\)
\citep[Lemma~2.3.6 with \(\Phi(x)=(x\vee0)^2\)]{2023vandervaartWeakConvergenceEmpirical}:
with i.i.d.\ Rademacher \(\varrho_{\i}\), \(\E[\|\mathbb
W^{\mathsf C}_n\|_{\calK}^2\mid\mathscr G_n]\le4\E[\|\sum_{\i}\varrho_{\i}
Y_{n,\i}\|_{\calK}^2\mid\mathscr G_n]\).
Given the \(Y_{n,\i}\), the process \(u\mapsto\sum_{\i}\varrho_{\i}Y_{n,\i}(u)\)
is sub-Gaussian for \(d_n\) and vanishes at \(u=0\), so chaining
\citep[Corollary~2.2.5 with \(\psi_2(x)=\mathrm e^{x^2}-1\), p.~148]{2023vandervaartWeakConvergenceEmpirical} and
\eqref{eqn--P1a-entropy} with \(T=a_n+b_n\), with the integral up to the
\(d_n\)-diameter \(2(a_n+b_n)\) of \(\calK\) and packing numbers bounded by
covering numbers at half the radius, bound its conditional \(\psi_2\)-norm,
and hence its conditional second moment, by \(C_2^{1/2}(a_n+b_n)\) and
\(C_2(a_n+b_n)^2\).
Hence \(\E\|\mathbb W^{\mathsf C}_n\|^2_{\calK}\le
8C_2(\E a_n^2+\E b_n^2)\le8C_2(4R\|\dot m\|_2+2R^2\|\dot m\|_2^2)\), and
\(\E\omega_\delta(\mathbb W^{\mathsf C}_n)^2\le4\E\|\mathbb
W^{\mathsf C}_n\|^2_{\calK}\) is bounded in \(n\).
Therefore, for every \(x>0\),
\[
  \E\omega_\delta(\mathbb W_n)
  \le x+\big\{\E\omega_\delta(\mathbb W^{\mathsf C}_n)^2\big\}^{1/2}
  \Pr\big(\omega_\delta(\mathbb W^{\mathsf C}_n)>x\big)^{1/2}
  +2\E\|\mathbb W^{\mathsf I}_n\|_{\calK},
\]
and Steps 1 and 3 give \(\lim_{\delta\downarrow0}\limsup_n\E\omega_\delta(\mathbb
W_n)\le x\).
Part (b) is proved in \Cref{sec--prf--lem--two-scale-local-b}.
\end{proof}

\begin{proof}[Proof of \Cref{thm--two-scale}]
The paragraph on the limit process above shows that \(\calZ\) is almost
surely unique and symmetric about zero.

\emph{Step 1: part (b).}
Let \(r_n^2\Omega_n\to\Omega\).
Then \(r_n\sigma_n=(\|A_0^{-1}r_n^2\Omega_nA_0^{-1}\|+1)^{1/2}\) is bounded, so
\(\hat u_n:=r_n(\hat\theta_n-\theta_0)=\Op(1)\) by
\Cref{lem--two-scale-rate}(a).
Let \(\mathbb M_n(u):=r_n^2\{\hat M_n(\theta_0+u/r_n)-\hat M_n(\theta_0)\}\)
if \(\theta_0+u/r_n\in\Theta\), and \(\mathbb M_n(u):=-2r_n^2\) otherwise.
As \(|m_\theta|\le1\), \(\mathbb M_n\) is bounded and \(\mathbb M_n(\hat
u_n)\ge\sup_u\mathbb M_n(u)-r_n^2\hat\zeta_n\), where
\(r_n^2\hat\zeta_n=\op(1)\) by
\Cref{asm--two-scale}\eqref{asm--two-scale-rates}.
Fix \(R>0\) and \(\calK=\{u:\|u\|\le R\}\).
For large \(n\), \eqref{eqn--P1a-decomp} with \(v=u/r_n\) and \(r_n^2\mathbb
G_n(u/r_n)=\mathbb W_n(u)\) gives, for \(u\in\calK\),
\[
  \mathbb M_n(u)=-\tfrac12u'A_0u+u'r_n\bar\psi_n+\mathbb W_n(u)
  +r_n^2\rho_n(u/r_n),
\]
and \(\sup_{u\in\calK}r_n^2|\rho_n(u/r_n)|\le\epsilon_n(R/r_n)R(R+1)=\op(1)\)
by \eqref{eqn--P1a-rho}, since \(r_n^2\|u/r_n\|(\|u/r_n\|+r_n^{-1})
=\|u\|(\|u\|+1)\).
By \Cref{lem--two-scale-local}(a), the continuous mapping theorem and
Slutsky's lemma, \(\mathbb M_n\leadsto\mathbb M\) in \(\ell^\infty(\calK)\),
where \(\mathbb M(u):=-\frac12u'A_0u+u'G+\mathbb W(u)\).
As \(R\) is arbitrary, \(\mathbb M_n\leadsto\mathbb M\) in the space of
locally bounded functions on \(\R^d\) with the topology of uniform convergence
on compacta \citep[Theorem~1.6.1]{2023vandervaartWeakConvergenceEmpirical},
the space \(\mathbf B_{\rm loc}(\R^d)\) of \citet[Section~2]{kim1990cube}.
Completing the square,
\(-\frac12u'A_0u+u'g=-\frac12(u-A_0^{-1}g)'A_0(u-A_0^{-1}g)
+\frac12g'A_0^{-1}g\), so
\[
  \mathbb M(A_0^{-1}G+v)=\mathbb W_G(v)-\tfrac12v'A_0v
  +\tfrac12G'A_0^{-1}G+\mathbb W(A_0^{-1}G),
  \qquad
  \mathbb W_g(v):=\mathbb W(A_0^{-1}g+v)-\mathbb W(A_0^{-1}g).
\]
For fixed \(g\), \(\mathbb W_g\overset d=\mathbb W\) by stationarity of the
increments, so \(\mathbb W_g\) lies almost surely in the set \(D\) of Step~4
of the proof of \Cref{lem--Z-law}, by the paragraph on the limit process.
As \((g,w)\mapsto w(A_0^{-1}g+\cdot)-w(A_0^{-1}g)\) is continuous, \(D\) is
measurable and \(G\) is independent of \(\mathbb W\), Fubini's theorem gives
\(\mathbb W_G\in D\) almost surely.
On this event, by the display, \(\mathbb M\) is continuous, tends to
\(-\infty\) at infinity and has the unique maximiser
\(u^\star:=A_0^{-1}G+\mathfrak t(\mathbb W_G)\), where \(\mathfrak t\) is the
argmax map on \(D\).
So the law of \(\mathbb M\) is concentrated on the set \(C_{\max}(\R^d)\) of
continuous functions that tend to \(-\infty\) at infinity and have a unique
maximiser, as condition~(i) of \citet[Theorem~2.7]{kim1990cube} requires; its
conditions~(ii) and~(iii) are \(\hat u_n=\Op(1)\) and the near-maximisation
property shown above.
Moreover, \(\mathfrak t(\mathbb W_g)\overset d=\mathfrak t(\mathbb
W)=\calZ\) for fixed \(g\), so, by Fubini's theorem again,
\(u^\star-A_0^{-1}G=\mathfrak t(\mathbb W_G)\) is independent of \(G\) and
distributed as \(\calZ\).
The argmax continuous mapping theorem \citep[Theorem~2.7]{kim1990cube}
therefore gives \(\hat u_n\leadsto u^\star\overset d=A_0^{-1}G+\calZ\), with
\(G\) independent of \(\calZ\).

\emph{Step 2: part (a).}
If \(c_n\to0\), then \(\Op(\sigma_nc_n^{3/4})=\op(\sigma_n)\), and part (a) is
\Cref{lem--two-scale-rate}(a).

\emph{Step 3: the approximation in \(\dBL\).}
Let \(T_n:=\sigma_n^{-1}(\hat\theta_n-\theta_0)\) and \(L_n:=\sigma_n^{-1}A_0^{-1}
G_n+c_n\calZ\).
Since \(\dBL\) metrises weak convergence on \(\R^d\), it suffices that every
subsequence has a further subsequence along which \(T_n\) and \(L_n\)
converge weakly to the same limit.
As \(\|\Gamma_n\|\le1\) and \(c_n\in(0,1]\), a given subsequence has a
further one with \(\Gamma_n\to\Gamma\) and \(c_n\to c_\infty\in[0,1]\).
All hypotheses are inherited by subsequences, so
\Cref{lem--two-scale-rate,lem--two-scale-local}, \eqref{eqn--P1a-clt} and
Step 1 apply along it.
Since \(\sigma_n^{-1}A_0^{-1}G_n\sim N(0,\Gamma_n)\) is independent of \(\calZ\),
\(L_n\leadsto\bm T:=\Gamma^{1/2}\bm Z_0+c_\infty\calZ\), with
\(\bm Z_0\sim N(0,\I_d)\) independent of \(\calZ\).
If \(c_\infty=0\), \Cref{lem--two-scale-rate}(a) gives
\(T_n=\sigma_n^{-1}A_0^{-1}\bar\psi_n+\op(1)\), and \eqref{eqn--P1a-clt} with
\(B_n=\sigma_n^{-1}A_0^{-1}\), for which \(\sigma_n\|B_n\|=\|A_0^{-1}\|\) and
\(B_n\Omega_nB_n'=\Gamma_n\), gives \(T_n\leadsto N(0,\Gamma)\), the law of \(\bm T\).
If \(c_\infty>0\), then \(r_n\sigma_n=c_n^{-1}\to c_\infty^{-1}\) and
\(r_n^2\Omega_n=(r_n\sigma_n)^2A_0\Gamma_nA_0\to\Omega:=c_\infty^{-2}A_0\Gamma A_0\).
Step 1 gives \(r_n(\hat\theta_n-\theta_0)\leadsto A_0^{-1}G+\calZ\) with
\(G\sim N(0,\Omega)\) independent of \(\calZ\).
Hence \(T_n=c_nr_n(\hat\theta_n-\theta_0)\leadsto
c_\infty A_0^{-1}G+c_\infty\calZ\), and \(c_\infty A_0^{-1}G\sim
N(0,c_\infty^2A_0^{-1}\Omega A_0^{-1})=N(0,\Gamma)\), so the limit is
again the law of \(\bm T\).

\Cref{asm--two-scale}\eqref{asm--two-scale-rates} enters this proof through
\(\hat\zeta_n=\op(r_n^{-2})\), in \Cref{lem--two-scale-rate}(a) and Step 1;
through the identification of \(r_n\) with \(n^{2/3}\) in the conventions
above, by which the statements of \Cref{sec--two-scale} are the ones proved
here; and, in \Cref{lem--two-scale-local}(b), through
\(N_{n,1}N_{n,2}/n^3\to0\).
\end{proof}

\subsubsection{Proof of \Cref{lem--two-scale-local}(b)}
\label{sec--prf--lem--two-scale-local-b}

\emph{Conventions.}
\(C\) denotes a constant that does not depend on \(n\) (it may depend on
\(\calK\)); \(\pi_{ij}:=\pi_{(i,j)}\) and
\(\mathsf I_{ij}:=\mathsf I_{(i,j)}\);
\(f_u\), \(\mathbb W_n\) and \(\mathbb W^*_n\) are as in
\Cref{sec--prf--two-scale}.
As in that subsection, we define all arrays and their weights on
one probability space, so a sequence converges in probability if and only if
every subsequence has a further subsequence along which it converges almost
surely.
We use two facts.

(LF) Let \(x_{n,1},\dots,x_{n,p_n}\) be \(\mathscr U_n\)-measurable random
vectors with
\[
  \textstyle\sum_lx_{n,l}x_{n,l}'\pto\Xi\quad\text{and}\quad
  \max_l\|x_{n,l}\|\pto0,
\]
and let \(\zeta_1,\dots,\zeta_{p_n}\) be independent of each other and of
\(\mathscr U_n\), centred, with unit variance and \(|\zeta_l|\le\bar L\).
Then \(\dBL(\mathcal L_{\mathscr U}(\sum_l\zeta_lx_{n,l}),N(0,\Xi))\pto0\).
This is the Lindeberg--Feller theorem, applied along almost surely convergent
subsequences; Lindeberg's sum is at most
\(\bar L^2\ind\{\bar L\max_l\|x_{n,l}\|>\epsilon\}\sum_l\|x_{n,l}\|^2\).

(P) The expectation of a product is zero if one factor is \(\pi_{ij}[f]\) with
\(f\) bounded, the product of the other factors is integrable, and no other
factor depends on \(U_{(0,j)}\) or \(U_{(i,j)}\).
Indeed, conditionally on all latent variables other than \(U_{(0,j)}\) and
\(U_{(i,j)}\), the other factors are fixed and that factor has mean
\(\E[\pi_{ij}[f]\mid U_{(i,0)}]=0\), by degeneracy of the top-level
projection.
The same holds with \(U_{(i,0)}\) in place of \(U_{(0,j)}\).

The proof has six steps.
Steps~2--4 verify the conditions of (LF) for the representation
\eqref{eqn--P1b-rep} in Step~1, Step~5 applies (LF), and Step~6 adds
stochastic equicontinuity.

\emph{Step 1: preliminaries.}
We record moment bounds, a representation of
\((r_n\tilde\Delta^*_n,\mathbb W^*_n)\) as a sum of conditionally
independent terms, and some rates.

(i)~\emph{Moment bounds.}
Fix a compact \(\calK\subset\R^d\), \(u\in\calK\) and \(n\) with
\(\sup_{v\in\calK}\|v\|\le\eta_0r_n\).
\Cref{asm--two-scale}\eqref{asm--two-scale-modulus} with \(\theta=\theta_0\)
and \(\delta=\|u\|/r_n\) gives
\(\E[|f_u(W_{(1,1)})|\mid\bm U^\circ]\le\dot m(\bm U^\circ)\|u\|/r_n\).
Since \(|f_u|\le2\),
\(\pi_{ij}[f]=f(W_{(i,j)})-\E[f\mid U_{(i,0)}]-\E[f\mid U_{(0,j)}]+\E f\)
and \(f_u^2\le2|f_u|\), Jensen's inequality yields
\begin{equation}
  |\pi_{ij}[f_u]|\le8,\qquad
  \E f_u^2\le\frac{C\|u\|}{r_n},\qquad
  \E\bigl[\pi_{ij}[f_u]^2\bigm|U_{(i,0)}\bigr]\le\frac{16\,\dot m_i\|u\|}{r_n},
  \qquad
  \E\,\mathsf I_{ij}[f_u]^2\le\frac{C\|u\|^2}{r_n^2},
  \label{eqn--P1b-moments}
\end{equation}
where \(\dot m_i:=\E[\dot m(\bm U^\circ_{(i,1)})\mid U_{(i,0)}]+\E\dot m(\bm U^\circ)\),
so that \(\E\dot m_i^2\le4\E\dot m(\bm U^\circ)^2\le C\).
The last bound holds because \(\mathsf I_{ij}[f]\) is the top-level component of
the Hoeffding decomposition of \(\E[f(W_{(i,j)})\mid\bm U^\circ_{(i,j)}]\),
which is at most \(\dot m(\bm U^\circ_{(i,j)})\|u\|/r_n\) in absolute value.
The column analogue of the third bound also holds.
By (P), the \(\pi_{ij}[f_u]\) are uncorrelated, so
\(\E\mathbb W_n(u)^2=r_n^{-2}|\calI_n|\E\pi_{11}[f_u]^2\le r_n\E f_u^2\le C\|u\|\).

(ii)~\emph{Representation.}
We have \(\xi_{n,\bm i}-1=(\xi_{1,i_1}-\bar\xi_{n,1})+(\xi_{2,i_2}-\bar\xi_{n,2})\),
\(r_n^2/|\calI_n|=r_n^{-1}\) and
\(\sum_i(\xi_{1,i}-\bar\xi_{n,1})x_i=\sum_i\xi_{1,i}(x_i-\bar x)\).
Hence
\begin{equation}
  \bigl(r_n\tilde\Delta^*_n,\mathbb W^*_n(u)\bigr)
  =\sum_{i=1}^{N_{n,1}}\xi_{1,i}\bigl(a_i-\bar a,\,A_i(u)-\bar A(u)\bigr)
  +\sum_{j=1}^{N_{n,2}}\xi_{2,j}\bigl(b_j-\bar b,\,B_j(u)-\bar B(u)\bigr),
  \label{eqn--P1b-rep}
\end{equation}
where \(a_i:=r_nN_{n,1}^{-1}\Delta_{1,i}\),
\(A_i(u):=r_n^{-1}\sum_{j}\pi_{ij}[f_u]\), \(b_j:=r_nN_{n,2}^{-1}\Delta_{2,j}\) and
\(B_j(u):=r_n^{-1}\sum_{i}\pi_{ij}[f_u]\).
Bars denote averages over \(i\) for \(a\) and \(A\), and over \(j\) for \(b\) and
\(B\); thus \(\bar A(u)=\mathbb W_n(u)/N_{n,1}\) and \(\bar B(u)=\mathbb W_n(u)/N_{n,2}\).
Given \(\mathscr U_n\), \eqref{eqn--P1b-rep} is a sum of independent centred
terms.
We call its first coordinate, built from the \(a_i\) and \(b_j\), the
first-order part, and its second coordinate, built from the \(A_i\) and
\(B_j\), the cell part.

(iii)~\emph{Rates.}
As \(n\le N_{n,k}\) and \(N_{n,1}\sim N_{n,2}\sim n\) by
\Cref{asm--two-scale}\eqref{asm--two-scale-rates}, so that
\(N_{n,1}N_{n,2}/n^3\to0\),
\begin{equation}
  \frac{N_{n,2}}{N_{n,1}r_n}=\Bigl(\frac{N_{n,2}}{N_{n,1}^2}\Bigr)^{2/3}
  \le\Bigl(\frac{N_{n,1}N_{n,2}}{n^3}\Bigr)^{2/3}\to0
  \quad\text{and, likewise,}\quad
  \frac{N_{n,1}}{N_{n,2}r_n}\to0.
  \label{eqn--P1b-rates}
\end{equation}
Since \(r_n^3=N_{n,1}N_{n,2}\), also
\(r_n/N_{n,1}=N_{n,2}/r_n^2=\{N_{n,2}/(N_{n,1}r_n)\}^{1/2}\to0\), and symmetrically.

\emph{Step 2: the first-order part.}
Let \(t_n>0\) satisfy \(\sup_nt_n^2\|\Omega_n\|<\infty\).
Then, for \(k=1,2\),
\begin{equation}
  \frac{t_n}{N_{n,k}}\max_{i\le N_{n,k}}\|\Delta_{k,i}\|\pto0
  \quad\text{and}\quad
  \frac{t_n^2}{N_{n,k}^2}\sum_{i=1}^{N_{n,k}}(\Delta_{k,i}-\hatE_{n,k}[\Delta_k])
  (\Delta_{k,i}-\hatE_{n,k}[\Delta_k])'-\frac{t_n^2}{N_{n,k}}\Sigma_{n,k}\pto0.
  \label{eqn--P1b-first}
\end{equation}
Indeed, recall that \(s_{n,k}^2=\E\|\Delta_k^{(n)}\|^2\); if it vanishes
there is nothing to prove.
Otherwise, as \(t_n^2s_{n,k}^2/N_{n,k}\le d\,t_n^2\|\Omega_n\|\le C\), it
suffices to prove \eqref{eqn--P1b-first} with \(t_n^2=N_{n,k}/s_{n,k}^2\).
Uniform integrability of \(\|\Delta_k^{(n)}\|^2/s_{n,k}^2\) gives
\(\Pr(\max_i\|\Delta_{k,i}\|^2>\epsilon N_{n,k}s_{n,k}^2)
\le\epsilon^{-1}\E[s_{n,k}^{-2}\|\Delta_k^{(n)}\|^2
\ind\{\|\Delta_k^{(n)}\|^2>\epsilon N_{n,k}s_{n,k}^2\}]\to0\).
It also gives the weak law of large numbers
\(\bigl\{N_{n,k}^{-1}\sum_i\Delta_{k,i}\Delta_{k,i}'-\Sigma_{n,k}\bigr\}/s_{n,k}^2\pto0\),
by truncation at a fixed level.
Finally, \(\E\|\hatE_{n,k}[\Delta_k]\|^2/s_{n,k}^2=N_{n,k}^{-1}\).
With \(t_n=r_n\), \eqref{eqn--P1b-first} gives
\(\max_i\|a_i-\bar a\|+\max_j\|b_j-\bar b\|\pto0\) and
\(\sum_i(a_i-\bar a)(a_i-\bar a)'+\sum_j(b_j-\bar b)(b_j-\bar b)'
=r_n^2\Omega_n+\op(1)\pto\Omega\).

\emph{Step 3: the cell part.}
We find the limit of the conditional covariance function of
\(\mathbb W^*_n\) and show that the largest summands of the cell part are
negligible.

(i)~\emph{Decomposition.}
Fix \(s,t\in\calK\).
Then
\[
  \sum_i\{A_i(s)-\bar A(s)\}\{A_i(t)-\bar A(t)\}
  =S^{\rm d}_n+S^{\rm o}_n-N_{n,1}^{-1}\mathbb W_n(s)\mathbb W_n(t),
\]
where the last term is \(\Op(N_{n,1}^{-1})\) and
\[
  S^{\rm d}_n:=\frac{1}{r_n^2}\sum_{i,j}\pi_{ij}[f_s]\pi_{ij}[f_t],
  \qquad
  S^{\rm o}_n:=\frac{1}{r_n^2}\sum_{i}\sum_{j\ne j'}\pi_{ij}[f_s]\pi_{ij'}[f_t].
\]

(ii)~\emph{Diagonal terms.}
By orthogonality of the Hoeffding components,
\[
  \E S^{\rm d}_n=r_n\E\{\pi_{11}[f_s]\pi_{11}[f_t]\}
  =r_n\Bigl\{\E[f_sf_t]-\E f_s\E f_t
  -\sum_{k=1}^2\E\bigl(\pi_{\e_k}[f_s]\pi_{\e_k}[f_t]\bigr)\Bigr\}.
\]
By \eqref{eqn--P1a-polar} with \(\delta=r_n^{-1}\) and \(\theta=\theta_0\),
\(r_n\E[f_sf_t]\to C_{\mathbb W}(s,t)\).
The other terms are \(O(r_n^{-1})\), because \(|\E f_u|\le\E\dot m(\bm U^\circ)\,\|u\|/r_n\) and
\(|\pi_{(i,0)}[f_u]|\le\dot m_i\|u\|/r_n\) (similarly for
\(\pi_{(0,j)}[f_u]\)).
Two summands \(\pi_{ij}[f_s]\pi_{ij}[f_t]\) are independent unless they share an
index.
There are at most \(|\calI_n|(N_{n,1}+N_{n,2})\) ordered pairs of summands that share
an index, and each summand
has second moment at most \(64\E f_s^2\le C/r_n\).
Hence \(\var S^{\rm d}_n\le C(N_{n,1}+N_{n,2})/r_n^2\to0\) by \eqref{eqn--P1b-rates},
and \(S^{\rm d}_n\pto C_{\mathbb W}(s,t)\).

(iii)~\emph{Off-diagonal terms.}
\(\E S^{\rm o}_n=0\) by (P).
In
\[
  \E(S^{\rm o}_n)^2=r_n^{-4}\sum_{i,i'}\sum_{j\ne j'}\sum_{l\ne l'}
  \E\{\pi_{ij}[f_s]\pi_{ij'}[f_t]\pi_{i'l}[f_s]\pi_{i'l'}[f_t]\},
\]
a term vanishes by (P) unless \(\{l,l'\}=\{j,j'\}\).
We bound the remaining terms separately for \(i=i'\) and for \(i\ne i'\).
There are at most \(2N_{n,1}N_{n,2}^2\) remaining terms with \(i=i'\); as
\(2|xy|\le x^2+y^2\), each is at most the larger of
\(\E\{\pi_{ij}[f_s]^2\pi_{ij'}[f_t]^2\}\) and
\(\E\{\pi_{ij'}[f_s]^2\pi_{ij}[f_t]^2\}\).
Given \(U_{(i,0)}\), the two factors are independent, so by
\eqref{eqn--P1b-moments} the term is at most \(C\E\dot m_i^2/r_n^2\).
There are at most \(2N_{n,1}^2N_{n,2}^2\) remaining terms with \(i\ne i'\).
In each of them, the four factors are
independent given \(\mathscr G_n\), with conditional means
\(\mathsf I_{ij}[f_s]\), \(\mathsf I_{ij'}[f_t]\), and so on.
Apply the Cauchy--Schwarz inequality, pairing \((i,j)\) with \((i',j')\)
and \((i,j')\) with \((i',j)\); the members of each pair depend on disjoint
sets of shocks.
The term is then at most
\(\E\mathsf I_{11}[f_s]^2\,\E\mathsf I_{11}[f_t]^2\le C/r_n^4\).
Hence
\(\E(S^{\rm o}_n)^2\le C(N_{n,1}N_{n,2}^2r_n^{-6}+N_{n,1}^2N_{n,2}^2r_n^{-8})\), and
the right-hand side equals \(C(N_{n,1}^{-1}+r_n^{-2})\to0\).

Combining (i)--(iii),
\(\sum_i\{A_i(s)-\bar A(s)\}\{A_i(t)-\bar A(t)\}\pto C_{\mathbb W}(s,t)\).
By symmetry the same holds for the \(B_j\), with the same diagonal
\(S^{\rm d}_n\).
The conditional covariance of \(\mathbb W^*_n(s)\) and \(\mathbb W^*_n(t)\) is
the sum of the two, and so tends to \(2C_{\mathbb W}(s,t)=\kappa C_{\mathbb W}(s,t)\).

(iv)~\emph{Largest summands.}
Given \(U_{(i,0)}\), the \(\pi_{ij}[f_u]\), \(j\le N_{n,2}\), are independent and
centred.
By \eqref{eqn--P1b-moments},
\[
  \E A_i(u)^4\le r_n^{-4}\E\bigl[N_{n,2}\E\{\pi_{i1}[f_u]^4\mid U_{(i,0)}\}
  +3N_{n,2}^2\E\{\pi_{i1}[f_u]^2\mid U_{(i,0)}\}^2\bigr]
  \le Cr_n^{-4}(N_{n,2}r_n^{-1}+N_{n,2}^2r_n^{-2}).
\]
Hence \(\E\max_iA_i(u)^4\le\sum_i\E A_i(u)^4\le C(r_n^{-2}+N_{n,1}^{-1})\).
Since also \(\bar A(u)=\Op(N_{n,1}^{-1})\), and similarly for the \(B_j\),
\(\max_i|A_i(u)-\bar A(u)|+\max_j|B_j(u)-\bar B(u)|\pto0\).

\emph{Step 4: the cross-covariance.}
We have
\[
  \sum_i(a_i-\bar a)\{A_i(u)-\bar A(u)\}
  =N_{n,1}^{-1}\sum_{i,j}\Delta_{1,i}\pi_{ij}[f_u]
  -N_{n,1}^{-1}r_n\hatE_{n,1}[\Delta_1]\mathbb W_n(u).
\]
The second term is \(\Op(N_{n,1}^{-1})\), as
\(\E\|r_n\hatE_{n,1}[\Delta_1]\|^2\le r_n^2\operatorname{tr}\Omega_n\le C\).
The first term has mean zero by (P).
In its second moment, (P) leaves only the pairs \((i,j),(i',j')\) with \(j=j'\).
The \(N_{n,1}N_{n,2}\) terms with \(i=i'\) are each at most \(64s_{n,1}^2\).
For \(i\ne i'\), conditioning on \(U_{(0,j)}\) makes the two factors
independent, each with conditional mean
\(\psi(U_{(0,j)}):=\E[\Delta_{1,i}\mathsf I_{ij}[f_u]\mid U_{(0,j)}]\);
by the Cauchy--Schwarz inequality and \eqref{eqn--P1b-moments},
\(\E\|\psi\|^2\le s_{n,1}^2\E\mathsf I_{11}[f_u]^2\le Cs_{n,1}^2/r_n^2\).
As \(s_{n,1}^2\le dN_{n,1}\|\Omega_n\|\le CN_{n,1}/r_n^2\), the second moment is at
most
\(Cs_{n,1}^2(N_{n,2}/N_{n,1}+N_{n,2}/r_n^2)\le C(N_{n,2}/r_n^2+N_{n,1}N_{n,2}/r_n^4)
=C[\{N_{n,2}/(N_{n,1}r_n)\}^{1/2}+r_n^{-1}]\).
This tends to zero by \eqref{eqn--P1b-rates}.
The column term \(\sum_j(b_j-\bar b)\{B_j(u)-\bar B(u)\}\) is handled in the
same way, using \(N_{n,1}/(N_{n,2}r_n)\to0\).
So the conditional covariance of \(r_n\tilde\Delta^*_n\) and \(\mathbb W^*_n(u)\)
tends to zero in probability.

\emph{Step 5: finite-dimensional laws.}
Let \(u_1,\dots,u_p\in\calK\) and stack \eqref{eqn--P1b-rep} over them.
By Steps~2--4, the largest summand tends to zero in probability and
the conditional covariance matrix converges in probability to
\(\mathrm{diag}(\Omega,\kappa[C_{\mathbb W}(u_l,u_m)]_{l,m})\), the covariance matrix of
\((G,\kappa^{1/2}\mathbb W(u_1),\dots,\kappa^{1/2}\mathbb W(u_p))\).
By (LF), with the \(N_{n,1}+N_{n,2}\) weights as the \(\zeta_l\),
\[
  \bigl(r_n\tilde\Delta^*_n,\mathbb W^*_n(u_1),\dots,\mathbb W^*_n(u_p)\bigr)
  \dtop\bigl(G,\kappa^{1/2}\mathbb W(u_1),\dots,\kappa^{1/2}\mathbb
  W(u_p)\bigr)
\]
conditionally given \(\mathscr U_n\).

\emph{Step 6: stochastic equicontinuity and conclusion.}
We combine Step~5 with stochastic equicontinuity of \(\mathbb W^*_n\).

(i)~\emph{Stochastic equicontinuity.}
For \(\delta>0\), let \(\calF_{n,\delta}:=\{f_u-f_{u'}:u,u'\in\calK,\
\|u-u'\|\le\delta\}\), with \(\omega_\delta\) as in the proof of part (a).
Then \(\omega_\delta(\mathbb W_n)=r_n^2\|H_{n,\bm\iota}\|_{\calF_{n,\delta}}\)
and \(\omega_\delta(\mathbb W^*_n)
=r_n^2\|H^*_{n,\bm\iota}-H_{n,\bm\iota}\|_{\calF_{n,\delta}}\).
The bound \(\E\|H^*_{n,\bm\iota}-H_{n,\bm\iota}\|_{\calF}\le5(L_1+L_2)
\E\|H_{n,\bm\iota}\|_{\calF}\), derived in the second paragraph of the proof of
\Cref{lem--cell-maximal}, holds with \(\calF=\calF_{n,\delta}\), as \(\calF_{n,\delta}\)
is also a non-stochastic class of differences \(m_{\theta'}-m_\theta\) indexed by a
Borel set of pairs.
This gives
\[
  \E\omega_\delta(\mathbb W^*_n)\le5(L_1+L_2)\,\E\omega_\delta(\mathbb W_n),
  \quad\text{so}\quad
  \lim_{\delta\downarrow0}\limsup_n\E\omega_\delta(\mathbb W^*_n)=0
\]
by \Cref{lem--two-scale-local}(a).

(ii)~\emph{Approximation in \(\dBL\).}
Let \(X^*_n:=(r_n\tilde\Delta^*_n,\mathbb W^*_n)\) and
\(X:=(G,\kappa^{1/2}\mathbb W)\), with values in \(\R^d\times\ell^\infty(\calK)\)
normed by \(\|g\|+\|z\|_{\calK}\).
Given \(\mathscr U_n\), \(X^*_n\) is a linear, hence continuous, image of the
finitely many weights.
For \(\delta>0\), choose a finite \(\calK_\delta\subseteq\calK\) and a
map \(p_\delta:\calK\to\calK_\delta\) with
\(\|p_\delta(u)-u\|\le\delta\), and let \(\Pi_\delta(g,z):=(g,z\circ p_\delta)\).
Then \(\|\Pi_\delta(g,z)-(g,z)\|\le\omega_\delta(z)\).
Moreover, \(\Pi_\delta(g,z)\) is a \(1\)-Lipschitz function of
\((g,(z(v))_{v\in\calK_\delta})\) when this vector is normed by
\(\|g\|+\max_v|z(v)|\).
Hence
\[
  \dBL\bigl(\mathcal L_{\mathscr U}(X^*_n),\mathcal L(X)\bigr)
  \le\E_{\mathscr U}\min\{2,\omega_\delta(\mathbb W^*_n)\}
  +\dBL\bigl(\mathcal L_{\mathscr U}(\Pi_\delta X^*_n),\mathcal L(\Pi_\delta X)\bigr)
  +\E\min\{2,\kappa^{1/2}\omega_\delta(\mathbb W)\}.
\]
The middle term tends to zero in probability by Step~5
(the bounded Lipschitz distance for the norm \(\|g\|+\max_v|z(v)|\) is at most
\(\sqrt2\) times that for the Euclidean norm, since
\(\|g\|+\max_v|z(v)|\le\sqrt2\,(\|g\|^2+\sum_vz(v)^2)^{1/2}\)).
The last tends to zero as \(\delta\downarrow0\), because the paths of
\(\mathbb W\) are continuous.
The first has expectation at most \(\E\omega_\delta(\mathbb W^*_n)\).
The right-hand side of the last display is measurable: \(\omega_\delta(\mathbb
W^*_n)\) is a supremum over a Borel set of pairs, and a bounded Lipschitz
distance between laws on a Euclidean space is a supremum over a countable
family.
Letting \(n\to\infty\) and then \(\delta\downarrow0\) therefore proves
\(\dBL(\mathcal L_{\mathscr U}(X^*_n),\mathcal L(X))\ptoout0\).
\qed

\subsubsection{Proof of \Cref{thm--two-scale-boot}}

We keep the conventions of \Cref{sec--prf--lem--two-scale-local-b}.
Any auxiliary randomisation in \(\hat\theta_n\) or \(\tilde A_n\), such as
random tie-breaking, is included in \(\mathscr U_n\); as it is independent of
the latent variables and the weights (\Cref{sec--prf--two-scale}), the
conditional law of the weights given \(\mathscr U_n\) is unchanged.
Hence, as \(\sigma(\bm W_n)\subseteq\mathscr U_n\), \(\hat\theta_n\) and
\(\hat u_n=r_n(\hat\theta_n-\theta_0)\) (as in the proof of
\Cref{thm--two-scale}) are \(\mathscr U_n\)-measurable.
Any randomisation used to compute \(\tilde\theta^*_n\) is integrated out by
\(\E_{\mathscr U}\) together with the weights.
We prove each claim conditionally on \(\mathscr U_n\) and transfer it to
\(\bm W_n\), with or without the randomisation in \(\hat\theta_n\), by
\Cref{lem--strong-cond-weak-suff}, which applies to every
sub-\(\sigma\)-algebra of \(\mathscr U_n\); unconditional \(\op\) and \(\Op\)
terms remain so conditionally on \(\mathscr U_n\)
(\Cref{lem--cond-equiv-uncond-op1-Op1}).
We use two elementary bounds:
\(\dBL(\mathcal L_{\mathscr U}(X+Y),Q)\le\dBL(\mathcal L_{\mathscr
U}(X),Q)+\E_{\mathscr U}\min\{2,\|Y\|\}\), and
\(\dBL(\mathcal L(\Phi(X)),\mathcal L(\Phi(Y)))\le\lambda\,\dBL(\mathcal
L(X),\mathcal L(Y))\) for a \(\lambda\)-Lipschitz map \(\Phi\) with
\(\lambda\ge1\).

\emph{Part (a).}
If \(c_n\to0\), then \(\Op(\sigma_nc_n^{3/4})=\op(\sigma_n)\), and (a) is
\Cref{lem--two-scale-rate}(b).

\emph{Part (b).}
Suppose \(r_n^2\Omega_n\to\Omega\).
Then \(r_n\sigma_n\to(1+\|A_0^{-1}\Omega A_0^{-1}\|)^{1/2}\), so \(\sigma_n\asymp r_n^{-1}\).
\Cref{lem--two-scale-rate} gives \(\hat u_n=\Op(1)\) and
\(\hat v^*:=r_n(\tilde\theta^*_n-\hat\theta_n)=\Op(1)\).

\emph{Step 1: localisation.}
Fix \(T>0\) and let \(\bar B_T\) be the closed ball of radius \(T\) about the
origin.
On the event \(E_n:=\{\hat\theta_n+r_n^{-1}\bar B_T\subseteq\calN\}\), whose
probability tends to one, define
\(\tilde{\mathbb M}^*_n(v):=r_n^2\{\tilde M^*_n(\hat\theta_n+v/r_n)
-\tilde M^*_n(\hat\theta_n)\}\) for \(\|v\|\le T\).
Off this event, which is \(\mathscr U_n\)-measurable, set
\(\tilde{\mathbb M}^*_n:=\mathbb M^\circ_n:=0\), with \(\mathbb M^\circ_n\) as
below; this affects no limit below.
For multiplier weights and \(K=2\),
\(\hat M^*_n(\theta)-\hat M_n(\theta)=\tilde Z^*_{n,1}(\theta)+(H^*_{n,\bm\iota}-H_{n,\bm\iota})[m_\theta]\)
exactly (see the proof of \Cref{lem--two-scale-rate}(b)); moreover
\(\hat\theta_n+v/r_n=\theta_0+(\hat u_n+v)/r_n\).
So \eqref{eqn--bootstrap-hoeffding-decomps-and-rems} gives
\[
  \tilde{\mathbb M}^*_n(v)=-\tfrac12v'\tilde A_nv+v'r_n\tilde\Delta^*_n
  +\mathbb W^*_n(\hat u_n+v)-\mathbb W^*_n(\hat u_n)
  +r_n^2\bigl\{\tilde R^*_{n,1}(\hat\theta_n+v/r_n)-\tilde R^*_{n,1}(\hat\theta_n)\bigr\}.
\]
Let \(\phi^*_n\) be as in that proof, so that
\(\phi^*_n(\eta_n)\pto0\) whenever the non-stochastic \(\eta_n\) tend to
\(0\).
For \(\rho>0\), on \(\{\|\hat u_n\|\le\rho\}\), both arguments of
\(\tilde R^*_{n,1}\) lie within \((\rho+T)/r_n\) of \(\theta_0\).
The last term is therefore at most
\(2(\rho+T)(\rho+T+1)\phi^*_n((\rho+T)/r_n)\) in absolute value.
As \(\hat u_n=\Op(1)\) and \(\tilde A_n\pto A_0\),
\begin{equation}
  \sup_{\|v\|\le T}\bigl|\tilde{\mathbb M}^*_n(v)-\mathbb M^\circ_n(v)\bigr|\pto0,
  \qquad
  \mathbb M^\circ_n(v):=-\tfrac12v'A_0v+v'r_n\tilde\Delta^*_n
  +\mathbb W^*_n(\hat u_n+v)-\mathbb W^*_n(\hat u_n).
  \label{eqn--P1b-loc}
\end{equation}

\emph{Step 2: conditional limit of the criterion.}
Let \(\mathbb M^\kappa(v):=-\frac12v'A_0v+v'G+\kappa^{1/2}\mathbb W(v)\), with
\(G\sim N(0,\Omega)\) independent of \(\mathbb W\).
For \(\|u\|\le\rho\), define
\(\Phi_u(g,z):=\bigl(v\mapsto-\tfrac12v'A_0v+v'g+z(u+v)-z(u)\bigr)\), a map
from \(\R^d\times\ell^\infty(\bar B_{\rho+T})\) to \(\ell^\infty(\bar B_T)\).
It is \((T+2)\)-Lipschitz, because \(z\mapsto z(u+\cdot)-z(u)\) is
\(2\)-Lipschitz in the supremum norm.
Moreover \(\Phi_u(G,\kappa^{1/2}\mathbb W)\overset d=\mathbb M^\kappa\) on
\(\bar B_T\): by stationary increments, \(\mathbb W(u+\cdot)-\mathbb W(u)\) is
a centred Gaussian process with continuous paths and covariance \(C_{\mathbb W}\),
independent of \(G\).
Since \(\hat u_n\) is \(\mathscr U_n\)-measurable, on \(\{\|\hat u_n\|\le\rho\}\)
\[
  \dBL\bigl(\mathcal L_{\mathscr U}(\mathbb M^\circ_n),\mathcal L(\mathbb M^\kappa)\bigr)
  \le(T+2)\,\dBL\bigl(\mathcal L_{\mathscr U}(r_n\tilde\Delta^*_n,\mathbb W^*_n),
  \mathcal L(G,\kappa^{1/2}\mathbb W)\bigr),
\]
with laws on \(\ell^\infty(\bar B_T)\) and on \(\R^d\times\ell^\infty(\bar
B_{\rho+T})\), respectively.
The right-hand side tends to zero in outer probability by \Cref{lem--two-scale-local}(b)
with \(\calK=\bar B_{\rho+T}\).
As \(\hat u_n=\Op(1)\), letting \(\rho\to\infty\) and using
\eqref{eqn--P1b-loc} gives, for every \(T\),
\(\dBL(\mathcal L_{\mathscr U}(\tilde{\mathbb M}^*_n),\mathcal L(\mathbb
M^\kappa))\ptoout0\) in \(\ell^\infty(\bar B_T)\).

\emph{Step 3: limit maximiser.}
Completing the square as in the proof of \Cref{thm--two-scale}, with
\(\mathbb W\) replaced by \(\kappa^{1/2}\mathbb W\) and \(\mathbb W_g\) as
defined there, \(\mathbb M^\kappa(A_0^{-1}G+w)=\kappa^{1/2}\mathbb W_G(w)-\frac12w'A_0w\)
plus a term free of \(w\), and \(\mathbb W_G\overset d=\mathbb W\) given \(G\).
Self-similarity gives
\(\kappa^{1/2}\mathbb W(\kappa^{1/3}x)-\frac12\kappa^{2/3}x'A_0x
\overset d=\kappa^{2/3}\{\mathbb W(x)-\frac12x'A_0x\}\) as processes in \(x\).
By \Cref{thm--two-scale}, \(\calZ\) is almost surely unique.
Hence \(\mathbb M^\kappa\) has an almost surely unique maximiser
\(\tau^\kappa\), and \(\tau^\kappa\overset d=A_0^{-1}G+\kappa^{1/3}\calZ\) with
\(\calZ\) independent of \(G\).
For \(R>0\), the restriction of \(\mathbb M^\kappa\) to \(\bar B_R\) also has an
almost surely unique maximiser \(\tau^\kappa_R\), by
\citet[Lemma~2.6]{kim1990cube}, since
\(\var\{\mathbb M^\kappa(s)-\mathbb M^\kappa(t)\}\ge\kappa\Lambda(s-t)>0\) for
\(s\ne t\).
Moreover \(\tau^\kappa_R=\tau^\kappa\) on \(\{\|\tau^\kappa\|<R\}\).

\emph{Step 4: conditional argmax.}
The argmax continuous mapping theorem concerns unconditional laws, and
\(\hat v^*=\Op(1)\), which also holds conditionally on \(\mathscr U_n\) in
probability (\Cref{lem--cond-equiv-uncond-op1-Op1}), does not give tightness
of \(\mathcal L_{\mathscr U}(\hat v^*)\) for a fixed realisation along a
single subsequence, so we argue in three stages:
(i) replace \(\hat v^*\) by a measurable near-maximiser \(\check v_n\) over a
fixed ball \(\bar B_R\);
(ii) along subsequences on which Step~2 holds almost surely, fix a realisation
of the conditioning variables and apply \citet[Theorem~2.7]{kim1990cube} to
the conditional laws;
(iii) remove the truncation, as \(\hat v^*\in\bar B_R\) with high probability.
Stage~(ii) is the device of \citet[Theorem~5.1 and Section~6,
pp.~499--500]{patra2018consistent}; there, the rate also holds for each
realisation, so no truncation is needed.

\emph{(i) Truncation.}
Fix \(\epsilon>0\).
Take \(R\) with \(\sup_n\Pr(\|\hat v^*\|>R)\le\epsilon^2\) and
\(\Pr(\|\tau^\kappa\|\ge R)\le\epsilon\), and use Steps~1 and~2 with \(T=R\).
Off \(E_n\), let \(\check v_n:=0\).
On it, let \(\check v_n:=\hat v^*\) if \(\|\hat v^*\|\le R\) and
\(\tilde A_n\) is positive semi-definite, a Borel condition under which
\(\sup_{\theta\in\Theta}\tilde M^*_n(\theta)<\infty\) (see the paragraph after
\eqref{eqn--m-est-boot-def}); otherwise let
\(\check v_n\) be a point of \(\bar B_R\) with \(\tilde{\mathbb
M}^*_n(\check v_n)\ge\sup_{\bar B_R}\tilde{\mathbb M}^*_n-r_n^{-1}\).
Its conditional law is defined only if \(\check v_n\) is measurable.
This is not automatic: \(\tilde{\mathbb M}^*_n\) need not be continuous in
\(v\), so neither \(\sup_{\bar B_R}\tilde{\mathbb M}^*_n\) nor a maximiser need
be a Borel function of the data and the weights.
We replace \(\hat\theta_n\) and \(\tilde A_n\) by almost surely equal Borel
functions of the latent variables, the weights and any auxiliary randomisation;
this changes no law.
Then, by \Cref{asm--m-est-general}\eqref{asm--m-est-meas},
\(\tilde{\mathbb M}^*_n(v)\) is a Borel function of \(v\) and of these
variables, which take values in a Borel space.
So \citet[Proposition~7.50(a), p.~184]{bertsekas1978stochastic}, applied to
\(-\tilde{\mathbb M}^*_n\) on the Borel, hence analytic, set \(E_n\times\bar
B_R\) (a Borel function is lower semianalytic by their Definition~7.21,
p.~177), gives such a point that is
analytically, hence universally, measurable, that is, measurable for the
completion of the Borel \(\sigma\)-algebra under every probability measure
(their Definition~7.18, p.~167, and Definitions~7.19--7.20 and the inclusions
on p.~171).
As \(\hat v^*\), \(\tilde A_n\) and \(E_n\) are Borel, \(\check v_n\) is then
universally measurable too.
By their Proposition~7.44 (p.~172), its section at each fixed value of the
latent variables and of the auxiliary randomisation is universally measurable
in the weights, so \(\mathcal L_{\mathscr U}(\check v_n)\) is defined for
every realisation and, by Fubini's theorem, is a version of the conditional
law of \(\check v_n\) given \(\mathscr U_n\).
Then \(\check v_n\in\bar B_R\) and, on \(E_n\), by
\eqref{eqn--m-est-boot-def},
\(\check v_n\) maximises \(\tilde{\mathbb M}^*_n\) over \(\bar B_R\) up to
\(r_n^2\tilde\zeta^*_n+r_n^{-1}\).

\emph{(ii) Argmax for each realisation.}
A sequence converges in (outer) probability if and only if every
subsequence has a further subsequence that converges (outer) almost surely
\citep[Lemma~1.9.2(ii)]{2023vandervaartWeakConvergenceEmpirical}; the bounded
Lipschitz distance below, on \(\ell^\infty(\bar B_R)\), converges in outer
probability by Step~2.
So along a subsequence of any subsequence, outside a null set,
\(\dBL(\mathcal L_{\mathscr U}(\tilde{\mathbb M}^*_n),\mathcal L(\mathbb
M^\kappa))\to0\) in \(\ell^\infty(\bar B_R)\), \(E_n\) occurs
eventually, and
\(\Pr_{\mathscr U}(r_n^2\tilde\zeta^*_n>\epsilon')\to0\) for every rational
\(\epsilon'>0\).
Fix such a realisation; under \(\Pr_{\mathscr U}\), the problem is then
unconditional.
Given the conditioning variables, \(\tilde{\mathbb M}^*_n\) is an affine
function of the finitely many weights \(\xi_{n,\i}\), \(\i\in\calI_n\), with
coefficients bounded on \(\bar B_R\), as \(|m_\theta|\le m_*\); so it is a
continuous, hence Borel, map of the weights into \(\ell^\infty(\bar B_R)\), and
\(\mathcal L_{\mathscr U}(\tilde{\mathbb M}^*_n)\) is a Borel law.
Apply \citet[Theorem~2.7]{kim1990cube} to the processes extended to \(\R^d\) by
\(v\mapsto z(p_R(v))-\|v-p_R(v)\|\), where \(p_R\) is the metric projection
onto \(\bar B_R\).
The extension is \(1\)-Lipschitz in the supremum norm, so it preserves the
convergence.
It leaves maximisers and near-maximisers over \(\bar B_R\) unchanged and makes
the paths tend to \(-\infty\).
So, by Step~3, the extended limit lies almost surely in the set of continuous
functions that tend to \(-\infty\) and have a unique maximiser
\citep[p.~195]{kim1990cube}; here this maximiser is \(\tau^\kappa_R\).
The other two conditions of the theorem hold, as \(\check v_n\in\bar B_R\) and
the slack \(r_n^2\tilde\zeta^*_n+r_n^{-1}\) tends to zero in conditional
probability.
This yields \(\mathcal L_{\mathscr U}(\check v_n)\leadsto\mathcal L(\tau^\kappa_R)\).
Unlike the distance in \(\ell^\infty(\bar B_R)\),
\(\dBL(\mathcal L_{\mathscr U}(\check v_n),\mathcal L(\tau^\kappa_R))\) is
measurable for the completion of the law of the conditioning variables.
On \(\R^d\) the supremum defining it may be taken over a fixed countable
family of functions \(h\) (as in the proof of
\Cref{lem--strong-cond-weak-suff}), and \(\E_{\mathscr U}h(\check v_n)\) is
such a measurable function by Fubini's theorem, since \(\check v_n\) is
universally measurable.
So almost sure convergence along the further subsequences is outer almost
sure convergence, and, as the initial subsequence was arbitrary,
\(\dBL(\mathcal L_{\mathscr U}(\check v_n),\mathcal L(\tau^\kappa_R))\pto0\).

\emph{(iii) Removing the truncation.}
On \(E_n\), whose probability tends to one,
\[
  \dBL\bigl(\mathcal L_{\mathscr U}(\hat v^*),\mathcal L_{\mathscr
  U}(\check v_n)\bigr)
  \le2\Pr_{\mathscr U}(\|\hat v^*\|>R)
  +2\Pr_{\mathscr U}\{\lambda_{\min}(\tilde A_n)<0\}.
\]
The expectation of the first term on the right is at most \(2\epsilon^2\),
and that of the second, \(2\Pr\{\lambda_{\min}(\tilde A_n)<0\}\), tends to
zero, as \(\tilde A_n\pto A_0\), which is positive definite; by
Markov's inequality, the right-hand side exceeds \(\epsilon\) with
probability at most \(2\epsilon+o(1)\).
Also \(\dBL(\mathcal L(\tau^\kappa_R),\mathcal L(\tau^\kappa))
\le2\Pr(\|\tau^\kappa\|\ge R)\le2\epsilon\).
With (ii) and the triangle inequality for \(\dBL\),
\[
  \limsup_n\Pr\bigl\{\dBL\bigl(\mathcal L_{\mathscr U}(\hat v^*),\mathcal
  L(\tau^\kappa)\bigr)>4\epsilon\bigr\}\le2\epsilon.
\]
As \(\epsilon\) is arbitrary,
\(\dBL(\mathcal L_{\mathscr U}(r_n(\tilde\theta^*_n-\hat\theta_n)),\mathcal
L(A_0^{-1}G+\kappa^{1/3}\calZ))\pto0\).
Part (b) follows by \Cref{lem--strong-cond-weak-suff}.

\emph{Main claim.}
It suffices that every subsequence has a further subsequence along which the
claim holds; as \(\|\Gamma_n\|\le1\) and \(c_n\in(0,1]\), we may choose the
further subsequence with
\(\Gamma_n\to\Gamma\) and \(c_n\to c_\infty\in[0,1]\).
Since \(\sigma_n^{-1}A_0^{-1}G_n\sim N(0,\Gamma_n)\), the target law then converges
in \(\dBL\) to the law of
\(\bm T^*:=\Gamma^{1/2}\bm Z_0+\kappa^{1/3}c_\infty\calZ\), with
\(\bm Z_0\) as in Step~3 of the proof of \Cref{thm--two-scale}.
So it suffices to show
\(\dBL(\mathcal L_{\mathscr U}(\sigma_n^{-1}(\tilde\theta^*_n-\hat\theta_n)),\mathcal
L(\bm T^*))\pto0\)
along this subsequence; the claim for the conditional law given \(\bm W_n\)
then follows as in part (b).
All lemmas are applied along the subsequence.

\emph{Case \(c_\infty>0\).}
Then \(r_n^2\Omega_n=c_n^{-2}A_0\Gamma_nA_0\to\Omega:=c_\infty^{-2}A_0\Gamma A_0\).
By Steps 1--4,
\(\dBL(\mathcal L_{\mathscr U}(r_n(\tilde\theta^*_n-\hat\theta_n)),\mathcal
L(\tau^\kappa))\pto0\), where
\(\tau^\kappa\overset d=A_0^{-1}G+\kappa^{1/3}\calZ\) and \(G\sim N(0,\Omega)\).
As \(\sigma_n^{-1}=c_nr_n\) and multiplication by \(c_n\le1\) is \(1\)-Lipschitz,
the conditional law of \(\sigma_n^{-1}(\tilde\theta^*_n-\hat\theta_n)\) is
within \(\op(1)+\E\min\{2,|c_n-c_\infty|\|\tau^\kappa\|\}\) of
\(\mathcal L(c_\infty\tau^\kappa)\) in \(\dBL\).
Finally \(\mathcal L(c_\infty\tau^\kappa)=\mathcal L(\bm T^*)\), because \(c_\infty A_0^{-1}G\sim N(0,\Gamma)\).

\emph{Case \(c_\infty=0\).}
By \Cref{lem--two-scale-rate}(b), as \(c_n\to0\),
\[
  \sigma_n^{-1}(\tilde\theta^*_n-\hat\theta_n)
  =\sigma_n^{-1}A_0^{-1}\tilde\Delta^*_n
  +\sigma_n^{-1}(\tilde A_n^{-1}-A_0^{-1})\tilde\Delta^*_n+\op(1).
\]
The middle term is \(\op(1)\), because \(\tilde A_n^{-1}\pto A_0^{-1}\) and
\(\E\|\sigma_n^{-1}\tilde\Delta^*_n\|^2\le \sigma_n^{-2}\operatorname{tr}\Omega_n\le
d\|A_0\|^2\).
Given \(\mathscr U_n\),
\[
  \sigma_n^{-1}A_0^{-1}\tilde\Delta^*_n
  =\sum_{k=1}^2\sum_{i=1}^{N_{n,k}}\xi_{k,i}\,
  \frac{A_0^{-1}(\Delta_{k,i}-\hatE_{n,k}[\Delta_k])}{\sigma_nN_{n,k}}.
\]
By \eqref{eqn--P1b-first} with \(t_n=\sigma_n^{-1}\), admissible because
\(\sigma_n^{-2}\|\Omega_n\|\le\|A_0\|^2\), the conditional covariance is
\(\sigma_n^{-2}A_0^{-1}\Omega_nA_0^{-1}+\op(1)=\Gamma_n+\op(1)\to\Gamma\), and the
largest summand tends to zero in probability.
By (LF),
\(\dBL(\mathcal L_{\mathscr U}(\sigma_n^{-1}A_0^{-1}\tilde\Delta^*_n),N(0,\Gamma))\pto0\).
The claim follows, since \(\mathcal L(\bm T^*)=N(0,\Gamma)\) when \(c_\infty=0\).
\qed

\subsection{\texorpdfstring{Proofs of \Cref{lem--anderson,lem--Z-law}, \Cref{thm--size} and \Cref{cor--boot-ci,cor--functional}}{Proofs of Lemmas \ref{lem--anderson} and \ref{lem--Z-law}, Theorem \ref{thm--size} and Corollaries \ref{cor--boot-ci} and \ref{cor--functional}}}
\label{sec--prf--size}

\begin{proof}[Proof of \Cref{lem--anderson}]
Fix \(z\in\R^d\) and let \(f_z(\lambda):=\Pr(G+\lambda z\in C)=\Pr(G\in
C-\lambda z)\).
It is even, since \(-G\overset d=G\) and \(-C=C\).
Let \(\gamma>1\), the case \(\gamma=1\) being trivial, and
\(\mu:=(\gamma+1)/(2\gamma)\in(0,1)\), so that \(\mu\gamma-(1-\mu)\gamma=1\).
As \(C\) is convex, \(\mu(C-\gamma z)+(1-\mu)(C+\gamma z)\subseteq C-z\).
So the Ehrhard--Lata{\l}a--Borell inequality
\citep[Proposition~A.2.7, p.~607]{2023vandervaartWeakConvergenceEmpirical},
which holds for every Gaussian vector, degenerate or not, applied to the sets
\(C-\gamma z\) and \(C+\gamma z\), gives
\[
  \Phi^{-1}f_z(1)\ge\mu\,\Phi^{-1}f_z(\gamma)
  +(1-\mu)\,\Phi^{-1}f_z(-\gamma)=\Phi^{-1}f_z(\gamma),
\]
where \(\Phi\) is the standard normal distribution function,
\(\Phi^{-1}(0)=-\infty\) and \(\Phi^{-1}(1)=\infty\); the middle expression
is defined because \(f_z(-\gamma)=f_z(\gamma)\).
Hence \(f_z(1)\ge f_z(\gamma)\); for
non-degenerate \(G\) this is Anderson's theorem \citep{anderson1955integral}.
Since \(G\) and \(\calY\) are independent, Fubini's theorem gives
\(\Pr(G+\gamma\calY\in C)=\int f_z(\gamma)\,\Pr(\calY\in\d z)\le\int
f_z(1)\,\Pr(\calY\in\d z)=\Pr(G+\calY\in C)\).
\end{proof}

\begin{proof}[Proof of \Cref{lem--Z-law}]
Throughout write \(q(t):=\frac12t'A_0t\) and \(\ell_u(t):=u't\) for
\(t,u\in\R^d\), and recall that
\(\calZ=\argmax_t\{\mathbb W(t)-q(t)\}\).

\emph{What the proof uses, and how it is organised.}
Only the following properties of \(\Lambda\) and \(\mathbb W\) enter, all of
them established in the paragraph on the limit process in
\Cref{sec--prf--two-scale}: \(\Lambda\) is even, homogeneous of degree one,
strictly positive off the origin and bounded by \(2\|\dot m\|_2\|\cdot\|\);
\(C_{\mathbb W}(s,t)=\frac12\{\Lambda(s)+\Lambda(t)-\Lambda(s-t)\}\); and
\(\mathbb W\) is a centred Gaussian process with continuous paths,
\(\mathbb W(0)=0\) and stationary increments, for which \(\mathbb W-q\)
almost surely tends to \(-\infty\) and has a unique maximiser. Note for later
use that \(\Lambda\) is continuous: since
\(\Lambda(v)=\var\mathbb W(v)=\|\mathbb W(v)\|_{2}^{2}\) and
\(\|\mathbb W(s)-\mathbb W(t)\|_{2}^{2}=\Lambda(s-t)\), the triangle
inequality in \(L_2\) gives
\(|\Lambda(s)^{1/2}-\Lambda(t)^{1/2}|\le\Lambda(s-t)^{1/2}
\le(2\|\dot m\|_2\|s-t\|)^{1/2}\), so \(\Lambda^{1/2}\) is
\(\tfrac12\)-H\"older and \(\Lambda\) is \(\tfrac12\)-H\"older on
bounded sets.

The argument has a single idea. Adding a linear function \(\ell_u\) to
\(\mathbb W\) translates the law of the maximiser by \(A_0^{-1}u\), by
completing the square and the stationarity of the increments (Step~4); so if
adding \(\ell_u\) also preserved null sets of the law of the
maximiser, that law would be quasi-invariant under translations, and a
measure on \(\R^d\) quasi-invariant under all translations is equivalent to
Lebesgue measure. The obstruction is that \(\ell_u\) is \emph{not} in the
Cameron--Martin space of \(\mathbb W\) --- the dual formula for that norm in
Step~3, applied to the evaluation at \(t\), gives
\(|h(t)|\le|h|_{H(P_{\mathbb W})}\Lambda(t)^{1/2}\) for every \(h\) in the
space, so its elements grow at most like \(\|t\|^{1/2}\) along a ray, whereas
\(\ell_u\) grows linearly when \(u\ne0\) --- so the theorem that preserves
null sets does not apply to it. Steps~1 and~2 construct a substitute: a function \(h\) that
\emph{is} in the Cameron--Martin space and agrees with \(\ell_u\) on a ball of
radius \(R\). Step~3 applies Cameron--Martin to \(h\). Step~4 shows that the
two maximisers agree unless one of them escapes the ball, an event of
probability tending to zero in \(R\), and so transfers the null set from
\(h\) to \(\ell_u\). Step~5 draws the conclusion.

\emph{Step 1: a spectral representation of \(\Lambda\).}
The aim is to write \(\Lambda(t)=\int|t's|\,\nu(\d s)\) for a finite measure
\(\nu\) on the unit sphere, which is what makes the construction of Step~2
possible: it exhibits \(\Lambda\) as an average of absolute linear
functionals, each of which has an elementary representation in \(L_2\).

Let \(t_1,\dots,t_p\in\R^d\) and \(z_1,\dots,z_p\in\mathbb C\), and set
\(\zeta_j:=\exp\{\sqrt2\,\mathbb W(t_j)-2\Lambda(t_j)\}\), a positive random
variable. Since \(\mathbb W(t_j)+\mathbb W(t_l)\) is centred Gaussian with
variance \(\Lambda(t_j)+\Lambda(t_l)+2C_{\mathbb W}(t_j,t_l)\), and since
\(2C_{\mathbb W}(t_j,t_l)=\Lambda(t_j)+\Lambda(t_l)-\Lambda(t_j-t_l)\), that
variance is \(2\Lambda(t_j)+2\Lambda(t_l)-\Lambda(t_j-t_l)\). Using
\(\E e^{G}=e^{\var(G)/2}\) for centred Gaussian \(G\), applied to
\(G=\sqrt2\{\mathbb W(t_j)+\mathbb W(t_l)\}\),
\[
  \E[\zeta_j\zeta_l]
  =\exp\bigl[-2\Lambda(t_j)-2\Lambda(t_l)
  +\var\{\mathbb W(t_j)+\mathbb W(t_l)\}\bigr]
  =\exp\{-\Lambda(t_j-t_l)\}.
\]
(At \(j=l\) this reads \(\E\zeta_j^2=1\): the \(\zeta_j\) are unit vectors in
\(L_2\).) As the \(\zeta_j\) are real,
\[
  \sum_{j,l}z_j\bar z_l\exp\{-\Lambda(t_j-t_l)\}
  =\sum_{j,l}z_j\bar z_l\,\E[\zeta_j\zeta_l]
  =\E\Bigl|\sum_jz_j\zeta_j\Bigr|^2\ge0,
\]
so \(t\mapsto\exp\{-\Lambda(t)\}\) is a positive definite function on
\(\R^d\). It is continuous, by the H\"older bound above, and equals one at the
origin. By Bochner's theorem \citep[Theorem~1.4.3]{rudin1962fourier} it is
therefore the characteristic function of a random vector \(X\) in \(\R^d\).

Two properties identify \(X\). It is symmetric, \(-X\overset d=X\), because
\(\Lambda\) is even and so the characteristic function is real and even. And
for \(A,B>0\), homogeneity of degree one gives
\[
  \E e^{\mathrm it'(AX^{(1)}+BX^{(2)})}
  =\exp\{-\Lambda(At)\}\exp\{-\Lambda(Bt)\}
  =\exp\{-(A+B)\Lambda(t)\}
  =\exp[-\Lambda\{(A+B)t\}]
\]
for independent copies \(X^{(1)},X^{(2)}\) of \(X\), that is
\(AX^{(1)}+BX^{(2)}\overset d=(A+B)X\). So \(X\) is a stable random vector,
the stability relation holding with \(C=A+B\) and no shift, and it is
symmetric because \(\Lambda\) is even. By
\citet[Theorem~2.1.2]{samorodnitsky1994stable} a stable vector has an index
\(\alpha\in(0,2]\) with \(C=(A^\alpha+B^\alpha)^{1/\alpha}\); taking
\(A=B=1\) gives \(2^{1/\alpha}=2\), so \(\alpha=1\) and \(X\) is symmetric
\(1\)-stable. Since
\(1\in(0,2)\), \citet[Theorem~2.4.3]{samorodnitsky1994stable} applies: there
is a finite measure
\(\nu\) on the unit sphere \(\sphere^{d-1}\), the spectral measure, with
\[
  \Lambda(t)=\int|t's|\,\nu(\d s),\qquad t\in\R^d.
\]

One consequence is recorded now because Step~2 needs it. Put
\(M_\nu:=\int ss'\,\nu(\d s)\), the second-moment matrix of \(\nu\). For
\(y\ne0\), \(y'M_\nu y=\int(y's)^2\nu(\d s)\), which vanishes only if \(\nu\)
is concentrated on the hyperplane \(\{s:y's=0\}\); but then
\(\Lambda(y)=\int|y's|\,\nu(\d s)=0\), contradicting \(\Lambda(y)>0\). Hence
\(M_\nu\succ0\) and \(M_\nu^{-1}\) exists.

\emph{Step 2: a function in the Cameron--Martin space that is linear on a
ball.}
This is the crux. Working inside \(L_2(\nu\otimes\lambda)\), where
\(\lambda\) is Lebesgue measure on \(\R\), we write \(h\) down explicitly and
record the covariance form of \(\mathbb W\) there; Step~3 then reads
membership in the Cameron--Martin space off a Cauchy--Schwarz bound.

(a) \emph{The path space.} Write \(\calC:=C(\R^d)\) for the space of
continuous real-valued functions on \(\R^d\) and let \(x\) denote a generic
element of it, so that \(x\) is a possible sample path of \(\mathbb W\) and
\(x(t)\in\R\) is its value at the point \(t\in\R^d\). Give \(\calC\) the
\(\sigma\)-algebra generated by the evaluations \(e_t(x):=x(t)\),
\(t\in\R^d\), which is its Borel \(\sigma\)-algebra for the topology of
uniform convergence on compacta
\citep[Theorem~A.3.7]{bogachev1998gaussian}, and let \(P_{\mathbb W}\) be the
law of \(\mathbb W\) on \(\calC\); under \(P_{\mathbb W}\) the coordinate
process \(t\mapsto x(t)\) therefore has the law of \(t\mapsto\mathbb W(t)\),
and an integral \(\int F(x)\,P_{\mathbb W}(\d x)\) is an expectation over
paths, \(\E[F(\mathbb W)]\). Let \(\mathcal V\) be the linear span of the
evaluations. An element \(l=\sum_ja_je_{t_j}\) of \(\mathcal V\) is a function
on \(\calC\) whose law under \(P_{\mathbb W}\) is that of
\(\sum_ja_j\mathbb W(t_j)\), so
\begin{equation}
  \int l(x)^2\,P_{\mathbb W}(\d x)
  =\E\Bigl[\Bigl(\sum_ja_j\mathbb W(t_j)\Bigr)^2\Bigr]
  =\sum_j\sum_ka_ja_kC_{\mathbb W}(t_j,t_k),
  \label{eqn--Zlaw-form}
\end{equation}
the second equality because \(\mathbb W\) is centred, so that its second
moments are its covariances, and the value finite because
\(\E[\mathbb W(t)^2]=\Lambda(t)<\infty\). This quadratic form
is the only feature of \(P_{\mathbb W}\) that Step~3 uses, and part~(b)
reproduces it.

(b) \emph{A concrete model.} For \(t\in\R^d\) define
\(f_t(s,r):=\ind\{0<r\le t's\}-\ind\{t's<r\le0\}\) on
\(\sphere^{d-1}\times\R\), where \(r\) is a real variable and not a path:
for each \(s\), \(f_t(s,\cdot)\) is the signed indicator of the interval
between \(0\) and \(t's\). The two signed indicators \(f_t(s,\cdot)\) and
\(f_v(s,\cdot)\) overlap only when \(t's\) and \(v's\) have the same sign, and
then on an interval of length \(\min(|t's|,|v's|)\), on which their product
equals \(+1\) --- both signed indicators are \(+1\) if \(t's\) and \(v's\) are
positive and both are \(-1\) if they are negative --- so that the integral of
the product is the length of that overlap:
\[
  \int f_t(s,r)f_v(s,r)\,\d\lambda(r)
  =\tfrac12\bigl(|t's|+|v's|-|(t-v)'s|\bigr),
\]
the integral being \(\min(|t's|,|v's|)\) when \((t's)(v's)>0\) and \(0\)
otherwise, which is the right-hand side in either case. Integrating against
\(\nu\) and using Step~1, \(\langle f_t,f_v\rangle
=\frac12\{\Lambda(t)+\Lambda(v)-\Lambda(t-v)\}=C_{\mathbb W}(t,v)\) in
\(L_2(\nu\otimes\lambda)\). Hence, for every finite combination,
\begin{equation}
  \Bigl\|\sum_ja_jf_{t_j}\Bigr\|^2
  =\sum_j\sum_ka_ja_kC_{\mathbb W}(t_j,t_k)
  =\int l(x)^2\,P_{\mathbb W}(\d x),
  \qquad l=\sum_ja_je_{t_j},
  \label{eqn--Zlaw-match}
\end{equation}
the second equality being \eqref{eqn--Zlaw-form}: the \(f_t\) reproduce in
\(L_2(\nu\otimes\lambda)\) the quadratic form that \(\mathbb W\) induces on
\(\mathcal V\).

(c) \emph{The clipped linear functional.} Fix \(u\in\R^d\) and set
\(\phi_u(s):=s'M_\nu^{-1}u\), so that
\[
  \int s\,\phi_u(s)\,\nu(\d s)=\Bigl(\int ss'\,\nu(\d s)\Bigr)M_\nu^{-1}u
  =M_\nu M_\nu^{-1}u=u,
\]
and \(C_u:=\int|\phi_u|\,\d\nu<\infty\) because \(\nu\) is finite and
\(\phi_u\) is bounded on the sphere. Note also
\(\|u\|=\|\int s\phi_u(s)\,\nu(\d s)\|\le\int\|s\|\,|\phi_u(s)|\,\nu(\d s)=C_u\),
since \(\nu\) lives on the unit sphere. Fix \(R>0\) and put
\(\phi_{u,R}(s,r):=\phi_u(s)\ind\{|r|\le R\}\), which lies in
\(L_2(\nu\otimes\lambda)\) because \(\int\phi_u^2\,\d\nu<\infty\) and the
\(r\)-section has length \(2R\). Then, for each \(s\),
\(\int f_t(s,r)\ind\{|r|\le R\}\,\d\lambda(r)\) is the signed length of the
part of the interval between \(0\) and \(t's\) that lies in \([-R,R]\), namely
\(t's\) clipped to \([-R,R]\); so
\begin{equation}
  h(t):=\langle f_t,\phi_{u,R}\rangle
  =\int\phi_u(s)\,\max\{-R,\min(R,t's)\}\,\nu(\d s),\qquad t\in\R^d.
  \label{eqn--Zlaw-h}
\end{equation}

(d) \emph{What \(h\) delivers.} Three properties, each immediate from
\eqref{eqn--Zlaw-h}. It is continuous, by dominated convergence. It is
\emph{exactly linear on the ball of radius \(R\)}: if \(\|t\|\le R\) then
\(|t's|\le\|t\|\le R\) for every \(s\) on the sphere, the clipping is
inactive, and \(h(t)=\int\phi_u(s)t's\,\nu(\d s)=t'\int s\phi_u(s)\nu(\d s)
=u't\). And it grows at most linearly, \(|h(t)|\le C_u\|t\|\), with \(C_u\)
\emph{not} depending on \(R\) --- which is what lets \(R\) grow in Step~4 with
the bound there held fixed.

\emph{Step 3: \(h\) lies in the Cameron--Martin space, and the shifted laws
are equivalent.}

The evaluations separate the points of \(\calC\), and every element of
\(\mathcal V\) is a finite linear combination \(\sum_ja_je_{t_j}\) of them,
whose law under \(P_{\mathbb W}\) is that of \(\sum_ja_j\mathbb W(t_j)\) and
so centred Gaussian --- degenerately so when that sum vanishes identically, as
at \(t=0\), which \citet[Definition~1.1.1]{bogachev1998gaussian} admits. Hence
\(P_{\mathbb W}\) is a centred Gaussian measure in the sense of
\citet[Definition~2.2.1]{bogachev1998gaussian}, and, by the remark following
that definition, \(\mathcal V\) is the dual of \(\calC\) for the weak topology
that it induces. For a centred measure the covariance operator of
\citet[(2.2.3)]{bogachev1998gaussian} is \(R_{P_{\mathbb W}}(l)(l)=\int
l(x)^2\,P_{\mathbb W}(\d x)\), so \eqref{eqn--Zlaw-match} and Cauchy--Schwarz
in \(L_2(\nu\otimes\lambda)\) give, for every
\(l=\sum_ja_je_{t_j}\in\mathcal V\),
\[
  l(h)=\sum_ja_jh(t_j)=\Bigl\langle\sum_ja_jf_{t_j},\phi_{u,R}\Bigr\rangle
  \le\{R_{P_{\mathbb W}}(l)(l)\}^{1/2}\,\|\phi_{u,R}\| ,
\]
by \eqref{eqn--Zlaw-h} for the second equality. So the Cameron--Martin norm of
\citet[Section~2.2]{bogachev1998gaussian}, \(|h|_{H(P_{\mathbb W})}
=\sup\{l(h):l\in\mathcal V,\ R_{P_{\mathbb W}}(l)(l)\le1\}\), is at most
\(\|\phi_{u,R}\|=(2R\int\phi_u^2\,\d\nu)^{1/2}<\infty\). As \(h\) is
continuous by Step~2(d) it is a vector of \(\calC\), so
\citet[Theorem~2.4.5(ii)]{bogachev1998gaussian} applies: \(h\) lies in the
Cameron--Martin space \(H(P_{\mathbb W})\), and the laws of \(\mathbb W+h\)
and of \(\mathbb W\) on \(\calC\) are equivalent. In particular the two laws
have the same null sets, which is the only use made of this step; that the
bound grows with \(R\) does no harm, since what Step~4 holds fixed as \(R\)
grows is \(|h(t)|\le C_u\|t\|\) from Step~2(d).

Only the finiteness of \(|h|_{H(P_{\mathbb W})}\) was needed, so no element of
the reproducing kernel Hilbert space of \(P_{\mathbb W}\) had to be exhibited.
One can be: by \eqref{eqn--Zlaw-match} the map \(e_t\mapsto f_t\) extends to a
linear isometry of the \(L_2(P_{\mathbb W})\)-closure of \(\mathcal V\) onto
the closed span of the \(f_t\), and pulling \(\phi_{u,R}\) back through it,
after projecting onto that span, gives a \(g\) with \(h=R_{P_{\mathbb W}}(g)\);
then \citet[Lemma~2.4.1 and Corollary~2.4.3]{bogachev1998gaussian} give the
equivalence again, with Cameron--Martin density
\(\exp\{g(x)-\frac12|h|^2_{H(P_{\mathbb W})}\}\). Nothing below needs it.

\emph{Step 4: quasi-invariance of the law of the maximiser.}
Let \(D\subseteq\calC\) be the set of \(x\) such that \(x(t)-q(t)\to-\infty\)
as \(\|t\|\to\infty\) and \(x-q\) has a unique maximiser, which we call
\(\mathfrak t(x)\).

(i) \emph{\(D\) and \(\mathfrak t\) are measurable.} Because \(x-q\) is
continuous, its supremum over a closed ball, or over a set
\(\{t:\|t\|\ge N'\}\), equals its supremum over the points of that set with
rational coordinates, a countable set. The divergence condition can therefore
be written as: for every \(N\in\N\) there is \(N'\in\N\) with
\(\sup_{\|t\|\ge N'}\{x(t)-q(t)\}\le-N\), a countable combination of
measurable conditions. Given divergence, the supremum is attained, and
\(x-q\) has two maximisers if and only if two disjoint closed balls with
rational centres and rational radii both attain it --- again a countable
condition. Hence \(D\) is measurable, and so is \(\mathfrak t\) on \(D\),
since for such balls \(\calK\) we have
\(\{\mathfrak t\in\calK\}=D\cap\{\sup_{\calK}(x-q)=\sup(x-q)\}\). By the
paragraph on the limit process in \Cref{sec--prf--two-scale},
\(P_{\mathbb W}(D)=1\); and \(\Pr(\mathbb W+h\in D)=1\) too, by Step~3.

(ii) \emph{Adding a linear function shifts the maximiser.} Put
\(y_u:=A_0^{-1}u\). Completing the square,
\(u't-q(t)=-q(t-y_u)+\frac12u'y_u\), so
\[
  \mathbb W(t)+u't-q(t)
  =\mathbb W_u(t-y_u)-q(t-y_u)+\mathbb W(y_u)+\tfrac12u'y_u,
  \qquad \mathbb W_u(w):=\mathbb W(y_u+w)-\mathbb W(y_u),
\]
and the last two terms do not depend on \(t\). The process \(\mathbb W_u\) ---
the shifted process of Step~1 of the proof of \Cref{thm--two-scale},
there taken at the index \(u\) --- has the law of \(\mathbb W\), by stationarity of the
increments. Hence \(\mathbb W+\ell_u\in D\) almost surely, with
\(\mathfrak t(\mathbb W+\ell_u)=y_u+\mathfrak t(\mathbb W_u)\), and therefore
\begin{equation}
  \mathfrak t(\mathbb W+\ell_u)\overset d=A_0^{-1}u+\calZ .
  \label{eqn--Zlaw-shift}
\end{equation}

(iii) \emph{Transferring a null set.} Let \(B\subseteq\R^d\) be Borel with
\(\Pr(\calZ\in B)=0\), that is
\(P_{\mathbb W}\{x\in D:\mathfrak t(x)\in B\}=0\); by Step~3,
\(\Pr\{\mathfrak t(\mathbb W+h)\in B\}=0\) as well. Now the two maximisers
coincide unless one of them leaves the ball of radius \(R\). Indeed, suppose
both \(\mathfrak t(\mathbb W+\ell_u)\) and \(\mathfrak t(\mathbb W+h)\) lie in
the open ball of radius \(R\), on which \(h=\ell_u\) by Step~2(d). Then
\(\mathbb W+\ell_u-q\) and \(\mathbb W+h-q\) agree on that ball, so the
maximiser of the first over the whole space maximises the second over the
ball; the ball also contains the \emph{global} maximiser of the second, so the
two agree. Consequently
\[
  \Pr\{\mathfrak t(\mathbb W+\ell_u)\in B\}
  \le\Pr\{\mathfrak t(\mathbb W+h)\in B\}
  +\Pr\{\|\mathfrak t(\mathbb W+\ell_u)\|\ge R\}
  +\Pr\{\|\mathfrak t(\mathbb W+h)\|\ge R\},
\]
and the first term is zero. For either of the remaining two, the maximum of
the objective is at least its value at \(t=0\), which is \(0\) because
\(\mathbb W(0)=0\), \(\ell_u(0)=h(0)=0\) and \(q(0)=0\); so if the maximiser
has norm at least \(R\) then the objective is non-negative somewhere on
\(\{\|t\|\ge R\}\). Since \(|u't|\le\|u\|\|t\|\le C_u\|t\|\) and
\(|h(t)|\le C_u\|t\|\), both by Step~2(c)--(d), each of the two terms is at
most \(\Pr(\sup_{\|t\|\ge R}\{\mathbb W(t)+C_u\|t\|-q(t)\}\ge0)\), whence by
\eqref{eqn--Zlaw-shift}
\begin{equation}
  \Pr(A_0^{-1}u+\calZ\in B)
  \le2\Pr\Bigl(\sup_{\|t\|\ge R}\bigl\{\mathbb W(t)+C_u\|t\|-q(t)\bigr\}
  \ge0\Bigr).
  \label{eqn--Zlaw-escape}
\end{equation}

(iv) \emph{Letting the radius grow.} The right-hand side of
\eqref{eqn--Zlaw-escape} does not depend on \(R\) except through the range of
the supremum, because \(C_u\) does not. Split \(q=\frac14t'A_0t+\frac14t'A_0t\)
and bound the linear term against one half: writing \(\varrho:=\|t\|\),
\(C_u\varrho-\frac14\lambda_{\min}(A_0)\varrho^2\le C_u^2/\lambda_{\min}(A_0)\),
the maximum over \(\varrho\ge0\) of a concave quadratic, so
\[
  \mathbb W(t)+C_u\|t\|-q(t)
  \le\Bigl\{\mathbb W(t)-\tfrac14t'A_0t\Bigr\}
  +\frac{C_u^2}{\lambda_{\min}(A_0)} .
\]
By \citet[Lemma~2.5]{kim1990cube}, \(\mathbb W(t)-\frac14t'A_0t\to-\infty\)
almost surely as \(\|t\|\to\infty\); hence the supremum in
\eqref{eqn--Zlaw-escape} is eventually negative and the probability tends to
zero as \(R\to\infty\). Therefore \(\Pr(A_0^{-1}u+\calZ\in B)=0\). As \(u\)
was arbitrary, taking \(u=A_0y\) gives \(\Pr(\calZ\in B-y)=0\) for every
\(y\in\R^d\): the law \(\mu\) of \(\calZ\) is quasi-invariant under
translations, in the strong sense that a null set stays null under every
shift.

\emph{Step 5: from quasi-invariance to the conclusion.}
Let \(B\subseteq\R^d\) be Borel with Lebesgue measure \(|B|\). By Tonelli's
theorem and the translation invariance of Lebesgue measure,
\begin{equation}
  \int_{\R^d}\mu(B-y)\,\d y
  =\int_{\R^d}\!\int_{\R^d}\ind_B(z+y)\,\d y\,\mu(\d z)
  =\int_{\R^d}|B|\,\mu(\d z)=|B| .
  \label{eqn--Zlaw-fubini}
\end{equation}
If \(\mu(B)=0\), Step~4 gives \(\mu(B-y)=0\) for every \(y\), so the
left-hand side of \eqref{eqn--Zlaw-fubini} vanishes and \(|B|=0\).
Conversely, if \(|B|=0\) then the left-hand side vanishes, so \(\mu(B-y)=0\)
for Lebesgue-almost every \(y\), and in particular for some \(y\); applying
Step~4 to the set \(B-y\) and the shift \(-y\) returns \(\mu(B)=0\). Hence
\(\mu\) and Lebesgue measure on \(\R^d\) are mutually absolutely continuous,
which is the first assertion.

For the second, fix \(a\ne0\) and a Borel set \(E\subseteq\R\), and work in
orthonormal coordinates whose first axis is \(a/\|a\|\). In those coordinates
\(\{z:a'z\in E\}\) is the product of \(\|a\|^{-1}E\) with \(\R^{d-1}\), so by
Tonelli's theorem it is Lebesgue-null in \(\R^d\) if and only if \(E\) is
Lebesgue-null in \(\R\). Combining this with the first assertion, the law of
\(a'\calZ\) and Lebesgue measure on \(\R\) are mutually absolutely
continuous. Absolute continuity of that law makes the distribution function
of \(a'\calZ\) continuous; and absolute continuity in the other direction
makes it strictly increasing, since a non-degenerate interval has positive
Lebesgue measure and hence positive probability.
\end{proof}

\begin{proof}[Proof of \Cref{thm--size}]
Let \(T_{n,a}:=\sigma_{n,a}^{-1}a'(\hat\theta_n-\theta_0)\),
\(T^*_{n,a}:=\sigma_{n,a}^{-1}a'(\tilde\theta^*_n-\hat\theta_n)\) and \(\hat
q'_n:=\sigma_{n,a}^{-1}\hat q_n\), a conditional \((1-\alpha)\)-quantile of
\(|T^*_{n,a}|\); then \(a'\theta_0\in \mathrm{CI}_n\) if and only if \(|T_{n,a}|\le\hat q'_n\).
For \(c\in[0,1]\) let \(T_{c,a}\) be as in the theorem,
\(T^*_{c,a}:=(1-c^2)^{1/2}Z_0+\kappa^{1/3}c\,a'\calZ\), \(F_c\) and \(F^*_c\)
the distribution functions of \(|T_{c,a}|\) and \(|T^*_{c,a}|\), and
\(q_{c,a}:=(F^*_c)^{-1}(1-\alpha)\).

\emph{Step 1: the limit laws.}
For \(c<1\), \(T_{c,a}\) and \(T^*_{c,a}\) are convolutions of a
non-degenerate normal law with another law, and so have positive continuous
densities on \(\R\); for \(c=1\), they are positive multiples of \(a'\calZ\),
whose law is mutually absolutely continuous with Lebesgue measure by
\Cref{lem--Z-law}.
Either way, \(F_c\) and \(F^*_c\) are continuous and strictly increasing on
\([0,\infty)\), so \(q_{c,a}>0\) is unique and \((F^*_c)^{-1}\) is continuous
at \(1-\alpha\).

\emph{Step 2: convergence along subsequences.}
Every subsequence has a further subsequence along which \(\Gamma_n\to\Gamma\),
\(c_n\to c_\infty\), \(a'\Gamma_na+c_n^2\to\lambda_a\) and
\(c_{n,a}\to c\in[0,1]\); all limits in this step are along it.
We show that (i) \(T_{n,a}\leadsto T_{c,a}\) and (ii)
\(\Pr(|T^*_{n,a}|\le x\mid\bm W_n)\pto F^*_c(x)\) for every \(x\in\R\).
These give the coverage.
For \(\varepsilon>0\), \(F^*_c(q_{c,a}-\varepsilon)<1-\alpha<
F^*_c(q_{c,a}+\varepsilon)\) by Step~1, so (ii) gives \(\Prout_n(|\hat
q'_n-q_{c,a}|>\varepsilon)\to0\), whichever conditional quantile is chosen:
with probability tending to one,
\(\Pr(|T^*_{n,a}|\le q_{c,a}+\varepsilon\mid\bm W_n)>1-\alpha\), so \(\hat
q'_n\le q_{c,a}+\varepsilon\), and
\(\Pr(|T^*_{n,a}|\le q_{c,a}-\varepsilon\mid\bm W_n)<1-\alpha\), so \(\hat
q'_n\ge q_{c,a}-\varepsilon\).
By (i) and Slutsky's lemma, \(|T_{n,a}|-\hat q'_n\leadsto
|T_{c,a}|-q_{c,a}\), whose distribution function is continuous at zero, so
\(\Pr_n(a'\theta_0\in \mathrm{CI}_n)\to F_c(q_{c,a})\).
If \(\hat q_n\) is not measurable, the same holds for
\(\Prout_n(a'\theta_0\in\mathrm{CI}_n)\) and for
\(1-\Prout_n(a'\theta_0\notin\mathrm{CI}_n)\): with
\(B_n:=\{|\hat q'_n-q_{c,a}|>\varepsilon\}\), whose outer probability tends to
zero, the event \(\{a'\theta_0\in\mathrm{CI}_n\}\) contains
\(\{|T_{n,a}|\le q_{c,a}-\varepsilon\}\setminus B_n\) and is contained in
\(\{|T_{n,a}|\le q_{c,a}+\varepsilon\}\cup B_n\), and \(F_c\) is continuous; Step~4
applies to both.
It remains to show (i) and (ii).
As \((\sigma_{n,a}/\sigma_n)^2=a'\Gamma_na+c_n^2\) and \(\liminf_n\sigma_{n,a}/\sigma_n>0\),
\(\lambda_a\in(0,1]\).
Then \(\sigma_n/\sigma_{n,a}=(a'\Gamma_na+c_n^2)^{-1/2}\to\lambda_a^{-1/2}\),
\(c=\lim c_n\sigma_n/\sigma_{n,a}=c_\infty\lambda_a^{-1/2}\) and \(a'\Gamma
a=\lambda_a(1-c^2)\).
By \Cref{thm--two-scale} and Step~3 of its proof,
\(\sigma_n^{-1}(\hat\theta_n-\theta_0)\leadsto\bm T=\Gamma^{1/2}\bm Z_0+c_\infty\calZ\).
Hence
\(T_{n,a}=(\sigma_n/\sigma_{n,a})\,a'\sigma_n^{-1}(\hat\theta_n-\theta_0)\leadsto
\lambda_a^{-1/2}\{(a'\Gamma a)^{1/2}Z_0+c_\infty a'\calZ\}\overset
d=T_{c,a}\).
Likewise, by \Cref{thm--two-scale-boot}, as in the last part of its proof,
\(\dBL(\mathcal L(\sigma_n^{-1}(\tilde\theta^*_n-\hat\theta_n)\mid\bm W_n),\mathcal L(\bm T^*))\pto0\),
with \(\bm T^*=\Gamma^{1/2}\bm Z_0+\kappa^{1/3}c_\infty\calZ\) as in that proof.
For large \(n\), \(x\mapsto(\sigma_n/\sigma_{n,a})a'x\) is Lipschitz with constant at
most \(2\lambda_a^{-1/2}\), which is at least \(1\), so
\(\dBL(\mathcal L(T^*_{n,a}\mid\bm W_n),\mathcal L((\sigma_n/\sigma_{n,a})a'\bm T^*))\le
2\lambda_a^{-1/2}\,\dBL(\mathcal L(\sigma_n^{-1}(\tilde\theta^*_n-\hat\theta_n)\mid\bm W_n),\mathcal
L(\bm T^*))\pto0\); moreover
\(\dBL(\mathcal L((\sigma_n/\sigma_{n,a})a'\bm T^*),\mathcal L(\lambda_a^{-1/2}a'\bm T^*))\le
\E\min\{2,|\sigma_n/\sigma_{n,a}-\lambda_a^{-1/2}|\,|a'\bm T^*|\}\to0\), and, as for
\(T_{n,a}\), \(\lambda_a^{-1/2}a'\bm T^*\overset d=T^*_{c,a}\).
Hence \(\dBL(\mathcal L(T^*_{n,a}\mid\bm W_n),\mathcal L(T^*_{c,a}))\pto0\), and (ii) follows
by the portmanteau theorem, applied along almost surely convergent
subsequences, as \(F^*_c\) is continuous.

\emph{Step 3: comparison of the limits.}
By Step~1, \(F^*_c(q_{c,a})=1-\alpha\).
\Cref{lem--anderson}, applied with \(G=(1-c^2)^{1/2}Z_0\), \(\calY=c\,a'\calZ\),
\(\gamma=\kappa^{1/3}>1\) and \(C=[-q_{c,a},q_{c,a}]\), gives
\(F_c(q_{c,a})\ge1-\alpha\).
If \(c=0\), then \(T_{0,a}=T^*_{0,a}=Z_0\), and equality holds.
If \(c=1\) (no Gaussian term),
\(F_1(q_{1,a})-(1-\alpha)=\Pr(\kappa^{-1/3}q_{1,a}<|a'\calZ|\le q_{1,a})>0\),
because \(q_{1,a}>0\) and the law of \(a'\calZ\) is mutually absolutely
continuous with Lebesgue measure (\Cref{lem--Z-law}).
If \(c\in(0,1)\), let \(\sigma:=(1-c^2)^{1/2}\), let \(\Phi\) be the standard
normal distribution function, and let
\(h(x):=\Pr(|\sigma Z_0+x|\le q_{c,a})=\Phi\{(q_{c,a}-x)/\sigma\}
-\Phi\{(-q_{c,a}-x)/\sigma\}\).
Then \(h\) is even, and \(h'(x)=\sigma^{-1}[\Phi'\{(q_{c,a}+x)/\sigma\}
-\Phi'\{(q_{c,a}-x)/\sigma\}]<0\) for \(x>0\), so \(h\) is strictly
decreasing on \([0,\infty)\).
Conditioning on \(\calZ\),
\(F_c(q_{c,a})-(1-\alpha)=\E\{h(c|a'\calZ|)-h(\kappa^{1/3}c|a'\calZ|)\}>0\),
because the integrand is positive on \(\{a'\calZ\ne0\}\), which has
probability one.

\emph{Step 4: conclusions.}
Part (c) follows from Steps~2 and~3, and part (b) from Steps~2 and~3 with
\(c=0\).
For (a): every subsequence has a further subsequence as in Step~2, along
which the coverage converges to \(F_c(q_{c,a})\ge1-\alpha\) by Steps~2
and~3.
For (d), the map \(c\mapsto F_c(q_{c,a})\) is continuous on \([0,1]\): if
\(c_m\to c\), then \(T_{c_m,a}\to T_{c,a}\) and \(T^*_{c_m,a}\to T^*_{c,a}\)
pointwise, so \(F^*_{c_m}\to F^*_c\) pointwise and, by Step~1,
\(q_{c_m,a}\to q_{c,a}\); hence \(|T_{c_m,a}|-q_{c_m,a}\to|T_{c,a}|-q_{c,a}\)
pointwise, and the distribution function of the limit is continuous at zero.
Every subsequence has a further subsequence as in Step~2, along which the
coverage tends to \(F_c(q_{c,a})\) and, by this continuity, so does
\(F_{c_{n,a}}(q_{c_{n,a},a})\); this gives the first claim of (d).
For the second, a subsequence along which the coverage tends to its limit
inferior has a further subsequence as in Step~2 with
\(c\ge\liminf_nc_{n,a}>0\), and by Step~3 its limit \(F_c(q_{c,a})\) exceeds
\(1-\alpha\).
For the uniformity statement, let \(\mathcal P\) be a class of distributions
with common \(\theta_0\), \(A_0\) and \(\Lambda\) such that every sequence
\(\{P_n\}\) in \(\mathcal P\) satisfies the conditions of the theorem.
Choose \(P_n\in\mathcal P\) with \(P_n(a'\theta_0\in \mathrm{CI}_n)\le\inf_{P\in\mathcal
P}P(a'\theta_0\in \mathrm{CI}_n)+n^{-1}\); (a) applied to \(\{P_n\}\) gives the claim.
\end{proof}

\begin{proof}[Proof of \Cref{cor--boot-ci}]
Let \(\sigma_g^2:=\nabla g(\theta_0)'V\nabla g(\theta_0)\).
Since \(\Sigma\succ0\) and \(A_0\succ0\), \(V\succ0\); as \(\nabla g(\theta_0)\ne0\),
\(\sigma_g>0\).
By \Cref{thm--m-est-non-smooth}, \(\sqrt n(\hat\theta_n-\theta_0)\leadsto
N(0,V)\).
By \Cref{thm--bootstrap}, \(\sqrt n(\tilde\theta^*_n-\hat\theta_n)\dtop
N(0,V)\) conditionally given \(\bm W_n\).
The proof of that theorem also gives
\(\sqrt n(\tilde\theta^*_n-\hat\theta_n)=A_0^{-1}\sqrt n\tilde\Delta^*_n+\op(1)
=\Op(1)\).
In particular, \(\hat\theta_n\) and \(\tilde\theta^*_n\) converge in
probability to \(\theta_0\).
Since \(\nabla g\) is continuous at \(\theta_0\), the mean value theorem gives,
on events with probability tending to one,
\(\sqrt n\{g(\hat\theta_n)-g(\theta_0)\}=\nabla g(\theta_0)'\sqrt n(\hat\theta_n-\theta_0)+\op(1)\)
and
\(\sqrt n\{g(\tilde\theta^*_n)-g(\hat\theta_n)\}=\nabla g(\theta_0)'\sqrt
n(\tilde\theta^*_n-\hat\theta_n)+R^*_n\), with \(R^*_n=\op(1)\).
Hence \(\sqrt n\{g(\hat\theta_n)-g(\theta_0)\}\leadsto N(0,\sigma_g^2)\).
For \(f\in\mathrm{BL}_1(\R)\) and any \(X^*_n\),
\(|\E[f(X^*_n+R^*_n)-f(X^*_n)\mid\bm W_n]|\le\E[\min\{2,|R^*_n|\}\mid\bm W_n]\),
and the
right-hand side has expectation \(\E\min\{2,|R^*_n|\}\to0\).
Thus \(\sqrt n\{g(\tilde\theta^*_n)-g(\hat\theta_n)\}\dtop N(0,\sigma_g^2)\)
conditionally given \(\bm W_n\).
Both convergences are preserved by the Lipschitz map \(|\cdot|\).
The limit law of \(|\sigma_gZ_0|\) has a continuous distribution function,
strictly increasing on \([0,\infty)\), because \(\sigma_g>0\).
The argument of \citet[Lemma~23.3]{vandervaart1998asymptotic}, applied to
\(|\cdot|\), gives \(\sqrt n\hat q_n\pto\sigma_gz_{1-\alpha/2}\), the
\((1-\alpha)\)-quantile of \(|\sigma_gZ_0|\), and hence
\(\Pr\{|g(\hat\theta_n)-g(\theta_0)|\le\hat q_n\}\to1-\alpha\).
If \(\hat q_n\) is not measurable, this applies to the smallest and the
largest \((1-\alpha)\)-quantiles, which are measurable and bracket
\(\hat q_n\); the corresponding events bracket that for \(\hat q_n\), so its
outer and inner probabilities also tend to \(1-\alpha\).
\end{proof}

\begin{proof}[Proof of \Cref{cor--functional}]
(a) Let
\(T^g_n:=\{g(\hat\theta_n)-g(\theta_0)\}/(\|\nabla g(\theta_0)\|\sigma_{n,a})\),
and let \(T^{g*}_n\) be the same with \(g(\tilde\theta^*_n)-g(\hat\theta_n)\)
in the numerator, so that \(g(\theta_0)\in\mathrm{CI}_n\) if and only if
\(|T^g_n|\) is at most the conditional \((1-\alpha)\)-quantile of
\(|T^{g*}_n|\).
By \Cref{lem--two-scale-rate}, \(\hat\theta_n-\theta_0\) and
\(\tilde\theta^*_n-\hat\theta_n\) are \(\Op(\sigma_n)\).
With \(T_{n,a}\) and \(T^*_{n,a}\) as in the proof of \Cref{thm--size}, the mean
value theorem shows that \(T^g_n-T_{n,a}\) and \(T^{g*}_n-T^*_{n,a}\) are
\(\op(\sigma_n/\sigma_{n,a})=\op(1)\), as \(\sigma_n/\sigma_{n,a}\) is bounded.
Such terms affect neither the limit of \(T_{n,a}\) in Step~2 of that proof nor,
as in the proof of \Cref{cor--boot-ci}, the conditional limit; so the proof
of \Cref{thm--size} applies.

(b) Along subsequences with \(\Gamma_n\to\Gamma\) and \(c_n\to c_\infty\),
\Cref{thm--two-scale,thm--two-scale-boot} give
\(\sigma_n^{-1}S(\hat\theta_n-\theta_0)\leadsto S\bm T\) and
\(\sigma_n^{-1}S(\tilde\theta^*_n-\hat\theta_n)\dtop S\bm T^*\) conditionally given
\(\bm W_n\), with \(\bm T\) and \(\bm T^*\) as in their proofs.
If \(c_\infty>0\), the laws of \(\bm T\) and \(\bm T^*\), convolutions with the
laws of \(c_\infty\calZ\) and \(\kappa^{1/3}c_\infty\calZ\), are mutually
absolutely continuous with Lebesgue measure by \Cref{lem--Z-law}, and hence so
are the laws of
\(\|S\bm T\|\) and \(\|S\bm T^*\|\) on \([0,\infty)\).
If \(c_\infty=0\), then \(\|\Gamma\|=1\), and \(\|S\Gamma^{1/2}\bm Z_0\|\) has a
continuous distribution function, strictly increasing on \([0,\infty)\).
So Steps~1, 2 and~4 of the proof of \Cref{thm--size} apply, with
\(\|\cdot\|\) in place of \(|\cdot|\), and in Step~3 \Cref{lem--anderson}
applies with \(G=\Gamma^{1/2}\bm Z_0\), \(\calY=c_\infty\calZ\),
\(\gamma=\kappa^{1/3}\) and the convex symmetric set \(C=\{x:\|Sx\|\le q\}\).
\end{proof}

\subsection{Measurability conventions and sufficient conditions}
\label{sec--measurability-conventions}

\paragraph{Measurability.}
Under \Cref{asm--m-est-general}\eqref{asm--m-est-meas}, by Tonelli's theorem, the centred conditional
expectations can be taken in the jointly measurable versions
\([P_{\e}m_\theta](u)=\int m_\theta(\tau(u,u'))\,\d u'-\E[m_\theta]\), so every
process below is jointly measurable in \(\theta\), the latent variables, the
bootstrap weights and the Rademacher variables of symmetrisation arguments.
Its supremum over a Borel set of \(\theta\) is then universally measurable in
these variables, by the measurable projection theorem
\citep[Corollary~5.3.5]{dudley1999uniform}: the classes involved are
image-admissible Suslin \citep[p.~138]{delapena1999decoupling}.
This is the measurability that symmetrisation needs.
A symmetrised supremum is integrated over the Rademacher variables and over
the data in turn, and Fubini's theorem, which does not hold in general for
outer expectations, holds for functions that are measurable for the
completion of the product measure
\citep[pp.~10 and~173]{2023vandervaartWeakConvergenceEmpirical}, as
universally measurable functions are.
We take all probability spaces to be complete, so outer expectations of these
suprema are expectations, and conditional bounds can be integrated by
Fubini's theorem.
The maximal inequality of \citet[Theorem~4]{chen2026cross} is stated for
pointwise measurable classes, but its proof uses this only to make such
suprema measurable, in the symmetrisation inequality of
\citet[Lemma~B.1]{chiang2023inference}, and to integrate a conditional entropy
bound, for which an outer expectation suffices; so it applies to the classes
used here.
For maximum score, \((\theta,y,x)\mapsto\ind\{y\cdot x'\beta(\theta)\ge0\}\) is the
indicator of a relatively closed subset of the open set
\(\Theta\times\R\times\R^d\), hence Borel, so \eqref{asm--m-est-meas} holds.

\paragraph{Sufficient conditions for
\Cref{asm--m-est-general}\eqref{asm--m-est-stoch-eq}.}
Condition \eqref{eqn--m-est-stoch-eq} follows from primitive conditions.
Write \(\|H_{n,\e}\|_{\calF}:=\sup_{f\in\calF}|H_{n,\e}[f]|\).
Under \Cref{asm--m-est-general}\eqref{asm--m-est-meas}, for \(\e\in\calE_k\) and a VC-type class
\(\calF\) with envelope \(F\), of the form \(\{m_\theta:\theta\in B\}\) or
\(\{m_\theta-m_{\theta_0}:\theta\in B\}\) with \(B\) Borel,
\citet[Theorem~4 and Lemma~5]{chen2026cross} give
\(\E\|H_{n,\e}\|_{\calF}\le C\|F\|_{P_0,2}|\calI_{n,\e}|^{-1/2}\le
C\|F\|_{P_0,2}\,n^{-k/2}\), with \(C\) depending only on \(K\) and the VC
characteristics.
Their theorem is stated for centred classes, which suffices: each
\(\pi_{\e}\) with \(\e\ne\bm0\) annihilates constants, so
\(H_{n,\e}[f]=H_{n,\e}[f-\E f]\), and \(\{f-\E f:f\in\calF\}\) is of VC-type
with envelope \(F+\E F\), whose \(\|\cdot\|_{P_0,2}\) norm is at most \(2\|F\|_{P_0,2}\),
and characteristics depending only on those of \(\calF\).
With \(\calF=\calM\), the VC-type condition in
\Cref{asm--m-est-general}\eqref{asm--m-est-stoch-eq} thus yields the first condition in
\eqref{eqn--m-est-stoch-eq}, even with \(\Theta\) in place of the shrinking
balls, and the second when \(k\ge3\).
For \(k=2\), the second condition holds, even for every \(\eta_n\to0\), if the
localised classes \(\{m_\theta-m_{\theta_0}:\|\theta-\theta_0\|\le\eta\}\) are
VC-type with characteristics free of \(\eta\), as in
\Cref{asm--two-scale}\eqref{asm--two-scale-modulus}, and have envelopes
whose \(\|\cdot\|_{P_0,2}\) norms tend to zero as \(\eta\to0\); for maximum score,
\eqref{eqn--beta-theta-1-Lipschitz} and \Cref{lem--max-score-local-envelope}
give the envelope \(\Psi_{\sqrt2\eta,\beta_0}\) of that lemma, with
\(\|\Psi_{\sqrt2\eta,\beta_0}\|_{P_0,2}=O(\eta^{1/2})\).
For the bounded weights of \Cref{asm--boot-wgt},
\Cref{lem--bootstrap-2nd-order-partial-control} carries these bounds over to
the weighted projections, so no multiplier analogue of
\citet[Lemmas~2.9.1 and~2.9.11]{2023vandervaartWeakConvergenceEmpirical} is
needed.

\section{\texorpdfstring{Proofs for \Cref{sec--max-score}}{Proofs for Section 3}}

\subsection{Proof of \Cref{lem--max-score-consistency}}
\label{sec--prf--lem--max-score-consistency}

Let \(M_{0} (\beta) = \E [m_{\beta} (W)]\).
Under \Cref{asm--exchangeable}, \Cref{lem--exchangeable-ulln-for-max-score}
shows that \(\sup_{\beta \in \mathcal{B}} |\hat{M}_{n} (\beta) - M_{0} (\beta)|
\pto 0\), as \(\mathcal{B} \subseteq \sphere^{d - 1}\) is Borel.
Under \Cref{asm--ms-consistency},
\(M_{0}\) is uniquely maximised at \(\beta_{0} \in \mathcal{B}\) and continuous
on the compact set \(\mathcal{B}\) --- as shown in
\Cref{lem--ms-b0-is-maximizer,lem--ms-continuity-M0}.
For an exact maximiser \(\hat\beta_n\), the conclusion follows from
Theorem~2.1 of \citet{newey1994large}.
In general, these properties give
\(M_0(\beta_0)>\sup_{\beta\in\mathcal{B}\setminus\mathcal{O}}M_0(\beta)\)
for every open \(\mathcal{O}\ni\beta_0\), and the conclusion follows from
\citet[Corollary~3.2.3(i), p.~395]{2023vandervaartWeakConvergenceEmpirical},
which allows the \(\op(1)\) tolerance in
\eqref{eqn--max-score-estimator-def}.
\qed

\subsubsection{Auxiliary results for the proof of
\Cref{lem--max-score-consistency}}

Denote
\begin{equation}
  q_{\beta} (x) = \ind \{x' \beta \geq 0\},
  \quad \text{so that} \quad m_{\beta} (y, x) = q_{\beta} (y \cdot x) = \ind
  \{y \cdot x' \beta \geq 0 \}.
  \label{eqn--max-score-objective-setup}
\end{equation}

\begin{lemma}
\label{lem--exchangeable-ulln-for-max-score}
Under \Cref{asm--exchangeable}, \(\sup_{\beta \in \mathcal{B}} |\hat{M}_{n}
(\beta) - M_{0} (\beta)| \pto 0\) for any Borel \(\mathcal{B}
\subseteq \sphere^{d - 1}\).
\end{lemma}

\begin{proof}[Proof of \Cref{lem--exchangeable-ulln-for-max-score}]
Let \(Z = Y \cdot X\), so that \(\hat{M}_{n} (\beta) = \hatE_{n} [q_{\beta}
(Z)]\)
and \(M_{0} (\beta) = \E [q_{\beta} (Z)]\) for \(q_{\beta} (z) = \ind \{z' \beta
\geq 0\}\) in \eqref{eqn--max-score-objective-setup}.
The half-spaces \(\{z\in\R^{d}:z'u\leq c\}\), \(u\in\R^{d}\), \(c\in\R\),
form a VC class of dimension \(d+1\), and the indicator functions of a VC
class of sets form a VC-subgraph class of the same index, hence of the same
dimension \citep[Problems~2.6.14 and~2.6.9,
pp.~219--220]{2023vandervaartWeakConvergenceEmpirical}.
As \(\{z:z'\beta\geq0\}=\{z:z'(-\beta)\leq0\}\) and a subclass of a VC class
has no larger dimension, \(\{q_{\beta}:\beta\in\mathcal{B}\}\) is a VC-subgraph
class of dimension at most \(d+1\), with envelope \(1\), for any
\(\mathcal{B}\subseteq\sphere^{d-1}\).
By \citet[Theorem~2.6.7, p.~206]{2023vandervaartWeakConvergenceEmpirical}, the class
\(\calH:=\{z\mapsto\ind\{z'u\le c\}:u\in\R^{d},\,c\in\R\}\), which contains
\(\{q_\beta:\beta\in\mathcal B\}\), is of VC-type with envelope \(1\).
It is also pointwise measurable: given \(u\) and \(c\), choose
\(u_m\in\mathbb{Q}^{d}\) with \(\|u_m-u\|\le m^{-2}\) and
\(c_m\in\mathbb{Q}\cap[c+1/m,c+2/m]\); as \(|z'u_m-z'u|\le\|z\|m^{-2}\),
\(\ind\{z'u_m\le c_m\}\to\ind\{z'u\le c\}\) for every \(z\), and, by dominated
convergence, the centred functions converge too.
Hence \citet[Theorem~4]{chen2026cross}, which is stated for centred classes,
applies as stated to \(\{h-\E h(Z):h\in\calH\}\), which is pointwise
measurable by the last sentence and of VC-type by the argument after
\Cref{thm--m-est-non-smooth}.
As \(H_{n,\e}[h]=H_{n,\e}[h-\E h(Z)]\), Step~2 of the proof of
\Cref{lem--boot-consistent}, with \(\calH\) in place of \(\calM\) and
\(m_*=1\), gives, under \Cref{asm--exchangeable},
\(\E\sup_{\beta\in\mathcal B}|\hat M_n(\beta)-M_0(\beta)|\le
\sum_{\e\in\calE}\E\|H_{n,\e}\|_{\calH}\le Cn^{-1/2}\); the supremum on the
left is measurable for the completed measure, by the measurable projection
theorem as in the paragraph after \Cref{asm--m-est-general}, because, for
Borel \(\mathcal B\), \((\beta,z)\mapsto\ind\{z'\beta\ge0\}\) is jointly Borel
measurable.
This proves the claim
\citep[cf.][Theorem~3.4(1), pp.~855--856, which also assumes that each
\(n/N_{n,k}\) converges]{2021daveziesEmpiricalProcessResults}.
\end{proof}

For \(y \in \{-1, 1\}\), \(m_{\beta}\) in \eqref{eqn--max-score-objective-setup}
satisfies the identity
\begin{equation*}
  \begin{split}
    m_{\beta} (y, x)  &= \ind \{y = 1\} q_{\beta} (x) + \ind \{y = -1\} (1 - q_{\beta} (x) +
    \ind \{x' \beta = 0\}) \\
     &= \ind \{y = -1\} + y q_{\beta} (x) + \ind \{y = -1,\ x' \beta = 0\}.
  \end{split}
\end{equation*}
By the law of iterated expectations and \(g_{0} (X) = \E [Y \mid X]\),
\begin{equation*}
  M_{0} (\beta) = \Pr \{Y = -1\} + \E [g_{0} (X) q_{\beta} (X)] + \Pr \{Y = -1,\ X' \beta
  = 0\}.
\end{equation*}
Under \Cref{asm--ms-consistency}\eqref{asm--continuity-under-constraint},
the last term is zero for each \(\beta \in \mathcal{B}\).
Since \(\beta_{0} \in \mathcal{B}\), it follows that for every \(\beta \in
\mathcal{B}\),
\begin{equation}
  M_{0} ( \beta_{0} ) - M_{0} (\beta) =
  \E [g_{0} (X) \ind \{X' \beta_{0} \geq 0 >
  X' \beta\} - g_{0} (X) \ind \{X' \beta \geq 0 > X'
  \beta_{0}\}].
  \label{eqn--Q0-b0-b-diff-1}
\end{equation}

\begin{lemma}
\label{lem--ms-b0-is-maximizer}
Under \Cref{asm--ms-consistency}\eqref{asm--latent-linear-median-regression}
and \eqref{asm--continuity-under-constraint}, \(g_{0}\) in \eqref{eqn--g0-def}
satisfies, almost surely,
\begin{equation}
  \begin{split}
    g_{0} (X) \ind \{X' \beta_{0} \geq 0\}  &= | g_{0} (X) | \ind \{X' \beta_{0} \geq 0\}, \\
    g_{0} (X) \ind \{X' \beta_{0} < 0\}  &= - | g_{0} (X) | \ind \{X' \beta_{0} < 0\}.
  \end{split}
  \label{eqn--psi-sign-xpb0}
\end{equation}
For every \(\beta \in \mathcal{B}\), the difference in
\eqref{eqn--Q0-b0-b-diff-1} can be written as
\begin{equation}
  \begin{split}
    M_{0} ( \beta_{0} ) - M_{0} (\beta)  &= \E [|g_{0} (X)| (\ind \{X' \beta_{0} \geq 0, X' \beta < 0\} +
    \ind \{ X' \beta_{0} < 0, X' \beta \geq 0\})]
    \\
     &= \E [|g_{0} (X)| \ind \{\sgn (X' \beta_{0}) \neq \sgn (X'
    \beta)\}],
  \end{split}
  \label{eqn--ms-M0-b0-b-diff}
\end{equation}
and hence \(M_{0} (\beta_{0}) \geq M_{0} (\beta)\) for every \(\beta \in
\mathcal{B}\).
If, in addition, \Cref{asm--ms-consistency}%
\eqref{asm--index-sgn-disagree-pos-pr} holds, then \(M_{0} (\beta_{0}) >
M_{0} (\beta)\) for every \(\beta \in \mathcal{B} \setminus \{\beta_{0}\}\).
\end{lemma}

\begin{proof}[Proof of \Cref{lem--ms-b0-is-maximizer}]
Under \Cref{asm--ms-consistency}\eqref{asm--latent-linear-median-regression},
we deduce \eqref{eqn--psi-sign-xpb0} from \(g_{0} (X) \geq 0\) when \(X'
\beta_{0} \geq 0\), and \(g_{0} (X) \leq 0\) when \(X' \beta_{0} < 0\).
To see this, the conditional median restriction implies \(\Pr \{\epsilon \leq 0
\mid X\} \geq \frac{1}{2}\) and \(\Pr \{\epsilon < 0 \mid X\} \leq \frac{1}{2}\)
almost surely.
Since \(g_{0} (X) = 2 \Pr \{\epsilon \leq X' \beta_{0} \mid X\} - 1\),
the claimed signs follow by monotonicity of conditional probabilities.
The representation \eqref{eqn--ms-M0-b0-b-diff} then follows by combining
\eqref{eqn--psi-sign-xpb0} with the expression in \eqref{eqn--Q0-b0-b-diff-1}.
Then \(M_{0} (\beta_{0}) - M_{0} (\beta) \geq 0\) for every \(\beta
\in \mathcal{B}\) follows since the integrand in
this representation is non-negative.
Under \Cref{asm--ms-consistency}\eqref{asm--index-sgn-disagree-pos-pr}, the
integrand is positive with positive probability.
\end{proof}

\begin{lemma}
\label{lem--ms-continuity-M0}
Under \Cref{asm--ms-consistency}\eqref{asm--continuity-under-constraint},
\(M_{0}\) is continuous everywhere on
\(\mathcal{B}\).
\end{lemma}

\begin{proof}[Proof of \Cref{lem--ms-continuity-M0}]
Let \(\beta\in\mathcal{B}\) and \(\beta_{n}\to\beta\) in \(\mathcal{B}\).
Since \(Y\in\{-1,1\}\), \(m_{\beta_{n}}(W)=\ind\{Y X'\beta_{n}\ge0\}\to
\ind\{Y X'\beta\ge0\}=m_{\beta}(W)\) on \(\{X'\beta\neq0\}\), which has
probability one by
\Cref{asm--ms-consistency}\eqref{asm--continuity-under-constraint}.
Hence \(M_{0}(\beta_{n})\to M_{0}(\beta)\) by bounded convergence.
\end{proof}

\subsection{\texorpdfstring{Proof of \Cref{prp--max-score-asymp-normal}(a)}{Proof of Proposition \ref{prp--max-score-asymp-normal}(a)}}
\label{sec--prf--thm-max-score-asymp-normal}

The proof proceeds by verifying the conditions of \Cref{thm--m-est-non-smooth}.
\Cref{asm--m-est-consistent-pop-quad-decay}\eqref{asm--m-est-consistent} holds
by \(\hat{\beta}_{n} \pto \beta_{0}\) as shown in
\Cref{lem--max-score-consistency}.
Since \(\hat\beta_n\) is measurable, \(E_n:=\{\hat\beta_n'\beta_0>0\}\) is an
event, and \(\Pr(E_n)\to1\) because \(\hat\beta_n\pto\beta_0\).
On \(E_n\), \(\hat\beta_n=\beta(\hat\theta_n)\) with \(\hat\theta_n\in\Theta\).
Redefining \(\hat\theta_n:=0\) on \(E_n^c\) changes neither its consistency nor
its limit law, and then \eqref{eqn--m-est-def} holds with
\(\hat\zeta_n:=\sup_\Theta\hat M_n(\beta(\theta))-\hat M_n(\beta(\hat\theta_n))\ge0\), which is measurable by the
paragraph after \Cref{asm--m-est-general} and equals \(\sup_\Theta\hat M_n(\beta(\theta))-\hat M_n(\hat\beta_n)\)
on \(E_n\); so \(\hat\zeta_n=\op(n^{-1})\) by \Cref{asm--X-ac-smooth-m-diff} and
\(\Pr(E_n)\to1\).

\begin{lemma}
\label{lem--max-score-primitive-density}
Suppose that \Cref{asm--X-ac-smooth-m-diff}\eqref{asm--X-ac-smooth-m-diff-primitive-density} holds.
Let \(Y = 2 \ind \{X' \beta_{0} - \epsilon \geq 0\} - 1\) and
\(Z = Y X\).
Then \(Z\) has the continuously differentiable density
\(p(z)=f_X(z)q(z)+f_X(-z)\{1-q(-z)\}\), where \(q(x):=F(x'\beta_0\mid x)\), and
\(|p(z)|+\|\nabla p(z)\|\le C_p(1+\|z\|)^{-d-2-\eta}\) with
\(C_p:=2C_X(1+C_F)\).
\end{lemma}

The proof is in \Cref{sec--prf--lem--max-score-primitive-density}.

Let \(p\) be the density of \(Z\) in \Cref{lem--max-score-primitive-density}.
Write \(z=\beta_0s+B_0t\), with
\(s=\beta_0'z\in\R\) and \(t=B_0'z\in\R^{d-1}\).
For \(\eta>0\) and a continuously differentiable density \(\varrho\) on
\(\R^d\) with
\(|\varrho(z)|+\|\nabla\varrho(z)\|\le C_\varrho(1+\|z\|)^{-d-2-\eta}\), let
\begin{equation*}
  G_\varrho(c):=\int_{\R^{d-1}}\int_{-t'c}^{\infty}\varrho(\beta_0s+B_0t)
  \,\d s\,\d t,\qquad c\in\R^{d-1},
\end{equation*}
the probability that \(Z'\beta_0+(B_0'Z)'c\ge0\) when \(Z\) has density
\(\varrho\).
Since \(\|B_0t-\beta_0t'c\|\ge\|t\|\), the functions
\(\|t\|^j\sup_{c}\{|\varrho|+\|\nabla\varrho\|\}(B_0t-\beta_0t'c)\),
\(j\le2\), are bounded by \(C_\varrho\|t\|^j(1+\|t\|)^{-d-2-\eta}\), which is
integrable on \(\R^{d-1}\).
By dominated convergence, \(G_\varrho\) is twice continuously differentiable,
with
\begin{equation}
  \nabla G_\varrho(c)=\int t\,\varrho(B_0t-\beta_0t'c)\,\d t,\qquad
  \nabla^2G_\varrho(c)=-\int tt'\,\beta_0'\nabla\varrho(B_0t-\beta_0t'c)\,\d t,
  \label{eqn--P4-G-derivs}
\end{equation}
and \(\|\nabla G_\varrho\|+\|\nabla^2G_\varrho\|\le C_\eta C_\varrho\), where \(C_\eta\)
depends only on \(d\) and \(\eta\).

In the chart \eqref{eqn--beta-local}, \(\beta_0'\beta(\theta)=
(1-\|\theta\|^2)^{1/2}>0\) and \(B_0'\beta(\theta)=\theta\), so
\begin{equation}
  z'\beta(\theta)=(1-\|\theta\|^2)^{1/2}\{s+t'\gamma(\theta)\},\qquad
  \gamma(\theta):=\theta/(1-\|\theta\|^2)^{1/2}.
  \label{eqn--P4-gamma}
\end{equation}
Hence \(M_0(\beta(\theta))=G_p(\gamma(\theta))\), where \(\gamma\) is smooth on
\(\Theta\) with \(\gamma(0)=0\) and \(D\gamma(0)=\I_{d-1}\).

\emph{Gradient.} By \eqref{eqn--P4-G-derivs},
\(\nabla G_p(0)=\int t\,p(B_0t)\,\d t=p_{Z'\beta_0}(0)\,
\E[B_0'Z\mid Z'\beta_0=0]\), with \(p_{Z'\beta_0}\) the density of
\(Z'\beta_0\).
Substituting \(t\mapsto-t\) in the second term of \(p\) gives
\(\nabla G_p(0)=\int t\,f_X(B_0t)\,g_0(B_0t)\,\d t\), with
\(g_0(x)=2q(x)-1\) as in \eqref{eqn--g0-def}.
Since \(\mathrm{Med}[\epsilon\mid X]=0\) and \(F(\cdot\mid x)\) is
continuous, \(F(0\mid x)=1/2\) for almost every \(x\) under the law of
\(X\), hence, by joint
continuity of \(F\), for every \(x\) in the support of \(X\).
Thus \(g_0(x)=2F(0\mid x)-1=0\) whenever \(x'\beta_0=0\) and \(x\) lies in
the support, whilst \(f_X=0\) outside it; so \(f_X(B_0t)g_0(B_0t)=0\) for all
\(t\), and \(\nabla G_p(0)=0\).

\emph{Hessian.} By the chain rule,
\(\nabla^2_\theta G_p(\gamma(\theta))|_{\theta=0}=D\gamma(0)'\nabla^2G_p(0)
D\gamma(0)+\sum_j\partial_jG_p(0)\nabla^2\gamma_j(0)=\nabla^2G_p(0)\): the
curvature term of the chart vanishes because the gradient does.
Therefore \(\theta\mapsto M_0(\beta(\theta))\) is twice continuously
differentiable near \(0\), with zero gradient and Hessian \(-A_0\), where
\begin{equation}
  \begin{gathered}
  A_0=-\nabla^2G_p(0)=\int tt'\,\beta_0'\nabla p(B_0t)\,\d t=\frac{\d\phi}{\d s}(0),
  \qquad\text{where}\\
  \phi(s):=p_{Z'\beta_0}(s)\,\E\bigl[(B_0'Z)(B_0'Z)'\bigm|Z'\beta_0=s\bigr]
  =\int tt'\,p(\beta_0s+B_0t)\,\d t.
  \end{gathered}
  \label{eqn--A0-ms}
\end{equation}
\emph{Sign.} The same substitution gives
\(A_0=\int tt'\{g_0(B_0t)\,\beta_0'\nabla f_X(B_0t)+f_X(B_0t)\,\beta_0'\nabla
g_0(B_0t)\}\,\d t\).
The first term vanishes for every \(t\): \(g_0(B_0t)=0\) if \(B_0t\) lies in
the support of \(X\), and otherwise \(B_0t\) lies in the open set on which
\(f_X\equiv0\), so that \(\nabla f_X(B_0t)=0\).
Hence \(A_0=\int tt'\,f_X(B_0t)\,\beta_0'\nabla g_0(B_0t)\,\d t\).
Almost surely \(g_0(X)\ge0\) if \(X'\beta_0\ge0\) and \(g_0(X)\le0\) if
\(X'\beta_0<0\) (\Cref{lem--ms-b0-is-maximizer}); as \(g_0\) is continuous,
these signs hold on the open set \(\{f_X>0\}\), so
\(f_X\,\beta_0'\nabla g_0\ge0\) on \(\{x:x'\beta_0=0\}\) and \(A_0\succeq0\).
Under \Cref{asm--X-ac-smooth-m-diff}\eqref{asm--X-ac-smooth-m-diff-m-increasing},
\(a'A_0a=\int(a't)^2f_X(B_0t)\beta_0'\nabla g_0(B_0t)\,\d t>0\) for
\(a\ne0\), so \(A_0\succ0\).
Since \(\beta_0'\nabla g_0=2f_{\epsilon|X}(0\mid\cdot)\) on
\(\{x:x'\beta_0=0,f_X(x)>0\}\) (see the discussion after
\Cref{asm--X-ac-smooth-m-diff}), and \(f_X(B_0t)=0\) at the other points of
the hyperplane, \(A_0\) also has the primitive form
\(A_0=2\int tt'\,f_{\epsilon|X}(0\mid B_0t)\,f_X(B_0t)\,\d t\).
Thus a second-order Taylor expansion verifies
\Cref{asm--m-est-consistent-pop-quad-decay}\eqref{asm--m-est-smooth-pop}.

Next, we verify condition \eqref{asm--m-est-stoch-eq} of
\Cref{asm--m-est-general}.
Since \(\|\beta(\theta)-\beta_0\|^2=2\{1-(1-\|\theta\|^2)^{1/2}\}\le
2\|\theta\|^2\) and \(\theta_0=0\),
\begin{equation}
  \|\beta (\theta) - \beta_{0}\| \leq \sqrt{2}\, \|\theta - \theta_{0}\|.
  \label{eqn--beta-theta-1-Lipschitz}
\end{equation}
Let \(\calF=\{\ind\{z'\beta\ge0\}:\|\beta\|=1\}\), \(z=yx\).
The proof of \Cref{lem--exchangeable-ulln-for-max-score} shows that
\(\calF\) is a VC-subgraph class of dimension at most \(d+1\).
The class
\(\calM(\delta)=\{m_\theta-m_{\theta_0}:\theta\in\Theta,\
\|\theta-\theta_0\|\le\delta\}\) of the proof of
\Cref{lem--bootstrap-2nd-order-control} is the translate by the fixed
function \(-\ind\{z'\beta_0\ge0\}\) of a subclass of \(\calF\); the bijection
\((z,t)\mapsto(z,t+\ind\{z'\beta_0\ge0\})\) maps its subgraphs onto those of
the subclass, so it is a VC-subgraph class of dimension at most \(d+1\); and, by
\eqref{eqn--beta-theta-1-Lipschitz}, the function \(\Psi_{\sqrt2\delta,\beta_0}\) of \Cref{lem--max-score-local-envelope}
is an envelope for it by that lemma.
By \citet[Theorem~2.6.7]{2023vandervaartWeakConvergenceEmpirical}, it is of
VC-type with envelope \(\Psi_{\sqrt2\delta,\beta_0}\) and characteristics depending only
on \(d\), not on \(\delta\).
As \(\calM(\delta)\) has the form considered after
\Cref{asm--m-est-general} and \Cref{thm--m-est-non-smooth}, where
measurability and centring are addressed,
\citet[Theorem~4 and Lemma~5]{chen2026cross} give, with constants
not depending on \(\delta\),
\begin{equation*}
  |\calI_{n,\e}|^{1/2} \Eout \left[ \sup_{\theta\in\Theta:\|\theta - \theta_{0}\| \leq \delta}
  |H_{n,\e} [m_{\theta} - m_{\theta_{0}}]| \right] \lesssim \sqrt{\E
  [\Psi_{\sqrt{2}\delta,\beta_0}]} \lesssim \sqrt{\delta}.
\end{equation*}
For \(\bm{\iota}'\e = k\geq2\), \(|\calI_{n,\e}|^{1/2} \geq n^{k/2}\geq n\),
and therefore, for every sequence \(\delta_{n} \to 0\),
\begin{equation*}
  \Eout \left[ n \sup_{\theta\in\Theta:\|\theta - \theta_{0}\| \leq \delta_{n}} |H_{n,\e}
  [m_{\theta} - m_{\theta_{0}}]| \right] \lesssim \sqrt{\delta_{n}} \to 0.
\end{equation*}
Moreover, \(\calM=\{m_\theta:\theta\in\Theta\}\) is a subclass of
\(\calF\), so, by \citet[Theorem~2.6.7]{2023vandervaartWeakConvergenceEmpirical}, it is of VC-type with
envelope \(1\), with
\(A\) enlarged if needed to meet the lower bound in condition
\eqref{asm--m-est-stoch-eq}.
This verifies condition \eqref{asm--m-est-stoch-eq} of
\Cref{asm--m-est-general}.

We now verify condition \eqref{asm--m-est-stoch-diff} of
\Cref{asm--m-est-general}.
Fix \(k\), and let \(\calN:=\{\theta:\|\theta\|\le1/2\}\subset\Theta\).
For almost every \(\bm u\), given \(\bm U_{\e_k}=\bm u\), \(Z\) has the density
\(\varrho_k(z\mid\bm u):=p_{\e_k}(z\mid\bm u)q_k(z,\bm u)+p_{\e_k}(-z\mid\bm u)
\{1-q_k(-z,\bm u)\}\), with \(q_k:=(1+g_{0,\e_k})/2\), as in the proof of \Cref{lem--max-score-primitive-density};
here \(0\le q_k\le1\) where \(p_{\e_k}(\cdot\mid\bm u)>0\), and \(p_{\e_k}=\nabla p_{\e_k}=0\) elsewhere, as \(p_{\e_k}\ge0\).
By \Cref{asm--X-ac-smooth-m-diff}%
\eqref{asm--X-ac-smooth-m-diff-p-smooth}--\eqref{asm--X-ac-smooth-m-diff-m-dominance},
\(\varrho_k(\cdot\mid\bm u)\) is continuously differentiable and
\(|\varrho_k(z\mid\bm u)|+\|\nabla_z\varrho_k(z\mid\bm u)\|\le2D_{\e_k}(\bm u)
(1+\|z\|)^{-d-2-\bar\eta}\), with \(\bar\eta\) the smaller of the exponents \(\eta\) in
\Cref{asm--X-ac-smooth-m-diff}\eqref{asm--X-ac-smooth-m-diff-m-dominance}
and in \Cref{asm--X-ac-smooth-m-diff}\eqref{asm--X-ac-smooth-m-diff-primitive-density}.
By \eqref{eqn--P4-gamma}, \(\E[m_\theta(W)\mid\bm U_{\e_k}=\bm u]=
G_{\varrho_k(\cdot\mid\bm u)}(\gamma(\theta))\), so
\begin{equation*}
  [\pi_{\e_k}m_\theta](\bm u)=G_{\varrho_k(\cdot\mid\bm u)}(\gamma(\theta))
  -G_p(\gamma(\theta)).
\end{equation*}
By \eqref{eqn--P4-G-derivs} and smoothness of \(\gamma\) on \(\Theta\), for
almost every \(\bm u\) the map \(\theta\mapsto[\pi_{\e_k}m_\theta](\bm u)\) is twice
continuously differentiable on \(\Theta\).
Its gradient at \(\theta_0=0\) is, using \(\nabla G_p(0)=0\),
\(D\gamma(0)=\I_{d-1}\) and the substitution \(t\mapsto-t\) as above,
\begin{equation}
  \Delta_k(\bm u):=\nabla[\pi_{\e_k}m_{\theta_0}](\bm u)
  =\int_{\R^{d-1}}t\,g_{0,\e_k}(B_0t,\bm u)\,p_{\e_k}(B_0t\mid\bm u)\,\d t.
  \label{eqn--P4-Delta-ms}
\end{equation}
Its Hessian \(\tilde H_{k,\theta}(\bm u):=\nabla^2_\theta[\pi_{\e_k}m_\theta]
(\bm u)\) is continuous in \(\theta\), and
\(\sup_{\|\theta\|\le3/4}\|\tilde H_{k,\theta}(\bm u)\|\le\bar D_k(\bm u):=
C\{D_{\e_k}(\bm u)+C_p\}\), with \(C\) depending only on \(d\), \(\bar\eta\) and
the derivatives of \(\gamma\) on \(\{\theta:\|\theta\|\le3/4\}\), a neighbourhood of \(\calN\).
Hence \(\E\|\Delta_k\|^2\le C\,\E[D_{\e_k}(\bm U_{\e_k})^2]<\infty\) and
\(\E\bar D_k<\infty\).
Since \(\E[\pi_{\e_k}m_\theta]=0\) for every \(\theta\), differentiating under
the expectation (dominated convergence) gives \(\E\Delta_k=0\).

For \(\theta\in\calN\), Taylor's theorem with integral remainder gives
\begin{equation*}
  r_{k,\theta}=(\theta-\theta_0)'\Bigl\{\int_0^1(1-\tau)\,
  \tilde H_{k,\theta_0+\tau(\theta-\theta_0)}\,\d\tau\Bigr\}(\theta-\theta_0).
\end{equation*}
Averaging and Fubini's theorem give
\begin{equation*}
  |(\hatE_{n,k}-\E)[r_{k,\theta}]|\le\tfrac12\|\theta-\theta_0\|^2R_{n,k},
  \qquad R_{n,k}:=\sup_{\theta\in\calN}\|(\hatE_{n,k}-\E)[\tilde H_{k,\theta}]\|.
\end{equation*}
The variables \(U_{i\cdot\e_k}\), \(i=1,\dots,N_{n,k}\), are i.i.d.;
\(\theta\mapsto\tilde H_{k,\theta}\) is continuous on the compact set
\(\calN\) and dominated by the integrable \(\bar D_k\).
So \(R_{n,k}\pto0\) by \citet[Lemma~2.4]{newey1994large} (\(R_{n,k}\) is
measurable, as the supremum may be taken over a countable dense subset).
Moreover \(R_{n,k}\le\hatE_{n,k}[\bar D_k]+\E\bar D_k\), and averages of
identically distributed integrable variables are uniformly integrable; hence
\(\E R_{n,k}\to0\).
Consequently, for all positive sequences \(\eta_n\to0\) and
\(\varpi_n\), and all \(n\) with \(\eta_n\le1/2\),
\begin{equation}
  \Eout\Biggl[\sup_{\theta\in\Theta:0<\|\theta-\theta_0\|\le\eta_n}
  \frac{|(\hatE_{n,k}-\E)[r_{k,\theta}]|}{\|\theta-\theta_0\|
  (\|\theta-\theta_0\|+\varpi_n)}\Biggr]\le\tfrac12\E R_{n,k}\to0.
  \label{eqn--P4-stoch-diff}
\end{equation}
With \(\varpi_n=n^{-1/2}\) this is \eqref{eqn--m-est-stoch-diff}; with
\(\varpi_n=n^{-2/3}\) it gives, as \(|m_\theta|\le1\),
\Cref{asm--two-scale}\eqref{asm--two-scale-first} for a fixed law when
\(K=2\).
\qed

\begin{lemma}
\label{lem--max-score-local-envelope}
Given any \(\beta,b\in\sphere^{d-1}\) with \(\|\beta-b\|\le\delta\),
\(|m_\beta-m_b|\le \Psi_{\delta,b}\), where
\(\Psi_{\delta,b}(y,x)=\ind\{|x'b|\le\delta\|x\|\}\).
If \(X\) has a density \(f_X(x)\le C_X(1+\|x\|)^{-d-\eta}\) for some
\(\eta>0\), as under \Cref{asm--X-ac-smooth-m-diff}\eqref{asm--X-ac-smooth-m-diff-primitive-density}, then
\(\E[\Psi_{\delta,b}]\le C\delta\), and hence
\(\sup_{\|\beta-b\|\le\delta}\E[(m_\beta-m_b)^2]\le C\delta\), with
\(C\) depending only on \(d\), \(\eta\) and \(C_X\).
\end{lemma}

\begin{proof}[Proof of \Cref{lem--max-score-local-envelope}]
Let \(z=yx\) with \(y\in\{-1,1\}\) (if \(y=0\), both indicators equal one;
for other \(y\ne0\), divide the inequalities below by \(|y|\)).
If \(m_\beta(y,x)\ne m_b(y,x)\), one of \(z'\beta\) and \(z'b\) is
non-negative and the other negative, so \(z'b\) and
\(z'(b-\beta)=z'b-z'\beta\) have the same sign and
\(|x'b|=|z'b|\le|z'(b-\beta)|\le\delta\|x\|\).
This proves the first claim, and \((m_\beta-m_b)^2=|m_\beta-m_b|\le
\Psi_{\delta,b}\).
It remains to bound \(\E[\Psi_{\delta,b}]\).

Let \(B_b\in\R^{d\times(d-1)}\) have orthonormal columns with \(B_b'b=0\),
and use the orthonormal coordinates \(x=tb+B_bs\), where \(t\in\R\) and
\(s\in\R^{d-1}\).
This change of variables has absolute Jacobian determinant one, with
\(x'b=t\) and \(\|x\|^2=t^2+\|s\|^2\).
For \(0<\delta\le\frac12\), the event defining \(\Psi_{\delta,b}\) satisfies
\begin{equation*}
  |t|\le\delta\,(t^2+\|s\|^2)^{1/2}
  \quad\implies\quad
  |t|\le\frac{\delta}{(1-\delta^2)^{1/2}}\|s\|\le2\delta\|s\|.
\end{equation*}
Thus, for each fixed \(s\), the possible values of \(t\) lie in an interval of
length at most \(4\delta\|s\|\).
Since \(\|tb+B_bs\|\ge\|s\|\), the density bound gives
\begin{equation*}
  \E[\Psi_{\delta,b}]\le\int_{\R^{d-1}}\int_{-2\delta\|s\|}^{2\delta\|s\|}
  f_X(tb+B_bs)\,\d t\,\d s\le4C_X\delta\,I_{d,\eta},\qquad
  I_{d,\eta}:=\int_{\R^{d-1}}\|s\|(1+\|s\|)^{-d-\eta}\,\d s.
\end{equation*}
The integral \(I_{d,\eta}\) is finite, as
\(\int_0^\infty r^{d-1}(1+r)^{-d-\eta}\,\d r<\infty\), and depends only on
\(d\) and \(\eta\).
For \(\delta>\frac12\), simply use \(\E[\Psi_{\delta,b}]\le1\le2\delta\).
Combining the two cases gives \(\E[\Psi_{\delta,b}]\le C\delta\) with
\(C:=\max\{4C_XI_{d,\eta},2\}\), which depends only on \(d\), \(\eta\) and
\(C_X\).
\end{proof}

\subsection{\texorpdfstring{Proof of \Cref{lem--max-score-primitive-density}}{Proof of Lemma \ref{lem--max-score-primitive-density}}}
\label{sec--prf--lem--max-score-primitive-density}

For every Borel set \(A\subseteq\R^d\), conditioning on \(X\) gives
\begin{equation*}
  \begin{split}
    \Pr(Z\in A)
    &=\int_{\R^d}\ind\{x\in A\}\,q(x)f_X(x)\,\d x
    +\int_{\R^d}\ind\{-x\in A\}\{1-q(x)\}f_X(x)\,\d x\\
    &=\int_A\bigl[f_X(z)q(z)+f_X(-z)\{1-q(-z)\}\bigr]\,\d z,
  \end{split}
\end{equation*}
so \(p\) is a density of \(Z\).

By the chain rule, \(q\) is continuously differentiable, with
\begin{equation*}
  \nabla q(x)=f_{\epsilon|X}(x'\beta_0\mid x)\,\beta_0
  +\nabla_xF(t\mid x)\big|_{t=x'\beta_0}.
\end{equation*}
On the support of \(X\), \(F(\cdot\mid x)\) is non-decreasing with values in
\([0,1]\), by continuity, as it is a conditional distribution function for
almost every \(x\); so there \(f_{\epsilon|X}\ge0\) and \(0\le q\le1\), and,
since \(\|\beta_0\|=1\), the assumed bound gives, for \(x\) in the support of
\(X\), \(\|\nabla q(x)\|\le C_F(1+\|x\|)^{a_F}\).
Consequently, \(p\) is continuously differentiable, with
\begin{equation*}
  \nabla p(z)=q(z)\nabla f_X(z)+f_X(z)\nabla q(z)
  -\{1-q(-z)\}\nabla f_X(-z)+f_X(-z)\nabla q(-z).
\end{equation*}
Off the support of \(X\), \(f_X=\nabla f_X=0\), so every term of \(p\) and
\(\nabla p\) that contains \(q\) or \(\nabla q\) at such a point vanishes.
Using this, \(0\le q\le1\) on the support of \(X\) and the bounds in
\Cref{asm--X-ac-smooth-m-diff}\eqref{asm--X-ac-smooth-m-diff-primitive-density},
we obtain
\begin{equation*}
  \begin{split}
    |p(z)|+\|\nabla p(z)\|
    &\le2C_X(1+\|z\|)^{-d-2-\eta-a_F}+2C_XC_F(1+\|z\|)^{-d-2-\eta}\\
    &\le2C_X(1+C_F)(1+\|z\|)^{-d-2-\eta}.
  \end{split}
\end{equation*}
\qed

\subsection{\texorpdfstring{Proof of \Cref{prp--max-score-asymp-normal}(b)}{Proof of Proposition \ref{prp--max-score-asymp-normal}(b)}}
\label{sec--prf--prp--ms-two-scale}

Here \(K=2\) and the DGP is fixed, so that the uniformity requirements of
\Cref{sec--prf--two-scale} hold automatically (\Cref{rem--fixed-law}).
\Cref{sec--two-scale} is applied with \(d-1\) in place of \(d\), with
\(\theta_0=0\), \(\eta_0=1/2\) and
\(\calN=\{\theta:\|\theta\|\le1/2\}\subset\Theta\).
For \(\theta\in\calN\) and \(a\in\R^{d-1}\),
\(\|D\beta(\theta)a\|^2=\|a\|^2+(\theta'a)^2/(1-\|\theta\|^2)\le\frac43\|a\|^2\),
so \(\beta\) is \(2\)-Lipschitz on the convex set \(\calN\).

\emph{Condition (i) of \Cref{asm--two-scale}.}
This is shown at the end of the proof of part (a).

\emph{Condition (ii).}
Let \(\theta\in\calN\) and \(\delta\in(0,1/2]\).
For \(\theta'\in\calN\) with \(\|\theta'-\theta\|\le\delta\),
\(\|\beta(\theta')-\beta(\theta)\|\le2\delta\), so
\Cref{lem--max-score-local-envelope} gives
\(|m_{\theta'}-m_\theta|\le\Psi_{2\delta,\beta(\theta)}\), a bound that does
not depend on \(\theta'\); the supremum of the left-hand side over these
\(\theta'\) is universally measurable, by the paragraph after
\Cref{asm--m-est-general}.
By \Cref{asm--ms-two-scale}\eqref{asm--ms-two-scale-density}, for almost every
value \(\bm u^\circ\) of \(\bm U^\circ\), the conditional law of \(X_{(1,1)}\) given
\(\bm U^\circ=\bm u^\circ\) has a density bounded by
\(\bar p(\bm u^\circ)(1+\|x\|)^{-d-\eta^\circ}\).
The proof of \Cref{lem--max-score-local-envelope} uses the density of \(X\)
only through such a bound, with \(C_X\) in place of \(\bar p(\bm u^\circ)\).
Applied to this conditional density, it gives
\(\E[\Psi_{\delta',b}(W_{(1,1)})\mid\bm U^\circ]\le\max\{4\bar p(\bm U^\circ)
I_{d,\eta^\circ},2\}\,\delta'\) for every \(\delta'>0\) and
\(b\in\sphere^{d-1}\), with \(I_{d,\eta^\circ}\) as there.
With \(\delta'=2\delta\),
\[
  \E\Bigl[\sup_{\theta'\in\calN:\|\theta'-\theta\|\le\delta}
  |m_{\theta'}(W_{(1,1)})-m_\theta(W_{(1,1)})|\Bigm|\bm U^\circ\Bigr]
  \le\dot m(\bm U^\circ)\,\delta,\qquad
  \dot m:=8I_{d,\eta^\circ}\bar p+4,
\]
and \(\E[\dot m(\bm U^\circ)^2]<\infty\).
For \(0<\eta\le1/2\), the class \(\{m_\theta-m_{\theta_0}:\|\theta\|\le\eta\}\) is
the class \(\calM(\eta)\) of the proof of part (a), a VC-subgraph class of
dimension at most \(d+1\).
By \citet[Theorem~2.6.7]{2023vandervaartWeakConvergenceEmpirical}, it is
therefore VC-type with the envelope \(F_\eta\) and characteristics that depend
only on \(d\).
Here \(F_\eta\) is Borel measurable.
It is the indicator of the set of \((y,x)\) at which
\(m_\theta(y,x)\ne m_{\theta_0}(y,x)\) for some \(\theta\) in the compact ball
\(\bar B_\eta:=\{\theta:\|\theta\|\le\eta\}\).
That set is the union of two sets.
The first is \(\{m_{\theta_0}=1\}\) intersected with the projection onto
\((y,x)\) of the relatively open set
\(\{(\theta,y,x)\in\bar B_\eta\times\R\times\R^d:y\,x'\beta(\theta)<0\}\), and
this projection is open.
The second is \(\{m_{\theta_0}=0\}\) intersected with the projection of the
closed set \(\{(\theta,y,x)\in\bar B_\eta\times\R\times\R^d:y\,x'\beta(\theta)
\ge0\}\), and this projection is closed, because \(\bar B_\eta\) is compact.

\emph{Condition (iii).}
As \(\beta_0\beta_0'+B_0B_0'=\I_d\), the map \((s,t)\mapsto\beta_0s+B_0t\)
from \(\R\times\R^{d-1}\) onto \(\R^d\) is orthogonal, with inverse
\(x\mapsto(\beta_0'x,B_0'x)\).
By \eqref{eqn--P4-gamma}, for \(\theta\in\Theta\), \(x'\beta(\theta)\) has
the sign of \(s+t'\gamma(\theta)\) when \(x=\beta_0s+B_0t\).
Let \(v\in\R^{d-1}\), \(\delta>0\) and \(\theta,\theta+\delta v\in\calN\).
As \(X\) has a density, almost surely neither \(X'\beta(\theta)\) nor
\(X'\beta(\theta+\delta v)\) vanishes.
Then, since \(Y\in\{-1,1\}\), \(m_{\theta+\delta v}(W)\ne m_\theta(W)\) if and
only if \(X'\beta(\theta)\) and \(X'\beta(\theta+\delta v)\) have opposite
signs, that is, if and only if \(\beta_0'X\) lies
strictly between \(-(B_0'X)'\gamma(\theta)\) and
\(-(B_0'X)'\gamma(\theta+\delta v)\).
Let \(w:=\{\gamma(\theta+\delta v)-\gamma(\theta)\}/\delta\), and let \(J(t)\)
be the closed interval with endpoints \(-t'\gamma(\theta)\) and
\(-t'\gamma(\theta+\delta v)\), of length \(\delta|t'w|\).
As \((m_{\theta+\delta v}-m_\theta)^2\) is the indicator that
\(m_{\theta+\delta v}\ne m_\theta\),
\[
  \delta^{-1}\E[(m_{\theta+\delta v}-m_\theta)^2]
  =\delta^{-1}\int_{\R^{d-1}}\int_{J(t)}f_X(\beta_0s+B_0t)\,\d s\,\d t
  =\int_{\R^{d-1}}|t'w|\,\bar f(t)\,\d t,
\]
where \(\bar f(t)\) is the average of \(s\mapsto f_X(\beta_0s+B_0t)\) over
\(J(t)\) if \(t'w\ne0\), and \(\bar f(t):=f_X(B_0t)\) otherwise.
Now let \(\delta_j>0\), \(\theta_j\in\R^{d-1}\) and \(v_j\in\R^{d-1}\) satisfy
\(\delta_j\to0\), \(\theta_j\to0\) and \(v_j\to v\), not necessarily
monotonically, and let \(w_j\) and \(\bar f_j\) be \(w\) and \(\bar f\) at
\((\delta,\theta,v)=(\delta_j,\theta_j,v_j)\), defined for large \(j\).
As \(\gamma\) is continuously differentiable with \(\gamma(0)=0\) and
\(D\gamma(0)=\I_{d-1}\),
\(w_j=\int_0^1D\gamma(\theta_j+\tau\delta_jv_j)v_j\,\d\tau\to v\).
Also, \(J(t)\) shrinks to \(\{0\}\) for each \(t\), so
\(\bar f_j(t)\to f_X(B_0t)\) by continuity of \(f_X\).
Since \(\|\beta_0s+B_0t\|\ge\|t\|\),
\Cref{asm--X-ac-smooth-m-diff}\eqref{asm--X-ac-smooth-m-diff-primitive-density}
bounds the integrands by \(\sup_l\|w_l\|\,C_X\|t\|(1+\|t\|)^{-d-2-\eta}\), which
is integrable on \(\R^{d-1}\).
By dominated convergence,
\[
  \delta_j^{-1}\E[(m_{\theta_j+\delta_jv_j}-m_{\theta_j})^2]\to
  \int_{\R^{d-1}}|t'v|\,f_X(B_0t)\,\d t=\Lambda(v).
\]
Dominated convergence also shows that \(\Lambda\) is continuous.
Hence the convergence is uniform in \(v\) on compact sets: otherwise, for
some \(\varepsilon>0\), sequences as above with \(v_j\) in a compact set would
satisfy
\(|\delta_j^{-1}\E[(m_{\theta_j+\delta_jv_j}-m_{\theta_j})^2]-\Lambda(v_j)|
\ge\varepsilon\), whereas, along a subsequence with \(v_j\to v\), both terms
would tend to \(\Lambda(v)\).
Finally, let \(v\ne0\).
The set \(\{t:f_X(B_0t)>0\}\) is open, because \(f_X\) is continuous.
It is non-empty, because
\(\{x:x'\beta_0=0\}=\{B_0t:t\in\R^{d-1}\}\) and, by
\Cref{asm--X-ac-smooth-m-diff}\eqref{asm--X-ac-smooth-m-diff-m-increasing},
\(f_X>0\) somewhere on this hyperplane.
As \(t'v\ne0\) off a hyperplane of \(\R^{d-1}\), \(\Lambda(v)>0\).

\emph{Condition (iv).}
By \Cref{asm--ms-two-scale}\eqref{asm--ms-two-scale-rates},
\(N_{n,1}\sim N_{n,2}\sim n\).
With \(\hat\theta_n\) redefined on \(E_n^c\) and \(\hat\zeta_n\) as in the
proof of part (a), \(\hat\zeta_n=\sup_\Theta\hat M_n(\beta(\theta))
-\hat M_n(\hat\beta_n)\) on \(E_n\), so \(\hat\zeta_n=\op(n^{-4/3})\) by
\Cref{asm--ms-two-scale}\eqref{asm--ms-two-scale-rates} and
\(\Pr(E_n)\to1\).

\emph{The limit.}
The proof of part (a) verifies
\Cref{asm--m-est-consistent-pop-quad-decay,asm--m-est-general} without using
the convergence of \(n/N_{n,k}\), and \Cref{asm--exchangeable} is
assumed.
So the conditions of \Cref{thm--two-scale} hold for the constant sequence
\(\{P_n\}\), with \(d-1\) in place of \(d\).
If \(\Delta_1=\Delta_2=0\) almost surely, then
\(\Sigma_{n,1}=\Sigma_{n,2}=0\), so \(\Omega_n=0\).
\Cref{thm--two-scale}(b) with \(\Omega=0\), for which \(G=0\), then gives
\(n^{2/3}\hat\theta_n=n^{2/3}(\hat\theta_n-\theta_0)\leadsto\calZ\).
For \(d=2\), the parameter is scalar, and the representation of \(\calZ\) for
a scalar parameter in \Cref{sec--two-scale} gives the last claim.

\emph{Vanishing scores.}
Suppose that \(\E[Y\mid X,\bm U_{\e_k}]=\E[Y\mid X]\) almost surely.
Then, for almost every value \(\bm u\) of \(\bm U_{\e_k}\),
\(g_{0,\e_k}(x,\bm u)=g_0(x)\) for almost every \(x\) with
\(p_{\e_k}(x\mid\bm u)>0\), with the version \(g_0(x)=2F(x'\beta_0\mid x)-1\)
of \Cref{asm--X-ac-smooth-m-diff}\eqref{asm--X-ac-smooth-m-diff-m-increasing}.
Both functions are continuous in \(x\), by
\Cref{asm--X-ac-smooth-m-diff}\eqref{asm--X-ac-smooth-m-diff-m-smooth}
and~\eqref{asm--X-ac-smooth-m-diff-primitive-density}, and
\(\{x:p_{\e_k}(x\mid\bm u)>0\}\) is open, so equality holds at every such
\(x\).
For almost every \(\bm u\), moreover, \(p_{\e_k}(\cdot\mid\bm u)\) vanishes
off the support of \(X\): the conditional law of \(X\) given
\(\bm U_{\e_k}=\bm u\) is concentrated on this closed set, and
\(p_{\e_k}(\cdot\mid\bm u)\) is continuous.
As shown in the proof of part (a), \(g_0=0\) on the part of
\(\{x:x'\beta_0=0\}\) that lies in the support of \(X\).
So the integrand in \eqref{eqn--P4-Delta-ms} vanishes for almost every
\(\bm u\), and \(\Delta_k=0\) almost surely.

\emph{Other forms of \(\Lambda\).}
By the orthogonal change of variables above,
\(f_{X'\beta_0}(s):=\int_{\R^{d-1}}f_X(\beta_0s+B_0t)\,\d t\) is a density of
\(X'\beta_0\), continuous by dominated convergence, and
\(f_{X'\beta_0}(0)>0\) by the last step for condition (iii).
The ratio \(f_X(B_0t)/f_{X'\beta_0}(0)\) is the continuous conditional
density of \(B_0'X\) given \(X'\beta_0=0\), which gives the expression for
\(\Lambda\) after the proposition.
If \(X\sim N(0,\sigma^2\I_d)\), then \(f_X(B_0t)=(2\pi\sigma^2)^{-1/2}
\varphi(t)\), with \(\varphi\) the density of \(T\sim N(0,\sigma^2\I_{d-1})\).
So \(\Lambda(v)=(2\pi\sigma^2)^{-1/2}\E|v'T|\), and
\(\E|v'T|=\sigma\|v\|(2/\pi)^{1/2}\) gives \(\Lambda(v)=\|v\|/\pi\).
\qed

\clearpage
\section{The plain bootstrap}
\label{sec--plain-bootstrap}

This appendix is about the plain bootstrap.
It defines the estimator, states its conditional limit, explains why we do not
recommend it when \(K=2\), and reports a three-way design in which the reason
no longer bites, so that the plain bootstrap becomes usable; the same design
serves as a check of
\Cref{thm--m-est-non-smooth,thm--bootstrap,cor--boot-ci} at \(K=3\).

The plain bootstrap estimator \(\hat\theta^*_n\) satisfies
\begin{equation}
  \begin{gathered}
    \hat{\theta}_{n}^{*} \in \Theta \text{ is measurable} \quad \text{and} \quad
    \hat{M}_{n}^{*} (\hat{\theta}_{n}^{*}) \geq \sup_{\theta \in \Theta}
    \hat{M}_{n}^{*} (\theta) - \hat{\zeta}_{n}^{*}, \\
    \text{where} \quad \hat{\zeta}_{n}^{*} \text{ measurable,} \quad \text{and}
    \quad 0 \leq \hat{\zeta}_{n}^{*} = \op (1/n).
  \end{gathered}
  \label{eqn--m-est-boot-def-basic}
\end{equation}

\begin{remark}[Plain bootstrap]
\label{rem--plain-bootstrap}
Suppose the conditions of \Cref{thm--bootstrap} hold, under either weighting
scheme in \Cref{asm--boot-wgt}, and that
\(\sup_{\|\theta-\theta_0\|\ge\varepsilon}M_0(\theta)<M_0(\theta_0)\) for every
\(\varepsilon>0\).
Then \(\hat\theta^*_n\pto\theta_0\) (\Cref{lem--boot-consistent}) and
\(n^{1/2}(\hat\theta^*_n-\hat\theta_n)\dtop N(0,V)\) conditionally given
\(\bm W_n\) (\Cref{sec--prf--prp--bootstrap-prev}).
We nevertheless do not recommend it.
When \(K=2\) and \(m_\theta\) is of cube-root type, as in maximum score, the
top-level (cell-level) component \(H_{n,(1,1)}[m_\theta]\) of
\eqref{eqn--gen-hoeffding-decomp} gives \(\hat M_n\) a sharp peak at
\(\hat\theta_n\), at a distance of order \(|\calI_n|^{-1/3}\) from the
first-order point \(\theta_0+A_0^{-1}\bar\psi_n\), where its first-order part
peaks; since \(\hat\theta^*_n\) maximises \(\hat M_n\) plus the reweighted
noise \(\hat M^*_n-\hat M_n\), that noise pulls \(\hat\theta^*_n\) towards the
first-order point and so shortens the bootstrap tail on the far side of
\(\hat\theta_n\), the side of the cube-root term.
For scalar \(\theta\), basic intervals
\([\hat\theta_n-q^*_{1-\alpha/2},\hat\theta_n-q^*_{\alpha/2}]\), with
\(q^*_p\) the conditional \(p\)-quantile of
\(\hat\theta^*_n-\hat\theta_n\), need that tail (percentile intervals use the
other one) when there is no Gaussian term or whenever the Gaussian and
cube-root terms share a sign.
In the notation of \Cref{sec--two-scale}, with \(d=1\),
\(n^{4/3}\Omega_n\to\Omega\) and \(G\sim N(0,\Omega)\) independent of
\(\mathbb W\), we expect the conditional law of
\(n^{2/3}(\hat\theta^*_n-\hat\theta_n)\) given \(\bm W_n\) to approach the
conditional law given \((\mathbb W,G)\) of
\(\argmax_v\{\mathbb W(v)+\kappa^{1/2}\mathbb W^*(v)+v(G+G^*)-\frac12A_0v^2\}
-\argmax_v\{\mathbb W(v)+vG-\frac12A_0v^2\}\), where \((\mathbb W^*,G^*)\) is
an independent copy of \((\mathbb W,G)\) and \(\kappa\), the limit of
\(\var(\xi_{n,\i})\), is \(2\) for multiplier and \(3\) for pigeonhole weights.
A simulation of this limit, not reported here, suggests coverage of
about 78\% (multiplier) and 85\% (pigeonhole) for nominal 95\% basic
intervals when \(\Omega=0\), returning to 95\% only as the weight \(c_n\) of
the cube-root term tends to zero; percentile and symmetric intervals overcover.
The design below reports how much of this undercoverage remains at \(K=3\),
where \(c_n\) does tend to zero.
\citet{cattaneo2020bootstrap} trace the inconsistency of the non-parametric
bootstrap for i.i.d.\ data \citep{abrevaya2005bootstrap} to the same
feature, which the reshaping in \eqref{eqn--m-est-boot-def} removes; the flat
regions of \(\hat M_n\), of width \(O(|\calI_n|^{-1})\), are not the cause.
\(\lozenge\)
\end{remark}

\Cref{sec--two-scale,sec--simulations} concern \(K=2\), where no cluster-level
averaging smooths the top-level component and the cube-root term enters at
relative order \(n^{-1/6}\).
For \(K\) ways and a balanced array that term is \(\Op(n^{-K/3})\), since
\(|\calI_n|^{1/3}=n^{K/3}\), so its order relative to the Gaussian term
\(\Op(n^{-1/2})\) is \(n^{1/2-K/3}\): \(n^{-1/6}\) at \(K=2\) but
\(n^{-1/2}\) at \(K=3\).
Under non-degeneracy the cube-root term should therefore be negligible once
\(K\ge3\); \Cref{thm--m-est-non-smooth,thm--bootstrap,cor--boot-ci} should
describe the estimator and the modified bootstrap; and the plain bootstrap,
whose conditional limit \(N(0,V)\) holds for every \(K\) and whose defect is
specific to the cube-root scale, should become usable as well.
The design below examines these predictions numerically.
The scale \(\sigma_n\) of \Cref{sec--two-scale} and the results proved there are
for \(K=2\); the quantities \(c_n\) and \(\omega_n\) used below extend its
formulae to \(K\) ways and serve here only to order the designs.

\paragraph{Design.}
Let \(K=3\), \(d=2\) and \(N_{n,1}=N_{n,2}=N_{n,3}=n\), so that
\(|\calI_n|=n^3\), so the cube-root rate is \(n\).
For \(\i=(i_1,i_2,i_3)\in\calI_n\),
\[
  X_{\i}=\alpha_{i_1}+\gamma_{i_2}+\delta_{i_3}+e_{\i}\in\R^2,\qquad
  \epsilon_{\i}=u_{i_1}+v_{i_2}+w_{i_3}+\nu_{\i},\qquad
  Y_{\i}=2\ind\{X_{\i}'\beta_0\ge\epsilon_{\i}\}-1,
\]
with \(\beta_0=(\cos\theta_0,\sin\theta_0)'\), \(\theta_0=\pi/4\), and
\(\theta\) estimated in the angle chart, exactly as in
\Cref{sec--simulations}.
All components are independent centred Gaussians:
\(\alpha,\gamma,\delta,e\) are isotropic bivariate with standard deviations
\(s_{x1},s_{x2},s_{x3},s_{ex}\), and \(u,v,w,\nu\) are scalar with standard
deviations \(s_{u1},s_{u2},s_{u3},s_\nu\).
Write \(\sigma_X^2=\sum_{k=1}^3s_{xk}^2+s_{ex}^2\),
\(\sigma_\epsilon^2=\sum_{k=1}^3s_{uk}^2+s_\nu^2\) and
\(\bar s_{xk}^2=\sigma_X^2-s_{xk}^2\), the variance of each coordinate of
\(X\) within the units of way \(k\).
Then \(A_0=(\sigma_X/\sigma_\epsilon)/\pi\), \(\Lambda(1)=1/\pi\) and
\(V=\Sigma/A_0^2\) with \(\Sigma=\sum_{k=1}^3\Sigma_{n,k}\),
\[
  \Sigma_{n,k}=\frac{\arcsin(s_{uk}^2/\sigma_\epsilon^2)}{\pi^2}\,
  \frac{s_{xk}^2}{\bar s_{xk}(\bar s_{xk}^2+2s_{xk}^2)^{1/2}} ,
\]
the three-way form of the expression in \Cref{sec--simulations}: all
components are jointly Gaussian and independent, so only the totals
\(\bar s_{xk}\) and \(\sigma_\epsilon\) enter.
We report two non-degenerate designs, in both of which every way has a
non-degenerate \(X\)-shock and a non-degenerate \(\epsilon\)-shock, so that
\(\Sigma>0\) by the condition in \Cref{sec--simulations}: T1, with
every standard deviation equal to one, and TS, with
\((s_{x1},s_{x2},s_{x3},s_{ex})=(s_{u1},s_{u2},s_{u3},s_\nu)
=(1,0.5,0.5,0.25)\), whose small cell shocks raise each way's share of the
score and hence \(V\).
A degenerate control T0, with every standard deviation equal to one
except \(s_{u1}=s_{u2}=s_{u3}=0\), has \(\Sigma=0\).

Extending the scale \(\sigma_n^2=V/n+n^{-2K/3}\) of \Cref{sec--two-scale} to
\(K\) ways and a balanced array gives
\[
  c_n=(n^{K/3}\sigma_n)^{-1}=(1+n^{2K/3-1}V)^{-1/2},\qquad
  \omega_n=\frac{\var(\calZ)n^{-2K/3}}{V/n+\var(\calZ)n^{-2K/3}}
  =\frac{\var(\calZ)n^{1-2K/3}}{V+\var(\calZ)n^{1-2K/3}} ,
\]
which at \(K=2\) are the quantities of \Cref{sec--simulations} and at
\(K=3\) are \(c_n=(1+nV)^{-1/2}\) and
\(\omega_n=\var(\calZ)n^{-1}/\{V+\var(\calZ)n^{-1}\}\).
In T1, \(A_0=1/\pi=0.318\), \(\Sigma=0.0198\), \(V=0.196\)
(\(\sqrt V=0.442\)) and \(\var(\calZ)^{1/2}=1.194\); in TS,
\(\Sigma=0.0639\) and \(V=0.631\) (\(\sqrt V=0.794\)), with the same \(A_0\)
and \(\var(\calZ)\).
Over the grid below, \(\omega_n\) falls from 0.38 to 0.05 in T1 and
from 0.16 to 0.02 in TS, against 0.67 to 0.45 for the two-way
D1 of \Cref{tab--designs} over \(n=25,\dots,400\): at \(K=3\) the
cube-root term is a minor part of the variance even at \(n=12\), and a
negligible one by \(n=128\).

\paragraph{Procedures.}
We use the multiplier weights of \Cref{asm--boot-wgt}, with
\(\xi_{k,i}\sim\mathrm{Uniform}[-\sqrt3,\sqrt3]\) drawn independently for
the three ways, so that
\(\xi_{n,\i}=1+\sum_{k=1}^3(\xi_{k,i_k}-\bar\xi_{n,k})\) and the
\(\kappa\) of \Cref{rem--plain-bootstrap}, the limit of
\(\var(\xi_{n,\i})\), is \(3\): with multiplier weights \(\kappa=K\), so
the value \(2\) quoted there is the one for \(K=2\).
We compare: \emph{Main}, the modified bootstrap \eqref{eqn--m-est-boot-def}
with \(\tilde A^{\rm 4pt}_n\) of \eqref{eqn--A-4pt} at the step
\(0.79\,|\calI_n|^{-1/7}\) of \Cref{sec--max-score} (\(\varsigma=1\) in both
designs), doubled while \(\tilde A^{\rm 4pt}_n\le0\); \emph{Oracle Hessian},
the same with \(\tilde A_n=A_0\); \emph{Plain},
\eqref{eqn--m-est-boot-def-basic}; and \emph{Normal},
\(\hat\theta_n\pm z_{1-\alpha/2}(V/n)^{1/2}\) with the true \(V\).
Symmetric and basic intervals are reported for each bootstrap.
We take \(n\in\{12,20,32,48,64,128\}\) (1{,}728 to 2{,}097{,}152 cells),
\(R=2{,}000\) replications (\(R=1{,}000\) at \(n=128\)), \(B=999\) draws and
\(\alpha=0.05\); the Monte Carlo standard error of a coverage near 95\% is
0.49 points, or 0.69 at \(n=128\).
The estimator and every draw are computed exactly, as in
\Cref{sec--simulations}, now by sorting \(2n^3\) arc endpoints;
\(n=128\), with \(2{,}097{,}152\) cells in each of \(1{,}000\) replications
and \(999\) draws apiece, is the largest cell we ran; it took about two hours
per design on four cores.

\begin{table}[htbp]
\centering\footnotesize
\setlength{\tabcolsep}{3pt}
\caption{Coverage (\%) of nominal 95\% intervals in the three-way designs.}
\begin{tabular}{lrrcccccccccc}
\hline
& & & & & & \multicolumn{2}{c}{Main, \(A_0\)} & \multicolumn{2}{c}{Main, 4pt}
& \multicolumn{2}{c}{Plain} & \\
Design & \(n\) & \(R\) & \(c_n\) & \(\omega_n\) & Rem.
& symm. & basic & symm. & basic & symm. & basic
& \(\tilde A^{\rm 4pt}_n/A_0\)\\
\hline
T1 & 12 & 2000 & 0.55 & 0.38 & 0.47 & 96.0 & 95.3 & 98.0 & 97.2 & 98.4 & 86.7 & 0.91\\
            & 20 & 2000 & 0.45 & 0.27 & 0.37 & 96.7 & 96.2 & 97.8 & 97.7 & 98.5 & 89.1 & 0.94\\
            & 32 & 2000 & 0.37 & 0.19 & 0.23 & 95.9 & 95.5 & 97.0 & 96.8 & 97.8 & 89.9 & 0.96\\
            & 48 & 2000 & 0.31 & 0.13 & 0.17 & 96.1 & 95.0 & 96.5 & 96.2 & 96.9 & 91.2 & 0.97\\
            & 64 & 2000 & 0.27 & 0.10 & 0.14 & 95.7 & 95.3 & 96.6 & 96.2 & 96.8 & 92.0 & 0.98\\
            & 128 & 1000 & 0.20 & 0.05 & 0.07 & 95.1 & 95.1 & 95.4 & 95.4 & 95.4 & 93.3 & 0.99\\
TS & 12 & 2000 & 0.34 & 0.16 & 0.45 & 93.3 & 91.6 & 95.3 & 94.0 & 97.5 & 88.0 & 0.99\\
            & 20 & 2000 & 0.27 & 0.10 & 0.30 & 95.1 & 93.0 & 95.5 & 94.2 & 96.7 & 88.4 & 1.01\\
            & 32 & 2000 & 0.22 & 0.07 & 0.21 & 93.8 & 92.3 & 94.5 & 93.2 & 95.9 & 89.2 & 1.02\\
            & 48 & 2000 & 0.18 & 0.04 & 0.16 & 94.1 & 92.7 & 95.0 & 93.3 & 95.3 & 89.9 & 1.01\\
            & 64 & 2000 & 0.16 & 0.03 & 0.11 & 93.6 & 92.3 & 94.1 & 93.1 & 95.0 & 90.9 & 1.00\\
            & 128 & 1000 & 0.11 & 0.02 & 0.07 & 94.7 & 95.0 & 95.5 & 94.3 & 96.2 & 93.2 & 1.01\\
\hline
\end{tabular}

\smallskip
\begin{minipage}{0.95\linewidth}\footnotesize\raggedright
\textit{Notes.} Multiplier weights throughout; \(B=999\).
Rem.: \(\var(\hat\theta_n-\theta_0-A_0^{-1}\bar\psi_n)/
\var(\hat\theta_n-\theta_0)\), to be compared with \(\omega_n\).
\(\tilde A^{\rm 4pt}_n/A_0\): Monte Carlo mean.  The step was doubled in
0.05\% of the replications of TS at \(n=12\) and in none elsewhere.
The oracle normal interval covers 85.2, 91.0, 91.0, 91.8, 93.0 and 93.9 in
T1 and 90.0, 92.5, 92.8, 93.8, 93.5 and 95.5 in TS.
Basic intervals are
\([\hat\theta_n-q^*_{0.975},\hat\theta_n-q^*_{0.025}]\), as in
\Cref{tab--sim-coverage}.
\end{minipage}
\end{table}

\paragraph{Results.}
The predictions hold, with one qualification.
In T1 the plain bootstrap's basic interval improves monotonically,
from 86.7\% at \(n=12\) to 93.3\% at \(n=128\), as \(c_n\) falls from 0.55 to
0.20, and the overcoverage of its symmetric interval disappears (98.4\% to
95.4\%).
This is the behaviour \Cref{rem--plain-bootstrap} leads one to expect: the
plain bootstrap is first-order valid for every \(K\), and what it gets wrong
is the cube-root part of the law, which the symmetric interval does not need
and the basic one does.
In this design both modified bootstraps cover at least nominally at every
\(n\) and reach nominal coverage at \(n=128\): with the true \(A_0\) they cover
95.0--96.7\% over the grid and 95.1\% at \(n=128\), and with
\(\tilde A^{\rm 4pt}_n\) 95.4--98.0\% and 95.4\% at \(n=128\), which is
0.4 points above nominal against a standard error of 0.69.
The feasible version thus comes down from 98.0/97.2 at \(n=12\) to the
oracle one's performance.
The oracle normal interval, by contrast, still covers only 93.9\% at
\(n=128\): the reverse of the ordering in \Cref{tab--sim-coverage}, where the
modified bootstrap was conservative and the normal interval the weaker
procedure, and a reminder that a vanishing \(\omega_n\) is not the same as a
negligible one.
The Monte Carlo standard deviation of \(\hat\theta_n\) agrees with the
extended scale: \(n^{1/2}\) times it is 0.608 at \(n=12\) and 0.454 at
\(n=128\), against \((V+\var(\calZ)/n)^{1/2}=0.561\) and \(0.455\), so the
ratio falls from 1.08 to 1.00.

The qualification is the first-order remainder.
Its share of the variance exceeds \(\omega_n\) by 0.02--0.10 in T1
but by up to 0.29 in TS, where the cell shocks are small --- the
pattern of DS, DM and DL in \Cref{sec--simulations}, and for the same reason: the factor multiplying
\(\delta\) in
\Cref{asm--two-scale}\eqref{asm--two-scale-modulus} has second moment of order
\(1/s_{ex}\), so the interaction in the top-level component need not be
negligible unless \(ns_{ex}\) is large.
In TS the modified bootstrap with the true \(A_0\) then undercovers
slightly at \(n\le64\) (93.3--95.1\% symmetric, 91.6--93.0\% basic), while
the feasible version covers 94.1--95.5\% there because
\(\tilde A^{\rm 4pt}_n\) is close to \(A_0\) and the step is unchanged.
At \(n=128\), where the remainder share has fallen to 0.07, that gap closes:
the two modified bootstraps cover 94.7/95.0\% and 95.5/94.3\%, the plain
basic interval 93.2\%, and the oracle normal interval 95.5\%, the last
because \(V\) is larger in TS than in T1, so the Gaussian term dominates
sooner (\(\omega_{128}=0.02\) against 0.05).

Two further comparisons locate the mechanism.
In the degenerate control T0, where \(\Sigma=0\) and \(c_n\equiv1\),
the plain bootstrap's basic interval does not improve at all --- 86.8, 85.2,
83.8, 85.9 and 85.6\% at \(n=12,\dots,64\) --- while the modified bootstrap
over-covers, at 99.2--99.8\% with the true \(A_0\) and 99.7--100.0\% with
\(\tilde A^{\rm 4pt}_n\), as in the degenerate two-way design D3.
And in the two-way D1, where \(c_n\) decays at rate \(n^{-1/6}\), the
plain basic interval moves only from 80.8 to 83.0\% over \(n=25,\dots,400\)
(\Cref{tab--sim-coverage}).
What rescues the plain bootstrap is thus the vanishing of the cube-root
component, not the number of ways as such; \(K=3\) matters because it makes
\(c_n\) decay at \(n^{-1/2}\) instead of \(n^{-1/6}\).
Even so, the recovery is slow: in T1 the plain basic interval is
still 1.7 points short at \(n=128\) and gains about 1.3 points per doubling
of \(n\), so nominal coverage would need roughly \(n=512\), with 64 times the
cells.
For \(K\ge3\) and a non-degenerate design we therefore still recommend the
modified bootstrap, or the symmetric interval if the plain bootstrap must be
used.

\IfFileExists{ecta.bst}{\bibliographystyle{ecta}}{\bibliographystyle{plainnat}}
\bibliography{main}

\end{document}